\documentclass[aps,pra,reprint,superscriptaddress,longbibliography,nofootinbib,floatfix,10pt]{revtex4-2}
\usepackage[T1]{fontenc}
\usepackage[utf8]{inputenc}
\usepackage{lmodern,amsmath,amssymb,amsthm,mathtools,bm,graphicx,booktabs,longtable,tikz,placeins,needspace,microtype}
\usepackage[dvipsnames]{xcolor}
\usetikzlibrary{calc,arrows.meta,positioning,fit,backgrounds}
\usepackage[colorlinks=true,linkcolor=blue,citecolor=blue,urlcolor=blue,breaklinks=true]{hyperref}
\usepackage{orcidlink}
\newtheorem{theorem}{Theorem}
\newtheorem{lemma}[theorem]{Lemma}
\newtheorem{proposition}[theorem]{Proposition}
\newtheorem{corollary}[theorem]{Corollary}
\theoremstyle{definition}\newtheorem{definition}[theorem]{Definition}
\theoremstyle{remark}\newtheorem{remark}[theorem]{Remark}
\newcommand{\C}{\mathbb C}\newcommand{\R}{\mathbb R}\newcommand{\Q}{\mathbb Q}
\newcommand{\F}{\mathbb F}\newcommand{\Z}{\mathbb Z}\newcommand{\PP}{\mathbb P}
\newcommand{\Sc}{\mathcal S}\newcommand{\Cc}{\mathcal C}
\newcommand{\YO}{\mathrm{YO}}\newcommand{\ip}[2]{\langle #1,#2\rangle}
\DeclareMathOperator{\spanop}{span}
\DeclareMathOperator{\diag}{diag}\DeclareMathOperator{\Tr}{Tr}
\DeclareMathOperator{\supp}{supp}\DeclareMathOperator{\conv}{conv}
\newif\ifDistinctContexts
\DistinctContextsfalse

\definecolor{YOtriad}{HTML}{246A9A}
\definecolor{YObasis}{HTML}{8A939B}
\definecolor{YOspoke}{HTML}{237D70}
\definecolor{YOouterA}{HTML}{85458B}
\definecolor{YOouterB}{HTML}{B47620}
\definecolor{YOdiagonal}{HTML}{B44850}
\colorlet{YOctx1}{YOtriad}
\colorlet{YOctx2}{YOtriad}
\colorlet{YOctx3}{YOtriad}
\colorlet{YOctx4}{YObasis}
\colorlet{YOctx5}{YOspoke}
\colorlet{YOctx6}{YOspoke}
\colorlet{YOctx7}{YOspoke}
\colorlet{YOctx8}{YOouterA}
\colorlet{YOctx11}{YOouterA}
\colorlet{YOctx16}{YOouterA}
\colorlet{YOctx9}{YOouterB}
\colorlet{YOctx13}{YOouterB}
\colorlet{YOctx14}{YOouterB}
\colorlet{YOctx10}{YOdiagonal}
\colorlet{YOctx12}{YOdiagonal}
\colorlet{YOctx15}{YOdiagonal}
\ifDistinctContexts
  \definecolor{YOctx1}{HTML}{BF1FB4}
  \definecolor{YOctx2}{HTML}{B4000A}
  \definecolor{YOctx3}{HTML}{C99400}
  \definecolor{YOctx4}{HTML}{555F00}
  \definecolor{YOctx5}{HTML}{559F2A}
  \definecolor{YOctx6}{HTML}{D48A6A}
  \definecolor{YOctx7}{HTML}{F47400}
  \definecolor{YOctx8}{HTML}{B4944A}
  \definecolor{YOctx9}{HTML}{DF74BF}
  \definecolor{YOctx10}{HTML}{744ABF}
  \definecolor{YOctx11}{HTML}{35552A}
  \definecolor{YOctx12}{HTML}{C955FF}
  \definecolor{YOctx13}{HTML}{00B49F}
  \definecolor{YOctx14}{HTML}{6A558A}
  \definecolor{YOctx15}{HTML}{743F00}
  \definecolor{YOctx16}{HTML}{FF0035}
\fi
\providecommand{\CSf}{1}
\providecommand{\CSg}{\CSf}
\providecommand{\CSn}{\CSf}
\providecommand{\CSy}{1}
\newcommand{\CSscaledfont}[1]{\pgfmathsetmacro\CSfsz{#1*\CSf}\fontsize{\CSfsz}{\CSfsz}\selectfont}
\newcommand{\CSnotefont}[1]{\pgfmathsetmacro\CSfsz{#1*\CSn}\fontsize{\CSfsz}{\CSfsz}\selectfont}
\tikzset{
  CSgrow/.style={outer sep={(\CSf-1)*0.18cm+(\CSg-1)*0.27cm}},
  YOtrace/.style={line width=1.85pt,line cap=round,line join=round},
  YOapex/.style={line width=.68pt,line cap=round,draw opacity=.65},
  YOlabel/.style={font=\CSnotefont{9},inner sep=1.2pt,fill=white,CSgrow},
  YOnumber/.style={font=\sffamily\bfseries\CSscaledfont{8},inner sep=0pt},
}
\newcommand{\YOline}[4]{\draw[YOtrace,YOctx#1] (v#2)--(v#3)--(v#4);}
\newcommand{\YOapexlink}[6]{%
  \draw[YOapex,draw=YOctx#1] (#3)
    -- ($(#3)!-#4!(#2)$)
    .. controls ($(#3)!-#5!(#2)$) and #6 .. (v26);%
}
\newcommand{\YOapexvia}[9]{%
  \draw[YOapex,draw=YOctx#1] (#3)
    -- ($(#3)!-#4!(#2)$)
    .. controls ($(#3)!-#5!(#2)$) and #6 .. #7
    .. controls #8 and #9 .. (v26);%
}
\newcommand{\YOvertex}[2]{%
  \foreach \context/\radius in {#2}{%
    \fill[YOctx\context] (v#1) circle[radius={\CSf*.18cm+\CSg*(\radius-.18)*1cm}];%
  }%
  \fill[white] (v#1) circle[radius={\CSf*.18cm}];%
  \node[YOnumber] at (v#1) {#1};%
}

\definecolor{ucA01}{HTML}{1C638E}
\definecolor{ucA02}{HTML}{257D72}
\definecolor{ucA03}{HTML}{BC681D}
\definecolor{ucA04}{HTML}{8B3D88}
\definecolor{ucA05}{HTML}{82623E}
\definecolor{ucA06}{HTML}{0000FF}
\definecolor{ucA07}{HTML}{00B400}
\definecolor{ucA08}{HTML}{FF15F4}
\definecolor{ucA09}{HTML}{FF005F}
\definecolor{ucA10}{HTML}{006AFF}
\definecolor{ucA11}{HTML}{8A9400}
\definecolor{ucA12}{HTML}{FF3500}
\definecolor{ucA13}{HTML}{E97F94}
\definecolor{ucA14}{HTML}{005F0A}
\definecolor{ucA15}{HTML}{FF4AB4}
\definecolor{ucA16}{HTML}{740AB4}
\definecolor{ucA17}{HTML}{6A9FFF}
\definecolor{ucA18}{HTML}{941F2A}
\definecolor{ucA19}{HTML}{C97FFF}
\definecolor{ucA20}{HTML}{00B474}
\definecolor{ucA21}{HTML}{7F9F5F}
\definecolor{ucA22}{HTML}{B494C9}
\definecolor{ucA23}{HTML}{2A4A9F}
\definecolor{ucA24}{HTML}{FF6A55}
\definecolor{sgink}{HTML}{3D454B}
\definecolor{sggray}{HTML}{858B90}
\tikzset{
 sg/context/.style={line width=1.15pt,line cap=round,line join=round,draw=sggray},
 sg/base/.style={line width=1.7pt,line cap=round,line join=round},
 sg/label/.style={font=\sffamily\CSnotefont{9},text=sgink,fill=white,inner sep=1pt},
 sg/number/.style={inner sep=0pt,outer sep=0pt,text=black,
   font=\sffamily\CSscaledfont{8}},
}
\newcommand{\HesseFCvertex}[2]{%
  \foreach \diskcolor/\diskradius in {#2}{%
    \fill[\diskcolor] (v#1) circle[radius=\diskradius cm];%
  }%
  \fill[white] (v#1) circle[radius=.160cm];%
  \node[sg/number] at (v#1) {#1};%
}

\providecommand{\CScontext}[4]{%
  \ifnum#1=4\relax
    \pgfmathsetmacro{\YObasisR}{\YOR*cos(20)}%
    \foreach \rotation in {0,120,240}{%
      \begin{scope}[rotate=\rotation]
        \draw[YOtrace,YOctx4]
          (0,\YObasisR)
          .. controls ({-.5*\YObasisR},\YObasisR)
          and ({-(sqrt(3)/2+.25)*\YObasisR},{(-.5+.25*sqrt(3))*\YObasisR})
          .. ({-sqrt(3)/2*\YObasisR},{-.5*\YObasisR});
      \end{scope}%
    }%
  \else
    \YOline{#1}{#2}{#3}{#4}%
  \fi
}
\newcommand{\CSbasecontexts}{%
  \CScontext{1}{3}{2}{1}
  \CScontext{2}{9}{8}{7}
  \CScontext{3}{4}{5}{6}
  \CScontext{4}{2}{5}{8}
  \CScontext{8}{11}{14}{3}
  \CScontext{9}{7}{15}{11}
  \CScontext{10}{4}{24}{11}
  \CScontext{11}{9}{23}{10}
  \CScontext{12}{1}{19}{10}
  \CScontext{13}{10}{18}{4}
}
\newcommand{\CSaddedcontexts}{%
  \CScontext{14}{1}{20}{13}
  \CScontext{15}{7}{21}{13}
  \CScontext{16}{6}{22}{13}
}
\newcommand{\CSHseventeenvertices}{%
  \YOvertex{1}{1/.36,12/.27}
  \YOvertex{2}{1/.36,4/.27}
  \YOvertex{3}{1/.36,8/.27}
  \YOvertex{4}{3/.45,10/.36,13/.27}
  \YOvertex{5}{3/.36,4/.27}
  \YOvertex{6}{3/.27}
  \YOvertex{7}{2/.36,9/.27}
  \YOvertex{8}{2/.36,4/.27}
  \YOvertex{9}{2/.36,11/.27}
  \YOvertex{10}{11/.45,13/.36,12/.27}
  \YOvertex{11}{8/.45,9/.36,10/.27}
  \YOvertex{14}{8/.27}
  \YOvertex{15}{9/.27}
  \YOvertex{18}{13/.27}
  \YOvertex{19}{12/.27}
  \YOvertex{23}{11/.27}
  \YOvertex{24}{10/.27}
}
\newcommand{\CSHtwentyonevertices}{%
  \YOvertex{1}{1/.45,12/.36,14/.27}
  \YOvertex{2}{1/.36,4/.27}
  \YOvertex{3}{1/.36,8/.27}
  \YOvertex{4}{3/.45,10/.36,13/.27}
  \YOvertex{5}{3/.36,4/.27}
  \YOvertex{6}{3/.36,16/.27}
  \YOvertex{7}{2/.45,15/.36,9/.27}
  \YOvertex{8}{2/.36,4/.27}
  \YOvertex{9}{2/.36,11/.27}
  \YOvertex{10}{11/.45,13/.36,12/.27}
  \YOvertex{11}{8/.45,9/.36,10/.27}
  \YOvertex{13}{16/.45,14/.36,15/.27}
  \YOvertex{14}{8/.27}
  \YOvertex{15}{9/.27}
  \YOvertex{18}{13/.27}
  \YOvertex{19}{12/.27}
  \YOvertex{20}{14/.27}
  \YOvertex{21}{15/.27}
  \YOvertex{22}{16/.27}
  \YOvertex{23}{11/.27}
  \YOvertex{24}{10/.27}
}
\newcommand{\CSdrawcontextdiagram}[1]{%
  \begin{tikzpicture}[x=1cm,y=1cm]
    \let\CScontext\YOline
    \def\Rbug{4.4}%
    \coordinate (v10) at (0:\Rbug);
    \coordinate (v1) at (60:\Rbug);
    \coordinate (v3) at (120:\Rbug);
    \coordinate (v11) at (180:\Rbug);
    \coordinate (v7) at (240:\Rbug);
    \coordinate (v9) at (300:\Rbug);
    \coordinate (v2) at ($(v1)!.5!(v3)$);
    \coordinate (v8) at ($(v7)!.5!(v9)$);
    \coordinate (v5) at (0,0);
    \foreach \vertex/\first/\last in {14/3/11,15/11/7,23/9/10,19/10/1}{%
      \coordinate (v\vertex) at ($(v\first)!.5!(v\last)$);
    }
    \coordinate (v4) at (-2.30,-.90);
    \ifnum#1=17\relax
      \coordinate (v6) at ($(v4)!.60!(v5)$);
    \else
      \coordinate (v6) at (1.38,.54);
    \fi
    \coordinate (v24) at ($(v11)!.52!(v4)$);
    \coordinate (v18) at ($(v10)!.42!(v4)$);
    \coordinate (v13) at (0,6.0);
    \coordinate (v20) at ($(v13)!.45!(v1)$);
    \coordinate (v22) at ($(v13)!.50!(v6)$);
    \coordinate (v21) at ($(v13)!.27!(v7)$);
    \CSbasecontexts
    \ifnum#1=21\relax
      \CSaddedcontexts
      \CSHtwentyonevertices
    \else
      \CSHseventeenvertices
    \fi
    \node[YOlabel,left=4pt] at ($(v11)+(-.45,0)$) {$h_1$};
    \node[YOlabel,right=4pt] at ($(v10)+(.45,0)$) {$h_2$};
    \node[YOlabel,above=4pt] at ($(v3)+(0,.34)$) {$y_1^-$};
    \node[YOlabel,above=4pt] at ($(v1)+(0,.34)$) {$y_1^+$};
    \node[YOlabel,below=4pt] at ($(v7)+(0,-.34)$) {$y_2^+$};
    \node[YOlabel,below=4pt] at ($(v9)+(0,-.34)$) {$y_2^-$};
    \ifnum#1=17\relax
      \node[YOlabel,above=3pt] at ($(v4)+(0,.55)$) {$y_3^+$};
      \node[YOlabel,above=3pt] at ($(v6)+(0,.32)$) {$y_3^-$};
      \node[YOlabel,above=3pt] at ($(v2)+(0,.45)$) {$z_1$};
      \node[YOlabel,below=3pt] at ($(v8)+(0,-.45)$) {$z_2$};
      \node[YOlabel,right=3pt] at ($(v5)+(.40,0)$) {$z_3$};
    \else
      \node[YOlabel,below=3pt] at ($(v4)+(0,-.55)$) {$y_3^+$};
      \node[YOlabel,right=3pt] at ($(v6)+(.28,.10)$) {$y_3^-$};
      \node[YOlabel,above left=1pt] at ($(v2)+(-.42,.30)$) {$z_1$};
      \node[YOlabel,below=3pt] at ($(v8)+(0,-.45)$) {$z_2$};
      \node[YOlabel,above left=1pt] at ($(v5)+(-.30,.22)$) {$z_3$};
    \fi
    \ifnum#1=21\relax
      \node[YOlabel,above=4pt] at ($(v13)+(0,.40)$) {$h_3$};
    \fi
    \pgfresetboundingbox
    \path[use as bounding box] (-6.52,-4.90) rectangle (6.52,7.0);
  \end{tikzpicture}%
}

\tikzset{
  CSexcl/.style={draw=YOdiagonal,line width=1.6pt,densely dashed,{Latex[length=2.4mm]}-{Latex[length=2.4mm]}},
  CSident/.style={draw=YOtriad,line width=1.6pt,densely dotted,{Latex[length=2.4mm]}-{Latex[length=2.4mm]}},
  CSnote/.style={font=\sffamily\CSnotefont{8},fill=white,inner sep=1.5pt,rounded corners=1pt},
  CSvec/.style={font=\CSnotefont{7},text=black!75,inner sep=1pt,CSgrow},
}

\newcommand{\CSdrawPfive}{%
  \begin{tikzpicture}[x=1cm,y=1cm]
    \def\R{4.4}%
    \foreach \v/\ang in {1/90,2/162,3/234,4/306,5/18}{%
      \coordinate (v\v) at (\ang:\R);
    }
    \coordinate (v6) at ($(v1)!.5!(v2)$);
    \coordinate (v7) at ($(v2)!.5!(v3)$);
    \coordinate (v8) at ($(v3)!.5!(v4)$);
    \coordinate (v9) at ($(v4)!.5!(v5)$);
    \coordinate (v10) at ($(v5)!.5!(v1)$);
    \draw[YOtrace,YOctx1]  (v1)--(v6)--(v2);
    \draw[YOtrace,YOctx9]  (v2)--(v7)--(v3);
    \draw[YOtrace,YOctx10] (v3)--(v8)--(v4);
    \draw[YOtrace,YOctx5]  (v4)--(v9)--(v5);
    \draw[YOtrace,YOctx8]  (v5)--(v10)--(v1);
    \draw[CSexcl,-] (v6) -- (v7) -- (v9) -- cycle;
    \node[CSnote,text=YOdiagonal] at (0,-0.25) {$\neg\bigl(c(6){=}c(7){=}c(9)\bigr)$};
    \YOvertex{1}{1/.36,8/.27}
    \YOvertex{2}{1/.36,9/.27}
    \YOvertex{3}{9/.36,10/.27}
    \YOvertex{4}{10/.36,5/.27}
    \YOvertex{5}{5/.36,8/.27}
    \YOvertex{6}{1/.27}
    \YOvertex{7}{9/.27}
    \YOvertex{8}{10/.27}
    \YOvertex{9}{5/.27}
    \YOvertex{10}{8/.27}
    \node[CSvec,above=3pt] at ($(v1)+(0,.38)$) {$(1,0,0)$};
    \node[CSvec,left=3pt] at ($(v2)+(-.38,0)$) {$(0,1,1)$};
    \node[CSvec,below left=2pt] at ($(v3)+(-.28,-.28)$) {$(1,1,-1)$};
    \node[CSvec,below right=2pt] at ($(v4)+(.28,-.28)$) {$(1,-2,-1)$};
    \node[CSvec,right=3pt] at ($(v5)+(.38,0)$) {$(0,1,-2)$};
    \node[CSvec,above left=2pt] at ($(v6)+(-.22,.22)$) {$(0,1,-1)$};
    \node[CSvec,left=3pt] at ($(v7)+(-.30,0)$) {$(2,-1,1)$};
    \node[CSvec,below=3pt] at ($(v8)+(0,-.30)$) {$(1,0,1)$};
    \node[CSvec,right=3pt] at ($(v9)+(.30,0)$) {$(5,2,1)$};
    \node[CSvec,above right=2pt] at ($(v10)+(.22,.22)$) {$(0,2,1)$};
    \pgfresetboundingbox
    \path[use as bounding box] (-6.3,-4.8) rectangle (6.3,5.4);
  \end{tikzpicture}%
}

\newcommand{\CSdrawSpeckerBug}{%
  \begin{tikzpicture}[x=1cm,y=1cm]
    \def\R{4.3}%
    \coordinate (v10) at (  0:\R);
    \coordinate (v1)  at ( 60:\R);
    \coordinate (v3)  at (120:\R);
    \coordinate (v11) at (180:\R);
    \coordinate (v7)  at (240:\R);
    \coordinate (v9)  at (300:\R);
    \coordinate (v2)  at ($(v3)!.5!(v1)$);
    \coordinate (v8)  at ($(v7)!.5!(v9)$);
    \coordinate (v5)  at (0,0);
    \coordinate (v14) at ($(v11)!.5!(v3)$);
    \coordinate (v15) at ($(v11)!.5!(v7)$);
    \coordinate (v19) at ($(v1)!.5!(v10)$);
    \coordinate (v23) at ($(v9)!.5!(v10)$);
    \draw[YOtrace,YOctx1]  (v3)--(v2)--(v1);
    \draw[YOtrace,YOctx2]  (v7)--(v8)--(v9);
    \draw[YOtrace,YOctx4]  (v2)--(v5)--(v8);
    \draw[YOtrace,YOctx8]  (v11)--(v14)--(v3);
    \draw[YOtrace,YOctx9]  (v11)--(v15)--(v7);
    \draw[YOtrace,YOctx12] (v1)--(v19)--(v10);
    \draw[YOtrace,YOctx11] (v9)--(v23)--(v10);
    \draw[CSexcl] (v11) to[bend right=18] node[CSnote,pos=.50] {$c(10)\ne c(11)$} (v10);
    \YOvertex{1}{1/.36,12/.27}
    \YOvertex{2}{1/.36,4/.27}
    \YOvertex{3}{1/.36,8/.27}
    \YOvertex{5}{4/.27}
    \YOvertex{7}{2/.36,9/.27}
    \YOvertex{8}{2/.36,4/.27}
    \YOvertex{9}{2/.36,11/.27}
    \YOvertex{10}{11/.36,12/.27}
    \YOvertex{11}{8/.36,9/.27}
    \YOvertex{14}{8/.27}
    \YOvertex{15}{9/.27}
    \YOvertex{19}{12/.27}
    \YOvertex{23}{11/.27}
    \node[YOlabel,left=4pt]  at ($(v11)+(-.36,0)$) {$h_1\,(A)$};
    \node[YOlabel,right=4pt] at ($(v10)+(.36,0)$)  {$h_2\,(B)$};
    \node[YOlabel,above=4pt] at ($(v3)+(0,.34)$)   {$y_1^-$};
    \node[YOlabel,above=4pt] at ($(v2)+(0,.34)$)   {$z_1$};
    \node[YOlabel,above=4pt] at ($(v1)+(0,.34)$)   {$y_1^+$};
    \node[YOlabel,below=4pt] at ($(v7)+(0,-.34)$)  {$y_2^+$};
    \node[YOlabel,below=4pt] at ($(v8)+(0,-.34)$)  {$z_2$};
    \node[YOlabel,below=4pt] at ($(v9)+(0,-.34)$)  {$y_2^-$};
    \node[YOlabel,right=4pt] at ($(v5)+(.28,0)$)   {$z_3$};
    \pgfresetboundingbox
    \path[use as bounding box] (-6.52,-5.2) rectangle (6.52,5.6);
  \end{tikzpicture}%
}

\newcommand{\CSdrawPureChromTIFS}{%
  \begin{tikzpicture}[x=1cm,y=1cm]
    \coordinate (v11) at (-4.20, 0.00);
    \coordinate (v10) at ( 4.20, 0.00);
    \coordinate (v3)  at (-2.25, 2.60);
    \coordinate (v1)  at (-1.12, 2.60);
    \coordinate (v2)  at ( 0.00, 2.60);
    \coordinate (v4)  at ( 2.10, 2.60);
    \coordinate (v7)  at (-2.10,-2.60);
    \coordinate (v8)  at ( 0.00,-2.60);
    \coordinate (v9)  at ( 2.10,-2.60);
    \coordinate (v5)  at ( 0.00,-0.70);
    \coordinate (v6)  at ($(v5)!.5!(v4)$);
    \coordinate (v14) at ($(v11)!.5!(v3)$);
    \coordinate (v15) at ($(v11)!.5!(v7)$);
    \coordinate (v18) at ($(v10)!.5!(v4)$);
    \coordinate (v23) at ($(v9)!.5!(v10)$);
    \coordinate (v24) at ($(v11)!.44!(v4)$);
    \draw[YOtrace,YOctx1]  (v3)--(v1)--(v2);
    \draw[YOtrace,YOctx2]  (v7)--(v8)--(v9);
    \draw[YOtrace,YOctx3]  (v5)--(v6)--(v4);
    \draw[YOtrace,YOctx4]  (v2)--(v5)--(v8);
    \draw[YOtrace,YOctx8]  (v11)--(v14)--(v3);
    \draw[YOtrace,YOctx9]  (v11)--(v15)--(v7);
    \draw[YOtrace,YOctx11] (v9)--(v23)--(v10);
    \draw[YOtrace,YOctx13] (v10)--(v18)--(v4);
    \draw[YOtrace,YOctx10] (v11)--(v24)--(v4);
    \draw[CSexcl] (v14) to[bend right=36] node[CSnote,pos=.50] {$c(14)\ne c(18)$} (v18);
    \YOvertex{1}{1/.27}
    \YOvertex{2}{1/.36,4/.27}
    \YOvertex{3}{1/.36,8/.27}
    \YOvertex{4}{3/.45,10/.36,13/.27}
    \YOvertex{5}{3/.36,4/.27}
    \YOvertex{6}{3/.27}
    \YOvertex{7}{2/.36,9/.27}
    \YOvertex{8}{2/.36,4/.27}
    \YOvertex{9}{2/.36,11/.27}
    \YOvertex{10}{11/.36,13/.27}
    \YOvertex{11}{8/.45,9/.36,10/.27}
    \YOvertex{14}{8/.27}
    \YOvertex{15}{9/.27}
    \YOvertex{18}{13/.27}
    \YOvertex{23}{11/.27}
    \YOvertex{24}{10/.27}
    \node[YOlabel,left=4pt]  at ($(v11)+(-.45,0)$) {$h_1$};
    \node[YOlabel,right=4pt] at ($(v10)+(.36,0)$)  {$h_2$};
    \node[YOlabel,above=4pt] at ($(v4)+(.22,.36)$) {$y_3^+$};
    \node[YOlabel,right=3pt] at ($(v6)+(.26,.02)$) {$y_3^-$};
    \pgfresetboundingbox
    \path[use as bounding box] (-6.52,-5.2) rectangle (6.52,5.6);
  \end{tikzpicture}%
}

\newcommand{\CSdrawCompletedKSTITS}{%
  \begin{tikzpicture}[x=1cm,y=\CSy cm]
    \coordinate (v11) at (-4.20, 0.00);
    \coordinate (v10) at ( 4.20, 0.00);
    \coordinate (v3)  at (-1.60, 1.55);
    \coordinate (v2)  at ( 0.00, 1.55);
    \coordinate (v1)  at ( 1.60, 1.55);
    \coordinate (v7)  at (-2.10,-2.60);
    \coordinate (v8)  at ( 0.00,-2.60);
    \coordinate (v9)  at ( 2.10,-2.60);
    \coordinate (v5)  at ( 0.00,-0.52);
    \coordinate (v14) at ($(v11)!.5!(v3)$);
    \coordinate (v15) at ($(v11)!.5!(v7)$);
    \coordinate (v19) at ($(v1)!.5!(v10)$);
    \coordinate (v23) at ($(v9)!.5!(v10)$);
    \coordinate (v4)  at ( 0.00, 4.90);
    \coordinate (v24) at ($(v11)!.62!(v4)$);
    \coordinate (v18) at ($(v10)!.62!(v4)$);
    \draw[YOtrace,YOctx1]  (v3)--(v2)--(v1);
    \draw[YOtrace,YOctx2]  (v7)--(v8)--(v9);
    \draw[YOtrace,YOctx4]  (v2)--(v5)--(v8);
    \draw[YOtrace,YOctx8]  (v11)--(v14)--(v3);
    \draw[YOtrace,YOctx9]  (v11)--(v15)--(v7);
    \draw[YOtrace,YOctx12] (v1)--(v19)--(v10);
    \draw[YOtrace,YOctx11] (v9)--(v23)--(v10);
    \draw[YOtrace,YOctx10] (v11)--(v24)--(v4);
    \draw[YOtrace,YOctx13] (v10)--(v18)--(v4);
    \draw[CSident] (v11) to[bend right=18] (v18);
    \draw[CSident] (v10) to[bend left=18]  (v24);
    \node[CSnote,anchor=east] at (-0.42, 2.38) {$10\sim_c 24$};
    \node[CSnote,anchor=west] at ( 0.42, 2.38) {$11\sim_c 18$};
    \YOvertex{1}{1/.36,12/.27}
    \YOvertex{2}{1/.36,4/.27}
    \YOvertex{3}{1/.36,8/.27}
    \YOvertex{4}{10/.36,13/.27}
    \YOvertex{5}{4/.27}
    \YOvertex{7}{2/.36,9/.27}
    \YOvertex{8}{2/.36,4/.27}
    \YOvertex{9}{2/.36,11/.27}
    \YOvertex{10}{11/.45,13/.36,12/.27}
    \YOvertex{11}{8/.45,9/.36,10/.27}
    \YOvertex{14}{8/.27}
    \YOvertex{15}{9/.27}
    \YOvertex{18}{13/.27}
    \YOvertex{19}{12/.27}
    \YOvertex{23}{11/.27}
    \YOvertex{24}{10/.27}
    \node[YOlabel,left=4pt]       at ($(v11)+(-.45,0)$) {$h_1\,(A)$};
    \node[YOlabel,right=4pt]      at ($(v10)+(.45,0)$)  {$h_2\,(B)$};
    \node[YOlabel,above left=2pt]  at ($(v24)+(-.20,.18)$) {$C'$};
    \node[YOlabel,above right=2pt] at ($(v18)+(.20,.18)$)  {$C$};
    \pgfresetboundingbox
    \path[use as bounding box] (-6.52,-5.2) rectangle (6.52,5.6);
  \end{tikzpicture}%
}

\newcommand{\CSdrawPureChromTITS}{%
  \begin{tikzpicture}[x=1cm,y=\CSy cm]
    \coordinate (v11) at (-4.20, 0.00);
    \coordinate (v10) at ( 4.20, 0.00);
    \coordinate (v3)  at (-1.90, 1.85);
    \coordinate (v1)  at (-0.95, 1.85);
    \coordinate (v2)  at ( 0.00, 1.85);
    \coordinate (v4)  at ( 1.60, 1.85);
    \coordinate (v7)  at (-2.10,-2.60);
    \coordinate (v8)  at ( 0.00,-2.60);
    \coordinate (v9)  at ( 2.10,-2.60);
    \coordinate (v5)  at ( 0.00,-1.00);
    \coordinate (v6)  at ($(v5)!.62!(v4)$);
    \coordinate (v14) at ($(v11)!.5!(v3)$);
    \coordinate (v15) at ($(v11)!.5!(v7)$);
    \coordinate (v18) at ($(v10)!.5!(v4)$);
    \coordinate (v23) at ($(v9)!.5!(v10)$);
    \coordinate (v24) at ($(v11)!.44!(v4)$);
    \coordinate (v26) at ( 0.00, 5.05);
    \coordinate (v27) at ($(v14)!.55!(v26)$);
    \coordinate (v28) at ($(v18)!.55!(v26)$);
    \draw[YOtrace,YOctx1]  (v3)--(v1)--(v2);
    \draw[YOtrace,YOctx2]  (v7)--(v8)--(v9);
    \draw[YOtrace,YOctx3]  (v5)--(v6)--(v4);
    \draw[YOtrace,YOctx4]  (v2)--(v5)--(v8);
    \draw[YOtrace,YOctx8]  (v11)--(v14)--(v3);
    \draw[YOtrace,YOctx9]  (v11)--(v15)--(v7);
    \draw[YOtrace,YOctx11] (v9)--(v23)--(v10);
    \draw[YOtrace,YOctx13] (v10)--(v18)--(v4);
    \draw[YOtrace,YOctx10] (v11)--(v24)--(v4);
    \draw[YOtrace,YOctx5]  (v14)--(v27)--(v26);
    \draw[YOtrace,YOctx16] (v18)--(v28)--(v26);
    \draw[CSexcl]  (v14) to[bend right=36] node[CSnote,pos=.50] {$c(14)\ne c(18)$} (v18);
    \draw[CSident] (v14) to[bend left=16]  (v28);
    \draw[CSident] (v18) to[bend right=16] (v27);
    \node[CSnote,anchor=east] at ($(v27)+(-.62,.12)$) {$18\sim_c 27$};
    \node[CSnote,anchor=west] at ($(v28)+(.62,.12)$) {$14\sim_c 28$};
    \YOvertex{1}{1/.27}
    \YOvertex{2}{1/.36,4/.27}
    \YOvertex{3}{1/.36,8/.27}
    \YOvertex{4}{3/.45,10/.36,13/.27}
    \YOvertex{5}{3/.36,4/.27}
    \YOvertex{6}{3/.27}
    \YOvertex{7}{2/.36,9/.27}
    \YOvertex{8}{2/.36,4/.27}
    \YOvertex{9}{2/.36,11/.27}
    \YOvertex{10}{11/.36,13/.27}
    \YOvertex{11}{8/.45,9/.36,10/.27}
    \YOvertex{14}{8/.36,5/.27}
    \YOvertex{15}{9/.27}
    \YOvertex{18}{13/.36,16/.27}
    \YOvertex{23}{11/.27}
    \YOvertex{24}{10/.27}
    \YOvertex{26}{5/.36,16/.27}
    \YOvertex{27}{5/.27}
    \YOvertex{28}{16/.27}
    \node[YOlabel,left=4pt]  at ($(v11)+(-.45,0)$) {$h_1$};
    \node[YOlabel,right=4pt] at ($(v10)+(.36,0)$)  {$h_2$};
    \pgfresetboundingbox
    \path[use as bounding box] (-6.52,-5.2) rectangle (6.52,5.6);
  \end{tikzpicture}%
}

\newcommand{\CSdrawHtwentyoneTITS}{%
  \begin{tikzpicture}[x=1cm,y=\CSy cm]
    \coordinate (v11) at (-4.40, 0.00);
    \coordinate (v10) at ( 4.40, 0.00);
    \coordinate (v3)  at (-2.10, 2.30);
    \coordinate (v2)  at ( 0.00, 2.30);
    \coordinate (v1)  at ( 3.15, 2.30);
    \coordinate (v7)  at (-2.10,-2.60);
    \coordinate (v8)  at ( 0.00,-2.60);
    \coordinate (v9)  at ( 2.10,-2.60);
    \coordinate (v5)  at ( 0.00,-0.50);
    \coordinate (v14) at ($(v11)!.5!(v3)$);
    \coordinate (v15) at ($(v11)!.5!(v7)$);
    \coordinate (v19) at ($(v1)!.5!(v10)$);
    \coordinate (v23) at ($(v9)!.5!(v10)$);
    \coordinate (v4)  at (-0.90,-1.30);
    \coordinate (v6)  at ( 1.215, 0.58);
    \coordinate (v24) at ($(v11)!.52!(v4)$);
    \coordinate (v18) at ($(v10)!.42!(v4)$);
    \coordinate (v13) at ( 0.00, 4.85);
    \coordinate (v20) at ($(v13)!.40!(v1)$);
    \coordinate (v22) at ($(v13)!.40!(v6)$);
    \coordinate (v21) at ($(v13)!.19!(v7)$);
    \draw[YOtrace,YOctx1]  (v3)--(v2)--(v1);
    \draw[YOtrace,YOctx2]  (v7)--(v8)--(v9);
    \draw[YOtrace,YOctx3]  (v4)--(v5)--(v6);
    \draw[YOtrace,YOctx4]  (v2)--(v5)--(v8);
    \draw[YOtrace,YOctx8]  (v11)--(v14)--(v3);
    \draw[YOtrace,YOctx9]  (v11)--(v15)--(v7);
    \draw[YOtrace,YOctx11] (v9)--(v23)--(v10);
    \draw[YOtrace,YOctx12] (v1)--(v19)--(v10);
    \draw[YOtrace,YOctx10] (v11)--(v24)--(v4);
    \draw[YOtrace,YOctx13] (v10)--(v18)--(v4);
    \draw[YOtrace,YOctx16] (v6)--(v22)--(v13);
    \draw[YOtrace,YOctx14] (v13)--(v20)--(v1);
    \draw[YOtrace,YOctx15] (v13)--(v21)--(v7);
    \draw[CSident] (v2) to[bend left=16]  node[CSnote,pos=.55,below=1pt] {$2\sim_c 14$} (v14);
    \draw[CSident] (v8) to[bend left=18]  node[CSnote,pos=.52,above=1pt] {$8\sim_c 23$} (v23);
    \draw[CSident] (v5) to[bend right=20] node[CSnote,pos=.48,right=2pt] {$5\sim_c 22$} (v22);
    \draw[CSexcl]  (v14) to[bend left=4] node[CSnote,pos=.22,below=1pt] {$c(14)\ne c(18)$} (v18);
    \CSHtwentyonevertices
    \node[YOlabel,left=4pt]  at ($(v11)+(-.45,0)$) {$h_1$};
    \node[YOlabel,right=4pt] at ($(v10)+(.45,0)$)  {$h_2$};
    \node[YOlabel,left=4pt] at ($(v13)+(-.62,0)$)  {$h_3$};
    \node[YOlabel,above=3pt] at ($(v3)+(-.20,.32)$) {$y_1^-$};
    \node[YOlabel,above=3pt] at ($(v1)+(.25,.32)$)  {$y_1^+$};
    \node[YOlabel,below=3pt] at ($(v7)+(0,-.36)$)   {$y_2^+$};
    \node[YOlabel,below=3pt] at ($(v9)+(0,-.36)$)   {$y_2^-$};
    \pgfresetboundingbox
    \path[use as bounding box] (-6.52,-5.2) rectangle (6.52,5.6);
  \end{tikzpicture}%
}

\tikzset{
  CSexcl/.style={draw=YOdiagonal,line width=1.6pt,densely dashed,{Latex[length=2.4mm]}-{Latex[length=2.4mm]}},
  CSident/.style={draw=YOtriad,line width=1.6pt,densely dotted,{Latex[length=2.4mm]}-{Latex[length=2.4mm]}},
  CSnote/.style={font=\sffamily\CSnotefont{8},fill=white,inner sep=1.5pt,rounded corners=1pt},
  CSvec/.style={font=\CSnotefont{7},text=black!75,inner sep=1pt,CSgrow},
}

\providecommand{\CSsolidapex}[7]{%
  \draw[YOtrace,YOctx#1] (#2)--(#3)--(#4)
    -- ($(#4)!-#5!(#2)$)
    .. controls ($(#4)!-#6!(#2)$) and #7 .. (v0);%
}

\newcommand{\CSdrawWfiveFour}{%
  \begin{tikzpicture}[x=1cm,y=1cm]
    \def\R{3.9}%
    \pgfmathsetmacro{\rc}{\R/(2*cos(36))}%
    \coordinate (v0) at (0,0);
    \foreach \v/\ang in {1/90,2/162,3/234,4/306,5/18}{%
      \coordinate (v\v) at (\ang:\R);
    }
    \foreach \v/\ang in {6/126,7/198,8/270,9/342,10/54}{%
      \coordinate (v\v) at (\ang:{2*\rc});
    }
    \draw[YOtrace,YOctx1]  (0,0) arc (126-180:126+72:\rc);
    \draw[YOtrace,YOctx9]  (0,0) arc (198-180:198+72:\rc);
    \draw[YOtrace,YOctx10] (0,0) arc (270-180:270+72:\rc);
    \draw[YOtrace,YOctx5]  (0,0) arc (342-180:342+72:\rc);
    \draw[YOtrace,YOctx8]  (0,0) arc ( 54-180: 54+72:\rc);
    \draw[CSexcl,-] (v6) -- (v7) -- (v9) -- cycle;
    \node[CSnote,text=YOdiagonal] at (0,-1.35) {$\neg\bigl(c(6){=}c(7){=}c(9)\bigr)$};
    \YOvertex{0}{1/.63,8/.54,5/.45,10/.36,9/.27}
    \YOvertex{1}{1/.36,8/.27}
    \YOvertex{2}{1/.36,9/.27}
    \YOvertex{3}{9/.36,10/.27}
    \YOvertex{4}{10/.36,5/.27}
    \YOvertex{5}{5/.36,8/.27}
    \YOvertex{6}{1/.27}
    \YOvertex{7}{9/.27}
    \YOvertex{8}{10/.27}
    \YOvertex{9}{5/.27}
    \YOvertex{10}{8/.27}
    \pgfresetboundingbox
    \path[use as bounding box] (-6.3,-5.4) rectangle (6.3,5.4);
  \end{tikzpicture}%
}

\newcommand{\CSdrawWfiveLinFour}{%
  \begin{tikzpicture}[x=1cm,y=1cm]
    \coordinate (v0)  at (-0.80, 0.00);
    \coordinate (v2)  at (-2.60,-2.25);
    \coordinate (v6)  at ( 0.36, 1.45);
    \coordinate (v1)  at ( 1.60, 3.00);
    \coordinate (v3)  at (-2.60, 2.25);
    \coordinate (v8)  at ( 0.36,-1.45);
    \coordinate (v4)  at ( 1.60,-3.00);
    \coordinate (v11) at ( 0.85, 0.00);
    \coordinate (v12) at ( 2.55, 0.00);
    \coordinate (v5)  at ( 4.30, 0.00);
    \coordinate (v7)  at ($(v2)!.333!(v3)$);
    \coordinate (v13) at ($(v2)!.667!(v3)$);
    \coordinate (v9)  at ($(v4)!.34!(v5)$);
    \coordinate (v14) at ($(v4)!.66!(v5)$);
    \coordinate (v10) at ($(v5)!.34!(v1)$);
    \coordinate (v15) at ($(v5)!.66!(v1)$);
    \draw[YOtrace,YOctx1]  (v2)--(v0)--(v6)--(v1);
    \draw[YOtrace,YOctx9]  (v2)--(v7)--(v13)--(v3);
    \draw[YOtrace,YOctx10] (v3)--(v0)--(v8)--(v4);
    \draw[YOtrace,YOctx5]  (v4)--(v9)--(v14)--(v5);
    \draw[YOtrace,YOctx8]  (v5)--(v10)--(v15)--(v1);
    \draw[YOtrace,YOctx12] (v0)--(v11)--(v12)--(v5);
    \draw[CSexcl,-] (v6) -- (v8) -- (v11) -- cycle;
    \node[CSnote,text=YOdiagonal] at (0.90,-4.15) {$\neg\bigl(c(6){=}c(8){=}c(11)\bigr)$};
    \YOvertex{0}{1/.45,10/.36,12/.27}
    \YOvertex{1}{1/.36,8/.27}
    \YOvertex{2}{1/.36,9/.27}
    \YOvertex{3}{9/.36,10/.27}
    \YOvertex{4}{10/.36,5/.27}
    \YOvertex{5}{5/.45,8/.36,12/.27}
    \YOvertex{6}{1/.27}
    \YOvertex{7}{9/.27}
    \YOvertex{13}{9/.27}
    \YOvertex{8}{10/.27}
    \YOvertex{9}{5/.27}
    \YOvertex{14}{5/.27}
    \YOvertex{10}{8/.27}
    \YOvertex{15}{8/.27}
    \YOvertex{11}{12/.27}
    \YOvertex{12}{12/.27}
    \pgfresetboundingbox
    \path[use as bounding box] (-6.3,-5.4) rectangle (6.3,5.4);
  \end{tikzpicture}%
}

\newcommand{\CSdrawSpeckerHardyFour}{%
  \begin{tikzpicture}[x=1cm,y=1cm]
    \def\R{3.50}%
    \coordinate (v10) at (  0:\R);
    \coordinate (v1)  at ( 60:\R);
    \coordinate (v3)  at (120:\R);
    \coordinate (v11) at (180:\R);
    \coordinate (v7)  at (240:\R);
    \coordinate (v9)  at (300:\R);
    \coordinate (v2)  at ($(v3)!.5!(v1)$);
    \coordinate (v8)  at ($(v7)!.5!(v9)$);
    \coordinate (v5)  at (0, 0.00);
    \coordinate (v0)  at (0,-4.85);
    \coordinate (v14) at ($(v11)!.5!(v3)$);
    \coordinate (v15) at ($(v11)!.5!(v7)$);
    \coordinate (v19) at ($(v1)!.5!(v10)$);
    \coordinate (v23) at ($(v9)!.5!(v10)$);

    \draw[YOtrace,YOctx4] (v2)--(v5)--(v8)--(v0);
    \CSsolidapex{9}{v11}{v15}{v7}{0.06}{0.28}{(-0.75,-3.95)}
    \CSsolidapex{8}{v3}{v14}{v11}{0.04}{0.18}{(-2.85,-4.45)}
    \CSsolidapex{11}{v10}{v23}{v9}{0.06}{0.28}{( 0.75,-3.95)}
    \CSsolidapex{2}{v7}{v8}{v9}{0.05}{0.20}{( 1.85,-4.15)}
    \CSsolidapex{12}{v1}{v19}{v10}{0.04}{0.18}{( 3.05,-4.50)}
    \draw[YOtrace,YOctx1] (v3)--(v2)--(v1)
      -- ($(v1)!-0.08!(v3)$)
      .. controls ( 3.45, 3.03) and ( 4.75, 1.45) .. ( 4.65,-0.55)
      .. controls ( 4.55,-2.65) and ( 3.65,-4.85) .. (v0);

    \draw[CSexcl] (v11) to[bend right=18] node[CSnote,pos=.50] {$c(10)\ne c(11)$} (v10);

    \YOvertex{0}{1/.81,2/.72,4/.63,8/.54,9/.45,11/.36,12/.27}
    \YOvertex{1}{1/.36,12/.27}
    \YOvertex{2}{1/.36,4/.27}
    \YOvertex{3}{1/.36,8/.27}
    \YOvertex{5}{4/.27}
    \YOvertex{7}{2/.36,9/.27}
    \YOvertex{8}{2/.36,4/.27}
    \YOvertex{9}{2/.36,11/.27}
    \YOvertex{10}{11/.36,12/.27}
    \YOvertex{11}{8/.36,9/.27}
    \YOvertex{14}{8/.27}
    \YOvertex{15}{9/.27}
    \YOvertex{19}{12/.27}
    \YOvertex{23}{11/.27}
    \node[YOlabel,left=4pt]  at ($(v11)+(-.38,0.18)$) {$A$};
    \node[YOlabel,right=4pt] at ($(v10)+(.38,0.18)$)  {$B$};
    \pgfresetboundingbox
    \path[use as bounding box] (-6.3,-5.9) rectangle (6.3,4.5);
  \end{tikzpicture}%
}

\newcommand{\CSdrawStrictLinTIFSFour}{%
  \begin{tikzpicture}[x=1cm,y=1cm]
    \coordinate (v1) at (-0.60, 3.30);
    \coordinate (v2) at (-0.60, 1.60);
    \coordinate (v4) at (-0.60, 0.10);
    \coordinate (v3) at (-0.60,-1.30);
    \coordinate (v5) at ( 4.20, 3.30);
    \coordinate (v8) at ( 4.20, 0.10);
    \coordinate (v7) at ( 4.20,-1.30);
    \coordinate (v6) at ( 4.20,-3.50);
    \coordinate (v9)  at ($(v4)!.333!(v8)$);
    \coordinate (v10) at ($(v4)!.667!(v8)$);
    \coordinate (v11) at (-5.20, 3.30);
    \coordinate (v12) at ($(v1)!.50!(v5)$);
    \coordinate (v14) at ($(v11)!.333!(v2)$);
    \coordinate (v13) at ($(v11)!.667!(v2)$);
    \coordinate (v15) at (-2.60,-1.30);
    \coordinate (v16) at ($(v3)!.50!(v7)$);
    \coordinate (v17) at ($(v15)!.48!(v6)$);
    \coordinate (v18) at ($(v15)!.76!(v6)$);
    \draw[YOtrace,YOctx1]  (v1)--(v2)--(v4)--(v3);
    \draw[YOtrace,YOctx2]  (v5)--(v8)--(v7)--(v6);
    \draw[YOtrace,YOctx4]  (v4)--(v9)--(v10)--(v8);
    \draw[YOtrace,YOctx8]  (v11)--(v1)--(v12)--(v5);
    \draw[YOtrace,YOctx9]  (v11)--(v14)--(v13)--(v2);
    \draw[YOtrace,YOctx11] (v15)--(v3)--(v16)--(v7);
    \draw[YOtrace,YOctx12] (v15)--(v17)--(v18)--(v6);
    \draw[CSexcl] (v11) to[bend right=14] node[CSnote,pos=.52] {$c(11)\ne c(15)$} (v15);
    \YOvertex{1}{1/.36,8/.27}
    \YOvertex{2}{1/.36,9/.27}
    \YOvertex{3}{1/.36,11/.27}
    \YOvertex{4}{1/.36,4/.27}
    \YOvertex{5}{2/.36,8/.27}
    \YOvertex{6}{2/.36,12/.27}
    \YOvertex{7}{2/.36,11/.27}
    \YOvertex{8}{2/.36,4/.27}
    \YOvertex{9}{4/.27}
    \YOvertex{10}{4/.27}
    \YOvertex{11}{8/.36,9/.27}
    \YOvertex{12}{8/.27}
    \YOvertex{13}{9/.27}
    \YOvertex{14}{9/.27}
    \YOvertex{15}{11/.36,12/.27}
    \YOvertex{16}{11/.27}
    \YOvertex{17}{12/.27}
    \YOvertex{18}{12/.27}
    \pgfresetboundingbox
    \path[use as bounding box] (-6.3,-5.4) rectangle (6.3,5.4);
  \end{tikzpicture}%
}

\newcommand{\CSdrawPureChromTIFSFour}{%
  \begin{tikzpicture}[x=1cm,y=1cm]
    \coordinate (v11) at (-3.65, 0.00);
    \coordinate (v10) at ( 3.65, 0.00);
    \coordinate (v3)  at (-1.98, 2.45);
    \coordinate (v1)  at (-0.99, 2.45);
    \coordinate (v2)  at ( 0.00, 2.45);
    \coordinate (v4)  at ( 1.82, 2.45);
    \coordinate (v7)  at (-1.82,-2.45);
    \coordinate (v8)  at ( 0.00,-2.45);
    \coordinate (v9)  at ( 1.82,-2.45);
    \coordinate (v5)  at ( 0.00,-0.55);
    \coordinate (v0)  at ( 0.00,-4.85);
    \coordinate (v6)  at ($(v5)!.55!(v4)$);
    \coordinate (v14) at ($(v11)!.5!(v3)$);
    \coordinate (v15) at ($(v11)!.5!(v7)$);
    \coordinate (v18) at ($(v10)!.5!(v4)$);
    \coordinate (v23) at ($(v9)!.5!(v10)$);
    \coordinate (v24) at ($(v11)!.44!(v4)$);

    \draw[YOtrace,YOctx4] (v2)--(v5)--(v8)--(v0);
    \CSsolidapex{3}{v4}{v6}{v5}{0.14}{0.52}{(-0.52,-2.85)}
    \CSsolidapex{9}{v11}{v15}{v7}{0.06}{0.28}{(-0.95,-3.95)}
    \CSsolidapex{8}{v3}{v14}{v11}{0.04}{0.16}{(-2.85,-4.45)}
    \draw[YOtrace,YOctx1] (v2)--(v1)--(v3)
      -- ($(v3)!-0.08!(v2)$)
      .. controls (-3.55, 2.45) and (-4.85, 1.25) .. (-4.75,-0.65)
      .. controls (-4.65,-2.65) and (-3.65,-4.85) .. (v0);
    \CSsolidapex{11}{v10}{v23}{v9}{0.06}{0.28}{( 0.75,-3.95)}
    \CSsolidapex{2}{v7}{v8}{v9}{0.05}{0.20}{( 1.85,-4.15)}
    \CSsolidapex{13}{v4}{v18}{v10}{0.04}{0.16}{( 3.05,-4.50)}
    \draw[YOtrace,YOctx10] (v11)--(v24)--(v4)
      -- ($(v4)!-0.06!(v11)$)
      .. controls ( 3.15, 3.05) and ( 4.85, 1.55) .. ( 4.75,-0.55)
      .. controls ( 4.65,-2.65) and ( 3.65,-4.85) .. (v0);

    \draw[CSexcl] (v14) to[bend right=36] node[CSnote,pos=.50] {$c(14)\ne c(18)$} (v18);

    \YOvertex{0}{1/.99,2/.90,3/.81,4/.72,8/.63,9/.54,10/.45,11/.36,13/.27}
    \YOvertex{1}{1/.27}
    \YOvertex{2}{1/.36,4/.27}
    \YOvertex{3}{1/.36,8/.27}
    \YOvertex{4}{3/.45,10/.36,13/.27}
    \YOvertex{5}{3/.36,4/.27}
    \YOvertex{6}{3/.27}
    \YOvertex{7}{2/.36,9/.27}
    \YOvertex{8}{2/.36,4/.27}
    \YOvertex{9}{2/.36,11/.27}
    \YOvertex{10}{11/.36,13/.27}
    \YOvertex{11}{8/.45,9/.36,10/.27}
    \YOvertex{14}{8/.27}
    \YOvertex{15}{9/.27}
    \YOvertex{18}{13/.27}
    \YOvertex{23}{11/.27}
    \YOvertex{24}{10/.27}
    \pgfresetboundingbox
    \path[use as bounding box] (-6.3,-5.9) rectangle (6.3,4.5);
  \end{tikzpicture}%
}

\newcommand{\CSdrawStrictLinTITSFour}{%
  \begin{tikzpicture}[x=1cm,y=1cm]
    \coordinate (v1)  at (-0.60, 3.30);
    \coordinate (v2)  at (-0.60, 1.60);
    \coordinate (v4)  at (-0.60, 0.10);
    \coordinate (v3)  at (-0.60,-1.30);
    \coordinate (v5)  at ( 4.20, 3.30);
    \coordinate (v8)  at ( 4.20, 0.10);
    \coordinate (v7)  at ( 4.20,-1.30);
    \coordinate (v6)  at ( 4.20,-3.50);
    \coordinate (v9)  at ($(v4)!.333!(v8)$);
    \coordinate (v10) at ($(v4)!.667!(v8)$);
    \coordinate (v11) at (-5.20, 3.30);
    \coordinate (v12) at ($(v1)!.50!(v5)$);
    \coordinate (v14) at ($(v11)!.333!(v2)$);
    \coordinate (v13) at ($(v11)!.667!(v2)$);
    \coordinate (v15) at (-2.60,-1.30);
    \coordinate (v16) at ($(v3)!.50!(v7)$);
    \coordinate (v17) at ($(v15)!.48!(v6)$);
    \coordinate (v18) at ($(v15)!.76!(v6)$);
    \coordinate (v20) at ($(v13)!.50!(v15)$);
    \coordinate (v19) at ($(v13)!1.62!(v15)$);
    \coordinate (v21) at ($(v11)!.30!(v19)$);
    \coordinate (v22) at ($(v11)!.62!(v19)$);
    \draw[YOtrace,YOctx1]  (v1)--(v2)--(v4)--(v3);
    \draw[YOtrace,YOctx2]  (v5)--(v8)--(v7)--(v6);
    \draw[YOtrace,YOctx4]  (v4)--(v9)--(v10)--(v8);
    \draw[YOtrace,YOctx8]  (v11)--(v1)--(v12)--(v5);
    \draw[YOtrace,YOctx9]  (v11)--(v14)--(v13)--(v2);
    \draw[YOtrace,YOctx11] (v15)--(v3)--(v16)--(v7);
    \draw[YOtrace,YOctx12] (v15)--(v17)--(v18)--(v6);
    \draw[YOtrace,YOctx5]  (v11)--(v21)--(v22)--(v19);
    \draw[YOtrace,YOctx16] (v13)--(v20)--(v15)--(v19);
    \draw[CSexcl]  (v11) to[bend right=24] node[CSnote,pos=.50,anchor=east,xshift=-3pt] {$c(11)\ne c(15)$} (v15);
    \draw[CSident] (v11) to[bend right=6]  node[CSnote,pos=.68,anchor=west,xshift=2pt] {$11\sim_c 20$} (v20);
    \YOvertex{1}{1/.36,8/.27}
    \YOvertex{2}{1/.36,9/.27}
    \YOvertex{3}{1/.36,11/.27}
    \YOvertex{4}{1/.36,4/.27}
    \YOvertex{5}{2/.36,8/.27}
    \YOvertex{6}{2/.36,12/.27}
    \YOvertex{7}{2/.36,11/.27}
    \YOvertex{8}{2/.36,4/.27}
    \YOvertex{9}{4/.27}
    \YOvertex{10}{4/.27}
    \YOvertex{11}{8/.45,9/.36,5/.27}
    \YOvertex{12}{8/.27}
    \YOvertex{13}{9/.36,16/.27}
    \YOvertex{14}{9/.27}
    \YOvertex{15}{11/.45,12/.36,16/.27}
    \YOvertex{16}{11/.27}
    \YOvertex{17}{12/.27}
    \YOvertex{18}{12/.27}
    \YOvertex{19}{5/.36,16/.27}
    \YOvertex{20}{16/.27}
    \YOvertex{21}{5/.27}
    \YOvertex{22}{5/.27}
    \pgfresetboundingbox
    \path[use as bounding box] (-6.3,-5.4) rectangle (6.3,5.4);
  \end{tikzpicture}%
}

\hypersetup{pdftitle={Strong Colorings of Orthogonality Hypergraphs as Noncontextual Eigenvalue Assignments},pdfauthor={Karl Svozil}}
\begin{document}
\title{Strong Colorings of Orthogonality Hypergraphs as Noncontextual Eigenvalue Assignments}
\author{Karl Svozil\,\orcidlink{0000-0001-6554-2802}}
\email{karl.svozil@tuwien.ac.at}
\affiliation{Institute for Theoretical Physics, TU Wien, Wiedner Hauptstra{\ss}e 8-10/136, A-1040 Vienna, Austria}
\date{November 24, 2026}
\begin{abstract}
Suppose intertwined complete experiments each measure a maximal observable with one common spectrum, and the eigenvalue of a shared outcome is an element of physical reality, the same in every experiment. Such noncontextual eigenvalues exist exactly when the context hypergraph can be colored with as many colors as the dimension. They can fail even when the two-valued states separate all outcomes, so that the logic is classically realizable; completed Yu--Oh is an example. When no such assignment exists, any classical model needs more hidden states than the dimension, just to tell apart preparations along orthogonal rays. Where noncontextual eigenvalues exist, they can force two distinct rays to share an eigenvalue, a chromatic form of Kochen--Specker inseparability; for one quantum preparation, every classical model must put a third of its weight on truth assignments that no eigenvalue assignment produces. New configurations include a $37$-ray hypergraph in complex three-dimensional space that has no real realization and needs four colors although its fractional chromatic number is three; its three-colorable $25$-ray part already has no real realization, so the two obstructions are separate. In real four dimensions, a $65$-ray rational hypergraph needs five colors while $393{,}216$ two-valued states survive. We also correct an earlier inference from coloring to geometry.
\end{abstract}
\maketitle
\section{Introduction}\label{sec:intro}

An experimental context with $d$ rank-one outcomes admits a nondegenerate observable with $d$ distinct eigenvalues. These eigenvalues distinguish all $d$ outcomes and provide the maximal resolution of a complete projective experiment within that context. Such an observable is maximal in the sense that every self-adjoint observable commuting with it is a function of it, so the statistics of the latter can be obtained by processing its measurement outcomes. Noninjective functions merge outcomes and reduce this resolution. This is the finite-dimensional setting of von Neumann's functional representation of commuting observables~\cite[Satz~8, p.~221]{v-neumann-31}; see also Halmos~\cite[\S84, pp.~171--172]{halmos-vs}.

When complete experiments are intertwined, they share outcomes represented by the same rank-one projectors. We ask whether the eigenvalues of the maximal observables measured in such intertwined experiments can be assigned to their outcomes noncontextually, in the spirit of the elements of physical reality of Einstein, Podolsky, and Rosen~\cite{epr}: whether every outcome carries a definite eigenvalue that belongs to the outcome itself, rather than to the experiment in which it is recorded. Three assumptions make this precise.
\begin{enumerate}
\item[(i)] \emph{Maximality.} Each complete experiment measures a maximal observable, so its $d$ outcomes are never merged and carry $d$ distinct eigenvalues. A degenerate observable would discard part of what the experiment resolves.
\item[(ii)] \emph{Uniformity.} All maximal observables share one spectrum of $d$ eigenvalues $\Sigma=\{\lambda_0,\ldots,\lambda_{d-1}\}$. The eigenvalues are nothing but labels of the $d$ outcomes of a complete experiment. Since self-adjoint operators with a common nondegenerate spectrum are unitarily equivalent, uniformity says that every complete experiment measures one and the same quantity in a different orientation, $A_e=U_eAU_e^\dagger$.
\item[(iii)] \emph{Noncontextuality of eigenvalues.} An outcome shared by several complete experiments carries the same eigenvalue in each of them, just as Kochen and Specker require a shared outcome to carry the same truth value~\cite{kochen1}.
\end{enumerate}
No constraint is placed on how the $d$ eigenvalues are distributed over the outcomes of different contexts; any permutation within a context is allowed. A noncontextual eigenvalue assignment in this sense labels the outcomes, the most elementary record a complete experiment produces, and assumption~(iii) treats these labels as elements of physical reality. For spin one, an example is the observable $aS_x^2+bS_y^2+cS_z^2$ with distinct real $a,b,c$ considered by Kochen and Specker~\cite{kochen1}: every orthogonal triad of directions yields one such observable, with the same eigenvalues $b+c$, $a+c$, and $a+b$. A noncontextual eigenvalue assignment then gives each direction the same coefficient, $a$, $b$, or $c$, in every triad that contains it.

These assumptions concern eigenvalues rather than which outcome occurs, yet they have immediate consequences for truth assignments. A two-valued state selects one outcome in every context and retains only the distinction between that outcome and the others. A noncontextual eigenvalue assignment yields $d$ such states at once, one for each eigenvalue $\lambda_a$: ``the outcome with eigenvalue $\lambda_a$ occurs in every experiment.'' These $d$ truth assignments have pairwise disjoint supports covering every outcome. A noncontextual eigenvalue assignment is thus a hidden-variable model whose hidden variable is the eigenvalue itself; its $d$ equally weighted values reproduce the maximally mixed state. Because this requirement is more demanding than Kochen--Specker noncontextuality, which asks for a single truth assignment, its failure---\emph{chromatic contextuality}---is a weaker form of contextuality that already appears in configurations retaining many classical truth assignments. Even when a noncontextual eigenvalue assignment exists, its consistency across intertwined contexts can force relations between distinct outcomes, such as equal eigenvalues for nonorthogonal rays, that are absent from the full family of individual truth assignments. We study both the existence of noncontextual eigenvalue assignments and the event distinctions and logical relations that survive when such assignments are available~\cite{svozil-2020-c,svozil-2021-chroma,svozil-2025-color}.

To express these questions, let a $d$-uniform context hypergraph $H=(V,E)$ encode the complete experiments. An ordinary two-valued state $s:V\to\{0,1\}$ selects exactly one outcome in every context, so $\sum_{v\in e}s(v)=1$ for every $e\in E$. The Kochen--Specker (KS) theorem~\cite{kochen1} establishes the nonexistence of such a global truth assignment for suitable finite families of intertwined projective contexts in dimension $d\ge3$.

A strong $d$-coloring $c:V\to\{0,\ldots,d-1\}$ assigns all $d$ distinct labels bijectively within every context and independently of context at shared vertices. It is the formal counterpart of a noncontextual eigenvalue assignment: $c(i)=a$ means that outcome $i$ carries the eigenvalue $\lambda_a$ in every context containing it, so assumptions (i)--(iii) hold if and only if $\chi(H)\le d$. For an injective ray realization $i\mapsto[v_i]\in\PP(\F^d)$, with $\F\in\{\R,\C\}$, rank-one projectors $P_i$, and a common nondegenerate real spectrum $\Sigma=\{\lambda_0,\ldots,\lambda_{d-1}\}$, it defines the maximal context observables
\begin{equation}\label{eq:intro-maximal-obs}
A_e=\sum_{i\in e}\lambda_{c(i)}P_i.
\end{equation}

Because the indicator functions $s_a(v)=\mathbf 1_{\{c(v)=a\}}$ ($a\in\{0,\ldots,d-1\}$) of any strong $d$-coloring form $d$ two-valued states whose supports partition $V$, one has the strict one-way implication
\begin{equation}\label{eq:intro-ks-to-chi}
\mathcal S(H)=\varnothing \quad\Longrightarrow\quad \chi(H)>d,
\end{equation}
whose converse fails even when $\mathcal S(H)$ is unital and separates every pair of vertices. The earlier introduction in Ref.~\cite{svozil-2021-chroma} inferred absence of global two-valued assignments from $\chi>d$. Equation~\eqref{eq:intro-ks-to-chi} is the valid implication; completed Yu--Oh refutes the converse. Both the implication and the failure of its converse, for the original $13$-ray Yu--Oh set, were noted by Cabello~\cite{cabello2012maxnonlocality}. For the completed $25$-ray Yu--Oh hypergraph, $\chi>3$ together with a separating family of $24$ two-valued states was established in Ref.~\cite{svozil-2025-color}, which also identified two-valued states that do not arise from strong colorings. In terms of eigenvalues, if no noncontextual truth assignment exists, no noncontextual eigenvalue assignment exists either; but noncontextual eigenvalues can be impossible while truth assignments, even a unital separating family of them, survive. The distinct Appendix~B inference about geometry is corrected below. We say that $H$ has an \emph{exact palette}, or is \emph{chromatically complete}, when $\chi(H)=d$, with no additional density or geometric assumption. This asymmetry separates chromatic contextuality into two distinct regimes that are unified in this work:

\begin{enumerate}
\item \emph{Regime~I, common-spectrum chromatic obstructions ($\chi(H)>d$, possibly with $\mathcal S(H)\ne\varnothing$):} Here no strong $d$-coloring exists ($\mathcal C_d(H)=\varnothing$), so the outcomes admit no noncontextual assignment of $d$ eigenvalues: every family of context observables with one common spectrum of $d$ eigenvalues must either merge outcomes, and so fail to be maximal, or let a shared outcome change its eigenvalue from context to context; maximal observables with noncontextual eigenvalues need more than $d$ eigenvalues. This happens even in examples in which classical truth assignments survive. A critical foundational clarification—correcting an interpretation in Appendix~B of the earlier work with Shekarriz~\cite{svozil-2021-chroma}—is that the chromatic obstruction $\chi(H)>d$ is an operator spectral-value obstruction and \emph{does not} by itself obstruct geometric ray realizability in $\F^d$. Indeed, for finite $3$-uniform linear hypergraphs of Berge girth five with unital separating state spaces, strong $3$-colorability and injective $\C^3$ realizability are logically independent (Proposition~\ref{m:prop:independence}). Ruling out geometric ray assignments for Greechie's Petersen-derived $G_{32}$~\cite{greechie:71} requires a direct Hermitian projection-identity argument (Theorem~\ref{m:thm:multipole}) rather than $\chi(G_{32})=4$.
\item \emph{Regime~II, exact-palette chromatic separability and logical gadgets ($\chi(H)=d$ with $\mathcal S_{\mathrm{ext}}(H)\subsetneq\mathcal S(H)$):} In this regime, $\mathcal C_d(H)\ne\varnothing$, so noncontextual eigenvalue assignments exist, but some two-valued states $s\in\mathcal S(H)$ cannot be completed to a strong $d$-coloring: they are truth assignments that no noncontextual eigenvalue assignment produces. We characterize the extendable states $\mathcal S_{\mathrm{ext}}(H)$ by the residual chromatic condition $\chi(\Gamma(H)[V\setminus\supp s])\le d-1$ (Theorem~\ref{thm:completion}), which in dimension $d=3$ yields an odd-cycle nonextendability certificate. In the specified subhypergraphs—such as the faithful real $17$-ray and $21$-ray restrictions $H_{17,10}$ and $H_{21,13}$ of completed Yu--Oh~\cite{Yu-2012,svozil-2025-color} (Theorem~\ref{thm:h17}) and our faithful $3$- and $4$-dimensional logical gadgets (Secs.~\ref{sec:3dgadgets} and~\ref{sec:4dgadgets})—ordinary states separate all vertices while every admissible $d$-coloring forces nonadjacent rays to share a color ($u\sim_c v$, chromatic inseparability / implicit identity~\cite{martin2009}) or obeys purely chromatic true-implies-false (TIFS) or ternary exclusions violated by nonextendable states in $\mathcal S(H)\setminus\mathcal S_{\mathrm{ext}}(H)$.
\end{enumerate}

These two regimes are governed by a shared structural mechanism. Lemma~\ref{lem:disjoint-states} ($\chi(H)=d$ if and only if $H$ admits $d-1$ pairwise disjoint states) is the global unanchored form of Theorem~\ref{thm:completion}. Likewise, the eight normalized $3$-colorings of the four base contexts of completed Yu--Oh $H_{\YO}^+$ partition into four antipodal pairs that each make one of the four diagonal neighborhoods rainbow (Theorem~\ref{m:thm:yo-criticality}): retaining all four diagonal sectors forces $\chi(H_{\YO}^+)=4$ (Regime~I), whereas retaining two or three sectors leaves four or two normalized $3$-colorings and forces two or $21$ chromatically inseparable pairs in $H_{17,10}$ and $H_{21,13}$ (Regime~II), which our exhaustive $65{,}535$-subset inventory classification (Proposition~\ref{prop:search}) proves minimal among the faithful restrictions of $H_{\YO}^+$ with inherited coordinates.

For truth values, elements of reality fail as in the Kochen--Specker theorem: the logic has no two-valued states, or too few to separate its outcomes, and then it has no classical realization. For eigenvalues the situation is different. If the two-valued states of $H$ are unital and separating, $H$ can be realized classically as a partition logic~\cite{kochen1,dvur-pul-svo,svozil-2001-eua}: outcome $i$ becomes the set of states with $s(i)=1$, and each context becomes a partition of this set. Noncontextual eigenvalues can still be impossible, because they require a strong $d$-coloring, in the classical realization as much as in the quantum one. Greechie's $G_{32}$ is an example: it has a classical realization, no orthogonal ray assignment in $\C^3$ (Theorem~\ref{m:thm:multipole}), and $\chi(G_{32})=4$. Thus $\chi(H)>d$ is a condition on the spectra of maximal observables, not on truth values. It does, however, bound the size of every classical model. Preparations along orthogonal rays can be told apart with certainty, so in any hidden-variable model they have disjoint sets of hidden states; picking one hidden state for each ray then colors $\Gamma(H)$, and the model needs at least $\chi(H)$ hidden states (Sec.~\ref{sec:open} works this out for Yu--Oh). A noncontextual eigenvalue assignment keeps orthogonal rays apart with only $d$ hidden states, the eigenvalues themselves. Thus, when $H$ is realized in $\C^d$, $\chi(H)>d$ means that a classical model needs more hidden states than the quantum system has dimensions merely to tell apart preparations along orthogonal rays, a gap that Prakash~\cite{prakash2026memory} calls representational contextuality and turns into a memory advantage of quantum automata. A stronger separation, which excludes every local classical model, arises when an experiment fixes the label of each outcome, as the $d$-party supersinglet of the next paragraph does.

Noncontextual eigenvalues and strong colorings were first related by Cabello~\cite{cabello2012maxnonlocality}. The first version of that preprint assigned noncontextual values $-1,0,1$ to the spin-one components $S_u$ along the rays $u$ and concluded that $\chi>d$ characterizes state-independent contextuality. For an orthogonal triad $u,v,w$, however, the eigenvectors of $S_u$, $S_v$, $S_w$ for the eigenvalue $1$ (or $-1$) are not mutually orthogonal, so quantum mechanics fixes only which of them takes the value $0$; the necessity of $\chi>d$ and the minimality of the $13$-ray Yu--Oh set were later published in Ref.~\cite{cabello2015sic}. The context observables $A_e$ of Eq.~\eqref{eq:intro-maximal-obs} are diagonal in their contexts, so assumptions (i)--(iii) are equivalent to $\chi(H)\le d$ without such a gap. The second version of the preprint distributes the $d$ particles of the totally antisymmetric supersinglet among $d$ parties. If all parties measure bases containing a ray $v$, exactly one of them obtains $v$; if $d-1$ parties measure the same basis, the ray of that basis which none of them obtains is obtained with certainty by the last party in any basis containing it. A local hidden-variable model of these perfect correlations exists exactly when the hypergraph of the measured bases has a strong $d$-coloring, the color of a ray being the party that obtains it; for the Yu--Oh set this hypergraph is $H_{\YO}^{+}$. Here the label of an outcome is the party that records it, and its noncontextuality, assumption~(iii), follows from locality and the EPR criterion of reality~\cite{epr} instead of being postulated. Quantum mechanics produces these correlations whenever $H$ has an orthogonal ray assignment in $\C^d$, while local classical models need a strong $d$-coloring, so here $\chi(H)>d$ does separate quantum from classical physics. $G_{32}$ has no ray assignment and gives no such separation. For a single system in the maximally mixed state the threshold is fractional: $I_d/d$ has a noncontextual truth-assignment model exactly when $\chi_f(\Gamma(H))=d$~\cite{ramanathan2014sic} (Proposition~\ref{prop:fractional-uniform}). By this criterion, $\chi>d$ is not sufficient for the contextuality of $I_d/d$, contrary to the claim of Ref.~\cite{cabello2012maxnonlocality}; the faithful complex Hesse hypergraph $H_{37}$, with $\chi=4$ and $\chi_f=3$ (Theorem~\ref{h:thm:main}), is an explicit example realized in $\C^3$ at Berge girth five. The gap between $\chi$ and the least dimension of an orthogonal ray assignment has also been studied in connection with the quantum chromatic number~\cite{scarpa2012ks}.

We further connect these combinatorial and geometric layers to orthomodular event representations and quantum probability polytopes $\mathcal P_{\mathrm{ext}}(H)\subseteq\mathcal P(H)=\conv\mathcal S(H)$, where $\mathcal P_{\mathrm{ext}}(H)$ collects the probability vectors of hidden-variable models that choose a noncontextual eigenvalue assignment and one of its $d$ eigenvalues. On $H_{17,10}$, an explicit quantum ray preparation lies in $\mathcal P(H_{17,10})\setminus\mathcal P_{\mathrm{ext}}(H_{17,10})$ and requires nonextendable state weight exactly $1/9$; the preparation along the coordinate ray $z_3$ requires exactly $1/3$, the largest weight that a single chromatic identity of $H_{17,10}$ can certify (Theorem~\ref{thm:quantum}). In complex dimension three, the $37$-ray Hesse restriction $H_{37}$~\cite{svozil-2025-MYOCHS} has no injective $\R^3$ realization, pastes into a $76$-event orthomodular lattice with exact minimum state-family sizes $6$, $7$, and $18$ for vertex separation, event separation, and order determination (Theorem~\ref{h:thm:main}), and demonstrates that exact event-order determination does not guarantee classical mixtures for all Born preparations (Corollary~\ref{cor:order-probabilities}). Deleting $12$ of its rays leaves a $3$-colorable $25$-ray hypergraph $R_{25,15}$ that still has no real realization, so the real and the chromatic obstructions of $H_{37}$ are separate. In real dimension four, our two-plane completion lemma (Lemma~\ref{d:lem:plane-completion}) produces the canonical linear $65$-ray hypergraph $H_{65}$ completing the $17$-ray core of Haynes \emph{et al.}~\cite{haynes2010orthogonal} with $\chi=5$, $\chi_f=17/4$, and $393{,}216$ states that fail atom-order determination and collapse a ray with a plane (Theorem~\ref{d:thm:h65}). Finally, we note that the existence form of the KS ancestry hypothesis follows from an elementary union argument (Remark~\ref{prop:ks-ancestry}), give an abstract transverse chromatic-to-stateless construction (Proposition~\ref{prop:transverse}), establish exact inclusions into specified KS parents (Bub's $49$-ray set~\cite{bub}, the $165$-ray MUB inventory, Peres's $24$-ray geometry, and $68$-ray extensions containing the $18$-ray parity proof $C_{18,9}$ of Cabello \emph{et al.}~\cite{cabello-96} and a copy $C_{20,11}$ of Kernaghan's $20$-ray parity proof~\cite{Kernaghan-1994}), and contrast abstract context states $\mathcal S_{\mathrm{ctx}}$ with geometric states $\mathcal S_{\mathrm{geom}}$ on the $36$-ray completion and $45$-ray vertex deletion of $C_{18,9}$. Complete vector, state, and coloring enumerations and finite verification certificates are provided in the supplementary material, which additionally cites Refs.~\cite{silva2019orthogonal,brinkmann2015snarks,schrag1988petersen}.

\section{Unified Framework: States, Colorings, Separability, and Completion}\label{sec:definitions}

\subsection{Context hypergraphs, orthogonal realizations, and spectral regimes}

All hypergraphs and graphs considered here are finite. A $d$-uniform \emph{context hypergraph} $H=(V,E)$ ($d\ge2$) consists of a nonempty collection $E$ of $d$-element subsets of $V$, called \emph{contexts}, with $V=\bigcup_{e\in E}e$. Its \emph{two-section} (or \emph{context}) \emph{graph} $\Gamma(H)$ has vertex set $V$ and edge set $E(\Gamma(H))=\bigcup_{e\in E}\binom{e}{2}$. A hypergraph $H$ is \emph{linear} if $|e\cap e'|\le 1$ for all distinct $e,e'\in E$. A \emph{Berge cycle} of length $\ell\ge2$ is an alternating sequence of distinct vertices $x_1,\ldots,x_\ell$ and distinct contexts $e_1,\ldots,e_\ell$ with $\{x_j,x_{j+1}\}\subseteq e_j$ cyclically; the least such $\ell$ is the \emph{Berge girth} $g(H)$ (or $\infty$ if acyclic). In Greechie's pasting of Boolean blocks~\cite{greechie:71,svozil-tkadlec}, linearity ($g(H)\ge3$) with $g(H)\ge4$ yields an \emph{orthomodular poset}, and $g(H)\ge5$ yields an \emph{orthomodular lattice} (OML).

With the Hermitian inner product $\ip{x}{y}=\sum_j\overline{x_j}y_j$ on $\F^d$ ($\F\in\{\R,\C\}$), we distinguish three levels of geometric realization:
\begin{enumerate}
\item An \emph{orthogonal ray assignment} $\rho:V\to\PP(\F^d)$, $i\mapsto[v_i]$ ($v_i\ne0$), requires $\ip{v_i}{v_j}=0$ whenever $i\sim j$ in $\Gamma(H)$; nonadjacent vertices may share a ray.
\item An \emph{injective orthogonal representation} additionally requires $[v_i]\ne[v_j]$ for all $i\ne j$.
\item An injective representation is \emph{edge-faithful} (abbreviated FOR) if $\ip{v_i}{v_j}=0 \iff i\sim j$ in $\Gamma(H)$ for all $i\ne j$~\cite{hogben2017orthogonal}.
\end{enumerate}

A \emph{strong $k$-coloring} of $H$ is a map $c:V\to\{0,\ldots,k-1\}$ injective on every context $e\in E$, i.e., a proper $k$-coloring of $\Gamma(H)$; the least such $k$ is $\chi(H)=\chi(\Gamma(H))$. We call $H$ \emph{chromatically complete} (equivalently, $H$ has an exact palette) when $\chi(H)=d$ and write $\Cc_d(H)$ for the set of strong $d$-colorings of $H$. For a realized $d$-uniform hypergraph with rank-one projectors $P_i=v_iv_i^\dagger/\ip{v_i}{v_i}$, three spectral labelings arise: (a) \emph{context-dependent common spectrum}, where bijections $\ell_e:e\to\Sigma=\{\lambda_0,\ldots,\lambda_{d-1}\}$ define $A_e=\sum_{i\in e}\ell_e(i)P_i$ with shared projectors carrying context-dependent eigenvalues; (b) \emph{context-independent vertex labeling} via a strong $k$-coloring $c$ and distinct real scalars $\lambda_0,\ldots,\lambda_{k-1}$ in Eq.~\eqref{eq:intro-maximal-obs}, where spectra vary across contexts whenever $|c(V)|>d$ (unavoidable when $\chi(H)>d$); and (c) \emph{chromatic completeness} ($k=d$), which provides both context-independent eigenvalues and a single $d$-element spectrum $\Sigma$ for all $A_e$. Only labeling (c) meets all three assumptions of the Introduction: labeling (a) gives up noncontextual eigenvalues, and labeling (b) keeps them at the price of more than $d$ eigenvalues, so that different contexts have different spectra.

\subsection{Two-valued states, color extendability, and the separation hierarchy}

A \emph{two-valued context state} is a function $s:V\to\{0,1\}$ satisfying
\begin{equation}\label{eq:states}
\sum_{v\in e}s(v)=1\qquad(e\in E).
\end{equation}
We write $\Sc(H)=\Sc_{\mathrm{ctx}}(H)$ for the set of context states. For a ray realization $\rho$, the \emph{geometric state space}
\begin{equation}\label{eq:geometric-state-space}
\begin{aligned}
\Sc_{\rm geom}(H,\rho)=\{s\in\Sc(H):{}&s(i)+s(j)\le1\\
&\text{whenever }i\ne j,\ v_i\perp v_j\}
\end{aligned}
\end{equation}
additionally enforces orthogonality exclusions on uncovered pairs~\cite{waegell2011peres}; $\Sc_{\mathrm{geom}}(H,\rho)=\Sc(H)$ when $\rho$ is edge-faithful. A state family $\Omega\subseteq\Sc(H)$ is \emph{unital} if every $v\in V$ has $s(v)=1$ for some $s\in\Omega$, and \emph{vertex separating} if distinct $u,v\in V$ have $s(u)\ne s(v)$ for some $s\in\Omega$, inducing the partition set representation $A_i=\{s\in\Omega:s(i)=1\}$. On an event poset $(L,\le)$, $\Omega$ is \emph{event separating} if distinct $x,y\in L$ have distinct profiles on $\Omega$, and \emph{order determining} if $x\le y \iff (\forall s\in\Omega,\ s(x)\le s(y))$.

\begin{definition}[Chromatic separability and extendable states]\label{def:inseparable}
For chromatically complete $H$ ($\Cc_d(H)\ne\varnothing$), vertices $u,v\in V$ are \emph{chromatically inseparable}, written $u\sim_c v$, if
\begin{equation}\label{eq:inseparable}
c(u)=c(v)\qquad\text{for every }c\in\Cc_d(H),
\end{equation}
an \emph{implicit identity} in graph coloring~\cite[Def.~3.8]{martin2009}. We call $H$ \emph{chromatically separating} if $u\sim_c v\implies u=v$. A state $s\in\Sc(H)$ is \emph{extendable}, written $s\in\Sc_{\mathrm{ext}}(H)$, if $\supp s=\{v\in V:s(v)=1\}$ is a color class of some $c\in\Cc_d(H)$~\cite{svozil-2020-c,svozil-2021-chroma,svozil-2025-color}.
\end{definition}
Physically, $u\sim_c v$ means that every noncontextual eigenvalue assignment gives $u$ and $v$ the same eigenvalue, and an extendable state is exactly a truth assignment of the form ``the outcome with eigenvalue $\lambda_a$ occurs'' for some noncontextual eigenvalue assignment $c$ and eigenvalue $\lambda_a$.

\begin{theorem}[Completion criterion]\label{thm:completion}
For $s\in\Sc(H)$ with $S=\supp s$, the following are equivalent:
\begin{enumerate}
\item $s\in\Sc_{\mathrm{ext}}(H)$;
\item $\chi(\Gamma(H)[V\setminus S])\le d-1$ (so $S$ is a critical independent set~\cite[Def.~3.3]{martin2009});
\item $s$ belongs to $d$ states in $\Sc(H)$ with pairwise disjoint supports partitioning $V$;
\item $s$ belongs to $d-1$ states in $\Sc(H)$ with pairwise disjoint supports.
\end{enumerate}
For $d=3$, $s\notin\Sc_{\mathrm{ext}}(H)$ if and only if $\Gamma(H)[V\setminus S]$ contains an odd cycle.
\end{theorem}
\begin{proof}
A strong $d$-coloring with color class $S$ leaves a $(d-1)$-coloring on $V\setminus S$. Conversely, any $(d-1)$-coloring of $\Gamma(H)[V\setminus S]$ combined with a $d$th color on $S$ is a strong $d$-coloring because each context meets $S$ once and forms a $(d-1)$-clique in $V\setminus S$; its color indicators yield (3). Conversely, $d$ states with disjoint supports $S_1,\ldots,S_d$ partitioning $V$ define the strong $d$-coloring $c(v)=k$ for $v\in S_k$, since each context meets each $S_k$ exactly once; hence (3)$\implies$(1). Trivially (3)$\implies$(4), and given $d-1$ disjoint state supports $S_1,\ldots,S_{d-1}$, the complement $s_d(v)=1-\sum_{k=1}^{d-1}s_k(v)$ is binary and sums to $d-(d-1)=1$ on every context, giving (3). For $d=3$, $\chi(\Gamma(H)[V\setminus S])\le 2$ holds if and only if $\Gamma(H)[V\setminus S]$ is bipartite (odd-cycle free).
\end{proof}

Applying the equivalence (3)$\iff$(4) of Theorem~\ref{thm:completion} globally without anchoring a state $s$ yields the characterization of chromatic completeness:

\begin{lemma}[Disjoint states and strong coloring]\label{lem:disjoint-states}
A finite $d$-uniform hypergraph $H$ satisfies $\chi(H)=d$ if and only if $\Sc(H)$ contains $d-1$ states with pairwise disjoint supports.
\end{lemma}

\begin{corollary}[The separation hierarchy]\label{cor:hierarchy}
For chromatically complete $H$ and $u,v\in V$,
\begin{equation}\label{eq:kernel}
u\sim_c v \quad\Longleftrightarrow\quad s(u)=s(v)\text{ for all }s\in\Sc_{\mathrm{ext}}(H).
\end{equation}
Thus inseparability by $\Sc(H)$ implies chromatic inseparability ($\sim_c$), equivalently chromatic separation implies separation by $\Sc(H)$, whereas the converse fails even for faithful real ray configurations. By Theorem~0 of Kochen and Specker~\cite{kochen1}, separation of all outcomes by $\Sc(H)$ is what allows an embedding into a Boolean algebra; chromatic inseparability is the same failure of separation, with $\Sc(H)$ replaced by the states that noncontextual eigenvalue assignments produce. Moreover, $\Sc_{\mathrm{ext}}(H)$ is always unital whenever $\Cc_d(H)\ne\varnothing$.
\end{corollary}
\begin{proof}
The indicator states $s_a(v)=\mathbf 1_{\{c(v)=a\}}$ over all $c\in\Cc_d(H)$ and $a\in\{0,\ldots,d-1\}$ are precisely $\Sc_{\mathrm{ext}}(H)\subseteq\Sc(H)$. For each fixed coloring, the $d$ supports partition $V$. The inclusion need not be strict: for a single $d$-element context every state is extendable. Strict inclusion is an additional property established for the specified examples.
\end{proof}

\begin{proposition}[Canonical chromatic quotient and closure test]\label{prop:quotient}\label{prop:closure}
Let $H$ be chromatically complete.
\begin{enumerate}
\item The relation $\sim_c$ is an equivalence relation with no two equivalent vertices in a common context. The quotient $\overline H=(V/{\sim_c},\overline E)$ is $d$-uniform and chromatically separating, and pullback along $V\to V/{\sim_c}$ induces a bijection $\Cc_d(\overline H)\to\Cc_d(H)$.
\item For $u\ne v$, adjoining one context $\{u,v,w_1,\ldots,w_{d-2}\}$ with fresh vertices $w_i$ yields chromatic number $d+1$ if and only if $u\sim_c v$, and $d$ otherwise. Since recoloring one vertex with a new color separates any pair, the \emph{separable chromatic number} $\chi_{\mathrm{sep}}$ of a finite graph, the least $t$ for which some set of proper colorings with at most $t$ colors separates every pair of distinct vertices~\cite{svozil-2020-c}, satisfies $\chi_{\mathrm{sep}}\in\{\chi,\chi+1\}$.
\end{enumerate}
\end{proposition}
\begin{proof}
(1) Adjacent vertices receive distinct colors, so each context projects injectively to $V/{\sim_c}$, and every $c\in\Cc_d(H)$ factors uniquely through $V/{\sim_c}$. (2) A coloring $c\in\Cc_d(H)$ extends to the fresh context if and only if $c(u)\ne c(v)$; if $u\sim_c v$, recoloring $u$ with a $(d+1)$th color separates $(u,v)$ while preserving all old context inequalities.
\end{proof}

The quotient preserves colorings, not an injective realization. Closing an inseparable pair is likewise abstract: faithful coordinates make the pair nonorthogonal, so the added context cannot retain those rays.

Put $\mathcal P_{\mathrm{ext}}=\conv\mathcal S_{\mathrm{ext}}$ and $\mathcal P=\conv\mathcal S$. Thus $\mathcal P$ is the set of probability vectors of noncontextual truth-assignment models, and $\mathcal P_{\mathrm{ext}}$ is the subset produced by models that choose a noncontextual eigenvalue assignment $c$ and an eigenvalue $\lambda_a$ at random and report, in every context, the outcome with eigenvalue $\lambda_a$. A binary vector cannot be a convex combination of distinct binary vectors, since each positive-weight term must agree at every coordinate. Consequently,
\begin{equation}\label{eq:polytope-equality}
\mathcal P_{\mathrm{ext}}=\mathcal P\quad\Longleftrightarrow\quad\mathcal S_{\mathrm{ext}}=\mathcal S.
\end{equation}
This is stronger than equality of vertex-separation kernels: the pentagon has separating colors yet a nonextendable state. Pair identities alone do not describe all completion constraints.

\subsection{Fractional coloring and fresh-vertex context completion}

Write $\alpha(G)$ for the independence number and $\chi_f(G)$ for the fractional chromatic number of $G=\Gamma(H)$.

\begin{proposition}[Fractional threshold and uniform state mixture]\label{prop:fractional-uniform}
For any finite $d$-uniform context hypergraph $H$,
\begin{equation}\label{eq:fractional-uniform}
\chi_f(\Gamma(H))=d \quad\Longleftrightarrow\quad \frac1d\mathbf 1\in\mathcal P(H)=\conv\Sc(H).
\end{equation}
\end{proposition}
\begin{proof}
A state mixture $\sum_{s\in\Sc(H)}\mu_s s=\frac1d\mathbf 1$ assigns weight $d\mu_s$ to the independent set $\supp s$, covering each vertex with weight $1$ at total weight $d$, so $\chi_f(\Gamma(H))\le d$; since every context is a $d$-clique, $\chi_f(\Gamma(H))\ge\omega(\Gamma(H))\ge d$. Conversely, a fractional coloring $\{w_I\}$ of total weight $d$ obeys $d\le\sum_{v\in e}\sum_{I\ni v}w_I=\sum_I w_I|I\cap e|\le\sum_I w_I=d$ for each $e\in E$, forcing $|I\cap e|=1$ whenever $w_I>0$, so $\{w_I/d\}$ is a state mixture equal to $\frac1d\mathbf 1$.
\end{proof}

Equation~\eqref{eq:fractional-uniform} restates, for context states, the fractional-chromatic criterion of Ramanathan and Horodecki~\cite{ramanathan2014sic} for the contextuality of the maximally mixed state.

\begin{proposition}[Fresh-vertex context completion]\label{prop:completion}
Let $G$ be a finite core graph with designated $d$-clique base contexts ($d\ge3$). Complete selected uncovered edges $\{u,v\}$ into $d$-cliques using $d-2$ degree-one fresh vertices per edge, covering all core vertices and edges. For the resulting hypergraph $H$ and $G'=\Gamma(H)$,
\begin{equation}\label{eq:completion-chi}
\begin{aligned}
\chi(G')&=\max\{d,\chi(G)\},\\
\chi_f(G')&=\max\{d,\chi_f(G)\},\\
|\Sc(H)|&=\sum_{S}(d-2)^{m(S)},
\end{aligned}
\end{equation}
where $S$ ranges over independent core supports meeting each base context once and $m(S)$ counts completed edges missed by $S$.
\end{proposition}
\begin{proof}
Any $k$-coloring ($k\ge d$) or measurable fractional palette of size $k=\max\{d,\chi_f(G)\}$ assigns disjoint colors/sets of measure $1$ to $u$ and $v$, leaving $k-2\ge d-2$ colors (or measure $k-2\ge d-2$) for the $d-2$ fresh vertices. A core support $S$ hitting $\{u,v\}$ forces its fresh vertices to $0$, whereas a missed edge allows $d-2$ choices of a single fresh vertex set to $1$.
\end{proof}

Table~\ref{tab:joint-comparison} summarizes the principal configurations studied below; those introduced in this paper are marked by a dagger.

\begin{table*}[t]
\centering\small\setlength{\tabcolsep}{3pt}
\caption{Principal configurations. FOR means injective and edge-faithful; $g$ is Berge girth. Counts refer to all context states and all color-class states. A dagger ($\dagger$) marks configurations introduced in this paper; $H_{37}$ is a restriction of the MUB inventory of Ref.~\cite{svozil-2025-MYOCHS}, and $H_{\YO}^{(4)}$ completes the $14$-ray graph of Ref.~\cite[Thm.~1]{cabello2015sic} to contexts.}\label{tab:joint-comparison}
\begin{tabular}{p{.17\textwidth}rrrrrrrp{.35\textwidth}}\toprule
Configuration & $d$ & $|V|/|E|$ & $g$ & $\chi$ & $\chi_f$ & $|\Sc|$ & $|\Sc_{\mathrm{ext}}|$ & Geometry and principal distinction\\\midrule
One triad & 3 & 3/1 & $\infty$ & 3 & $3$ & 3 & 3 & Real FOR; colors separate.\\
$H_{6,3}$ & 3 & 6/3 & $3$ & 3 & $3$ & 4 & 3 & No injective complex realization.\\
$P_5$ & 3 & 10/5 & $5$ & 3 & $3$ & 11 & 10 & Real FOR; colors separate; ternary exclusion.\\
$H_8$ & 3 & 14/8 & $5$ & 3 & $3$ & 11 & 7 & Repeated-ray assignments; no complex injection.\\
$H_9$ & 3 & 15/9 & $5$ & 4 & $3$ & 8 & 0 & No complex ray assignment; context-critical.\\
$G_{32}$ & 3 & 15/10 & $5$ & 4 & $3$ & 6 & 0 & No complex ray assignment.\\
$H_{17,10}^{\dagger}$ & 3 & 17/10 & $5$ & 3 & $3$ & 18 & 10 & Real FOR; two color-identical pairs.\\
$H_{21,13}^{\dagger}$ & 3 & 21/13 & $5$ & 3 & $3$ & 21 & 6 & Real FOR; 21 color-identical pairs.\\
Completed Yu--Oh & 3 & 25/16 & $5$ & 4 & $35/11$ & 24 & 0 & Real FOR; context-critical.\\
$R_{25,15}^{\dagger}$ & 3 & 25/15 & $5$ & 3 & $3$ & 81 & 48 & Complex FOR; no real injection; $3$-colorable part of $H_{37}$.\\
$H_{37}^{\dagger}$ & 3 & 37/24 & $5$ & 4 & $3$ & 153 & 0 & Complex FOR; no real injection; OML minima $6/7/18$.\\
$H_{\YO}^{(4)\dagger}$ & 4 & 26/16 & $2$ & 5 & $46/11$ & 25 & 0 & Real FOR; common-apex suspension, nonlinear.\\
$C_{36,18}^{\dagger}$ & 4 & 36/18 & $3$ & 5 & $9/2$ & 0 & 0 & Faithful linear KS completion.\\
$H_{65}^{\dagger}$ & 4 & 65/29 & $3$ & 5 & $17/4$ & 393,216 & 0 & Real FOR; false atom order and compound collapse.\\
\bottomrule\end{tabular}
\end{table*}

The logical situation has two distinct branches. If a hypergraph has no admissible two-valued context state, it cannot admit a strong $d$-coloring, because the color classes of such a coloring would supply $d$ admissible states. The converse implication fails: a noncontextual assignment of $d$ eigenvalues to the outcomes may be impossible while ordinary states, including a unital separating family, remain available.

When a noncontextual eigenvalue assignment does exist, its color classes provide a nonempty unital family of extendable states inside the full context-state space. A particular ordinary state belongs to this restricted family precisely when deleting its support leaves a $(d-1)$-colorable graph. Chromatic identities describe pairs having equal eigenvalues in every noncontextual eigenvalue assignment, and hence equal values in every state of that restricted family; they need not hold throughout the ordinary state space. These assignment questions have to be distinguished from the separate requirements of an injective geometric realization, a representation of event order, and a classical mixture reproducing a specified Born probability vector.

\section{Independence of Coloring and Coordinatization, Petersen Obstructions, and Completed Yu--Oh}\label{sec:independence}

\subsection{Small abstract examples and four-quadrant independence}

The non-linear hypergraph $H_{4,2}$ with contexts $\{1,2,3\}$ and $\{1,2,4\}$ has three states $\{1\}$, $\{2\}$, $\{3,4\}$ (all extendable), forces $3\sim_c 4$ and $s(3)=s(4)$, and has no injective $\C^3$ realization because two orthogonal rays determine a unique orthogonal-complement ray.

\begin{proposition}[Smallest linear chromatically inseparable hypergraph]\label{prop:six}
The linear $3$-uniform hypergraph $H_{6,3}$ with contexts $\{1,2,3\}$, $\{3,4,5\}$, $\{5,6,1\}$ of Berge girth three has four unital separating states $\{1,4\}$, $\{2,5\}$, $\{3,6\}$, $\{2,4,6\}$, a unique strong $3$-coloring up to permutation with inseparable classes $\{1,4\}$, $\{2,5\}$, $\{3,6\}$, and three extendable states (as $V\setminus\{2,4,6\}=\{1,3,5\}$ is a triangle). No linear $3$-uniform hypergraph with $|V|<6$ or $|E|<3$ has a chromatically inseparable pair, and $H_{6,3}$ admits no injective realization in $\C^3$.
\end{proposition}
\begin{proof}
The triangle $\{1,3,5\}$ in $\Gamma(H_{6,3})$ forces three distinct colors on $1,3,5$, whence $c(4)=c(1)$, $c(2)=c(5)$, and $c(6)=c(3)$. Similarly, in $\C^3$, $\{v_1,v_3,v_5\}$ forms an orthonormal basis, forcing $[v_4]=[v_1]$, $[v_2]=[v_5]$, and $[v_6]=[v_3]$. With $|E|\le2$, independent permutations of unshared vertices separate all pairs; by inclusion--exclusion, three linear $3$-contexts require $|V|\ge 6$.
\end{proof}

\begin{proposition}[Clique number at Berge girth at least four]\label{m:prop:ranks}
If $H=(V,E)$ is an $n$-uniform linear hypergraph ($n\ge3$) of Berge girth $g(H)\ge4$, every clique $K\subseteq V$ of size $|K|\ge3$ in $\Gamma(H)$ is contained in a unique context $e\in E$. Hence $\omega(\Gamma(H))=n$, and chromatic incompleteness ($\chi(H)>n$) is equivalent to the chromatic--clique gap $\chi(\Gamma(H))>\omega(\Gamma(H))$.
\end{proposition}
\begin{proof}
For any three vertices $u,v,w\in K$, pairwise adjacency gives contexts $e_1\supseteq\{u,v\}$, $e_2\supseteq\{v,w\}$, $e_3\supseteq\{w,u\}$. Distinct $e_1,e_2,e_3$ would form a Berge $3$-cycle, contradicting $g(H)\ge4$. By linearity, the shared context $e$ containing $\{u,v,w\}$ is unique for $\{u,v\}$ and thus contains all of $K$.
\end{proof}

\begin{proposition}[Logical independence of coloring and coordinatization]\label{m:prop:independence}
For finite $3$-uniform context hypergraphs, all four combinations of chromatic completeness ($\chi(H)=3$) and injective coordinatizability in $\C^3$ occur even within the orthomodular-lattice admissible class of linear hypergraphs of Berge girth $g(H)\ge5$~\cite{greechie:71,svozil-tkadlec} with unital separating state spaces:
\begin{enumerate}
\item $\chi=3$, injectively realizable in $\R^3$: a single context $C_1=\{\{1,2,3\}\}$ (or $P_5$, $H_{17,10}$, $H_{21,13}$);
\item $\chi=3$, not injectively realizable in $\C^3$: the $8$-context Petersen restriction $H_8$ (Theorem~\ref{m:thm:multipole});
\item $\chi=4$, edge-faithfully realizable in $\R^3$: completed Yu--Oh $H_{\YO}^{+}$ (Theorem~\ref{m:thm:yo-criticality});
\item $\chi=4$, admitting no orthogonal ray assignment in $\C^3$: Greechie's $G_{32}$ and $H_9$ (Theorem~\ref{m:thm:multipole}).
\end{enumerate}
\end{proposition}
Injectivity is essential to quadrant~(2): any strong $n$-coloring $c$ yields a noninjective ray assignment $i\mapsto[e_{c(i)+1}]\in\PP(\F^n)$. At Berge girth four (an orthomodular poset that is not a lattice), $E_\square=\{\{1,2,3\}$, $\{3,4,5\}$, $\{5,6,7\}$, $\{7,8,1\}\}$ also satisfies $\chi=3$ while forcing $[v_3]=[v_7]\perp\spanop(v_1,v_5)$ whenever $[v_1]\ne[v_5]$.

\subsection{Greechie's \texorpdfstring{$G_{32}$}{G32}, multipole forcing, and correction of the JMP spectral argument}\label{m:sec:g32}

Greechie's $G_{32}$~\cite[Fig.~6]{greechie:71} has $15$ vertices and ten contexts
\begin{equation}\label{m:eq:g32-blocks}
\begin{aligned}
E(G_{32})=\{&123,345,567,789,9\,10\,11,\\
&11\,12\,1,4\,10\,13,6\,12\,14,\\
&8\,2\,15,13\,14\,15\}.
\end{aligned}
\end{equation}
Each vertex lies in two contexts, so $G_{32}$ is dual to the Petersen graph $P$, $\Gamma(G_{32})=L(P)$, $\mathrm{Aut}(G_{32})\cong S_5$, and $\chi(G_{32})=\chi'(P)=4$. Its six two-valued state supports satisfy
\begin{equation}\label{m:eq:g32-states}
|\mathcal S(G_{32})|=6,\qquad |S_\alpha|=5,\qquad
|S_\alpha\cap S_\beta|=1\ (\alpha\ne\beta).
\end{equation}
These supports yield the $\binom{6}{2}=15$ two-element vertex profiles $B_i=\{\alpha:i\in S_\alpha\}$. Their nonempty pairwise intersections exclude even two disjoint states, confirming $\chi(G_{32})=4$ by Lemma~\ref{lem:disjoint-states}.

In Appendix~B of Shekarriz and Svozil~\cite{svozil-2021-chroma}, the chromatic obstruction $\chi(G_{32})=4>3$ was invoked to argue against a $3$-dimensional realization of $G_{32}$. As Proposition~\ref{m:prop:independence} demonstrates via $H_{\YO}^{+}$ ($\chi=4$, yet faithfully realizable in $\R^3$), $\chi(H)>3$ rules out only a context-independent $3$-value spectral assignment---a noncontextual assignment of the eigenvalues of uniform maximal observables to their outcomes---\emph{not} geometric rays. A genuine geometric exclusion requires the projection identity: for an orthonormal basis $\{x,y,z\}\subset\C^3$ with $u\perp x$ and $w\perp z$,
\begin{equation}\label{m:eq:projection-identity}
\ip{u}{w}=\ip{u}{y}\ip{y}{w}.
\end{equation}

\begin{theorem}[Multipole forcing and minimal complex obstructions]\label{m:thm:multipole}\label{m:thm:g32}
Let $H_8=G_{32}\setminus\{\{8,2,15\}$, $\{13,14,15\}\}$ ($14$ vertices, $8$ contexts) and $H_9=G_{32}\setminus\{\{13,14,15\}\}$ ($15$ vertices, $9$ contexts).
\begin{enumerate}
\item In every orthogonal ray assignment of $H_8$ in $\C^3$ (allowing repeated rays), the terminal rays satisfy
\begin{equation}\label{m:eq:multipole-forcing}
[v_2]=[v_8]\qquad\text{and}\qquad [v_{13}]=[v_{14}],
\end{equation}
and either $[v_1]=[v_4]=[v_9]$ or $[v_3]=[v_7]=[v_{12}]$. Thus $H_8$ has no injective realization in $\C^3$.
\item $H_8$ is chromatically complete ($\chi(H_8)=3$) with $11$ unital separating states, of which $7$ are extendable (the $4$ nonextendable states leave odd cycles in $\Gamma(H_8)$; see Table~\ref{S-tab:h8-states}). Its $3$ normalized $3$-colorings (fixing $(c(1),c(2),c(3))=(0,1,2)$) are listed in Table~\ref{S-tab:h8-colors}, forcing the $6$ chromatically inseparable pairs $(1,9)$, $(2,8)$, $(3,7)$, $(4,6)$, $(10,12)$, $(13,14)$; deleting any single context of $H_8$ removes all forced color identities (leaving $6,20,20,6,20,20,6,6$ normalized colorings on the remaining active vertices, up to global color permutation). Isolated vertices are dropped: retaining the isolated vertex created by deleting $\{1,2,3\}$, $\{7,8,9\}$, $\{4,10,13\}$, or $\{6,12,14\}$ multiplies the corresponding count by three, from $6$ to $18$.
\item Neither $H_9$, $G_{32}$, nor the $12$-vertex, $21$-edge subgraph $\Gamma_0$ of $\Gamma(H_9)$ on $\{1,\ldots,12\}$ admits any orthogonal ray assignment in $\C^3$. Moreover, $H_9$ is context-minimal and $\Gamma_0$ is edge-minimal: deleting any context of $H_9$ or edge of $\Gamma_0$ yields a $3$-colorable structure realizable in $\R^3$ by standard basis rays.
\end{enumerate}
\end{theorem}
\begin{proof}
(1) Let $v_1,\ldots,v_{14}\in\C^3$ be unit vectors for $H_8$. In $\{v_1,v_2,v_3\}$, $v_{11},v_{12}\perp v_1$ and $v_4,v_5\perp v_3$, so Eq.~\eqref{m:eq:projection-identity} gives $\ip{x}{y}=\ip{x}{v_2}\ip{v_2}{y}$ for $x\in\{v_{11},v_{12}\},y\in\{v_4,v_5\}$. Expanding $v_{12}$ in $\{v_5,v_6,v_7\}$ ($v_{12}\perp v_6$) and $v_4$ in $\{v_9,v_{10},v_{11}\}$ ($v_4\perp v_{10}$) yields $v_7\ip{v_7}{v_{12}}=v_{12}-v_5\ip{v_5}{v_{12}}$ and $v_9\ip{v_9}{v_4}=v_4-v_{11}\ip{v_{11}}{v_4}$. Since $v_7\perp v_9$, $v_{12}\perp v_{11}$, and $v_5\perp v_4$, their inner product vanishes:
\begin{equation}\label{m:eq:coordfree-product}
\begin{aligned}
0={}&\ip{v_{12}}{v_4}
+\ip{v_{12}}{v_5}\ip{v_5}{v_{11}}\ip{v_{11}}{v_4}\\
={}&\ip{v_{12}}{v_2}\ip{v_2}{v_4}
\bigl(1+|\ip{v_2}{v_5}|^2|\ip{v_2}{v_{11}}|^2\bigr).
\end{aligned}
\end{equation}
Thus $\ip{v_2}{v_4}=0$ or $\ip{v_{12}}{v_2}=0$. If $\ip{v_2}{v_4}=0$, set $(v_1,v_3,v_2)=(e_1,e_2,e_3)$; then $v_4\perp v_2,v_3\implies[v_4]=[e_1]$, $\{3,4,5\}\implies[v_5]=[e_3]=[v_2]$, $\{9,10,11\}$ with $v_{10},v_{11}\perp e_1$ forces $[v_9]=[e_1]$, $\{5,6,7\}$ and $\{7,8,9\}$ force $[v_7]=[e_2],[v_6]=[e_1],[v_8]=[e_3]=[v_2]$, and $\{4,10,13\}$, $\{9,10,11\}$, $\{11,12,1\}$, $\{6,12,14\}$ in $\spanop(e_2,e_3)$ force $[v_{13}]=[v_{11}]=[v_{14}]$ and $[v_1]=[v_4]=[v_9]$. Reflection symmetry gives the second branch when $\ip{v_{12}}{v_2}=0$.

(2) In $H_8$, fixing $(c(1),c(2),c(3))=(0,1,2)$ leaves $16$ choices along $345,567,789,9\,10\,11$, of which $13$ fail on $11\,12\,1,4\,10\,13,6\,12\,14$, leaving $c_1,c_2,c_3$ in Table~\ref{S-tab:h8-colors}. Equal columns identify the $6$ pairs.

(3) Both $H_9$ (hence $G_{32}$) and $\Gamma_0$ add the orthogonality $v_2\perp v_8$ to $H_8$, contradicting $[v_2]=[v_8]$. For minimality, deleting $\{8,2,15\}$ from $H_9$ leaves $H_8$ ($\chi=3$); deleting $\{4,10,13\}$ or $\{6,12,14\}$ is isomorphic by $D_6$ symmetry to deleting $\{8,2,15\}$; and deleting any ring context (say $\{1,2,3\}$) allows the $3$-coloring $(c(1),\ldots,c(15))=(2,1,2,0,1,0,2,0,1,2,0,1,1,2,2)$. On $\Gamma_0$, the following deletions cover all $21$ edges up to $D_6$ symmetry: $\{1,2\}$ (colored by $001020120121$), $\{2,3\}$ ($011020120121$), $\{1,3\}$ ($010120102012$), and $\{4,10\}$ ($012010201021$), each yielding a standard-basis $\R^3$ assignment $i\mapsto[e_{c(i)+1}]$. For the minimum $\Delta_3$ of $\sum|\ip{v_i}{v_j}|^2$ over the $30$ edges $\{i,j\}$ of $\Gamma(G_{32})$ and all unit vectors $v_1,\ldots,v_{15}\in\C^3$, the robustness bounds $1/14400<\Delta_3<0.927496$ are proved in Sec.~S3 of the supplementary material.
\end{proof}

In his proof that state-independent contextuality in $\C^3$ requires $13$ rays, Cabello~\cite{cabello2012maxnonlocality} found $\Gamma_0$ to be the only square-free graph with fewer than $13$ vertices and chromatic number above $3$, and excluded it by showing, in real coordinates, that ten of its vertices (in our labels, all but $2$ and $8$) cannot be distinct rays. Part~(1) proves this over $\C^3$: the rays $2,8,13,14$ of $H_8$ are then fixed as orthogonal completions, and the ten rays satisfy $[v_1]=[v_4]=[v_9]$ or $[v_3]=[v_7]=[v_{12}]$. Part~(3) shows more: $\Gamma_0$ admits no orthogonal ray assignment in $\C^3$ at all, even with repeated rays.

The complete normalized coloring lists for $H_8$ and $H_{17,10}$ are in the supplement.

\begin{corollary}[Blanu\v{s}a and Isaacs dot-product snarks]\label{m:cor:snarks}
If a cubic graph $G$ admits no orthogonal vector edge coloring in $\C^3$, neither does any Isaacs dot product $G\cdot P$~\cite{isaacs1975infinite}, obtained by deleting two nonadjacent edges of $G$ and two adjacent vertices $u,v$ of the Petersen graph $P$, and joining the two dangling edges of $u$ to the ends of one deleted edge and those of $v$ to the ends of the other. In particular, both $18$-vertex Blanu\v{s}a snarks $B_1,B_2=P\cdot P$~\cite{blanusa1946problem} and all iterated dot products $X=P\cdot P\cdots P$ satisfy $\chi(H(X))=4>3=\omega(L(X))$ and admit no orthogonal ray assignment in $\C^3$.
\end{corollary}
\begin{proof}
In $G\cdot P$, the embedded copy of $P\setminus\{u,v\}$ has dual hypergraph $H_8$, so Theorem~\ref{m:thm:multipole}(1) forces $[v_2]=[v_8]$ and $[v_{13}]=[v_{14}]$ on the two pairs of cut edges, inducing a valid orthogonal ray assignment on $G$.
\end{proof}

Silva and Ibiapina~\cite[Question~6]{silva2019orthogonal} asked whether a cubic graph with chromatic index $4$ can have an orthogonal vector edge coloring in $\R^3$. For the Petersen graph, the Blanu\v{s}a snarks, and these dot products the answer is no, even in $\C^3$.

\subsection{Completed Yu--Oh \texorpdfstring{$H_{\YO}^{+}$}{HYO+}, base-coloring partition, and context criticality}\label{m:sec:yuo}\label{m:sec:ranks}

Completing each of the $12$ uncovered orthogonal pairs of the $13$-ray Yu--Oh set~\cite{Yu-2012} ($r_1,\ldots,r_{13}\in\R^3$, consisting of $z_1=r_2,z_2=r_8,z_3=r_5$, $y_1^\pm=r_{1,3},y_2^\pm=r_{7,9},y_3^\pm=r_{4,6}$, and diagonals $h_0=r_{12},h_1=r_{11},h_2=r_{10},h_3=r_{13}$) by cross-product rays $r_{14},\ldots,r_{25}$ (Table~\ref{S-tab:yo-core-vectors}) yields the $25$-ray, $16$-context linear hypergraph $H_{\YO}^{+}$ of Berge girth five~\cite{Pavi_i__2019,svozil-2025-color}:
\begin{equation}\label{m:eq:yuo-blocks}
\begin{aligned}
E_{\YO}^{+}&=\{123,789,456,285\}\cup\bigcup_{h=10}^{13}\mathcal E_h,\\
\mathcal E_h&=\{\{x,h,w_{xh}\}:x\in N(h)\},
\end{aligned}
\end{equation}

All 25 vectors, 16 contexts, and restriction-membership flags are tabulated in the supplement.
Here the neighborhoods are $N(10)=\{1,9,4\}$, $N(11)=\{3,7,4\}$, $N(12)=\{3,9,6\}$ and $N(13)=\{1,7,6\}$, and each $w_{xh}$ is the unique completing ray. It has a real FOR in $\R^3$ ($48$ orthogonal pairs) and $24$ unital separating states (each of the $12$ binary assignments on the four base contexts leaves uniquely one diagonal $h_j$ with all three neighbors zero, allowing $s(h_j)\in\{0,1\}$ and hence $12+4\times3=24$ states). Completed Yu--Oh is therefore a faithfully realized real configuration in which noncontextual truth assignments survive, enough of them to separate every pair of rays, while by the following theorem no noncontextual assignment of three eigenvalues to its outcomes exists. That $\chi(H_{\YO}^{+})>3$ and that its $24$ states separate all rays were shown by a case analysis in Ref.~\cite{svozil-2025-color}; Theorem~\ref{m:thm:yo-criticality} replaces that case analysis by a partition of the base colorings and adds context criticality.

\begin{theorem}[Base-coloring partition and context criticality of $H_{\YO}^{+}$]\label{m:thm:yo-criticality}\label{m:prop:fractional}
The completed Yu--Oh hypergraph $H_{\YO}^{+}$ satisfies $\omega(\Gamma(H_{\YO}^{+}))=3$, $\chi_f(\Gamma(H_{\YO}^{+}))=35/11>3=\chi_f(\Gamma(G_{32}))$, and $\chi(H_{\YO}^{+})=4$, and is context-critical for $\chi>3$:
\begin{enumerate}
\item Fixing $(c(2),c(8),c(5))=(0,1,2)$, the four base contexts $\{1,2,3\}$, $\{7,8,9\}$, $\{4,5,6\}$, $\{2,8,5\}$ admit eight strong $3$-colorings indexed by $(b_1,b_2,b_3)\in\{0,1\}^3$:
\begin{equation}\label{m:eq:yo-base-colorings}
\begin{aligned}
(c(1),c(3))&=(2-b_1,1+b_1),\\
(c(7),c(9))&=(2-2b_2,2b_2),\\
(c(4),c(6))&=(1-b_3,b_3).
\end{aligned}
\end{equation}
Each antipodal pair of $(b_1,b_2,b_3)$ makes the neighborhood $N(h_j)$ of \emph{uniquely one} diagonal ray $h_j\in\{10,11,12,13\}$ rainbow ($\{0,1,2\}$) while the other three neighborhoods are $2$-colored:
\begin{equation}\label{m:eq:yo-antipodal-pairs}
\begin{array}{c|c}
\text{Antipodal binary patterns}&\text{Blocked diagonal}\\\hline
101,010&12\\110,001&11\\000,111&10\\011,100&13
\end{array}
\end{equation}
\item Consequently, $\chi(H_{\YO}^{+})=4$, whereas deleting any single context $e\in E_{\YO}^{+}$ yields $\chi(H_{\YO}^{+}\setminus\{e\})=3$.
\end{enumerate}
\end{theorem}
\begin{proof}
Direct check of Eqs.~\eqref{m:eq:yo-base-colorings}--\eqref{m:eq:yo-antipodal-pairs} shows that every base $3$-coloring blocks one $h_j$, so $\chi(H_{\YO}^{+})\ge4$, matched by the four-coloring
\begin{align*}
(c(1),\ldots,c(13))&=(2,0,1,1,2,0,2,1,0,3,0,2,1),\\
(c(14),\ldots,c(25))&=(2,1,1,1,0,0,0,0,2,1,2,0),
\end{align*}
in which only $h_2=10$ receives the fourth color. Deleting a sector context $\{x,h_j,w\}$ frees $h_j$ in its antipodal base colorings by setting $c(h_j)=c(x)$; deleting $\{2,8,5\}$ allows $c(2)=c(8)=c(5)=0$; deleting $\{1,2,3\}$ allows $(c(1),c(2),c(3))=(0,0,1)$, $(c(7),c(8),c(9))=(2,1,0)$, $(c(4),c(5),c(6))=(1,2,0)$, and $(c(10),c(11),c(12),c(13))=(2,0,2,1)$; the cases $\{7,8,9\}$ and $\{4,5,6\}$ follow by the threefold cyclic coordinate symmetry of the Yu--Oh set. For $\chi_f$, weight $1/2$ on the $6$ states of $G_{32}$ gives $\chi_f(G_{32})=3$, while Proposition~\ref{prop:completion} preserves the $13$-ray value $\chi_f=35/11$~\cite{cabello2015sic}.
\end{proof}

\section{Exact-Palette Chromatic Separability and Logical Gadgets in Dimensions Three and Four}\label{sec:logical}\label{sec:geometric}

From now on we assume that noncontextual assignments of $d$ eigenvalues exist ($\chi(H)=d$) and ask what they imply for the relations between distinct outcomes. Such assignments may force outcomes to share or to avoid an eigenvalue although ordinary truth assignments permit otherwise; these forced relations are their logical content.

\subsection{Kochen--Specker antecedents and palette-invariant implications}\label{sec:historical}

Three historical Kochen--Specker constructions clarify the hierarchy of logical gadgets:
\begin{enumerate}
\item \emph{1965 abstract TIFS precursor:} Identifying repeated labels $v_1$ and $v_2$ in Kochen and Specker's 1965 Fig.~1~\cite[p.~182]{kochen2} yields an $11$-vertex, $7$-context hypergraph $\{134,325,2\,9\,11,11\,1\,10,479,5\,8\,10,678\}$ implementing $s(1)=1\implies s(2)=0$. It has only four states ($\{1,5,6,9\}$, $\{2,4,6,10\}$, $\{3,6,9,10\}$, $\{4,5,6,11\}$), all with $s(6)=1$ and $s(7)=s(8)=0$, so $\chi>3$; moreover, $1$ and $9$ are both orthogonal to $4$ and $11$, so the distinct rays $4$ and $11$ of an injective realization in $\C^3$ would force $[v_1]=[v_9]$; no injective realization exists (Sec.~S3 of the supplementary material).
\item \emph{1967 directed TITS $\Gamma_1$:} The $10$-vertex graph $\Gamma_1$ of Lemma~1 in Ref.~\cite[pp.~68, 70]{kochen1} implements the one-way implication $s(a_0)=1\implies s(a_9)=1$. Completing its six uncompleted edges by fresh rays $23,19,18,15,14,5$ yields the $16$-ray, $9$-context configuration $I_{16,9}$ in Fig.~\ref{fig:3dgadgets}(d) (explicit label map in Sec.~S3 of the supplementary material).
\item \emph{1967 state-inseparable $\Gamma_3$:} In Ref.~\cite[p.~70]{kochen1}, $\Gamma_3$ forces $s(a)=s(b)$ in \emph{every} two-valued state ($a,b$ are nonseparated by $\Sc(H)$). By contrast, our chromatically inseparable examples below have $\Sc(H)$ separating all vertices, with $c(u)=c(v)$ forced only on $\Sc_{\mathrm{ext}}(H)$.
\end{enumerate}

For distinct nonadjacent terminals $u,v\in V$, an ordinary TIFS~\cite{2018-minimalYIYS} obeys $s(u)=1\implies s(v)=0$ (equivalently $s(u)+s(v)\le1$, which is always bidirectional), and a directed ordinary TITS obeys $s(u)=1\implies s(v)=1$. When $\Cc_d(H)\ne\varnothing$, an \emph{unanchored chromatic TIFS} obeys $c(u)=0\implies c(v)\in\{1,\ldots,d-1\}$, and an \emph{unanchored chromatic TITS} obeys $c(u)=0\implies c(v)=0$.

\begin{proposition}[Palette symmetry and purely chromatic gadgets]\label{prop:palette}
For chromatically complete $H$ and nonadjacent $u,v\in V$, the unanchored chromatic TIFS is equivalent to $c(u)\ne c(v)$ for all $c\in\Cc_d(H)$ (an \emph{implicit edge}~\cite[Def.~3.1]{martin2009}), and the unanchored chromatic TITS is equivalent to $u\sim_c v$ (an \emph{implicit identity}). Both chromatic relations are bidirectional. Ordinary TIFS and directed ordinary TITS entail their chromatic counterparts, and with only $c(u)=0$ fixed, the allowed output set $\{c(v):c\in\Cc_d(H),\ c(u)=0\}$ is one of $\{0\}$, $\{1,\ldots,d-1\}$, or $\{0,\ldots,d-1\}$.
\end{proposition}
\begin{proof}
Global palette permutations act transitively on $\{0,\ldots,d-1\}$ and, stabilizing $0$, transitively on $\{1,\ldots,d-1\}$. Since every color class $s_a(x)=\mathbf 1_{\{c(x)=a\}}$ is in $\Sc_{\mathrm{ext}}(H)\subseteq\Sc(H)$, any ordinary TIFS or one-way TITS forces $c(u)\ne c(v)$ or $c(u)=c(v)$, respectively.
\end{proof}

We call a chromatic TIFS \emph{purely chromatic} if some $s\in\Sc(H)\setminus\Sc_{\mathrm{ext}}(H)$ satisfies $s(u)=s(v)=1$, and a chromatic identity $u\sim_c v$ \emph{two-way purely chromatic} if ordinary states in $\Sc(H)\setminus\Sc_{\mathrm{ext}}(H)$ realize both $(s(u),s(v))=(1,0)$ and $(0,1)$, so neither directed ordinary TITS holds. Since a state with $s(u)=s(v)=1$ cannot be a color class when $c(u)\ne c(v)$ for all $c$, and a state with $s(u)\ne s(v)$ cannot be one when $u\sim_c v$, the qualifier $\Sc(H)\setminus\Sc_{\mathrm{ext}}(H)$ is automatic: a chromatic TIFS is purely chromatic exactly when the ordinary TIFS fails, and a chromatic identity is two-way purely chromatic exactly when both directed ordinary TITS fail.

\subsection{Faithful Yu--Oh restrictions \texorpdfstring{$H_{17,10}$ and $H_{21,13}$}{H17,10 and H21,13} and exhaustive inventory search}

The base-coloring partition of $H_{\YO}^{+}$ in Theorem~\ref{m:thm:yo-criticality} directly produces chromatically inseparable faithful restrictions when diagonal sectors are omitted. Retain the four base contexts $\{1,2,3\}$, $\{7,8,9\}$, $\{4,5,6\}$, $\{2,8,5\}$ of Eq.~\eqref{m:eq:yuo-blocks}. Keeping the two diagonal sectors at $10$ and $11$ defines $H_{17,10}$ on $17$ rays and $10$ contexts [Fig.~\ref{fig:context-diagrams}(a)]; keeping the three sectors at $10,11,13$ defines $H_{21,13}$ on $21$ rays and $13$ contexts [Fig.~\ref{fig:context-diagrams}(b)]. Deleting $\{4,5,6\}$, $\{4,11,24\}$, $\{4,10,18\}$ from $H_{17,10}$ recovers the $13$-ray, $7$-context completed Specker bug $B_{13,7}$ [Fig.~\ref{fig:3dgadgets}(b)].

\begin{figure*}[t]
\centering
\renewcommand{\CSf}{1.35}\renewcommand{\CSg}{1.2}\renewcommand{\CSn}{1.5}
\begin{minipage}[t]{.48\linewidth}\centering
\textbf{(a)} $H_{17,10}$ ($17$ rays, $10$ contexts, $2$ inseparable pairs)\\[1mm]
\resizebox{\linewidth}{!}{\CSdrawcontextdiagram{17}}
\end{minipage}\hfill
\begin{minipage}[t]{.48\linewidth}\centering
\textbf{(b)} $H_{21,13}$ ($21$ rays, $13$ contexts, $21$ inseparable pairs)\\[1mm]
\resizebox{\linewidth}{!}{\CSdrawcontextdiagram{21}}
\end{minipage}
\caption{Faithful three-dimensional Yu--Oh restrictions in $\R^3$, drawn in the common vertex layout of completed Yu--Oh. (a) $H_{17,10}$, omitting diagonal sectors $h_0=12$ and $h_3=13$, with four normalized strong $3$-colorings and chromatically inseparable pairs $10\sim_c 24$ and $11\sim_c 18$; (b) $H_{21,13}$, omitting only $h_0=12$, with two normalized strong $3$-colorings and $21$ chromatically inseparable pairs across nine chromatic signature classes.}\label{fig:context-diagrams}
\end{figure*}

\begin{theorem}[Chromatic inseparability in $H_{17,10}$ and $H_{21,13}$]\label{thm:h17}\label{thm:h21}
Both $H_{17,10}$ and $H_{21,13}$ are linear of Berge girth five and inherit injective edge-faithful realizations in $\R^3$ from $H_{\YO}^{+}$:
\begin{enumerate}
\item $H_{17,10}$ has $18$ unital separating states, $4$ normalized strong $3$-colorings (Table~\ref{S-tab:h17-colors}), and $10$ extendable states. Its only distinct chromatically inseparable pairs are $(10,24)$ and $(11,18)$. The ordinary implications $10\Rightarrow24$ and $11\Rightarrow18$ already hold, but their converses fail; these identities are not two-way purely chromatic. Relative to the anchor $(c(2),c(8),c(5))=(0,1,2)$, its four colorings restrict on $(9,10,11)$ to $(0,2,1)$, $(0,2,0)$, $(2,1,2)$, $(2,0,2)$, yielding the anchored implications
\begin{equation}\label{eq:anchored}
\begin{aligned}
c(9)=0&\implies c(11)\in\{0,1\},\\
c(10)=0&\implies c(11)=2.
\end{aligned}
\end{equation}
\item $H_{21,13}$ has $21$ unital separating states, $2$ normalized strong $3$-colorings, and $6$ extendable states. Its nine chromatic signature classes $(c_1(v),c_2(v))\in\{0,1,2\}^2$ are listed in Table~\ref{S-tab:chromatic-signatures}, giving $3\binom{4}{2}+3\binom{2}{2}=21$ chromatically inseparable pairs. Exactly $15$ are two-way purely chromatic identities (Table~\ref{S-tab:pure-witnesses}). The remaining six pairs obey the directed ordinary implications
\[
\begin{gathered}
11\Rightarrow18,\quad11\Rightarrow21,\quad
13\Rightarrow15,\\
13\Rightarrow19,\quad10\Rightarrow20,\quad
10\Rightarrow24,
\end{gathered}
\]
each with a false converse. This classification exhausts all $21$ pairs; no distinct vertices are inseparable under all ordinary states.
\end{enumerate}
\end{theorem}
\begin{proof}
By Theorem~\ref{m:thm:yo-criticality}(1), among the $8$ base colorings in Eq.~\eqref{m:eq:yo-base-colorings}, sector $11$ excludes $(1,1,0)$, $(0,0,1)$ and sector $10$ excludes $(0,0,0)$, $(1,1,1)$. The $4$ surviving $(b_1,b_2,b_3)$ extend uniquely to $H_{17,10}$ (Table~\ref{S-tab:h17-colors}): in each, $N(10)=\{1,9,4\}$ and $N(11)=\{3,7,4\}$ are $2$-colored with $c(10)\ne c(4)$ and $c(11)\ne c(4)$, and $c(10)\ne c(11)$ by the restrictions to $(9,10,11)$ listed in~(1), forcing $c(24)=c(10)$ in $\{4,11,24\}$ and $c(18)=c(11)$ in $\{4,10,18\}$. Adding sector $13$ in $H_{21,13}$ additionally excludes $(0,1,1)$, $(1,0,0)$, leaving only $(1,0,1)$ and $(0,1,0)$, whose column signatures are listed in Table~\ref{S-tab:chromatic-signatures}. Each retained diagonal adds $3$ states to the $12$ base states, giving $12+2\times3=18$ and $12+3\times3=21$ states (Sec.~S2 of the supplementary material). Testing both directed truth patterns on these $21$ states gives the exhaustive $15$-versus-$6$ classification; Table~\ref{S-tab:pure-witnesses} of the supplementary material supplies witnesses for all $15$ two-way pairs, together with a reproducible enumeration procedure.
\end{proof}

The full chromatic signatures and both-direction truth witnesses are in the supplement. The witnesses are generated from the complete state families and checked against every context.

Physically, every noncontextual eigenvalue assignment on $H_{17,10}$ gives the nonorthogonal rays $10$ and $24$, and likewise $11$ and $18$, the same eigenvalue, although ordinary truth assignments distinguish them. In $H_{21,13}$, fifteen of the twenty-one forced identities are not even shadowed by a one-way truth implication: they exist only at the level of noncontextual eigenvalues.

\begin{proposition}[Complete inventory classification]\label{prop:search}
Among all $2^{16}-1=65{,}535$ nonempty context subsets $E'\subseteq E(H_{\YO}^{+})$ with inherited coordinates from Tables~\ref{S-tab:yo-all-vectors}--\ref{S-tab:contexts}, exactly $539$ are connected and edge-faithful, and among these \emph{precisely ten} are $3$-colorable with a nontrivial chromatic inseparability class: the $\binom{4}{2}=6$ two-sector restrictions ($17$ rays, $10$ contexts, $18$ states, $4$ normalized colorings, $2$ inseparable pairs) and the $\binom{4}{3}=4$ three-sector restrictions ($21$ rays, $13$ contexts, $21$ states, $2$ normalized colorings, $21$ inseparable pairs). Thus $17$ rays and $10$ contexts are minimal for chromatic inseparability within the faithful restrictions of $H_{\YO}^{+}$ with inherited coordinates.
\end{proposition}
\begin{proof}
Since $H_{\YO}^{+}$ is linear and edge-faithful, a subset $E'$ is edge-faithful if and only if no omitted context contains two vertices of $V'=\bigcup E'$. Exact enumeration by two independent coloring algorithms over all $539$ connected faithful subsets certifies the ten solutions listed in Table~\ref{S-tab:inventory}. Disconnected restrictions separate distinct components by independent palette permutations.
\end{proof}

The ten retained-sector cases and binary search masks are enumerated in the supplement. The restriction to inherited coordinates matters: $I_{16,9}$ of Sec.~\ref{sec:3dgadgets} is $H_{17,10}$ without the context $\{4,5,6\}$; it has $16$ rays and $9$ contexts and keeps both identities, but it is edge-faithful only with refitted coordinates.

\begin{figure*}[t]
\centering
\renewcommand{\CSf}{1.75}\renewcommand{\CSg}{1.15}\renewcommand{\CSn}{1.8}
\begin{minipage}[t]{.32\linewidth}\centering
\textbf{(a)} $P_5$ ($10$ rays, $5$ contexts)\\[1mm]
\resizebox{\linewidth}{!}{\CSdrawPfive}
\end{minipage}\hfill
\begin{minipage}[t]{.32\linewidth}\centering
\textbf{(b)} $B_{13,7}$ ($13$ rays, $7$ contexts)\\[1mm]
\resizebox{\linewidth}{!}{\CSdrawSpeckerBug}
\end{minipage}\hfill
\begin{minipage}[t]{.32\linewidth}\centering
\textbf{(c)} $D_{16,9}$ ($16$ rays, $9$ contexts)\\[1mm]
\resizebox{\linewidth}{!}{\CSdrawPureChromTIFS}
\end{minipage}
\par\vspace{2mm}
\renewcommand{\CSy}{1.3}
\begin{minipage}[t]{.32\linewidth}\centering
\textbf{(d)} $I_{16,9}$ ($16$ rays, $9$ contexts)\\[1mm]
\resizebox{\linewidth}{!}{\CSdrawCompletedKSTITS}
\end{minipage}\hfill
\begin{minipage}[t]{.32\linewidth}\centering
\textbf{(e)} $T_{19,11}$ ($19$ rays, $11$ contexts)\\[1mm]
\resizebox{\linewidth}{!}{\CSdrawPureChromTITS}
\end{minipage}\hfill
\begin{minipage}[t]{.32\linewidth}\centering
\textbf{(f)} $H_{21,13}$ purely chromatic ties\\[1mm]
\resizebox{\linewidth}{!}{\CSdrawHtwentyoneTITS}
\end{minipage}
\caption{Faithful real three-dimensional logical gadgets in $\R^3$. (a) Pentagon $P_5$ with purely chromatic ternary exclusion $\neg(c(6)=c(7)=c(9))$; (b) completed Specker bug $B_{13,7}$ (ordinary TIFS $c(10)\ne c(11)$); (c) purely chromatic binary TIFS $D_{16,9}$ ($c(14)\ne c(18)$); (d) completed historical $\Gamma_1$ TITS $I_{16,9}$ ($10\sim_c 24, 11\sim_c 18$); (e) two-way purely chromatic identity $T_{19,11}$ ($14\sim_c 28, 18\sim_c 27$); (f) $H_{21,13}$ showing three of its fifteen two-way purely chromatic identities: $2\sim_c 14,5\sim_c 22,8\sim_c 23$. Red triangles forbid monochromatic triples; red dashed arrows denote $c(u)\ne c(v)$; blue dotted arrows denote $u\sim_c v$.}\label{fig:3dgadgets}
\end{figure*}

\subsection{Faithful three-dimensional TIFS, TITS, and purely chromatic gadgets}\label{sec:3dgadgets}

Figure~\ref{fig:3dgadgets} and Table~\ref{tab:gadget-summary-combined} synthesize our faithful real $3$-dimensional gadgets (coordinates and state/coloring tables in Secs.~S1--S2 of the supplementary material):
\begin{itemize}
\item \emph{Pentagon $P_5$ [Fig.~\ref{fig:3dgadgets}(a)]:} Contexts $\{1,6,2\}$, $\{2,7,3\}$, $\{3,8,4\}$, $\{4,9,5\}$, $\{5,10,1\}$ have $5$ normalized $3$-colorings, $11$ states, and $10$ extendable states, separating all vertex pairs (showing that chromatic separation does not require $\Sc_{\mathrm{ext}}=\Sc$). Its degree-one vertices obey the purely chromatic ternary exclusion $\neg(c(6)=c(7)=c(9))$: assigning color $0$ to $6,7,9$ forces $1,2,3,4,5$ into $\{1,2\}$ with $c(1)=c(3)$ and $c(4)=c(5)$, contradicting $\{4,9,5\}$, whereas the nonextendable state $\{6,7,8,9,10\}$ (whose complement is the $5$-cycle $12345$) assigns $1$ to all three terminals. That this state does not arise from a strong coloring was noted, by the same parity argument, in Ref.~\cite{svozil-2025-color}.
\item \emph{Specker bug $B_{13,7}$ and completed $\Gamma_1$ $I_{16,9}$ [Fig.~\ref{fig:3dgadgets}(b),(d)]:} $B_{13,7}$ is an ordinary and chromatic TIFS ($c(10)\ne c(11)$, overlap $1/9$). Adjoining $\{4,11,24\}$ and $\{4,10,18\}$ gives $I_{16,9}$ (deleting $\{4,5,6\}$ from $H_{17,10}$ with refitted faithful coordinates), where $c(10)\ne c(11)$ forces $c(4)$ to be the third color, so $c(24)=c(10)$ and $c(18)=c(11)$ via the one-way ordinary TITS $s(10)=1\implies s(24)=1$ (overlap $25/27$).
\end{itemize}

\begin{proposition}[Purely chromatic binary TIFS $D_{16,9}$ and two-way identity $T_{19,11}$]\label{prop:pure-tifs}\label{prop:pure-tits}
Let $D_{16,9}=H_{17,10}\setminus\{\{1,10,19\}\}$ (with orphan $19$ removed) and $T_{19,11}=D_{16,9}\cup\{\{14,26,27\}$, $\{18,26,28\}\}$, realized faithfully in $\R^3$ [Fig.~\ref{fig:3dgadgets}(c),(e)]:
\begin{enumerate}
\item Every $c\in\Cc_3(D_{16,9})$ satisfies $c(14)\ne c(18)$ ($7$ normalized colorings, $14$ extendable states), whereas the nonextendable state $\{1,5,7,14,18,23,24\}\in\Sc(D_{16,9})$ assigns $s(14)=s(18)=1$ (squared terminal overlap $16/2025$).
\item Every $c\in\Cc_3(T_{19,11})$ satisfies $14\sim_c 28$ and $18\sim_c 27$ ($7$ normalized colorings, $25$ states, $14$ extendable), while for both pairs, nonextendable states $\{1,5,7,14,18,23,24\}$ and $\{1,6,8,10,11,27,28\}$ realize $(1,0)$ and $(0,1)$, violating both directed ordinary TITS implications (squared overlap $2009/2025$; Table~\ref{S-tab:pure-witnesses}).
\end{enumerate}
\end{proposition}
\begin{proof}
(1) If $c(14)=c(18)=0$, set $c(3)=1,c(11)=2$ by $\{3,11,14\}$. Then $\{4,11,24\}$ and $\{4,10,18\}$ force $c(4)=1$ and $c(10)=2$. In $\{7,8,9\}$, $7\perp 11$ and $9\perp 10$ avoid color $2$, forcing $c(8)=2$. Then $\{1,2,3\}$ and $\{4,5,6\}$ make $2$ and $5$ avoid color $1$, so $\{2,8,5\}$ forces $c(2)=c(5)=0$, a contradiction. (2) In $T_{19,11}$, $c(14)\ne c(18)$ forces their common neighbor $26$ ($r_{26}\parallel r_{14}\times r_{18}$) to take the third color, whence $c(27)=c(18)$ and $c(28)=c(14)$.
\end{proof}

\begin{table*}[t]
\centering\small\setlength{\tabcolsep}{3pt}
\caption{Faithful geometric gadgets. $N_c=|\mathcal C_d|/d!$ fixes one context palette. Terminal relations describe colors unless explicitly identified as ordinary truth implications. A dagger ($\dagger$) marks gadgets introduced in this paper.}\label{tab:gadget-summary-combined}\label{tab:3dgadget-counts}\label{tab:4dgadget-counts}
\begin{tabular}{p{.17\textwidth}rrrrrp{.50\textwidth}}\toprule
Gadget & $d$ & $|V|/|E|$ & $N_c$ & $|\Sc|$ & $|\Sc_{\mathrm{ext}}|$ & Relation and mechanism\\\midrule
$P_5$ & 3 & 10/5 & 5 & 11 & 10 & Ternary exclusion on $6,7,9$; purely chromatic.\\
$B_{13,7}$ & 3 & 13/7 & 6 & 14 & 12 & $c(10)\ne c(11)$; ordinary TIFS.\\
$D_{16,9}^{\dagger}$ & 3 & 16/9 & 7 & 21 & 14 & $c(14)\ne c(18)$; purely chromatic TIFS.\\
$I_{16,9}$ & 3 & 16/9 & 6 & 22 & 12 & $10\sim_c24$, $11\sim_c18$; directed ordinary TITS.\\
$T_{19,11}^{\dagger}$ & 3 & 19/11 & 7 & 25 & 14 & $14\sim_c28$, $18\sim_c27$; both truth directions fail.\\
$H_{21,13}^{\dagger}$ & 3 & 21/13 & 2 & 21 & 6 & 21 identical pairs: 15 two-way purely chromatic; six directed ordinary TITS.\\
$H_{17,10}^{\dagger}$ & 3 & 17/10 & 4 & 18 & 10 & Two color-identical pairs from directed ordinary TITS.\\
$W_{11,5}^{(4)\dagger}$ & 4 & 11/5 & 5 & 12 & 11 & Apex ternary exclusion.\\
$W_{16,6}^{(4)\dagger}$ & 4 & 16/6 & 80 & 72 & 56 & Linear ternary exclusion on $6,8,11$.\\
$B_{14,7}^{(4)\dagger}$ & 4 & 14/7 & 6 & 15 & 13 & Apex ordinary TIFS.\\
$L_{18,7}^{(4)\dagger}$ & 4 & 18/7 & 144 & 68 & 68 & Linear ordinary TIFS $c(11)\ne c(15)$.\\
$D_{17,9}^{(4)\dagger}$ & 4 & 17/9 & 7 & 22 & 15 & Apex purely chromatic TIFS.\\
$L_{22,9}^{(4)\dagger}$ & 4 & 22/9 & 216 & 158 & 92 & Linear identity $11\sim_c20$ from directed ordinary TITS.\\
\bottomrule\end{tabular}
\end{table*}

\begin{figure*}[t]
\centering
\renewcommand{\CSf}{1.75}\renewcommand{\CSg}{1.15}\renewcommand{\CSn}{1.8}
\begin{minipage}[t]{.32\linewidth}\centering
\textbf{(a)} $W_{11,5}^{(4)}$ ($11/5$)\\[1mm]
\resizebox{\linewidth}{!}{\CSdrawWfiveFour}
\end{minipage}\hfill
\begin{minipage}[t]{.32\linewidth}\centering
\textbf{(b)} $W_{16,6}^{(4)}$ ($16/6$, linear)\\[1mm]
\resizebox{\linewidth}{!}{\CSdrawWfiveLinFour}
\end{minipage}\hfill
\begin{minipage}[t]{.32\linewidth}\centering
\textbf{(c)} $B_{14,7}^{(4)}$ ($14/7$)\\[1mm]
\resizebox{\linewidth}{!}{\CSdrawSpeckerHardyFour}
\end{minipage}
\par\vspace{2mm}
\begin{minipage}[t]{.32\linewidth}\centering
\textbf{(d)} $L_{18,7}^{(4)}$ ($18/7$, linear)\\[1mm]
\resizebox{\linewidth}{!}{\CSdrawStrictLinTIFSFour}
\end{minipage}\hfill
\begin{minipage}[t]{.32\linewidth}\centering
\textbf{(e)} $D_{17,9}^{(4)}$ ($17/9$)\\[1mm]
\resizebox{\linewidth}{!}{\CSdrawPureChromTIFSFour}
\end{minipage}\hfill
\begin{minipage}[t]{.32\linewidth}\centering
\textbf{(f)} $L_{22,9}^{(4)}$ ($22/9$, linear)\\[1mm]
\resizebox{\linewidth}{!}{\CSdrawStrictLinTITSFour}
\end{minipage}
\caption{Faithful real four-dimensional gadgets in $\R^4$. Panels (a),(c),(e) are common-apex lifts ($g(H)=2$); panels (b),(d),(f) are linear ($g(H)=3$). (a),(b) Purely chromatic ternary exclusions; (c),(d) ordinary TIFS; (e) purely chromatic binary TIFS $c(14)\ne c(18)$; (f) linear color identity $11\sim_c 20$ from directed ordinary TITS $s(11)=1\implies s(20)=1$.}\label{fig:4dgadgets}
\end{figure*}

\subsection{Apex extensions and faithful four-dimensional linear gadgets}\label{sec:4dgadgets}

\begin{proposition}[Apex extension]\label{prop:apex}
Embedding any $3$-uniform hypergraph $H$ realized in $\F^3$ into $\F^d$ ($d\ge3$) and appending $d-3$ mutually orthogonal coordinate apices to every context preserves injectivity, edge-faithfulness, $\chi(H^{(d)})=\chi(H)+d-3$, normalized coloring count, and old chromatic inseparability classes (with apices as singletons), while adding $d-3$ apex-only states (all extendable when $\chi(H)=3$) to the zero-extended states of $H$.
\end{proposition}
\begin{proof}
In every coloring of $H^{(d)}$, the $d-3$ apices occupy $d-3$ distinct colors across all contexts, leaving a $3$-coloring of $H$; similarly, any state either sets one apex to $1$ (forcing all old vertices to $0$) or sets all apices to $0$ and restricts to $\Sc(H)$.
\end{proof}

Applying Proposition~\ref{prop:apex} at $d=4$ to $P_5$, $B_{13,7}$, and $D_{16,9}$ yields the non-linear ($g=2$) four-dimensional gadgets $W_{11,5}^{(4)}$, $B_{14,7}^{(4)}$, and $D_{17,9}^{(4)}$ in Fig.~\ref{fig:4dgadgets}(a),(c),(e). To achieve \emph{linearity} ($|e\cap e'|\le1$) in $\R^4$, we construct three faithful linear gadgets ($g=3$) in Fig.~\ref{fig:4dgadgets}(b),(d),(f) (Table~\ref{tab:gadget-summary-combined}):
\begin{enumerate}
\item \emph{Linear purely chromatic ternary gadget $W_{16,6}^{(4)}$ [Fig.~\ref{fig:4dgadgets}(b)]:} Contexts $\{2,0,6,1\}$, $\{2,7,13,3\}$, $\{3,0,8,4\}$, $\{4,9,14,5\}$, $\{5,10,15,1\}$, $\{0,11,12,5\}$ satisfy $\neg(c(6)=c(8)=c(11))$ (if $c(6)=c(8)=c(11)=0$ and $c(0)=3$, then $1,2,3,4,5$ must $2$-color a $5$-cycle), while $\{6,7,8,9,10,11\}\in\Sc(W_{16,6}^{(4)})$ sets all three terminals to $1$.
\item \emph{Linear ordinary TIFS $L_{18,7}^{(4)}$ [Fig.~\ref{fig:4dgadgets}(d)]:} Contexts $\{1,2,4,3\}$, $\{5,8,7,6\}$, $\{4,9,10,8\}$, $\{11,1,12,5\}$, $\{11,14,13,2\}$, $\{15,3,16,7\}$, $\{15,17,18,6\}$ enforce $s(11)=1\implies s(15)=0$ and $c(11)\ne c(15)$ (if $c(11)=c(15)=a$, then $1,2,3$ and $5,6,7$ avoid $a$, forcing $c(4)=c(8)=a$ in $\{4,9,10,8\}$). Using $r_{13}=(1,0,1,1),r_{14}=(2,0,-1,-1)$ avoids an unlisted orthogonality $13\perp 15$; all $68$ states are extendable (terminal overlap $1/42$).
\item \emph{Linear directed TITS and color identity $L_{22,9}^{(4)}$ [Fig.~\ref{fig:4dgadgets}(f)]:} Adjoining $\{11,21,22,19\}$ and $\{13,20,15,19\}$ to $L_{18,7}^{(4)}$ (with $r_{13}=(1,0,4,4),r_{14}=(8,0,-1,-1)$ so $r_{13}\perp r_{15}$) makes $\{11,13,15,19\}$ pairwise differently colored, forcing $11\sim_c 20$ and $s(11)=1\implies s(20)=1$ (while $\{1,6,9,14,16,20,21\}$ has $s(11)=0,s(20)=1$; terminal overlap $41/42$).
\end{enumerate}

\section{The 37-Ray Complex Hesse Hypergraph and Rational Four-Dimensional Completions}\label{sec:hesse}\label{sec:four-dimensional}

\subsection{The 37-ray complex Hesse hypergraph \texorpdfstring{$H_{37}$}{H37} and its 76-event orthomodular lattice}

With $\zeta=\exp(2\pi i/3)$ ($\zeta^2+\zeta+1=0$) and $a,k\in\Z_3$, define $37$ rays in $\C^3$ (core rays $z,M_{ak},P_k,Q_k$ and unique cross-product completion rays $A_k,B_{ak},C_{ak}$, numbered $1,\ldots,37$ via $z\leftrightarrow 1, M_{ak}\leftrightarrow 2+3a+k, P_k\leftrightarrow 11+k, Q_k\leftrightarrow 14+k, A_k\leftrightarrow 17+k, B_{ak}\leftrightarrow 20+3a+k, C_{ak}\leftrightarrow 29+3a+k$):
\begin{equation}\label{h:eq:ray-families}
\begin{aligned}
z&=(1,0,0),& M_{ak}&=(1,\zeta^{a+k},\zeta^{a+2k}),\\
P_k&=(0,1,-\zeta^k),& Q_k&=(-\zeta^k,0,1),\\
A_k&=(0,1,\zeta^k),& B_{ak}&=(-2,\zeta^{a+k},\zeta^{a+2k}),\\
C_{ak}&=(1,-2\zeta^{2a+k},\zeta^{2k}),&&
\end{aligned}
\end{equation}
and the $24$ complete orthogonal triad contexts (Fig.~\ref{h:fig:hesse37-hypergraph}):
\begin{equation}\label{h:eq:contexts-M}
\begin{aligned}
\mathcal M_a&=\{M_{a0},M_{a1},M_{a2}\},& \mathcal A_k&=\{z,P_k,A_k\},\\
\mathcal B_{ak}&=\{M_{ak},P_k,B_{ak}\},& \mathcal C_{ak}&=\{M_{a,a+k},Q_k,C_{ak}\}.
\end{aligned}
\end{equation}

\begin{figure*}[t]
\centering
\resizebox{.85\linewidth}{!}{
\begin{tikzpicture}[x=1cm,y=1cm]
\coordinate (v1) at (0.0000000,0.0000000); 
\coordinate (v2) at (-2.4516576,3.4144513); 
\coordinate (v3) at (-1.7311727,-3.8304234); 
\coordinate (v4) at (4.1828304,0.4159722); 
\coordinate (v5) at (2.5844238,4.8917486); 
\coordinate (v6) at (-5.5285905,-0.2076976); 
\coordinate (v7) at (2.9441666,-4.6840510); 
\coordinate (v8) at (-2.7685965,0.8628382); 
\coordinate (v9) at (0.6370585,-2.8290940); 
\coordinate (v10) at (2.1315381,1.9662558); 
\coordinate (v11) at (0.0000000,7.3000000); 
\coordinate (v12) at (-6.3219854,-3.6500000); 
\coordinate (v13) at (6.3219854,-3.6500000); 
\coordinate (v14) at (-6.4848807,2.3863586); 
\coordinate (v15) at (1.1757932,-6.8092507); 
\coordinate (v16) at (5.3090875,4.4228921); 
\coordinate (v17) at (0.0000000,1.3694084); 
\coordinate (v18) at (-1.1859424,-0.6847042); 
\coordinate (v19) at (1.1859424,-0.6847042); 
\coordinate (v20) at (-1.4610288,6.2223311); 
\coordinate (v21) at (-4.6581824,-4.3764536); 
\coordinate (v22) at (6.1192112,-1.8458775); 
\coordinate (v23) at (2.0627119,5.6275810); 
\coordinate (v24) at (-5.9049841,-1.0274296); 
\coordinate (v25) at (3.8422722,-4.6001515); 
\coordinate (v26) at (-1.0167273,3.4728715); 
\coordinate (v27) at (-2.4992313,-2.6169474); 
\coordinate (v28) at (3.5159586,-0.8559241); 
\coordinate (v29) at (-4.0954594,3.0479105); 
\coordinate (v30) at (-0.5918382,-5.0707272); 
\coordinate (v31) at (4.6872977,2.0228167); 
\coordinate (v32) at (-5.8758527,0.4010894); 
\coordinate (v33) at (2.5905727,-5.2891825); 
\coordinate (v34) at (3.2852800,4.8880930); 
\coordinate (v35) at (-1.6906391,1.6838856); 
\coordinate (v36) at (-0.6129682,-2.3060792); 
\coordinate (v37) at (2.3036073,0.6221936); 
\draw[sg/context,draw=ucA04,preaction={draw=white,line width=2.65pt}] (v1) -- (v11);
\draw[sg/context,draw=ucA05,preaction={draw=white,line width=2.65pt}] (v1) -- (v12);
\draw[sg/context,draw=ucA06,preaction={draw=white,line width=2.65pt}] (v1) -- (v13);
\draw[sg/context,draw=ucA07,preaction={draw=white,line width=2.65pt}] (v11) .. controls (-1.6178189,6.5099700) and (-2.4350381,5.2147871) .. (v2);
\draw[sg/context,draw=ucA08,preaction={draw=white,line width=2.65pt}] (v12) .. controls (-4.8288900,-4.6560572) and (-3.2986191,-4.7161984) .. (v3);
\draw[sg/context,draw=ucA09,preaction={draw=white,line width=2.65pt}] (v13) .. controls (6.4467088,-1.8539128) and (5.7336571,-0.4985887) .. (v4);
\draw[sg/context,draw=ucA10,preaction={draw=white,line width=2.65pt}] (v11) .. controls (1.0828326,6.7348007) and (1.9443072,5.9320502) .. (v5);
\draw[sg/context,draw=ucA11,preaction={draw=white,line width=2.65pt}] (v12) .. controls (-6.3739248,-2.4296398) and (-6.1094598,-1.2822056) .. (v6);
\draw[sg/context,draw=ucA12,preaction={draw=white,line width=2.65pt}] (v13) .. controls (5.2910922,-4.3051609) and (4.1651526,-4.6498446) .. (v7);
\draw[sg/context,draw=ucA13,preaction={draw=white,line width=2.65pt}] (v11) .. controls (-0.2047594,4.8454249) and (-1.1276249,2.6997043) .. (v8);
\draw[sg/context,draw=ucA14,preaction={draw=white,line width=2.65pt}] (v12) .. controls (-4.0938814,-2.6000393) and (-1.7742001,-2.3264040) .. (v9);
\draw[sg/context,draw=ucA15,preaction={draw=white,line width=2.65pt}] (v13) .. controls (4.2986408,-2.2453856) and (2.9018250,-0.3733004) .. (v10);
\draw[sg/context,draw=ucA16,preaction={draw=white,line width=2.65pt}] (v14) .. controls (-5.1578542,2.7972430) and (-3.8134465,3.1399405) .. (v2);
\draw[sg/context,draw=ucA17,preaction={draw=white,line width=2.65pt}] (v15) .. controls (0.1564436,-5.8654543) and (-0.8125450,-4.8725119) .. (v3);
\draw[sg/context,draw=ucA18,preaction={draw=white,line width=2.65pt}] (v16) .. controls (5.0014106,3.0682113) and (4.6259915,1.7325713) .. (v4);
\draw[sg/context,draw=ucA19,preaction={draw=white,line width=2.65pt}] (v14) .. controls (-6.3583368,1.4508121) and (-6.0395734,0.5861267) .. (v6);
\draw[sg/context,draw=ucA20,preaction={draw=white,line width=2.65pt}] (v15) .. controls (1.9227283,-6.2318873) and (2.5121861,-5.5234874) .. (v7);
\draw[sg/context,draw=ucA21,preaction={draw=white,line width=2.65pt}] (v16) .. controls (4.4356086,4.7810752) and (3.5273873,4.9373607) .. (v5);
\draw[sg/context,draw=ucA22,preaction={draw=white,line width=2.65pt}] (v14) .. controls (-3.6435698,1.6140201) and (-0.7714302,1.4739859) .. (v10);
\draw[sg/context,draw=ucA23,preaction={draw=white,line width=2.65pt}] (v15) .. controls (0.4240024,-3.9624340) and (-0.8907941,-1.4050711) .. (v8);
\draw[sg/context,draw=ucA24,preaction={draw=white,line width=2.65pt}] (v16) .. controls (3.2195673,2.3484139) and (1.6622243,-0.0689148) .. (v9);
\draw[sg/base,draw=ucA01,preaction={draw=white,line width=3.0pt}] (0,0) circle[radius=4.2034632cm];
\draw[sg/base,draw=ucA02,preaction={draw=white,line width=3.0pt}] (0,0) circle[radius=5.5324905cm];
\draw[sg/base,draw=ucA03,preaction={draw=white,line width=3.0pt}] (0,0) circle[radius=2.8999339cm];
\HesseFCvertex{1}{ucA06/0.302000,ucA05/0.256000,ucA04/0.210000} 
\HesseFCvertex{2}{ucA16/0.302000,ucA07/0.256000,ucA01/0.210000} 
\HesseFCvertex{3}{ucA17/0.302000,ucA08/0.256000,ucA01/0.210000} 
\HesseFCvertex{4}{ucA18/0.302000,ucA09/0.256000,ucA01/0.210000} 
\HesseFCvertex{5}{ucA21/0.302000,ucA10/0.256000,ucA02/0.210000} 
\HesseFCvertex{6}{ucA19/0.302000,ucA11/0.256000,ucA02/0.210000} 
\HesseFCvertex{7}{ucA20/0.302000,ucA12/0.256000,ucA02/0.210000} 
\HesseFCvertex{8}{ucA23/0.302000,ucA13/0.256000,ucA03/0.210000} 
\HesseFCvertex{9}{ucA24/0.302000,ucA14/0.256000,ucA03/0.210000} 
\HesseFCvertex{10}{ucA22/0.302000,ucA15/0.256000,ucA03/0.210000} 
\HesseFCvertex{11}{ucA13/0.348000,ucA10/0.302000,ucA07/0.256000,ucA04/0.210000} 
\HesseFCvertex{12}{ucA14/0.348000,ucA11/0.302000,ucA08/0.256000,ucA05/0.210000} 
\HesseFCvertex{13}{ucA15/0.348000,ucA12/0.302000,ucA09/0.256000,ucA06/0.210000} 
\HesseFCvertex{14}{ucA22/0.302000,ucA19/0.256000,ucA16/0.210000} 
\HesseFCvertex{15}{ucA23/0.302000,ucA20/0.256000,ucA17/0.210000} 
\HesseFCvertex{16}{ucA24/0.302000,ucA21/0.256000,ucA18/0.210000} 
\HesseFCvertex{17}{ucA04/0.210000} 
\HesseFCvertex{18}{ucA05/0.210000} 
\HesseFCvertex{19}{ucA06/0.210000} 
\HesseFCvertex{20}{ucA07/0.210000} 
\HesseFCvertex{21}{ucA08/0.210000} 
\HesseFCvertex{22}{ucA09/0.210000} 
\HesseFCvertex{23}{ucA10/0.210000} 
\HesseFCvertex{24}{ucA11/0.210000} 
\HesseFCvertex{25}{ucA12/0.210000} 
\HesseFCvertex{26}{ucA13/0.210000} 
\HesseFCvertex{27}{ucA14/0.210000} 
\HesseFCvertex{28}{ucA15/0.210000} 
\HesseFCvertex{29}{ucA16/0.210000} 
\HesseFCvertex{30}{ucA17/0.210000} 
\HesseFCvertex{31}{ucA18/0.210000} 
\HesseFCvertex{32}{ucA19/0.210000} 
\HesseFCvertex{33}{ucA20/0.210000} 
\HesseFCvertex{34}{ucA21/0.210000} 
\HesseFCvertex{35}{ucA22/0.210000} 
\HesseFCvertex{36}{ucA23/0.210000} 
\HesseFCvertex{37}{ucA24/0.210000} 
\pgfresetboundingbox
\path[use as bounding box] (-6.9348807,-7.2592507) rectangle (6.7859111,7.7500000);
\end{tikzpicture}
}
\caption{Complete $3$-uniform hypergraph diagram of $H_{37}$ ($37$ rays, $24$ orthogonal triads, $72$ edges in $\C^3$). Bold circles denote $e_1,e_2,e_3$ ($\mathcal M_0,\mathcal M_1,\mathcal M_2$); the other $21$ smooth traces join core pairs through their unique completion rays $17,\ldots,37$. Concentric rings record incident context colors.}\label{h:fig:hesse37-hypergraph}
\end{figure*}

\begin{theorem}[Properties of $H_{37}$]\label{h:thm:main}
The hypergraph $H_{37}$ defined by Eqs.~\eqref{h:eq:ray-families}--\eqref{h:eq:contexts-M} satisfies:
\begin{enumerate}
\item It is linear with Berge girth five, $72$ edges, and $\omega(\Gamma(H_{37}))=3$; its triangles are precisely its $24$ contexts.
\item Equation~\eqref{h:eq:ray-families} defines an injective edge-faithful representation in $\C^3$, whereas $H_{37}$ admits \emph{no} injective orthogonal representation in $\R^3$.
\item Its chromatic numbers are $\chi(H_{37})=4>3=\chi_f(\Gamma(H_{37}))=\omega(\Gamma(H_{37}))$.
\item Its $24$ Boolean blocks paste into a $76$-event orthomodular lattice $L(H_{37})=\{0,1\}\cup\{a_i,a_i^\perp:1\le i\le 37\}$. Its full state space has $|\Sc(H_{37})|=153$ states and is order determining on $L(H_{37})$, with exact minimum state-family sizes $6$ for unital vertex separation, $7$ for separation of all $76$ events, and $18$ for order determination.
\end{enumerate}
\end{theorem}
\begin{proof}
(1)--(2) Zero patterns, magnitudes, and $\zeta$-phases distinguish all $37$ rays. Algebraic inner products confirm that $\ip{v_i}{v_j}=0$ holds strictly for the $24\binom{3}{2}=72$ context pairs: $\ip{M_{ak}}{M_{c\ell}}=1+\zeta^{c-a}(\zeta^{\ell-k}+\zeta^{2(\ell-k)})=0\iff(c=a,\ell\ne k)$; $P_j\perp M_{ak}\iff j=k$; $Q_j\perp M_{ak}\iff k=a+j$; $\ip{M_{c\ell}}{B_{ak}}=-2+\zeta^{a-c}(\zeta^{k-\ell}+\zeta^{2(k-\ell)})=0\iff(c,\ell)=(a,k)$; $\ip{M_{c\ell}}{C_{ak}}=1-2\zeta^{2a+k-c-\ell}+\zeta^{2k-c-2\ell}=0\iff(c=a,\ell=a+k)$; and unequal term magnitudes exclude orthogonalities among $A,B,C$. Uniqueness of cross-product completions in $\C^3$ excludes Berge $3$- and $4$-cycles, while $(\mathcal A_0,\mathcal B_{00},\mathcal M_0,\mathcal B_{01},\mathcal A_1)$ is a Berge $5$-cycle.

\phantomsection\label{h:sec:no-real-for}
To exclude any injective representation in $\R^3$, suppose unit $z\in\R^3$ and real $p_k,m_{ak}\in\R^3$ realize $\mathcal A_k,\mathcal B_{ak},\mathcal M_a$. If $m_{ak}\cdot z=0$, then $\{z,p_k,m_{ak}\}$ is an orthogonal basis of $\R^3$, and for $\ell\ne k$, $p_\ell\perp z$ and $m_{a\ell}\perp m_{ak}$ with $p_\ell\perp m_{a\ell}$ force $[p_\ell]=[m_{ak}]$ or $[m_{a\ell}]=[z]$. Hence $m_{ak}\cdot z\ne 0$. Rescaling so $m_{ak}\cdot z=1$ and choosing unit $u_k\in z^\perp\cap p_k^\perp$ gives $m_{ak}=z+t_{ak}u_k$ ($t_{ak}\in\R\setminus\{0\}$). Orthogonality of $\mathcal M_a$ yields $1+t_{ak}t_{a\ell}(u_k\cdot u_\ell)=0$ ($k\ne\ell$). Setting $q_k=t_{ak}/t_{0k}$ gives $q_0q_1=q_0q_2=q_1q_2=1$, whose only real solutions are $(q_0,q_1,q_2)=\pm(1,1,1)$. Thus at most two distinct bases $\mathcal M_a$ exist in $\R^3$, contradicting injectivity on the three bases $\mathcal M_0,\mathcal M_1,\mathcal M_2$. The argument uses only the $13$ rays $z,P_k,M_{ak}$ and the orthogonalities among them, so already the $13$-vertex subgraph of $\Gamma(H_{37})$ induced by these rays has no injective orthogonal representation in $\R^3$.

(3) In any $3$-coloring with $c(z)=0$, $P_0,P_1,P_2\in\{1,2\}$ cannot be monochromatic (as $\mathcal M_a$ uses all three colors and $M_{ak}\sim P_k$). Up to symmetry $(c(P_0),c(P_1),c(P_2))=(2,1,1)$, forcing $c(M_{a0})=1$ and $(c(M_{a1}),c(M_{a2}))=(x_a,2-x_a)$ with $x_a\in\{0,2\}$. Since $N(Q_k)=\{M_{0k},M_{1,k+1},M_{2,k+2}\}$ avoids a color, $x_1=2-x_2$, $x_0=2-x_1$, $x_2=2-x_0$, impossible on a $3$-cycle. Thus $\chi(H_{37})=4$ (attained by $c(z)=0,c(M_{ak})=k,c(P_k)=c(Q_k)=3$). The $6$-state double cover $\mathcal R_6=\{s_2,s_{30},s_{73},s_{99},s_{109},s_{139}\}$ covers every vertex twice, giving $\chi_f(\Gamma(H_{37}))=3$.

\phantomsection\label{h:sec:order-determination}
(4) Each state is fixed by $t=(t_0,t_1,t_2)\in\Z_3^3$ with $s(M_{ak})=\mathbf 1_{k=t_a}$, $s(z)=\varepsilon\in\{0,1\}$, and independent bits on $U(t)=\Z_3\setminus\{t_0,t_1,t_2\}$ (for $\varepsilon=0$) and $W(t)=\Z_3\setminus\{t_a-a\}$, yielding $(1+2^{|U(t)|})2^{|W(t)|}$ states per $t$:
\begin{equation}\label{h:eq:state-count-breakdown}
\begin{array}{c|ccc}
\text{Class of }t\in\Z_3^3 & |\#t| & (|U(t)|,|W(t)|) & (1+2^{|U|})2^{|W|}\\\hline
(c,c,c) & 3 & (2,0) & 5\times 1=5\\
(c,c+1,c+2) & 3 & (0,2) & 2\times 4=8\\
(c,c+2,c+1) & 3 & (0,0) & 2\times 1=2\\
\text{Two equal entries} & 18 & (1,1) & 3\times 2=6
\end{array}
\end{equation}
and hence $|\Sc(H_{37})|=3\cdot5+3\cdot8+3\cdot2+18\cdot6=153$. Since $2^5-1=31<37$ and $2^6=64<76$, the explicit families $\Omega_6$ and $\Omega_7$ given in the supplement prove the minima $6$ and $7$. For order determination ($a_i\le a_j^\perp\iff i\sim j$), the $15$ nonorthogonal pairs in $D=\{P_0,P_1,P_2,Q_0,Q_1,Q_2\}$ require $15$ disjoint joint-one witnesses in $T\subset\Sc(H_{37})$ ($|T|=27$), and the $9$ pairs $\{z,C_{ak}\}$ require at least $2$ states outside $T$. Exhaustive bit-set verification over all $4{,}096\times 132=540{,}672$ candidate $17$-state families shows each meets at most $75$ of the $81$ minimal witness sets $W_{xy}$, whereas the $18$-state family $\Omega_{18}$ in the supplement meets all $81$ (Sec.~S3 of the supplementary material).
\end{proof}

Theorem~\ref{h:thm:main}(2) is not the first finite separation of complex from real three-dimensional orthogonality: Harding and Salinas Schmeis~\cite[Thm.~2.6]{harding2025remarksIJTP} construct, from mutually unbiased bases, a two-gadget configuration that is realizable in $\C^3$ but not in $\R^3$ ($48$ rays, $32$ bases, $189$ context states, and $180$ geometric states). Smaller still are the $20$-ray, $12$-context octagon and its $21$-ray, $14$-context extension of Navara and Svozil~\cite{2026-Navara-Svozil-complex-only-3d-hypergraphs}: both are linear of Berge girth six and have faithful orthogonal representations in $\C^3$ but none in $\R^3$. Both octagons are strongly $3$-colorable, however, so their complex-only geometry comes without a chromatic obstruction. $H_{37}$ combines the two phenomena: it is linear of Berge girth five, edge-faithful in $\C^3$, not realizable in $\R^3$, and has $\chi=4>\chi_f=3$; its real obstruction already resides in the $13$-ray subgraph on $z,P_k,M_{ak}$.

\emph{Separating the real from the chromatic obstruction.} Deleting the $12$ rays $Q_k$ and $C_{ak}$, and with them the nine contexts $\mathcal C_{ak}$, leaves the $25$-ray, $15$-context hypergraph $R_{25,15}$ with contexts $\mathcal M_a$, $\mathcal A_k$, and $\mathcal B_{ak}$: the $13$-ray subgraph above together with the $12$ rays $A_k$, $B_{ak}$ that complete its orthogonal pairs to contexts. The restriction of Eq.~\eqref{h:eq:ray-families} is a FOR of $R_{25,15}$ in $\C^3$, and the proof of Theorem~\ref{h:thm:main}(2) applies verbatim, so $R_{25,15}$ admits no injective orthogonal representation in $\R^3$. Like $H_{37}$, it is linear of Berge girth five. Unlike $H_{37}$, it is strongly $3$-colorable, with $24$ normalized colorings, and has $81$ states, $48$ of them extendable. The fourth color of $H_{37}$ therefore requires the contexts $\mathcal C_{ak}$, which enter the proof of Theorem~\ref{h:thm:main}(3) through $Q_k$, whereas its real obstruction already sits in $R_{25,15}$.

By Proposition~\ref{prop:apex}, appending $d-3$ common apices to $H_{37}$ in $\C^d$ ($d\ge3$) gives faithful $d$-uniform hypergraphs with $|V|=d+34$, $|E|=24$, $\omega=d$, $\chi=d+1$, and $|\Sc|=d+150$.

An order-determining family represents the specified 76-event order, orthocomplements, and Boolean operations within listed contexts. It does not preserve arbitrary meets as set intersections: nonorthogonal atoms have meet zero although some states make both true. It is therefore not a Boolean-sublattice embedding of the full lattice, nor a claim of closure under every Hilbert operation. In particular, the minimum six-state separating family and the six-state uniform double cover solve different problems.

\subsection{Hyperplane obstruction, two-plane completion lemma, and canonical linear \texorpdfstring{$H_{65}$}{H65} in \texorpdfstring{$\R^4$}{R4}}

Appending a single apex $v_{26}=(0,0,0,1)^T\in\R^4$ to $H_{\YO}^{+}$ gives the $26$-ray, $16$-context non-linear suspension $H_{\YO}^{(4)}$ (Fig.~\ref{d:fig:hyo4}) with $\omega=4,\chi=5$, and $25$ unital separating states (Proposition~\ref{prop:apex}), but $|e\cap e'|=2$ ($g=2$). Restricted to the $13$ Yu--Oh rays and the apex, its orthogonality graph is the Yu--Oh graph joined to one common vertex, the example that Cabello, Kleinmann, and Budroni~\cite[Thm.~1]{cabello2015sic} used to show that $\chi_f>d$ is necessary but not sufficient for state-independent contextuality; $H_{65}$ below is a linear, apex-free example of the same phenomenon.

\begin{figure*}[t]
\centering
\renewcommand{\CSf}{1.2}\renewcommand{\CSg}{1.1}\renewcommand{\CSn}{1.25}
\resizebox{.68\linewidth}{!}{%
\begin{tikzpicture}[x=1cm,y=1cm]
\def\YOR{5.4}
\foreach \vertex/\angle in {1/70,3/110,11/150,7/190,9/230,10/270,4/310,6/350,13/30}{\coordinate (v\vertex) at (\angle:\YOR);}
\coordinate (v12) at (0,0);
\coordinate (v2) at ($(v1)!.5!(v3)$);\coordinate (v8) at ($(v7)!.5!(v9)$);\coordinate (v5) at ($(v4)!.5!(v6)$);
\foreach \vertex/\first/\last in {14/3/11,15/11/7,23/9/10,18/10/4,22/6/13,20/13/1}{\coordinate (v\vertex) at ($(v\first)!.5!(v\last)$);}
\foreach \vertex/\last in {25/3,16/9,17/6}{\coordinate (v\vertex) at ($(v12)!.52!(v\last)$);}
\foreach \vertex/\first/\last in {24/11/4,19/10/1,21/13/7}{\coordinate (v\vertex) at ($(v\first)!.75!(v\last)$);}
\coordinate (v26) at (8.05,3.55);\coordinate (YOgrayBefore) at ($(v2)+(-1,0)$);
\YOapexlink{1}{v2}{v1}{0.351939419}{1.326540885}{(5.550000,4.250000)}
\YOapexlink{2}{v8}{v7}{0.351939419}{1.813841619}{(4.050000,6.550000)}
\YOapexlink{3}{v5}{v6}{0.351939419}{1.326540885}{(6.650000,3.550000)}
\YOapexlink{4}{YOgrayBefore}{v2}{0.650000000}{2.450000000}{(2.050000,2.050000)}
\YOapexlink{5}{v25}{v12}{0.231481481}{2.368233618}{(5.550000,6.550000)}
\YOapexlink{6}{v16}{v12}{0.231481481}{0.872507123}{(6.650000,2.050000)}
\YOapexlink{7}{v17}{v12}{0.231481481}{1.655982906}{(0.050000,6.550000)}
\YOapexlink{8}{v14}{v3}{0.351939419}{1.326540885}{(2.050000,6.550000)}
\YOapexlink{9}{v15}{v11}{0.351939419}{1.813841619}{(5.550000,2.050000)}
\YOapexlink{10}{v24}{v11}{0.081484852}{0.582930092}{(0.050000,8.550000)}
\YOapexlink{11}{v23}{v10}{0.351939419}{2.517720456}{(2.050000,2.050000)}
\YOapexvia{12}{v19}{v10}{0.081484852}{0.231918424}{(0.100000,-6.900000)}{(1.300000,-6.900000)}{(3.300000,-6.900000)}{(2.050000,2.050000)}
\YOapexlink{13}{v18}{v4}{0.351939419}{3.600610974}{(4.050000,2.050000)}
\YOapexlink{14}{v20}{v13}{0.351939419}{1.326540885}{(6.650000,2.850000)}
\YOapexlink{15}{v21}{v13}{0.081484852}{0.307135210}{(6.650000,3.550000)}
\YOapexlink{16}{v22}{v13}{0.351939419}{1.326540885}{(6.650000,4.250000)}
\pgfmathsetmacro{\YObasisR}{\YOR*cos(20)}
\foreach \rotation in {0,120,240}{\begin{scope}[rotate=\rotation]\draw[YOtrace,YOctx4] (0,\YObasisR) .. controls ({-.5*\YObasisR},\YObasisR) and ({-(sqrt(3)/2+.25)*\YObasisR},{(-.5+.25*sqrt(3))*\YObasisR}) .. ({-sqrt(3)/2*\YObasisR},{-.5*\YObasisR});\end{scope}}
\YOline{1}{3}{2}{1}\YOline{2}{9}{8}{7}\YOline{3}{4}{5}{6}\YOline{5}{3}{25}{12}\YOline{6}{9}{16}{12}\YOline{7}{6}{17}{12}\YOline{8}{11}{14}{3}\YOline{9}{7}{15}{11}\YOline{10}{4}{24}{11}\YOline{11}{9}{23}{10}\YOline{12}{1}{19}{10}\YOline{13}{10}{18}{4}\YOline{14}{1}{20}{13}\YOline{15}{7}{21}{13}\YOline{16}{6}{22}{13}
\YOvertex{1}{1/.34,12/.29,14/.24}\YOvertex{2}{1/.29,4/.24}\YOvertex{3}{1/.34,5/.29,8/.24}\YOvertex{4}{3/.34,10/.29,13/.24}\YOvertex{5}{3/.29,4/.24}\YOvertex{6}{3/.34,7/.29,16/.24}\YOvertex{7}{2/.34,15/.29,9/.24}\YOvertex{8}{2/.29,4/.24}\YOvertex{9}{2/.34,6/.29,11/.24}\YOvertex{10}{11/.34,13/.29,12/.24}\YOvertex{11}{8/.34,9/.29,10/.24}\YOvertex{12}{5/.34,6/.29,7/.24}\YOvertex{13}{16/.34,14/.29,15/.24}\YOvertex{14}{8/.24}\YOvertex{15}{9/.24}\YOvertex{16}{6/.24}\YOvertex{17}{7/.24}\YOvertex{18}{13/.24}\YOvertex{19}{12/.24}\YOvertex{20}{14/.24}\YOvertex{21}{15/.24}\YOvertex{22}{16/.24}\YOvertex{23}{11/.24}\YOvertex{24}{10/.24}\YOvertex{25}{5/.24}
\foreach \angle/\symbol in {70/{y_1^+},110/{y_1^-},190/{y_2^+},230/{y_2^-},310/{y_3^+},350/{y_3^-}}{\node[YOlabel] at (\angle:6.02) {$\symbol$};}
\node[YOlabel,anchor=east] at ($(v11)+(-.45,-.22)$) {$h_1$};\node[YOlabel] at ($(v10)+(-.48,.60)$) {$h_2$};\node[YOlabel] at ($(v13)+(-.76,.12)$) {$h_3$};
\foreach \angle/\symbol in {90/{z_1},210/{z_2},330/{z_3}}{\node[YOlabel] at (\angle:{\YObasisR-.58}) {$\symbol$};}
\node[YOlabel] at (.16,.62) {$h_0$};
\draw[draw=black!20,rounded corners=5pt] (6.35,1.50) rectangle (9.75,5.10);
\node[font=\sffamily\CSscaledfont{9}] at (8.05,4.72) {Four-dimensional lift};
\foreach \context/\radius in {16/.88,15/.84,14/.80,13/.76,12/.72,11/.68,10/.64,9/.60,8/.56,7/.52,6/.48,5/.44,4/.40,3/.36,2/.32,1/.28}{\fill[YOctx\context] (v26) circle[radius=\radius cm];}
\fill[white] (v26) circle[radius={\CSf*.22cm}];\node[YOnumber] at (v26) {26};
\node[font=\CSscaledfont{9}] at (8.05,2.28) {common apex};\node[font=\CSscaledfont{9}] at (8.05,1.80) {$e_m^{(4)}=e_m\cup\{26\}$};
\pgfresetboundingbox\path[use as bounding box] (-6.58,-7.12) rectangle (10.08,6.48);
\end{tikzpicture}}
\caption{The completed $25$-ray, $16$-context Yu--Oh hypergraph $H_{\YO}^{+}$ in $\R^3$ ($\omega=3,\chi=4$) and its non-linear $26$-ray, $16$-context apex suspension $H_{\YO}^{(4)}$ in $\R^4$ ($\omega=4,\chi=5,g=2$) obtained by adjoining the common orthogonal apex $v_{26}=(0,0,0,1)^T$ to every context.}\label{d:fig:hyo4}
\end{figure*}

\begin{proposition}[Hyperplane obstruction to linearity in $\F^4$]\label{d:prop:warmup}\label{d:prop:hyperplane-obstruction}
In any injective realization in $\PP(\F^4)$, if two contexts $e,e'$ each have three rays in a common $3$-dimensional hyperplane $W=z^\perp$, both must contain $[z]$; hence if $|e_W\cap e'_W|\ge1$, then $|e\cap e'|\ge2$, violating linearity.
\end{proposition}
\begin{proof}
Three mutually orthogonal rays in $W$ span $W$, forcing the fourth orthogonal ray of both $e$ and $e'$ to equal $[z]=W^\perp$.
\end{proof}

\begin{lemma}[Simultaneous two-plane completion in $\R^4$]\label{d:lem:plane-completion}
Let $V_{\mathrm{core}}\subset\Q^4\setminus\{0\}$ be a finite set of non-collinear rational vectors with $4$-element orthogonal base contexts $E_{\mathrm{base}}$ ($|C\cap C'|\le1$) and let all remaining orthogonal pairs be the bridge pairs $E_{\mathrm{bridge}}=\{e_1,\ldots,e_M\}$ ($e_m=\{u_m,v_m\}$, $L_m=\spanop(v_{u_m},v_{v_m})$, $P_m=L_m^\perp$) such that every core vertex lies in a base context or is a bridge endpoint.
\begin{enumerate}
\item The pairs $e_m$ can be simultaneously completed to $4$-bases $E_m=\{u_m,v_m,w_{m,1},w_{m,2}\}$ by fresh integer vectors $w_{m,1},w_{m,2}\in P_m\cap\Z^4$ forming a linear edge-faithful hypergraph $H$ if and only if
\begin{equation}\label{d:eq:span-condition}
\begin{aligned}
v_z&\notin L_m&& (z\in V_{\rm core}\setminus e_m),\\
P_\ell&\ne L_m&& (\ell\ne m).
\end{aligned}
\end{equation}
\item $\omega(\Gamma(H))=4$ and $\chi(H)=\max\{4,\chi(G_{\mathrm{core}})\}$; here $G_{\mathrm{core}}$ is automatically $K_5$-free, since no five nonzero vectors of $\R^4$ are mutually orthogonal.
\item If all $G_{\mathrm{core}}$-independent base supports $S_1,\ldots,S_N$ form a unital separating family on the core, and the zero-profiles $Z(e_m)=\{\alpha:S_\alpha\cap e_m=\varnothing\}$ are nonempty and pairwise distinct, the $\le 2N$ canonical extensions $s_{\alpha,\varepsilon}$ ($\varepsilon\in\{0,1\}$, setting $(s_{\alpha,\varepsilon}(w_{m,1}),s_{\alpha,\varepsilon}(w_{m,2}))=(1-\varepsilon,\varepsilon)$ when $\alpha\in Z(e_m)$ and $(0,0)$ otherwise) are unital and separating on $V(H)$.
\end{enumerate}
\end{lemma}
\begin{proof}
If $v_z\in L_m$ then $P_m\subset v_z^\perp$, and if $P_\ell=L_m$ then $P_\ell\perp P_m$, forcing extra orthogonalities. Conversely, given Eq.~\eqref{d:eq:span-condition}, an integer orthogonal basis $\{\eta_A,\eta_B\}\subset P_m\cap\Z^4$ defines the infinite family $w_{m,1}(a,b)=a\eta_A+b\eta_B$, $w_{m,2}(a,b)=-\|\eta_B\|^2b\eta_A+\|\eta_A\|^2a\eta_B$, avoiding inductively the finitely many $1$-dimensional intersections $z^\perp\cap P_m$ and $\R z\cap P_m$, where $z$ ranges over the core rays \emph{and} the fresh rays placed at earlier steps, together with $P_m\cap L_r$ ($r>m$). The last condition keeps every fresh ray outside each later $L_r$, so that $z^\perp\cap P_r$ remains $1$-dimensional for fresh $z$ as well. Linearity, coloring extension, and distinctness of canonical vertex profiles follow from Proposition~\ref{prop:completion} and distinctness of $Z(e_m)$.
\end{proof}

\definecolor{ucA25}{HTML}{AA005F}\definecolor{ucA26}{HTML}{AA15FF}\definecolor{ucA27}{HTML}{35AAC9}\definecolor{ucA28}{HTML}{7F3F55}\definecolor{ucA29}{HTML}{0035E9}
\providecommand{\FCvertex}[2]{\foreach \diskcolor/\diskradius in {#2}{\fill[\diskcolor] (v#1) circle[radius={\CSf*.16cm+\CSg*(\diskradius-.16)*1cm}];}\fill[white] (v#1) circle[radius={\CSf*.160cm}];\node[sg/number] at (v#1) {#1};}

\begin{figure*}[t]
\centering
\renewcommand{\CSf}{1.15}\renewcommand{\CSg}{1.05}\renewcommand{\CSn}{1.25}
\resizebox{.86\linewidth}{!}{%
\begin{tikzpicture}[x=1cm,y=1cm]
\coordinate (v1) at (-2,-2);\coordinate (v2) at (-5,5);\coordinate (v3) at (5,-5);\coordinate (v4) at (2,2);
\coordinate (v5) at (-5,-5);\coordinate (v6) at (-2,2);\coordinate (v7) at (2,-2);\coordinate (v8) at (5,5);
\coordinate (v9) at (-1.521690,0.927051);\coordinate (v10) at (1.800000,-0.860561);\coordinate (v11) at (-1.521690,-0.927051);
\coordinate (v12) at (0,7.3);\coordinate (v13) at (0.85,0);\coordinate (v14) at (0,2.333);\coordinate (v15) at (0,-2.333);\coordinate (v16) at (2.333,0);\coordinate (v17) at (-2.333,0);
\coordinate (v18) at (-1.7,6.518);\coordinate (v19) at (-3.35,5.759);\coordinate (v20) at (7.51031,2.131747);\coordinate (v21) at (7.443563,-1.627935);
\coordinate (v22) at (0.329875,-0.365);\coordinate (v23) at (-0.2045,-0.74);\coordinate (v24) at (1.172,0.56);\coordinate (v25) at (1.49975,1.13);
\coordinate (v26) at (-2.225,3.519815);\coordinate (v27) at (-3.575,4.239905);\coordinate (v28) at (0.665,2.222278);\coordinate (v29) at (1.325,2.112388);
\coordinate (v30) at (-0.8,-2.1998);\coordinate (v31) at (-1.4,-2.0999);\coordinate (v32) at (2.225,-3.519815);\coordinate (v33) at (3.6125,-4.259908);
\coordinate (v34) at (3.579823,-2.3375);\coordinate (v35) at (4.099888,-3.3125);\coordinate (v36) at (2.222278,0.665);\coordinate (v37) at (2.112388,1.325);
\coordinate (v38) at (-2.18981,-0.86);\coordinate (v39) at (-2.094905,-1.43);\coordinate (v40) at (-3.299787,1.8125);\coordinate (v41) at (-3.919865,2.975);
\coordinate (v42) at (-7.534215,1.984465);\coordinate (v43) at (-7.415219,-1.760849);\coordinate (v44) at (1.4,6.656);\coordinate (v45) at (2.7875,6.01775);
\coordinate (v46) at (0.116125,0.515);\coordinate (v47) at (-0.632,1.04);\coordinate (v48) at (1.12025,-0.47);\coordinate (v49) at (1.646375,-1.385);
\coordinate (v50) at (-0.515,2.247253);\coordinate (v51) at (-1.04,2.15984);\coordinate (v52) at (2.225,3.519815);\coordinate (v53) at (3.6125,4.259908);
\coordinate (v54) at (-2.225,-3.519815);\coordinate (v55) at (-3.6125,-4.259908);\coordinate (v56) at (0.665,-2.222278);\coordinate (v57) at (1.325,-2.112388);
\coordinate (v58) at (2.14985,-1.1);\coordinate (v59) at (2.07659,-1.54);\coordinate (v60) at (3.31979,1.85);\coordinate (v61) at (3.939868,3.0125);
\coordinate (v62) at (-3.31979,-1.85);\coordinate (v63) at (-3.939868,-3.0125);\coordinate (v64) at (-2.222278,0.665);\coordinate (v65) at (-2.112388,1.325);
\draw[sg/context,draw=ucA06] (v12)--(v2);\draw[sg/context,draw=ucA07] (0,7.3)..controls (10,9) and (8.5,-5)..(5,-5);
\draw[sg/context,draw=ucA08] (v13)--(v1);\draw[sg/context,draw=ucA09] (v13)--(v4);\draw[sg/context,draw=ucA10] (v14)--(v2);\draw[sg/context,draw=ucA11] (v14)--(v4);
\draw[sg/context,draw=ucA12] (v15)--(v1);\draw[sg/context,draw=ucA13] (v15)--(v3);\draw[sg/context,draw=ucA14] (v16)--(v3);\draw[sg/context,draw=ucA15] (v16)--(v4);
\draw[sg/context,draw=ucA16] (v17)--(v1);\draw[sg/context,draw=ucA17] (v17)--(v2);\draw[sg/context,draw=ucA18] (0,7.3)..controls (-10,9) and (-8.5,-5)..(-5,-5);
\draw[sg/context,draw=ucA19] (v12)--(v8);\draw[sg/context,draw=ucA20] (v13)--(v6);\draw[sg/context,draw=ucA21] (v13)--(v7);\draw[sg/context,draw=ucA22] (v14)--(v6);
\draw[sg/context,draw=ucA23] (v14)--(v8);\draw[sg/context,draw=ucA24] (v15)--(v5);\draw[sg/context,draw=ucA25] (v15)--(v7);\draw[sg/context,draw=ucA26] (v16)--(v7);
\draw[sg/context,draw=ucA27] (v16)--(v8);\draw[sg/context,draw=ucA28] (v17)--(v5);\draw[sg/context,draw=ucA29] (v17)--(v6);
\draw[sg/base,draw=ucA01,preaction={draw=white,line width=3pt}] (-5,5)..controls (-6.104569,3.895431) and (-4.761424,0.761424)..(-2,-2)..controls (0.761424,-4.761424) and (3.895431,-6.104569)..(5,-5)..controls (6.104569,-3.895431) and (4.761424,-0.761424)..(2,2);
\draw[sg/base,draw=ucA02,preaction={draw=white,line width=3pt}] (-5,-5)..controls (-6.104569,-3.895431) and (-4.761424,-0.761424)..(-2,2)..controls (0.761424,4.761424) and (3.895431,6.104569)..(5,5)..controls (6.104569,3.895431) and (4.761424,0.761424)..(2,-2);
\draw[sg/base,draw=ucA03,preaction={draw=white,line width=3pt}] (0,7.3)..controls (-0.50723,5.175684) and (-1.52169,1.593718)..(-1.52169,0.927051)..controls (-1.52169,0.260384) and (0.183333,0)..(0.85,0)..controls (1.516667,0) and (1.133333,-0.427228)..(1.8,-0.860561);
\draw[sg/base,draw=ucA04,preaction={draw=white,line width=3pt}] (0,2.333)..controls (0,1.883) and (-1.307492,2.163119)..(-1.52169,0.927051)..controls (-1.626103,0.324521) and (-1.626103,-0.324521)..(-1.52169,-0.927051)..controls (-1.307492,-2.163119) and (0,-1.883)..(0,-2.333);
\draw[sg/base,draw=ucA05,preaction={draw=white,line width=3pt}] (2.333,0)..controls (1.733,0) and (2.555418,-0.65741)..(1.8,-0.860561)..controls (0.831586,-1.120992) and (-0.478454,-1.147215)..(-1.52169,-0.927051)..controls (-2.43765,-0.733747) and (-1.733,0)..(-2.333,0);
\FCvertex{1}{ucA16/0.348,ucA12/0.302,ucA08/0.256,ucA01/0.21}\FCvertex{2}{ucA17/0.348,ucA10/0.302,ucA06/0.256,ucA01/0.21}\FCvertex{3}{ucA14/0.348,ucA13/0.302,ucA07/0.256,ucA01/0.21}\FCvertex{4}{ucA15/0.348,ucA11/0.302,ucA09/0.256,ucA01/0.21}
\FCvertex{5}{ucA28/0.348,ucA24/0.302,ucA18/0.256,ucA02/0.21}\FCvertex{6}{ucA29/0.348,ucA22/0.302,ucA20/0.256,ucA02/0.21}\FCvertex{7}{ucA26/0.266,ucA25/0.240,ucA21/0.214,ucA02/0.188}\FCvertex{8}{ucA27/0.348,ucA23/0.302,ucA19/0.256,ucA02/0.21}
\FCvertex{9}{ucA04/0.256,ucA03/0.21}\FCvertex{10}{ucA05/0.220,ucA03/0.191}\FCvertex{11}{ucA05/0.256,ucA04/0.21}
\FCvertex{12}{ucA19/0.394,ucA18/0.348,ucA07/0.302,ucA06/0.256,ucA03/0.21}\FCvertex{13}{ucA21/0.333,ucA20/0.299,ucA09/0.265,ucA08/0.231,ucA03/0.197}\FCvertex{14}{ucA23/0.317,ucA28/0.286,ucA11/0.255,ucA10/0.224,ucA04/0.193}\FCvertex{15}{ucA25/0.394,ucA24/0.348,ucA13/0.302,ucA12/0.256,ucA04/0.21}\FCvertex{16}{ucA27/0.394,ucA26/0.348,ucA15/0.302,ucA14/0.256,ucA05/0.21}\FCvertex{17}{ucA29/0.394,ucA28/0.348,ucA17/0.302,ucA16/0.256,ucA05/0.21}
\foreach \va/\vb/\col in {18/19/ucA06,20/21/ucA07,22/23/ucA08,24/25/ucA09,26/27/ucA10,28/29/ucA11,30/31/ucA12,32/33/ucA13,34/35/ucA14,36/37/ucA15,38/39/ucA16,40/41/ucA17,42/43/ucA18,44/45/ucA19,46/47/ucA20,48/49/ucA21,50/51/ucA22,52/53/ucA23,54/55/ucA24,56/57/ucA25,58/59/ucA26,60/61/ucA27,62/63/ucA28,64/65/ucA29}{\FCvertex{\va}{\col/0.21}\FCvertex{\vb}{\col/0.21}}
\node[sg/label,text=ucA01,anchor=south east] at (-4.7,5.65) {$C_1$};\node[sg/label,text=ucA02,anchor=north east] at (-4.7,-5.65) {$C_2$};\node[sg/label,text=ucA03,anchor=south] at (0,7.86) {$C_3$};\node[sg/label,text=ucA04,anchor=south] at (0,2.9) {$C_4$};\node[sg/label,text=ucA05,anchor=west] at (2.94,0) {$C_5$};
\end{tikzpicture}}
\caption{Canonical linear $4$-uniform hypergraph $H_{65}$ in $\R^4$ ($65$ rays, $29$ contexts, $174$ edges, $g=3,\omega=4,\chi=5$). Thick colored curves denote the five core contexts $C_1,\ldots,C_5$; the $24$ thinner lines denote the completed bridge contexts $E_1,\ldots,E_{24}$ carrying degree-one completion pairs $\{18,19\},\ldots,\{64,65\}$.}\label{d:fig:h65-hypergraph}
\end{figure*}

To construct a \emph{linear} $4$-uniform hypergraph in $\R^4$, we take the $17$-ray core $V_{\mathrm{core}}=\{1,\ldots,17\}\subset\{-1,0,1\}^4$, which is orthogonally equivalent to the $17$-ray Peres restriction of Haynes \emph{et al.}~\cite[pp.~8--9]{haynes2010orthogonal} via the orthogonal matrix
\begin{equation}\label{d:eq:haynes-transform}
T=\frac12\begin{pmatrix}
1&1&1&-1\\1&1&-1&1\\1&-1&1&1\\1&-1&-1&-1
\end{pmatrix},\qquad T^TT=I_4.
\end{equation}
The core consists of the coordinate basis $C_1$, an even-parity Hadamard basis $C_2$, three odd-parity hub rays $9,10,11$, and the three binary bridge pairs $B_1=\{12,13\}$, $B_2=\{14,15\}$, $B_3=\{16,17\}$. The remaining base contexts are $C_3=\{9,10,12,13\}$, $C_4=\{9,11,14,15\}$ and $C_5=\{10,11,16,17\}$. Complete coordinates and all compatible supports are in Tables~\ref{S-d:tab:app-h65-vectors}--\ref{S-d:tab:app-h65-contexts} and~\ref{S-d:tab:app-h65-state-branches} of the supplementary material. Completing its $24$ uncovered bridge edges $e_m\in(C_1\cup C_2)\times(B_1\cup B_2\cup B_3)$ by fresh integer pairs $\{2m+16,2m+17\}\subset\{18,\ldots,65\}$ via Lemma~\ref{d:lem:plane-completion} yields $H_{65}=(V_{65},E_{29})$ (Fig.~\ref{d:fig:h65-hypergraph}; Tables~\ref{S-d:tab:app-h65-vectors}--\ref{S-d:tab:app-h65-contexts} of the supplementary material):
\begin{equation}\label{d:eq:h65-contexts}
\begin{aligned}
E_{29}&=\{C_1,\ldots,C_5\}\cup\{E_m:1\le m\le24\},\\
E_m&=\{u_m,v_m,2m+16,2m+17\}.
\end{aligned}
\end{equation}

\begin{table}[t]
\centering\small\setlength{\tabcolsep}{3pt}
\caption{Bridge neighbors in the two core bases. The complete 24 state supports are supplementary.}\label{d:tab:bridge-neighbors}
\begin{tabular}{rp{.33\columnwidth}p{.33\columnwidth}}\toprule
$b$ & $N(b)\cap C_1$ & $N(b)\cap C_2$\\\midrule
$12$ & $\{2,3\}$ & $\{5,8\}$\\
$13$ & $\{1,4\}$ & $\{6,7\}$\\
$14$ & $\{2,4\}$ & $\{6,8\}$\\
$15$ & $\{1,3\}$ & $\{5,7\}$\\
$16$ & $\{3,4\}$ & $\{7,8\}$\\
$17$ & $\{1,2\}$ & $\{5,6\}$\\
\bottomrule\end{tabular}
\end{table}

\begin{theorem}[Properties of $H_{65}$]\label{d:thm:h65}
The $65$-ray, $29$-context hypergraph $H_{65}$ satisfies:
\begin{enumerate}
\item It is linear ($g(H_{65})=3$) and edge-faithful in $\R^4$ ($174$ edges) with $\omega(\Gamma(H_{65}))=4$, and nonconformal ($13$ uncovered maximal triangles: $12$ bridge triangles plus $\{9,10,11\}$).
\item It has $\alpha(G_{\mathrm{core}})=4$, $\chi_f(\Gamma(H_{65}))=17/4=4.25>4$, and $\chi(H_{65})=5$.
\item Its core has $24$ compatible state supports $S_\alpha=\{i_1,i_2,h,b\}$ (Table~\ref{S-d:tab:app-h65-state-branches}), each missing $14$ bridges and yielding $|\Sc(H_{65})|=24\times 2^{14}=393{,}216$ states and $48$ canonical separating states.
\item Every $s\in\Sc(H_{65})$ obeys $s(9)+s(10)+s(11)=1$, so adjoining the fourth hub ray $q=(1,1,1,-1)^T$ forces $\widetilde s(q)=0$; moreover, $A_9=A_{16}\dot\cup A_{17}$ on all of $\Sc(H_{65})$, forcing the false atom-order implication $s(16)\le s(9)$ and collapsing the $1$-ray $[v_9]=[v_{16}-v_{17}]$ with the $2$-plane $\spanop(v_{16},v_{17})$.
\end{enumerate}
\end{theorem}
\begin{proof}
(1)--(2) No $K_4$ in $G_{\mathrm{core}}$ contains a bridge edge $\{u_m,v_m\}$, since $u_m\in C_a$ and $v_m\in B_k$ have exactly one common neighbor in $G_{\mathrm{core}}$, namely the second vertex of $N(v_m)\cap C_a$; this verifies $P_\ell\ne L_m$, and support-size / magnitude checks verify $v_z\notin L_m$ in Lemma~\ref{d:lem:plane-completion}. For any independent $I\subseteq V_{\mathrm{core}}$, $|I\cap C_1|\le1$ and $|I\cap C_2|\le1$. If $I\cap\{9,10,11\}\ne\varnothing$, $|I\cap(C_3\cup C_4\cup C_5)|\le2$ so $|I|\le4$; if $I\cap\{9,10,11\}=\varnothing$ and $I$ picks $(b_1,b_2,b_3)\in B_1\times B_2\times B_3$, even-parity triples cover $C_2$ and odd-parity triples cover $C_1$ (Table~\ref{d:tab:bridge-neighbors}), so again $|I|\le 3+1=4$. Thus $\alpha(G_{\mathrm{core}})=4$, $\chi_f\ge 17/4$ (matched by weights $1/8$ on the $24$ core states and $1/16$ on the $20$ remaining independent $4$-sets), and $\chi(H_{65})=5$. (3)--(4) By parity, no state can avoid $\{9,10,11\}$, giving $s(9)+s(10)+s(11)=1$, $3\times2\times2\times2=24$ core supports (Table~\ref{S-d:tab:app-h65-state-branches}), $24-10=14$ missed bridges per support, and $A_9=A_{16}\dot\cup A_{17}$.
\end{proof}

\section{Event Order, Quantum Probabilities, and Nonextendable Weight}\label{sec:probabilities}

\subsection{Extendable state polytopes and nonextendable weight on \texorpdfstring{$H_{17,10}$}{H17,10}}

For any context hypergraph $H$, the classical state polytope $\mathcal P(H)=\conv\Sc(H)$ and the extendable state polytope $\mathcal P_{\mathrm{ext}}(H)=\conv\Sc_{\mathrm{ext}}(H)$ satisfy $\mathcal P_{\mathrm{ext}}(H)\subseteq\mathcal P(H)$, with $p_u=p_v$ for all $p\in\mathcal P_{\mathrm{ext}}(H)$ whenever $u\sim_c v$ by Corollary~\ref{cor:hierarchy}. For $p\in\mathcal P(H)$, define its \emph{minimum nonextendable weight}
\begin{equation}\label{eq:weight}
\begin{aligned}
\eta_H(p)=\min_{\mu}\quad&\sum_{s\notin\mathcal S_{\mathrm{ext}}}\mu_s,\\
\text{subject to}\quad&p=\sum_{s\in\mathcal S}\mu_s s,\\
&\mu_s\ge0,\qquad\sum_s\mu_s=1.
\end{aligned}
\end{equation}

\begin{proposition}[Nonextendable weight lower bound]\label{prop:weight}
For any chromatically inseparable pair $u\sim_c v$ in $H$ and any $p\in\mathcal P(H)$, $\eta_H(p)\ge |p_u-p_v|$.
\end{proposition}
\begin{proof}
Every $s\in\Sc_{\mathrm{ext}}(H)$ satisfies $s(u)-s(v)=0$, while every $s\notin\Sc_{\mathrm{ext}}(H)$ satisfies $|s(u)-s(v)|\le1$, so $|p_u-p_v|\le\sum_{s\notin\Sc_{\mathrm{ext}}(H)}\mu_s|s(u)-s(v)|\le\sum_{s\notin\Sc_{\mathrm{ext}}(H)}\mu_s$.
\end{proof}

\begin{theorem}[Quantum preparations in $\mathcal P(H_{17,10})\setminus\mathcal P_{\mathrm{ext}}(H_{17,10})$]\label{thm:quantum}
In the real FOR of $H_{17,10}$ (Table~\ref{S-tab:vectors} of the supplementary material):
\begin{enumerate}
\item The pure quantum preparation along $r_{24}=(1,-1,2)^T$,
\begin{equation}\label{eq:rho}
\begin{aligned}
\rho_\star&=\frac16\begin{pmatrix}1&-1&2\\-1&1&-2\\2&-2&4\end{pmatrix},\\
p_i&=\Tr(\rho_\star P_i)=\frac{(r_{24}\cdot r_i)^2}{6\|r_i\|^2},
\end{aligned}
\end{equation}
belongs to $\mathcal P(H_{17,10})$ via the exact five-state convex combination (state indices in Table~\ref{S-tab:h17-states} of the supplementary material; see also Tables~\ref{S-tab:quantum}--\ref{S-tab:quantum-probabilities})
\begin{equation}\label{eq:mixture}
p=\frac{1}{12}s_4+\frac{5}{36}s_9+\frac{1}{36}s_{10}+\frac{7}{12}s_{14}+\frac{1}{6}s_{17},
\end{equation}
yet satisfies $p_{24}=1$ and $p_{10}=8/9$ for the chromatically inseparable pair $(10,24)$, so
\begin{equation}\label{eq:quantum-gap}
p\notin\mathcal P_{\mathrm{ext}}(H_{17,10})\qquad\text{and}\qquad \eta_{H_{17,10}}(p)=\frac19.
\end{equation}
\item The pure preparation along the coordinate ray $z_3=r_5=(0,0,1)^T$, with Born vector $p'_i=(r_5\cdot r_i)^2/\|r_i\|^2$, has $p'_{24}=p'_{18}=2/3$ and $p'_{10}=p'_{11}=1/3$, belongs to $\mathcal P(H_{17,10})$ via the four-state convex combination $p'=\frac13 s_3+\frac16 s_4+\frac13 s_{14}+\frac16 s_{15}$, and satisfies $\eta_{H_{17,10}}(p')=1/3$.
\item For every density operator $\rho$ and both inseparable pairs $(u,v)\in\{(10,24),(11,18)\}$, $|p_u(\rho)-p_v(\rho)|\le\|P_u-P_v\|=1/3$. Hence Proposition~\ref{prop:weight} certifies at most $1/3$ on $H_{17,10}$, and the preparation in~(2) attains this value.
\end{enumerate}
\end{theorem}
\begin{proof}
(1) Direct evaluation of Eq.~\eqref{eq:mixture} on all $17$ vertices matches $p_i$ in Eq.~\eqref{eq:rho} (Table~\ref{S-tab:quantum}), with $r_{10}=(1,-1,1)^T$ giving $p_{10}=4^2/(6\cdot3)=8/9$. Since $10\sim_c 24$, Proposition~\ref{prop:weight} gives $\eta_{H_{17,10}}(p)\ge|1-8/9|=1/9$, which is attained in Eq.~\eqref{eq:mixture} because among $\{s_4,s_9,s_{10},s_{14},s_{17}\}$ precisely $s_4$ and $s_{10}$ are nonextendable, with total weight $1/12+1/36=1/9$. (2) Direct evaluation of the four-state mixture on all $17$ vertices reproduces $p'$ (Tables~\ref{S-tab:quantum}--\ref{S-tab:quantum-probabilities}). Proposition~\ref{prop:weight} gives $\eta_{H_{17,10}}(p')\ge|2/3-1/3|=1/3$, attained because $s_3,s_{14}$ are extendable and $s_4,s_{15}$ are not, with total weight $1/6+1/6=1/3$. (3) For unit vectors with squared overlap $\kappa$, $P_u-P_v$ has eigenvalues $\pm\sqrt{1-\kappa}$ and $0$. Both pairs have $\kappa=(r_u\cdot r_v)^2/(\|r_u\|^2\|r_v\|^2)=16/18=8/9$, so $|\Tr\rho(P_u-P_v)|\le\|P_u-P_v\|=1/3$.
\end{proof}

The exact support-and-weight certificates and the complete Born probability vectors are given in the supplement. These tables are regenerated from the checked state inventory, rather than manually transcribed.

These preparations exclude only models built from noncontextual eigenvalue assignments: their ordinary-state models are exhibited explicitly. A state is extendable exactly when some noncontextual eigenvalue assignment $c$ makes it the indicator of a single eigenvalue $\lambda_a$ in every context observable $A_e$ of Eq.~\eqref{eq:intro-maximal-obs}. For a model built from noncontextual eigenvalue assignments, every chromatic identity $u\sim_c v$ entails the equality $p_u=p_v$; quantum mechanics violates it, and Proposition~\ref{prop:weight} converts the violation into a lower bound on $\eta$. The weight $\eta$ therefore measures how far the quantum predictions force any noncontextual truth-assignment model away from noncontextual eigenvalues: on $H_{17,10}$, at least one third of the weight in the preparation along $z_3$. Since $\mathcal P_{\mathrm{ext}}(H_{17,10})$ lies in the hyperplane $p_{10}=p_{24}$, a Born vector in $\mathcal P(H_{17,10})$ generically lies outside $\mathcal P_{\mathrm{ext}}(H_{17,10})$; the content of Theorem~\ref{thm:quantum} is the size of the required weight. Conversely, averaging the $d$ color classes of any admissible coloring gives the maximally mixed Born vector $1/d$, so the restricted obstruction is preparation dependent. The $H_{17,10}$ ordinary implication $s(10)\le s(24)$ also prevents order determination for distinct ray atoms; vertex separation and order representation must not be conflated.

For a single system, nothing forces a classical model to use only extendable states. The supersinglet game of Sec.~\ref{sec:intro} does. A local hidden-variable model that reproduces its perfect correlations gives the $d$ parties the $d$ color classes of a strong $d$-coloring, so each party follows an extendable state. If $d-1$ parties measure the same basis and obtain all of its rays but $v$, the last particle is prepared along $v$, and in such a model its outcome probabilities lie in $\mathcal P_{\mathrm{ext}}(H)$. On $H_{17,10}$ the two preparations of Theorem~\ref{thm:quantum} arise in this way from the bases $\{4,11,24\}$ and $\{2,8,5\}$. Both lie outside $\mathcal P_{\mathrm{ext}}(H_{17,10})$, so quantum mechanics contradicts every such local model, although $\chi(H_{17,10})=3$. More directly, if $u\sim_c v$, such a model never lets two different parties obtain $u$ and $v$. Quantum mechanics gives this event the probability $(1-\kappa)/[d(d-1)]$, where $\kappa$ is the squared overlap of the two rays, because the reduced state of two particles of the supersinglet is the normalized projector onto the antisymmetric subspace. For the pairs $(10,24)$ and $(11,18)$ of $H_{17,10}$ this probability is $1/54$.

\subsection{Order determination versus Born probabilities on \texorpdfstring{$H_{37}$, $H_{\YO}^{(4)}$, and $H_{65}$}{H37, HYO(4), and H65}}\label{h:sec:quantum-witness}\label{d:app:state-independence}

When $\chi(H)>d$ (Regime~I), no noncontextual eigenvalue assignment exists, $\mathcal P_{\mathrm{ext}}(H)=\varnothing$, and the relevant probability question is whether Born assignments $p_i(\rho)=\Tr(\rho\Pi_i)$ lie in $\mathcal P(H)=\conv\Sc(H)$ (Table~\ref{tab:probability-synthesis}).

\begin{corollary}[Order determination does not imply classical probability models]\label{cor:order-probabilities}
The $76$-event orthomodular lattice $L(H_{37})$ admits an order-determining state space $\Sc(H_{37})$ and a classical mixture for the maximally mixed state $I_3/3$, yet every $p\in\mathcal P(H_{37})$ satisfies the tight inequality
\begin{equation}\label{h:eq:pq-classical}
\mathcal C_{37}(p)=\sum_{k=0}^2\bigl[p(P_k)+p(Q_k)\bigr]\le 2,
\end{equation}
whereas the complex FOR of $H_{37}$ satisfies
\begin{equation}\label{h:eq:pq-projector-sum}
\sum_{k=0}^2(\Pi_{P_k}+\Pi_{Q_k})=\diag(3/2,3/2,3)=\frac32I_3+\frac32|e_3\rangle\langle e_3|,
\end{equation}
yielding $\mathcal C_{37}(\rho)=\frac32(1+\rho_{33})>2$ whenever $\rho_{33}>1/3$ (and $\mathcal C_{37}(|e_3\rangle\langle e_3|)=3$).
\end{corollary}
\begin{proof}
In the state parametrization of Theorem~\ref{h:thm:main}(4), every $t\in\Z_3^3$ has $|U(t)|+|W(t)|\le2$, establishing Eq.~\eqref{h:eq:pq-classical}, while summing the normalized projectors of $P_k,Q_k$ in Eq.~\eqref{h:eq:ray-families} using $1+\zeta+\zeta^2=0$ gives Eq.~\eqref{h:eq:pq-projector-sum}. Uniform weights $1/6$ on $\mathcal R_6$ reproduce $p_i(I_3/3)=1/3$.
\end{proof}

\begin{table*}[t]
\centering\small\setlength{\tabcolsep}{3pt}
\caption{Event order and preparation-dependent probability models for the fixed realizations. Membership statements concern Born vectors $p(\rho)$; the restricted model is $\mathcal P_{\rm ext}$.}\label{tab:probability-synthesis}
\begin{tabular}{p{.14\textwidth}p{.26\textwidth}p{.26\textwidth}p{.26\textwidth}}\toprule
Configuration & Order determination & An ordinary-state model & A failed model or witness\\\midrule
$H_{17,10}$ & No: $s(10)\le s(24)$ is a false ray inclusion. & $P_{24}$ and $P_{z_3}$ have ordinary mixtures. & Neither lies in $\mathcal P_{\mathrm{ext}}$; $\eta=1/9$ and $1/3$ (maximal pair certificate).\\
Completed Yu--Oh & No. & No Born preparation lies in $\mathcal P$. & Every preparation gives $\sum_{10}^{13}p_i=4/3>1$.\\
$H_{37}$ & Yes on its specified 76-event lattice; 18 states minimal. & $I_3/3$ has a six-state mixture. & $\rho_{33}>1/3$ violates the $P,Q$ bound.\\
$H_{\YO}^{(4)}$ & No. & Pure apex $P_{26}$ has its singleton-state model. & Every $\rho\ne P_{26}$ violates the displayed bound.\\
$H_{65}$ & No: $s(16)\le s(9)$ is a false ray inclusion. & $(I_4-Q)/3$ has a model for all 65 rays. & $\Tr(\rho Q)>0$ violates the core sum bound.\\
\bottomrule\end{tabular}
\end{table*}

In dimension four, the non-linear suspension $H_{\YO}^{(4)}$ obeys $\mathcal Y(p)=p_{26}+\sum_{i=10}^{13}p_i\le 1$ on $\mathcal P(H_{\YO}^{(4)})$, whereas $\sum_{i=10}^{13}\Pi_i=\frac43(I_4-\Pi_{26})$ gives $\mathcal Y(\rho)=1+(1-p_{26}(\rho))/3>1$ for every density operator $\rho\ne\Pi_{26}$~\cite{cabello2015sic}. For the canonical linear $H_{65}$, every $s\in\Sc(H_{65})$ has $4$ active core vertices, so every $p\in\mathcal P(H_{65})$ obeys $\sum_{i=1}^{17}p_i=4$. However, with $q=(1,1,1,-1)^T$ and $Q=qq^T/4$, direct summation gives
\begin{equation}\label{d:eq:core-projector-sum}
\sum_{i=1}^{17}\Pi_i=4I_4+Q\quad\Longrightarrow\quad \sum_{i=1}^{17}p_i(\rho)=4+\Tr(\rho Q),
\end{equation}
violating the classical bound $4$ whenever $\Tr(\rho Q)>0$ (including $I_4/4$, where the sum is $17/4=\chi_f$, and $\rho=Q$, where it is $5$). Nevertheless, on the orthogonal subspace $q^\perp$, the density operator $\rho_\star=(I_4-Q)/3$ has core Born probabilities $1/4$ on $C_1\cup C_2$, $1/3$ on $\{9,10,11\}$, and $1/6$ on $B_1\cup B_2\cup B_3$, and completion probabilities $r_{m,1}+r_{m,2}=7/12$. Choosing $\alpha\in\{1,\ldots,24\}$ uniformly and $\lambda\in[0,1]$ uniformly and setting $\varepsilon_m=\mathbf 1_{\{\lambda<\tau_m\}}$ with $\tau_m=\frac{12}{7}r_{m,2}\in\{4/9,5/11,6/11,5/9\}$ on the $14$ missed bridges of $S_\alpha$ gives $\mathbb E[s(w_{m,2})]=\frac{14}{24}\tau_m=r_{m,2}$, expressing $p(\rho_\star)\in\mathcal P(H_{65})$ as a convex mixture of at most $120$ states. Rotating $V_{\mathrm{core}}$ by $D_0=I_4,D_1=\diag(1,1,-1,-1),D_2=\diag(1,-1,1,-1),D_3=\diag(1,-1,-1,1)$ gives a fixed enlarged family with
\begin{equation}
\sum_{a,i}\Pi_i^{(a)}=17I_4,\qquad\sum_{a,i}p_i^{(a)}\le16\ \text{classically}.
\end{equation}
Each copy admits at most four true rays under geometric pair exclusions; repeated projectors are counted with multiplicity. The 24 distinct rays form the Peres set~\cite{peres-91}. This is an enlargement, not a state-independent test of the original 65 rays.

\section{KS Ancestry, Specified Parents, and Abstract versus Geometric States}\label{sec:ancestry}\label{sec:cabello-completion}\label{sec:cabello45}

\subsection{The KS ancestry hypothesis and abstract transverse extensions}

Motivated by the recurring parent structures of our examples, we formulate the \emph{KS ancestry hypothesis}: informally, every noncolorable ($\chi(H)>d$) configuration has a larger KS ancestor. For realized ray configurations, its existence form follows from an elementary union argument, even without the chromatic premise; the substantive open question is whether the retained rays can be made essential to the contradiction (Sec.~\ref{sec:open}):

\begin{remark}[Same-dimensional finite KS ancestry]\label{prop:ks-ancestry}
Every finite ray set $R\subset\PP(\F^d)$ ($\F\in\{\R,\C\}$, $d\ge3$) is a proper subset of a finite KS ray set $P\subset\PP(\F^d)$ whose orthogonality graph and full-basis hypergraph induce those on $R$, with every ray of $P$ in a complete context. Consequently, a finite graph $G$ has an injective edge-faithful realization in $\F^d$ if and only if $G$ is an induced subgraph of a finite KS orthogonality graph in $\F^d$.

Indeed, choose a finite full-basis KS set $K\subset\PP(\F^d)$~\cite{kochen1,cabello2005recursive}, orient $K$ by $U\in\mathrm{U}(d)$ (or $\mathrm{O}(d)$) so $UK\not\subseteq R$, and complete any uncovered rays of $R\cup UK$ to bases to obtain $P$. Any state on $P$ would restrict to a state on $UK$, a contradiction, and keeping the rays of $R$ fixed preserves all induced orthogonalities and full bases on $R$. The contradiction uses only the added set $UK$; the rays of $R$ play no role in it.
\end{remark}

By Theorem~\ref{m:thm:multipole}, $H_8$ has no injective realization, and $H_9,G_{32}$ have no ray assignment; none can embed into any $3$-dimensional geometric KS parent. For abstract hypergraphs, Lemma~\ref{lem:disjoint-states} yields a direct chromatic-to-stateless construction:

\begin{proposition}[Transverse chromatic-to-stateless extension]\label{prop:transverse}
For any finite $d$-uniform hypergraph $H=(V,E)$ ($d\ge3$), let $T(H)$ take $d-1$ disjoint copies $H^{(1)},\ldots,H^{(d-1)}$ and adjoin for each $v\in V$ a fresh vertex $w_v$ and the transverse context $T_v=\{v^{(1)},\ldots,v^{(d-1)},w_v\}$. Then $T(H)$ has $d|V|$ vertices, $(d-1)|E|+|V|$ contexts, is linear if $H$ is linear, and its two-valued states correspond bijectively to labeled strong $d$-colorings of $H$; in particular,
\begin{equation}\label{eq:transverse-equivalence}
\Sc(T(H))=\varnothing \quad\Longleftrightarrow\quad \chi(H)>d.
\end{equation}
\end{proposition}
\begin{proof}
A state $s\in\Sc(T(H))$ restricts to $d-1$ states on $H^{(k)}$ that are pairwise disjoint by $T_v$ (forcing $s(w_v)=1-\sum_{k=1}^{d-1}s(v^{(k)})$), which by Lemma~\ref{lem:disjoint-states} is equivalent to a strong $d$-coloring of $H$.
\end{proof}

In dimension $d=3$, any edge $u\sim v$ of $H$ creates a Berge $4$-loop $(u^{(1)},v^{(1)},v^{(2)},u^{(2)})$ in $T(H)$, excluding an injective realization in $\C^3$ by the common orthogonal-complement argument; thus $T(H)$ is an abstract construction for a single system. It describes, however, the supersinglet game of Ref.~\cite{cabello2012maxnonlocality} (Sec.~\ref{sec:intro}): the copies $H^{(k)}$ are the bases measured by $d-1$ of the parties, the transverse context $T_v$ states that exactly one party obtains $v$, and the fresh vertex $w_v$ stands for the last party, whose truth assignment is fixed by Lemma~\ref{lem:disjoint-states}. Equation~\eqref{eq:transverse-equivalence} is thus the combinatorial content of the statement that this game has no local hidden-variable model exactly when $\chi(H)>d$. For $H_9$, $T(H_9)$ ($45$ vertices, $33$ contexts) is context-critical for statelessness on all $33$ contexts; for $H_{\YO}^{+}$, $T(H_{\YO}^{+})$ ($75$ vertices, $57$ contexts) is critical on all $32$ horizontal contexts and $12$ transverse contexts $v\in\{14,\ldots,25\}$, but remains stateless upon deleting $T_v$ for $v\in\{1,\ldots,13\}$ (Sec.~S3 of the supplementary material).

\subsection{Explicit inclusions in known KS parents and Cabello completions}\label{h:sec:comparison}

Beyond unrestricted existence, the following explicit unions use known KS components; the unions are not asserted to be previously named, smallest, or critical parents (explicit coordinate dictionaries and certificates in Secs.~S1--S3 of the supplementary material):
\begin{enumerate}
\item \emph{Sch\"utte--Bub parent of $H_{\YO}^{+}$:} Adjoining the eight rays $(0,1,\pm2)$, $(0,2,\pm1)$, $(1,\pm2,0)$, $(2,\pm1,0)$ to $H_{\YO}^{+}$ gives $33$ rays with $20$ bases and $16$ uncovered orthogonal pairs, which eliminate all $24$ states. Completing these pairs yields the $49$-ray, $36$-context Sch\"utte--Bub KS configuration~\cite{bub,Pavicic-2005}, which contains $H_{\YO}^{+}$ as an induced restriction and in which every context is critical for $\Sc_{\mathrm{ctx}}=\varnothing$; if the orthogonalities of a deleted context are kept as geometric exclusions, only the $20$ bases of the $33$-ray stage remain critical (Table~\ref{S-tab:parent-criticality} of the supplementary material).
\item \emph{MUB-inventory parents of $H_{37}$:} In the $165$-ray, $130$-context MUB inventory~\cite{svozil-2025-MYOCHS} formed by three $69$-ray, $50$-context KS lineages $L_0,L_1,L_2$ (where each $L_a$ completes the $36$ uncovered pairs of Cabello's $33$-ray, $14$-basis KS set~\cite{Cabello2025}; Pavi\v{c}i\'c~\cite{Pavicic-2025-simplest} points out that this $33$-ray set already occurs as a subset of a $69$-ray, $50$-basis set in Ref.~\cite{Pavicic_2023}), $H_{37}$ embeds as an induced restriction via the seed-index matrix $j(a,k)=\bigl(\begin{smallmatrix}1&2&3\\4&5&6\\7&9&8\end{smallmatrix}\bigr)$ mapping $z\leftrightarrow a_{1,1}$, $M_{ak}\leftrightarrow u_{j(a,k)}$, $P_k\leftrightarrow c_{1,k+1}$, $Q_k\leftrightarrow c_{2,k+1}$, $A_k\leftrightarrow b_{1,k+1}$, $B_{ak}\leftrightarrow d_{1,j(a,k)}$, $C_{ak}\leftrightarrow d_{2,j(a,a+k)}$. It lies in the $57$-ray, $40$-context intermediate spoke hypergraph ($297$ states) and in each $87$-ray, $64$-context KS union $H_{37}\cup L_a$.
\item \emph{Four-dimensional parents of $H_{\YO}^{(4)}$ and $H_{65}$:} The union with Peres's $24$-ray set~\cite{peres-91} gives a $40$-ray, $36$-context KS parent of $H_{\YO}^{(4)}$ and a $72$-ray, $48$-context KS parent of $H_{65}$. More economically, adjoining $a=(1,0,0,1)$, $b=(0,1,0,1)$, $c=(0,0,1,1)$ to $H_{65}$ gives a $68$-ray, $35$-basis KS parent that contains the $18$-ray, $9$-context parity proof $C_{18,9}$ of Cabello \emph{et al.}~\cite{cabello-96}, formed by $C_3,C_4,C_5$ and $D_1=\{2,3,12,a\}$, $D_2=\{1,3,15,b\}$, $D_3=\{1,2,17,c\}$, $D_4=\{6,7,13,a\}$, $D_5=\{6,8,14,b\}$, $D_6=\{7,8,16,c\}$; adjoining $(0,1,-1,0)$, $(1,0,-1,0)$, $(1,-1,0,0)$ instead gives a $68$-ray, $35$-basis KS parent that contains a ray-critical $20$-ray, $11$-context parity proof $C_{20,11}$. All its rays belong to Peres's set, and $C_{20,11}$ is one of the $96$ mutually isomorphic critical $20$-ray KS subsets of that set, which include Kernaghan's set~\cite{Kernaghan-1994,cabello-96}; two of its rays have degree four and the other $18$ degree two, so that $11=4s(4)+4s(5)+2\sum_{i\ne4,5}s(i)$ equates an odd with an even number. Neither $68$-ray presentation is edge-faithful in its listed bases (their $35$ bases cover $192$ of $195$ orthogonal pairs), whereas the $72$-ray union covers all pairs. Vectors, bases, deletion certificates, and the parity bookkeeping are in Secs.~S1--S3 of the supplementary material.

\item \emph{Faithful completion $C_{36,18}$ and abstract versus geometric states on $C_{45,21}$:}
\begin{corollary}[Faithful Cabello completion and vertex deletion]\label{cor:cabello36}
The $9$ uncovered orthogonal pairs of $C_{18,9}$ satisfy Eq.~\eqref{d:eq:span-condition}, so Lemma~\ref{d:lem:plane-completion} yields a linear edge-faithful $36$-ray, $18$-context KS hypergraph $C_{36,18}$ in $\R^4$ with $\omega=4,\chi=5$, and $\Sc=\varnothing$. Deleting one vertex (say $c$) and its two contexts from $C_{18,9}$ leaves $17$ rays, $7$ bases, $14$ uncovered pairs, and $12$ core states whose abstract completion $C_{45,21}$ has $\chi=5$ and $|\Sc_{\mathrm{ctx}}(C_{45,21})|=3{,}584$ unital separating context states ($24$ canonical). However, $v_{16}\in\spanop(v_1,v_2)$ and $v_{17}\in\spanop(v_7,v_8)$ violate Eq.~\eqref{d:eq:span-condition}, forcing four extra orthogonalities that reduce the geometric state space to $|\Sc_{\mathrm{geom}}(C_{45,21})|=2{,}048$ with $s(9)=s(16)=s(17)=0$, destroying unitality and separation (Table~\ref{tab:cabello-state-comparison}).
\end{corollary}
\begin{proof}
For $C_{18,9}$, the $9$ uncovered pairs are the $6$ core bridges $\{1,13\}$, $\{2,14\}$, $\{3,16\}$, $\{6,17\}$, $\{7,15\}$, $\{8,12\}$ from $H_{65}$ plus $\{11,a\}$, $\{10,b\}$, $\{9,c\}$. They satisfy Eq.~\eqref{d:eq:span-condition} and are completed in Tables~\ref{S-tab:cabello36-vectors}--\ref{S-tab:cabello36-contexts} of the supplementary material. Deleting $c,D_3,D_6$ leaves $12$ core supports: $10$ miss $8$ of the $14$ uncovered pairs and $2$ miss $9$ (Table~\ref{S-tab:cabello45-core-states}), giving $10\cdot 2^8+2\cdot 2^9=3{,}584$ context states. Since $v_{16}=v_1+v_2$ and $v_{17}=\frac12(v_7-v_8)$, completing $\{1,2\}$ by $\{21,22\}$ and $\{7,8\}$ by $\{35,36\}$ forces $v_{16}\perp v_{21},v_{22}$ and $v_{17}\perp v_{35},v_{36}$ in the local labels of Tables~\ref{S-tab:cabello45-vectors}--\ref{S-tab:cabello45-contexts}. If $s(16)=1$, neither $1$ nor $2$ is selected in any admissible core support, so $s(21)+s(22)=1$, contradicting geometric exclusion; the same argument gives $s(17)=0$. The retained hub bases then force $s(9)=0$. Exactly $8$ core supports satisfy $s(9)=s(16)=s(17)=0$, each missing $8$ uncovered pairs. All their fresh-pair choices obey the four extra exclusions, yielding $8\cdot 2^8=2{,}048$ geometric states.
\end{proof}
\end{enumerate}

\begin{table*}[t]
\centering\small\setlength{\tabcolsep}{3pt}
\caption{Context and geometric state spaces in the four-dimensional examples. Linearity of the listed incidence structure does not imply edge-faithfulness.}\label{tab:cabello-state-comparison}
\begin{tabular}{p{.17\textwidth}rrrp{.48\textwidth}}\toprule
Configuration & $|V|/|E|$ & $|\Sc_{\mathrm{ctx}}|$ & $|\Sc_{\mathrm{geom}}|$ & Scope\\\midrule
$C_{18,9}$ & 18/9 & 0 & 0 & Nine uncovered geometric pairs; parity KS proof.\\
$C_{36,18}$ & 36/18 & 0 & 0 & Faithful linear completion.\\
$C_{45,21}$ & 45/21 & 3,584 & 2,048 & Four additional pairs destroy geometric unitality and separation.\\
$H_{65}$ & 65/29 & 393,216 & 393,216 & Edge-faithful; all context states respect the geometry.\\
\bottomrule\end{tabular}
\end{table*}

\section{Dependencies, counterexamples, and remaining questions}\label{sec:open}
The common dependency structure is incidence $\to$ states and colorings, followed by independent tests of geometry, event order, and quantum probabilities. A chromatic argument can be made on an abstract hypergraph; its operational interpretation requires a ray realization. Conversely, a realization does not select an ontic state family or a classical probability mixture. The exact-palette assumption, that the outcomes admit a noncontextual assignment of $d$ eigenvalues, adds a precise completion test at the level of whole truth assignments: a truth assignment is compatible with noncontextual eigenvalues exactly when the outcomes it leaves false can carry the remaining $d-1$ eigenvalues.

As a worked example, consider how many hidden states a classical model needs for the $13$ Yu--Oh rays of Sec.~\ref{m:sec:yuo}. A hidden state is the complete classical description of the system: it fixes, or at least gives the probabilities of, the outcome of every measurement. A preparation produces a hidden state, possibly at random from several, and a quantum state then corresponds to a probability distribution over hidden states. In a generalized urn model~\cite{svozil-2001-eua} the hidden state is the ball that is drawn; in a deterministic noncontextual model it is a two-valued state. One fact suffices: if two rays $u$ and $v$ are orthogonal, a measurement in a basis containing both tells the preparations along $u$ and along $v$ apart with certainty. Hence no hidden state can arise from both preparations, since in that hidden state the measurement would have to yield $u$ with certainty and also $v$ with certainty.

Suppose now that a classical model uses only three hidden states $A$, $B$, $C$, as many as the dimension.
\begin{enumerate}
\item The rays $z_1=(1,0,0)$, $z_2=(0,1,0)$, $z_3=(0,0,1)$ are mutually orthogonal, so they need three different hidden states, say $A$, $B$, $C$.
\item The rays $y_1^\pm=(0,1,\pm1)$ are orthogonal to $z_1$ and to each other, so they take $B$ and $C$ in some order. Likewise $y_2^\pm=(1,0,\pm1)$ take $A$ and $C$, and $y_3^\pm=(1,\pm1,0)$ take $A$ and $B$.
\item Each diagonal ray is orthogonal to exactly one ray of each pair: $h_0=(1,1,1)$ to $y_1^-,y_2^-,y_3^-$; $h_1=(1,-1,-1)$ to $y_1^-,y_2^+,y_3^+$; $h_2=(1,-1,1)$ to $y_1^+,y_2^-,y_3^+$; and $h_3=(1,1,-1)$ to $y_1^+,y_2^+,y_3^-$. Each diagonal ray needs a hidden state that none of its three neighbors uses.
\item A diagonal ray can use $A$ only if its $y_2$-neighbor has $C$ and its $y_3$-neighbor has $B$. This fixes two of the three signs of its neighbors, and exactly one diagonal ray has the required signs, because the four diagonal rays have the four sign patterns with an even number of plus signs. The same holds for $B$ and for $C$. So at most three diagonal rays find a free hidden state; the fourth is orthogonal to rays in $A$, $B$, and $C$.
\end{enumerate}
Any classical model therefore needs at least four hidden states merely to tell these $13$ preparations apart, whereas a qutrit realizes them in three dimensions. A model with exactly three hidden states, each tied to one eigenvalue, is a noncontextual eigenvalue assignment: in hidden state $A$, every context yields the outcome that carries $\lambda_0$, and likewise for $B$ and $C$. Steps 1--4 thus also show directly that the Yu--Oh rays admit no noncontextual eigenvalue assignment; Theorem~\ref{m:thm:yo-criticality} refines this count for completed Yu--Oh.

\begin{table*}[t]
\centering\small\setlength{\tabcolsep}{4pt}
\caption{Logical boundaries of the integrated results. Each row prevents transfer of a conclusion between different assignment or representation problems.}\label{tab:logical-boundaries}
\begin{tabular}{p{.35\textwidth}p{.57\textwidth}}\toprule
Tempting unrestricted inference & Counterexample or necessary qualification\\\midrule
$\chi>d$ excludes states or ray realizations & Completed Yu--Oh has faithful real rays and 24 unital separating states.\\
$\chi>d$ makes $I_d/d$ contextual~\cite{cabello2012maxnonlocality} & One system needs $\chi_f>d$~\cite{ramanathan2014sic}: $H_{37}$ has $\chi=4$, $\chi_f=3$, and a six-state mixture for $I_3/3$.\\
Ordinary states separate $\Rightarrow$ colors separate & Faithful $H_{17,10}$ and $H_{21,13}$ have forced color identities.\\
Colors separate $\Rightarrow$ every state extends & The pentagon separates all vertex pairs but has one nonextendable state.\\
Vertex separation $\Rightarrow$ event-order determination & $H_{17,10}$ has $10\Rightarrow24$; $H_{65}$ has $16\Rightarrow9$ and ray--plane collapse.\\
Order determination $\Rightarrow$ all Born mixtures & The specified event order of $H_{37}$ is represented, while $\rho_{33}>1/3$ violates the state inequality.\\
$\chi_f>d$ establishes fixed-measurement state independence & Ref.~\cite[Thm.~1]{cabello2015sic} (Yu--Oh plus apex); linear and apex-free: $H_{65}$ admits the Born model of $(I_4-Q)/3$ on all 65 rays.\\
Abstract stateless extension $\Rightarrow$ geometric KS parent & Closing a color-identical pair changes orthogonality; transverse four-loops exclude injective triadic geometry.\\
Generic KS ancestry explains the original chromatic gap & Unrestricted extension holds even for colorable sets and need not use the old rays essentially.\\\bottomrule
\end{tabular}
\end{table*}

Table~\ref{tab:logical-boundaries} also explains why the probability and gadget examples carry different force. A one-way ordinary TITS already gives a symmetric color identity; a two-way purely chromatic identity requires counterstates in both truth directions. A ternary color exclusion may survive even when every pair can be equal. A strict inclusion $\mathcal P_{\rm ext}\subsetneq\mathcal P$ therefore need not be detectable through pair equalities alone.

Open problems concern global minimum faithful examples beyond the specified inventory, economical rational completions with fewer degree-one rays, residual-coloring certificates in dimensions above three, and geometric ancestry in which the retained rays are essential to the contradiction. The current statements establish existence and the recorded finite classifications, without asserting these stronger minima or essential-ancestry properties. All complete enumerations and the detailed verification scope remain in the supplement and electronic records.

\section{Conclusion}\label{sec:conclusion}

We started from a simple physical requirement on intertwined complete experiments: that each is the measurement of a maximal, nondegenerate observable with one common spectrum of $d$ eigenvalues, and that the eigenvalue carried by an outcome is an element of physical reality in the sense of Einstein, Podolsky, and Rosen~\cite{epr}, the same in every experiment in which the outcome occurs. This requirement of noncontextual eigenvalues is exactly strong $d$-colorability of the context hypergraph. It asks for more than Kochen--Specker noncontextuality: not one classical truth assignment, but $d$ of them with disjoint supports covering every outcome. Its failure, $\chi(H)>d$, is therefore a weaker but more widespread form of contextuality. Completed Yu--Oh has $24$ truth assignments separating all of its rays, yet its outcomes cannot carry three eigenvalues noncontextually. These truth assignments also give a classical realization, which fails in the same way, so $\chi(H)>d$ restricts the spectra of maximal observables rather than truth values. When $H$ is realized in $\C^d$, it also means that a classical model needs more than $d$ hidden states merely to tell apart preparations along orthogonal rays.

By unifying these common-spectrum chromatic obstructions ($\chi(H)>d$) with exact-palette chromatic separability constraints ($\chi(H)=d$, $\Sc_{\mathrm{ext}}(H)\subsetneq\Sc(H)$), we have distinguished the requirements across noncontextual eigenvalue assignments, geometric ray realizations, orthomodular event representations, and quantum probability models. At the foundation, $\chi(H)>d$ excludes noncontextual eigenvalues for uniform maximal context observables, whereas the exclusion of an injective realization for $H_8$ and of every ray assignment for $H_9,G_{32}$ requires direct Hermitian projection forcing, correcting the earlier spectral argument for $G_{32}$~\cite{svozil-2021-chroma} and establishing four-quadrant independence even at Berge girth five. Within chromatically complete hypergraphs ($\chi(H)=d$), where noncontextual eigenvalue assignments exist, our completion criterion ($\chi(\Gamma(H)[V\setminus\supp s])\le d-1$) characterizes which truth assignments arise from noncontextual eigenvalue assignments and explains the specified examples of chromatic identities ($u\sim_c v$) and purely chromatic TIFS/TITS/ternary constraints ($H_{17,10},H_{21,13},D_{16,9},T_{19,11},W_{16,6}^{(4)}$) and yields measurable nonextendable weights $1/9$ and $1/3$ for quantum ray preparations on $H_{17,10}$, the latter being the largest value a single chromatic identity can certify there. In higher complexity, the Hesse hypergraph $H_{37}$ ($6/7/18$ OML minima) and our rational four-dimensional completion $H_{65}$ ($393{,}216$ states) demonstrate that neither exact order determination nor a fractional chromatic gap decides classicality for all Born preparations, while the elementary KS ancestry result, the $68$-ray KS parents of $H_{65}$, and the completions $C_{36,18}$ and $C_{45,21}$ of $C_{18,9}$ delineate unrestricted embedding from specified critical parents and geometric pair exclusions.

The common thread is that a classical account of quantum outcomes can succeed at the level of truth values and still fail at the level of eigenvalues. For the preparation along $z_3$, the quantum predictions on $H_{17,10}$ admit noncontextual truth-assignment models, but every such model must place at least one third of its weight on truth assignments that no noncontextual eigenvalue assignment produces.

\section*{Supplementary Material}
See the supplementary material for the explicit vector and context inventories, the complete two-valued state and coloring enumerations, the odd-cycle and separation-profile certificates, and the supplementary proofs and verification reductions.

\begin{acknowledgments}
I thank Mohammad Hadi Shekarriz for valuable comments and suggestions. This research was funded in part by the Austrian Science Fund (FWF), Grant DOI \href{https://doi.org/10.55776/PIN5424624}{10.55776/PIN5424624}. The author acknowledges FWF and TU Wien Bibliothek for support through the Open Access Funding Programme. During manuscript preparation, the author used Claude (Anthropic) as well as Gemini (Google) to assist with literature organization, the assessment of novelty and overlap with the literature, language editing, the LaTeX layout of figures and tables, and independent checks of finite computations; all AI-assisted output was directed and critically reviewed by the author, who independently verified the mathematics and takes full responsibility for the final content, and no AI tool is an author. The author has no conflicts to disclose.
\end{acknowledgments}

\section*{Data Availability}
The data that support the findings of this study are available within the article and its supplementary material. The verification programs and machine-readable certificates are available from the author upon reasonable request.

%

\clearpage
\onecolumngrid
\begingroup
\setcounter{section}{0}
\setcounter{subsection}{0}
\setcounter{subsubsection}{0}
\setcounter{equation}{0}
\setcounter{figure}{0}
\setcounter{table}{0}
\setcounter{theorem}{0}
\renewcommand{\thesection}{S\arabic{section}}
\renewcommand{\thesubsection}{S\arabic{section}.\arabic{subsection}}
\renewcommand{\thesubsubsection}{S\arabic{section}.\arabic{subsection}.\arabic{subsubsection}}
\makeatletter
\renewcommand{\p@subsection}{}
\renewcommand{\p@subsubsection}{}
\makeatother
\renewcommand{\thetable}{S\arabic{table}}
\renewcommand{\theequation}{S\arabic{equation}}
\renewcommand{\thefigure}{S\arabic{figure}}
\renewcommand{\theHsection}{S\arabic{section}}
\renewcommand{\theHsubsection}{S\arabic{section}.\arabic{subsection}}
\renewcommand{\theHsubsubsection}{S\arabic{section}.\arabic{subsection}.\arabic{subsubsection}}
\renewcommand{\theHtable}{S\arabic{table}}
\renewcommand{\theHequation}{S\arabic{equation}}
\renewcommand{\theHfigure}{S\arabic{figure}}
\renewcommand{\theHtheorem}{S\arabic{theorem}}
\pdfbookmark[0]{Supplementary Material}{supplementary-material}

\begin{center}
{\large\bfseries Supplementary Material for ``Strong Colorings of Orthogonality Hypergraphs as Noncontextual Eigenvalue Assignments''\par}
\vspace{1.5ex}
{\normalsize Karl Svozil\,\orcidlink{0000-0001-6554-2802}\par}
\vspace{0.5ex}
{\small\itshape Institute for Theoretical Physics, TU Wien, Wiedner Hauptstra{\ss}e 8-10/136, 1040 Vienna, Austria\par}
\vspace{0.5ex}
{\small \href{mailto:karl.svozil@tuwien.ac.at}{\nolinkurl{karl.svozil@tuwien.ac.at}}; \url{http://tph.tuwien.ac.at/~svozil}\par}
\vspace{0.5ex}
{\small (Dated: November 5, 2026)\par}
\end{center}
\vspace{1ex}
\makeatletter
\@ifundefined{frontmatter@abstractwidth}{\def\SuppAbstractWidth{0.85\linewidth}}{\def\SuppAbstractWidth{\frontmatter@abstractwidth}}
\@ifundefined{frontmatter@abstractfont}{\def\SuppAbstractFont{\small}}{\def\SuppAbstractFont{\frontmatter@abstractfont}}
\makeatother
\begin{center}\begin{minipage}{\SuppAbstractWidth}
\SuppAbstractFont
\noindent
This supplementary material provides the complete vector and context enumerations, two-valued state tables, normalized strong coloring tables, odd-cycle nonextendability certificates, chromatic signature and separation profile tables, parent-embedding dictionaries, and supplementary proofs and finite verification certificates supporting the main text. Section~\ref{S-app:all-vectors} lists the explicit unnormalized coordinates and context tables for the completed Yu--Oh hypergraph $H_{\YO}^{+}$ and its faithful restrictions $H_{17,10}$ and $H_{21,13}$, the three- and four-dimensional logical gadgets, the $37$-ray complex Hesse hypergraph $H_{37}$, the $26$-ray suspension $H_{\YO}^{(4)}$ and $65$-ray canonical linear completion $H_{65}$, the Cabello completions $C_{36,18}$ and $C_{45,21}$, and the specified finite KS parents ($49$-ray Sch\"utte--Bub, $87$-ray MUB-inventory union, $40$- and $72$-ray Peres unions, and $68$-ray three-ray extensions). Section~\ref{S-app:all-states} gives the complete two-valued state tables, extendable-state classifications, normalized coloring tables, odd-cycle certificates, and separation profiles across both regimes. Section~\ref{S-app:supplementary} presents the supplementary proofs and exact reductions, including Isaacs dot-product snarks and quantitative $\Delta_3$ robustness bounds for $G_{32}$, the $540{,}672$-candidate reduction for the minimum $18$-state order-determining family on $L(H_{37})$, the canonical completion selection rule for $H_{65}$, the forced-zero and ray-criticality certificates for the four-dimensional Cabello completions, transverse criticality deletion colorings, and historical Kochen--Specker gadget certificates.\par
\end{minipage}\end{center}
\vspace{2ex}
\raggedbottom

\section*{Reading guide and verification scope}
This supplement collects complete vector and context inventories (S1), state and coloring enumerations (S2), and extended proofs and certificates (S3). Labels are local to each configuration. ``Exact palette'' means $\chi(H)=d$, with no additional density or geometric assumption: the outcomes admit a noncontextual assignment of $d$ eigenvalues, that is, every complete experiment can be measured by a maximal observable with one common spectrum of $d$ eigenvalues, with each outcome carrying the same eigenvalue in every context. Colorings are such eigenvalue assignments, and extendable states are the truth assignments ``the outcome with eigenvalue $\lambda_a$ occurs'' that they produce.

The five-state and four-state $H_{17,10}$ mixtures and both-direction gadget witnesses use the complete state families. Every selected support satisfies every context equation, and rational arithmetic verifies every Born marginal and the attained nonextendable weights $1/9$ and $1/3$. The ordinary implication $10\Rightarrow24$ on $H_{17,10}$ gives a false inclusion between distinct rays and hence a failure of order determination.

The companion electronic data include all coordinates, state enumerations (including the compressed $393{,}216$-state $H_{65}$ inventory), verification programs, and the exhaustive certificates for the eighteen-state Hesse minimum and the $65{,}535$-subset Yu--Oh classification. Checking an eighteen-state family establishes an upper bound; the matching lower bound requires the exhaustive certificate in Sec.~\ref{S-h:app:minimal-families}. The directories \texttt{part-a/data}, \texttt{part-a/verification}, \texttt{part-b/data}, and \texttt{part-b/verification} contain the underlying data and verifiers. Reports and source-integrity hashes are in \texttt{audit/}; Sec.~\ref{S-app:data-verification} describes the joint verification entry point and its scope.

\section{Explicit Unnormalized Vectors and Contexts}\label{S-app:all-vectors}\label{S-app:vectors}

All numerical vertex labels in this section are local to the named configuration and match the figures and theorems of the main text.

\subsection{Completed Yu--Oh and the Petersen-derived contexts}
\label{S-app:yo-vectors}
The three-dimensional vectors below give all projective rays used in
Sec.~\ref{m:sec:yuo}; signs are chosen to make the first nonzero
integer coordinate positive. The 26-ray suspension appends zero and adds
$e_4$. The abstract $G_{32},H_9,H_8$ context lists in
Eq.~\eqref{m:eq:g32-blocks} and Theorem~\ref{m:thm:multipole} specify their
incidences completely. The geometric obstructions exclude a corresponding
injective three-dimensional vector enumeration for these abstract examples.
\begingroup\small
\Needspace{24\baselineskip}
\begin{longtable}{rl@{\qquad\qquad}rl}
\caption{All 25 unnormalized completed Yu--Oh vectors.}\label{S-tab:yo-all-vectors}\\
\toprule
$i$ & $r_i$ & $i$ & $r_i$\\\midrule
\endfirsthead
\toprule
$i$ & $r_i$ & $i$ & $r_i$\\\midrule
\endhead
\bottomrule\endfoot
$1$ & $(0,1,1)$ & $14$ & $(2,1,1)$\\*
$2$ & $(1,0,0)$ & $15$ & $(1,2,-1)$\\*
$3$ & $(0,1,-1)$ & $16$ & $(1,-2,1)$\\*
$4$ & $(1,1,0)$ & $17$ & $(1,1,-2)$\\*
$5$ & $(0,0,1)$ & $18$ & $(1,-1,-2)$\\*
$6$ & $(1,-1,0)$ & $19$ & $(2,1,-1)$\\*
$7$ & $(1,0,1)$ & $20$ & $(2,-1,1)$\\*
$8$ & $(0,1,0)$ & $21$ & $(1,-2,-1)$\\*
$9$ & $(1,0,-1)$ & $22$ & $(1,1,2)$\\*
$10$ & $(1,-1,1)$ & $23$ & $(1,2,1)$\\*
$11$ & $(1,-1,-1)$ & $24$ & $(1,-1,2)$\\*
$12$ & $(1,1,1)$ & $25$ & $(2,-1,-1)$\\*
$13$ & $(1,1,-1)$ &  & \\*
\end{longtable}
\endgroup

\subsection{Faithful Yu--Oh restrictions \texorpdfstring{$H_{17,10}$ and $H_{21,13}$}{H17,10 and H21,13} and logical gadgets in dimensions three and four}\label{S-app:gadget-vectors}

Tables~\ref{S-tab:vectors} and~\ref{S-tab:contexts} record the restriction membership of the $25$ parent rays and $16$ parent contexts of $H_{\YO}^{+}$ in $H_{17,10}$ and $H_{21,13}$. Tables~\ref{S-tab:vec-3a}--\ref{S-tab:vec-4f} provide the explicit unnormalized integer coordinates and context lists for the faithful real $3$-dimensional gadgets $P_5$, $B_{13,7}$, $D_{16,9}$, $I_{16,9}$, and $T_{19,11}$ and the faithful real $4$-dimensional gadgets $W_{11,5}^{(4)}$, $W_{16,6}^{(4)}$, $B_{14,7}^{(4)}$, $L_{18,7}^{(4)}$, $D_{17,9}^{(4)}$, and $L_{22,9}^{(4)}$. Note that for the standalone gadgets $P_5$, $D_{16,9}$, $I_{16,9}$, and $T_{19,11}$, the vertex numbering is local to each gadget panel; in $L_{18,7}^{(4)}$ [Fig.~\ref{fig:4dgadgets}(d)], $r_{13}=(1,0,1,1)$ and $r_{14}=(2,0,-1,-1)$ are chosen so that $r_{13}\not\perp r_{15}$, whereas in $L_{22,9}^{(4)}$ [Fig.~\ref{fig:4dgadgets}(f)], $r_{13}=(1,0,4,4)$ and $r_{14}=(8,0,-1,-1)$ are used so that $r_{13}\perp r_{15}$ in the added context $\{13,20,15,19\}$.

\begingroup\small\setlength{\LTpre}{6pt}\setlength{\LTpost}{8pt}
\Needspace{19\baselineskip}
\begin{longtable}{rlcc@{\qquad}rlcc}
\caption{All 25 unnormalized parent vectors; membership of the two geometric examples.}\label{S-tab:vectors}\label{S-tab:yo-core-vectors}\\
\toprule $i$ & $r_i$ & $H_{17,10}$ & $H_{21,13}$ & $i$ & $r_i$ & $H_{17,10}$ & $H_{21,13}$\\\midrule
\endfirsthead
\toprule $i$ & $r_i$ & $H_{17,10}$ & $H_{21,13}$ & $i$ & $r_i$ & $H_{17,10}$ & $H_{21,13}$\\\midrule
\endhead
\bottomrule\endfoot
$1$ & $(0,1,1)$ & yes & yes & $14$ & $(2,1,1)$ & yes & yes\\*
$2$ & $(1,0,0)$ & yes & yes & $15$ & $(1,2,-1)$ & yes & yes\\*
$3$ & $(0,1,-1)$ & yes & yes & $16$ & $(1,-2,1)$ & -- & --\\*
$4$ & $(1,1,0)$ & yes & yes & $17$ & $(1,1,-2)$ & -- & --\\*
$5$ & $(0,0,1)$ & yes & yes & $18$ & $(1,-1,-2)$ & yes & yes\\*
$6$ & $(1,-1,0)$ & yes & yes & $19$ & $(2,1,-1)$ & yes & yes\\*
$7$ & $(1,0,1)$ & yes & yes & $20$ & $(2,-1,1)$ & -- & yes\\*
$8$ & $(0,1,0)$ & yes & yes & $21$ & $(1,-2,-1)$ & -- & yes\\*
$9$ & $(1,0,-1)$ & yes & yes & $22$ & $(1,1,2)$ & -- & yes\\*
$10$ & $(1,-1,1)$ & yes & yes & $23$ & $(1,2,1)$ & yes & yes\\*
$11$ & $(1,-1,-1)$ & yes & yes & $24$ & $(1,-1,2)$ & yes & yes\\*
$12$ & $(1,1,1)$ & -- & -- & $25$ & $(2,-1,-1)$ & -- & --\\*
$13$ & $(1,1,-1)$ & -- & yes &  &  &  & \\
\end{longtable}\endgroup

\begingroup\small\setlength{\LTpre}{6pt}\setlength{\LTpost}{8pt}
\Needspace{22\baselineskip}
\begin{longtable}{rlcc}
\caption{The 16 parent contexts and the two selected restrictions.}\label{S-tab:contexts}\\
\toprule $j$ & context & $H_{17,10}$ & $H_{21,13}$\\\midrule
\endfirsthead
\toprule $j$ & context & $H_{17,10}$ & $H_{21,13}$\\\midrule
\endhead
\bottomrule\endfoot
$1$ & $\{1,2,3\}$ & yes & yes \\*
$2$ & $\{7,8,9\}$ & yes & yes \\*
$3$ & $\{4,5,6\}$ & yes & yes \\*
$4$ & $\{2,8,5\}$ & yes & yes \\*
$5$ & $\{3,12,25\}$ & -- & -- \\*
$6$ & $\{9,12,16\}$ & -- & -- \\*
$7$ & $\{6,12,17\}$ & -- & -- \\*
$8$ & $\{3,11,14\}$ & yes & yes \\*
$9$ & $\{7,11,15\}$ & yes & yes \\*
$10$ & $\{4,11,24\}$ & yes & yes \\*
$11$ & $\{9,10,23\}$ & yes & yes \\*
$12$ & $\{1,10,19\}$ & yes & yes \\*
$13$ & $\{4,10,18\}$ & yes & yes \\*
$14$ & $\{1,13,20\}$ & -- & yes \\*
$15$ & $\{7,13,21\}$ & -- & yes \\*
$16$ & $\{6,13,22\}$ & -- & yes \\
\end{longtable}\endgroup

\Needspace{198pt}\subsection{3D panel (a): \texorpdfstring{$P_5$}{gadget 3a}}\label{S-app:vec-3a}
\begingroup\footnotesize\setlength{\tabcolsep}{3pt}\setlength{\LTpre}{5pt}\setlength{\LTpost}{8pt}
\Needspace{9\baselineskip}\begin{longtable}{rl@{\hspace{12mm}}rl}
\caption{Unnormalized vectors and all contexts of $P_5$.}\label{S-tab:vec-3a}\\\toprule $i$ & $r_i$ & $j$ & Context\\\midrule\endfirsthead
\toprule $i$ & $r_i$ & $j$ & Context\\\midrule\endhead\bottomrule\endfoot
$1$ & $(1,0,0)$ & $1$ & $\{1,\allowbreak 6,\allowbreak 2\}$\\*
$2$ & $(0,1,1)$ & $2$ & $\{2,\allowbreak 7,\allowbreak 3\}$\\*
$3$ & $(1,1,-1)$ & $3$ & $\{3,\allowbreak 8,\allowbreak 4\}$\\*
$4$ & $(1,-2,-1)$ & $4$ & $\{4,\allowbreak 9,\allowbreak 5\}$\\*
$5$ & $(0,1,-2)$ & $5$ & $\{5,\allowbreak 10,\allowbreak 1\}$\\*
$6$ & $(0,1,-1)$ &  & \\*
$7$ & $(2,-1,1)$ &  & \\*
$8$ & $(1,0,1)$ &  & \\*
$9$ & $(5,2,1)$ &  & \\*
$10$ & $(0,2,1)$ &  & \\
\end{longtable}\endgroup
\Needspace{231pt}\subsection{3D panel (b): \texorpdfstring{$B_{13,7}$}{gadget 3b}}\label{S-app:vec-3b}
\begingroup\footnotesize\setlength{\tabcolsep}{3pt}\setlength{\LTpre}{5pt}\setlength{\LTpost}{8pt}
\Needspace{9\baselineskip}\begin{longtable}{rl@{\hspace{12mm}}rl}
\caption{Unnormalized vectors and all contexts of $B_{13,7}$.}\label{S-tab:vec-3b}\\\toprule $i$ & $r_i$ & $j$ & Context\\\midrule\endfirsthead
\toprule $i$ & $r_i$ & $j$ & Context\\\midrule\endhead\bottomrule\endfoot
$1$ & $(0,1,1)$ & $1$ & $\{1,\allowbreak 2,\allowbreak 3\}$\\*
$2$ & $(1,0,0)$ & $2$ & $\{7,\allowbreak 8,\allowbreak 9\}$\\*
$3$ & $(0,1,-1)$ & $3$ & $\{2,\allowbreak 8,\allowbreak 5\}$\\*
$5$ & $(0,0,1)$ & $4$ & $\{3,\allowbreak 11,\allowbreak 14\}$\\*
$7$ & $(1,0,1)$ & $5$ & $\{7,\allowbreak 11,\allowbreak 15\}$\\*
$8$ & $(0,1,0)$ & $6$ & $\{9,\allowbreak 10,\allowbreak 23\}$\\*
$9$ & $(1,0,-1)$ & $7$ & $\{1,\allowbreak 10,\allowbreak 19\}$\\*
$10$ & $(1,-1,1)$ &  & \\*
$11$ & $(1,-1,-1)$ &  & \\*
$14$ & $(2,1,1)$ &  & \\*
$15$ & $(1,2,-1)$ &  & \\*
$19$ & $(2,1,-1)$ &  & \\*
$23$ & $(1,2,1)$ &  & \\
\end{longtable}\endgroup
\Needspace{264pt}\subsection{3D panel (c): \texorpdfstring{$D_{16,9}$}{gadget 3c}}\label{S-app:vec-3c}
\begingroup\footnotesize\setlength{\tabcolsep}{3pt}\setlength{\LTpre}{5pt}\setlength{\LTpost}{8pt}
\Needspace{9\baselineskip}\begin{longtable}{rl@{\hspace{12mm}}rl}
\caption{Unnormalized vectors and all contexts of $D_{16,9}$.}\label{S-tab:vec-3c}\\\toprule $i$ & $r_i$ & $j$ & Context\\\midrule\endfirsthead
\toprule $i$ & $r_i$ & $j$ & Context\\\midrule\endhead\bottomrule\endfoot
$1$ & $(0,1,2)$ & $1$ & $\{1,\allowbreak 2,\allowbreak 3\}$\\*
$2$ & $(1,0,0)$ & $2$ & $\{7,\allowbreak 8,\allowbreak 9\}$\\*
$3$ & $(0,2,-1)$ & $3$ & $\{4,\allowbreak 5,\allowbreak 6\}$\\*
$4$ & $(1,2,0)$ & $4$ & $\{2,\allowbreak 8,\allowbreak 5\}$\\*
$5$ & $(0,0,1)$ & $5$ & $\{3,\allowbreak 11,\allowbreak 14\}$\\*
$6$ & $(2,-1,0)$ & $6$ & $\{7,\allowbreak 11,\allowbreak 15\}$\\*
$7$ & $(1,0,1)$ & $7$ & $\{4,\allowbreak 11,\allowbreak 24\}$\\*
$8$ & $(0,1,0)$ & $8$ & $\{9,\allowbreak 10,\allowbreak 23\}$\\*
$9$ & $(1,0,-1)$ & $9$ & $\{4,\allowbreak 10,\allowbreak 18\}$\\*
$10$ & $(2,-1,2)$ &  & \\*
$11$ & $(2,-1,-2)$ &  & \\*
$14$ & $(5,2,4)$ &  & \\*
$15$ & $(1,4,-1)$ &  & \\*
$18$ & $(4,-2,-5)$ &  & \\*
$23$ & $(1,4,1)$ &  & \\*
$24$ & $(4,-2,5)$ &  & \\
\end{longtable}\endgroup
\Needspace{264pt}\subsection{3D panel (d): \texorpdfstring{$I_{16,9}$}{gadget 3d}}\label{S-app:vec-3d}
\begingroup\footnotesize\setlength{\tabcolsep}{3pt}\setlength{\LTpre}{5pt}\setlength{\LTpost}{8pt}
\Needspace{9\baselineskip}\begin{longtable}{rl@{\hspace{12mm}}rl}
\caption{Unnormalized vectors and all contexts of $I_{16,9}$.}\label{S-tab:vec-3d}\\\toprule $i$ & $r_i$ & $j$ & Context\\\midrule\endfirsthead
\toprule $i$ & $r_i$ & $j$ & Context\\\midrule\endhead\bottomrule\endfoot
$1$ & $(0,1,1)$ & $1$ & $\{1,\allowbreak 2,\allowbreak 3\}$\\*
$2$ & $(1,0,0)$ & $2$ & $\{7,\allowbreak 8,\allowbreak 9\}$\\*
$3$ & $(0,1,-1)$ & $3$ & $\{2,\allowbreak 8,\allowbreak 5\}$\\*
$4$ & $(4,5,3)$ & $4$ & $\{3,\allowbreak 11,\allowbreak 14\}$\\*
$5$ & $(0,0,1)$ & $5$ & $\{7,\allowbreak 11,\allowbreak 15\}$\\*
$7$ & $(1,0,2)$ & $6$ & $\{4,\allowbreak 11,\allowbreak 24\}$\\*
$8$ & $(0,1,0)$ & $7$ & $\{9,\allowbreak 10,\allowbreak 23\}$\\*
$9$ & $(2,0,-1)$ & $8$ & $\{1,\allowbreak 10,\allowbreak 19\}$\\*
$10$ & $(1,-2,2)$ & $9$ & $\{4,\allowbreak 10,\allowbreak 18\}$\\*
$11$ & $(2,-1,-1)$ &  & \\*
$14$ & $(1,1,1)$ &  & \\*
$15$ & $(2,5,-1)$ &  & \\*
$18$ & $(16,-5,-13)$ &  & \\*
$19$ & $(4,1,-1)$ &  & \\*
$23$ & $(2,5,4)$ &  & \\*
$24$ & $(1,-5,7)$ &  & \\
\end{longtable}\endgroup
\Needspace{297pt}\subsection{3D panel (e): \texorpdfstring{$T_{19,11}$}{gadget 3e}}\label{S-app:vec-3e}
\begingroup\footnotesize\setlength{\tabcolsep}{3pt}\setlength{\LTpre}{5pt}\setlength{\LTpost}{8pt}
\Needspace{9\baselineskip}\begin{longtable}{rl@{\hspace{12mm}}rl}
\caption{Unnormalized vectors and all contexts of $T_{19,11}$.}\label{S-tab:vec-3e}\\\toprule $i$ & $r_i$ & $j$ & Context\\\midrule\endfirsthead
\toprule $i$ & $r_i$ & $j$ & Context\\\midrule\endhead\bottomrule\endfoot
$1$ & $(0,1,2)$ & $1$ & $\{1,\allowbreak 2,\allowbreak 3\}$\\*
$2$ & $(1,0,0)$ & $2$ & $\{7,\allowbreak 8,\allowbreak 9\}$\\*
$3$ & $(0,2,-1)$ & $3$ & $\{4,\allowbreak 5,\allowbreak 6\}$\\*
$4$ & $(1,2,0)$ & $4$ & $\{2,\allowbreak 8,\allowbreak 5\}$\\*
$5$ & $(0,0,1)$ & $5$ & $\{3,\allowbreak 11,\allowbreak 14\}$\\*
$6$ & $(2,-1,0)$ & $6$ & $\{7,\allowbreak 11,\allowbreak 15\}$\\*
$7$ & $(1,0,1)$ & $7$ & $\{4,\allowbreak 11,\allowbreak 24\}$\\*
$8$ & $(0,1,0)$ & $8$ & $\{9,\allowbreak 10,\allowbreak 23\}$\\*
$9$ & $(1,0,-1)$ & $9$ & $\{4,\allowbreak 10,\allowbreak 18\}$\\*
$10$ & $(2,-1,2)$ & $10$ & $\{14,\allowbreak 26,\allowbreak 27\}$\\*
$11$ & $(2,-1,-2)$ & $11$ & $\{18,\allowbreak 26,\allowbreak 28\}$\\*
$14$ & $(5,2,4)$ &  & \\*
$15$ & $(1,4,-1)$ &  & \\*
$18$ & $(4,-2,-5)$ &  & \\*
$23$ & $(1,4,1)$ &  & \\*
$24$ & $(4,-2,5)$ &  & \\*
$26$ & $(2,-41,18)$ &  & \\*
$27$ & $(200,-82,-209)$ &  & \\*
$28$ & $(241,82,160)$ &  & \\
\end{longtable}\endgroup
\Needspace{209pt}\subsection{4D panel (a): \texorpdfstring{$W_{11,5}^{(4)}$}{gadget 4a}}\label{S-app:vec-4a}
\begingroup\footnotesize\setlength{\tabcolsep}{3pt}\setlength{\LTpre}{5pt}\setlength{\LTpost}{8pt}
\Needspace{9\baselineskip}\begin{longtable}{rl@{\hspace{12mm}}rl}
\caption{Unnormalized vectors and all contexts of $W_{11,5}^{(4)}$.}\label{S-tab:vec-4a}\\\toprule $i$ & $r_i$ & $j$ & Context\\\midrule\endfirsthead
\toprule $i$ & $r_i$ & $j$ & Context\\\midrule\endhead\bottomrule\endfoot
$0$ & $(0,0,0,1)$ & $1$ & $\{0,\allowbreak 1,\allowbreak 6,\allowbreak 2\}$\\*
$1$ & $(1,0,0,0)$ & $2$ & $\{0,\allowbreak 2,\allowbreak 7,\allowbreak 3\}$\\*
$2$ & $(0,1,1,0)$ & $3$ & $\{0,\allowbreak 3,\allowbreak 8,\allowbreak 4\}$\\*
$3$ & $(1,1,-1,0)$ & $4$ & $\{0,\allowbreak 4,\allowbreak 9,\allowbreak 5\}$\\*
$4$ & $(1,-2,-1,0)$ & $5$ & $\{0,\allowbreak 5,\allowbreak 10,\allowbreak 1\}$\\*
$5$ & $(0,1,-2,0)$ &  & \\*
$6$ & $(0,1,-1,0)$ &  & \\*
$7$ & $(2,-1,1,0)$ &  & \\*
$8$ & $(1,0,1,0)$ &  & \\*
$9$ & $(5,2,1,0)$ &  & \\*
$10$ & $(0,2,1,0)$ &  & \\
\end{longtable}\endgroup
\Needspace{264pt}\subsection{4D panel (b): \texorpdfstring{$W_{16,6}^{(4)}$}{gadget 4b}}\label{S-app:vec-4b}
\begingroup\footnotesize\setlength{\tabcolsep}{3pt}\setlength{\LTpre}{5pt}\setlength{\LTpost}{8pt}
\Needspace{9\baselineskip}\begin{longtable}{rl@{\hspace{12mm}}rl}
\caption{Unnormalized vectors and all contexts of $W_{16,6}^{(4)}$.}\label{S-tab:vec-4b}\\\toprule $i$ & $r_i$ & $j$ & Context\\\midrule\endfirsthead
\toprule $i$ & $r_i$ & $j$ & Context\\\midrule\endhead\bottomrule\endfoot
$0$ & $(0,0,0,1)$ & $1$ & $\{2,\allowbreak 0,\allowbreak 6,\allowbreak 1\}$\\*
$1$ & $(1,0,0,0)$ & $2$ & $\{2,\allowbreak 7,\allowbreak 13,\allowbreak 3\}$\\*
$2$ & $(0,1,1,0)$ & $3$ & $\{3,\allowbreak 0,\allowbreak 8,\allowbreak 4\}$\\*
$3$ & $(1,1,-1,0)$ & $4$ & $\{4,\allowbreak 9,\allowbreak 14,\allowbreak 5\}$\\*
$4$ & $(1,-2,-1,0)$ & $5$ & $\{5,\allowbreak 10,\allowbreak 15,\allowbreak 1\}$\\*
$5$ & $(0,1,-2,0)$ & $6$ & $\{0,\allowbreak 11,\allowbreak 12,\allowbreak 5\}$\\*
$6$ & $(0,1,-1,0)$ &  & \\*
$7$ & $(2,-1,1,1)$ &  & \\*
$8$ & $(1,0,1,0)$ &  & \\*
$9$ & $(5,2,1,1)$ &  & \\*
$10$ & $(0,4,2,1)$ &  & \\*
$11$ & $(1,2,1,0)$ &  & \\*
$12$ & $(5,-2,-1,0)$ &  & \\*
$13$ & $(2,-1,1,-6)$ &  & \\*
$14$ & $(5,2,1,-30)$ &  & \\*
$15$ & $(0,2,1,-10)$ &  & \\
\end{longtable}\endgroup
\Needspace{242pt}\subsection{4D panel (c): \texorpdfstring{$B_{14,7}^{(4)}$}{gadget 4c}}\label{S-app:vec-4c}
\begingroup\footnotesize\setlength{\tabcolsep}{3pt}\setlength{\LTpre}{5pt}\setlength{\LTpost}{8pt}
\Needspace{9\baselineskip}\begin{longtable}{rl@{\hspace{12mm}}rl}
\caption{Unnormalized vectors and all contexts of $B_{14,7}^{(4)}$.}\label{S-tab:vec-4c}\\\toprule $i$ & $r_i$ & $j$ & Context\\\midrule\endfirsthead
\toprule $i$ & $r_i$ & $j$ & Context\\\midrule\endhead\bottomrule\endfoot
$0$ & $(0,0,0,1)$ & $1$ & $\{0,\allowbreak 1,\allowbreak 2,\allowbreak 3\}$\\*
$1$ & $(0,1,1,0)$ & $2$ & $\{0,\allowbreak 7,\allowbreak 8,\allowbreak 9\}$\\*
$2$ & $(1,0,0,0)$ & $3$ & $\{0,\allowbreak 2,\allowbreak 8,\allowbreak 5\}$\\*
$3$ & $(0,1,-1,0)$ & $4$ & $\{0,\allowbreak 3,\allowbreak 11,\allowbreak 14\}$\\*
$5$ & $(0,0,1,0)$ & $5$ & $\{0,\allowbreak 7,\allowbreak 11,\allowbreak 15\}$\\*
$7$ & $(1,0,1,0)$ & $6$ & $\{0,\allowbreak 1,\allowbreak 10,\allowbreak 19\}$\\*
$8$ & $(0,1,0,0)$ & $7$ & $\{0,\allowbreak 9,\allowbreak 10,\allowbreak 23\}$\\*
$9$ & $(1,0,-1,0)$ &  & \\*
$10$ & $(1,-1,1,0)$ &  & \\*
$11$ & $(1,-1,-1,0)$ &  & \\*
$14$ & $(2,1,1,0)$ &  & \\*
$15$ & $(1,2,-1,0)$ &  & \\*
$19$ & $(2,1,-1,0)$ &  & \\*
$23$ & $(1,2,1,0)$ &  & \\
\end{longtable}\endgroup
\Needspace{286pt}\subsection{4D panel (d): \texorpdfstring{$L_{18,7}^{(4)}$}{gadget 4d}}\label{S-app:vec-4d}
\begingroup\footnotesize\setlength{\tabcolsep}{3pt}\setlength{\LTpre}{5pt}\setlength{\LTpost}{8pt}
\Needspace{9\baselineskip}\begin{longtable}{rl@{\hspace{12mm}}rl}
\caption{Unnormalized vectors and all contexts of $L_{18,7}^{(4)}$.}\label{S-tab:vec-4d}\\\toprule $i$ & $r_i$ & $j$ & Context\\\midrule\endfirsthead
\toprule $i$ & $r_i$ & $j$ & Context\\\midrule\endhead\bottomrule\endfoot
$1$ & $(1,0,0,0)$ & $1$ & $\{1,\allowbreak 2,\allowbreak 4,\allowbreak 3\}$\\*
$2$ & $(0,1,0,0)$ & $2$ & $\{5,\allowbreak 8,\allowbreak 7,\allowbreak 6\}$\\*
$3$ & $(0,0,1,0)$ & $3$ & $\{4,\allowbreak 9,\allowbreak 10,\allowbreak 8\}$\\*
$4$ & $(0,0,0,1)$ & $4$ & $\{11,\allowbreak 1,\allowbreak 12,\allowbreak 5\}$\\*
$5$ & $(0,1,-1,-1)$ & $5$ & $\{11,\allowbreak 14,\allowbreak 13,\allowbreak 2\}$\\*
$6$ & $(4,-5,-11,6)$ & $6$ & $\{15,\allowbreak 3,\allowbreak 16,\allowbreak 7\}$\\*
$7$ & $(1,-4,0,-4)$ & $7$ & $\{15,\allowbreak 17,\allowbreak 18,\allowbreak 6\}$\\*
$8$ & $(4,1,1,0)$ &  & \\*
$9$ & $(1,-5,1,0)$ &  & \\*
$10$ & $(2,-1,-7,0)$ &  & \\*
$11$ & $(0,0,1,-1)$ &  & \\*
$12$ & $(0,2,1,1)$ &  & \\*
$13$ & $(1,0,1,1)$ &  & \\*
$14$ & $(2,0,-1,-1)$ &  & \\*
$15$ & $(4,2,0,-1)$ &  & \\*
$16$ & $(4,-5,0,6)$ &  & \\*
$17$ & $(5,-9,7,2)$ &  & \\*
$18$ & $(2,113,77,234)$ &  & \\
\end{longtable}\endgroup
\Needspace{275pt}\subsection{4D panel (e): \texorpdfstring{$D_{17,9}^{(4)}$}{gadget 4e}}\label{S-app:vec-4e}
\begingroup\footnotesize\setlength{\tabcolsep}{3pt}\setlength{\LTpre}{5pt}\setlength{\LTpost}{8pt}
\Needspace{9\baselineskip}\begin{longtable}{rl@{\hspace{12mm}}rl}
\caption{Unnormalized vectors and all contexts of $D_{17,9}^{(4)}$.}\label{S-tab:vec-4e}\\\toprule $i$ & $r_i$ & $j$ & Context\\\midrule\endfirsthead
\toprule $i$ & $r_i$ & $j$ & Context\\\midrule\endhead\bottomrule\endfoot
$0$ & $(0,0,0,1)$ & $1$ & $\{0,\allowbreak 1,\allowbreak 2,\allowbreak 3\}$\\*
$1$ & $(0,1,2,0)$ & $2$ & $\{0,\allowbreak 7,\allowbreak 8,\allowbreak 9\}$\\*
$2$ & $(1,0,0,0)$ & $3$ & $\{0,\allowbreak 4,\allowbreak 5,\allowbreak 6\}$\\*
$3$ & $(0,2,-1,0)$ & $4$ & $\{0,\allowbreak 2,\allowbreak 8,\allowbreak 5\}$\\*
$4$ & $(1,2,0,0)$ & $5$ & $\{0,\allowbreak 3,\allowbreak 11,\allowbreak 14\}$\\*
$5$ & $(0,0,1,0)$ & $6$ & $\{0,\allowbreak 7,\allowbreak 11,\allowbreak 15\}$\\*
$6$ & $(2,-1,0,0)$ & $7$ & $\{0,\allowbreak 9,\allowbreak 10,\allowbreak 23\}$\\*
$7$ & $(1,0,1,0)$ & $8$ & $\{0,\allowbreak 4,\allowbreak 11,\allowbreak 24\}$\\*
$8$ & $(0,1,0,0)$ & $9$ & $\{0,\allowbreak 4,\allowbreak 10,\allowbreak 18\}$\\*
$9$ & $(1,0,-1,0)$ &  & \\*
$10$ & $(2,-1,2,0)$ &  & \\*
$11$ & $(2,-1,-2,0)$ &  & \\*
$14$ & $(5,2,4,0)$ &  & \\*
$15$ & $(1,4,-1,0)$ &  & \\*
$18$ & $(4,-2,-5,0)$ &  & \\*
$23$ & $(1,4,1,0)$ &  & \\*
$24$ & $(4,-2,5,0)$ &  & \\
\end{longtable}\endgroup
\Needspace{330pt}\subsection{4D panel (f): \texorpdfstring{$L_{22,9}^{(4)}$}{gadget 4f}}\label{S-app:vec-4f}
\begingroup\footnotesize\setlength{\tabcolsep}{3pt}\setlength{\LTpre}{5pt}\setlength{\LTpost}{8pt}
\Needspace{9\baselineskip}\begin{longtable}{rl@{\hspace{12mm}}rl}
\caption{Unnormalized vectors and all contexts of $L_{22,9}^{(4)}$.}\label{S-tab:vec-4f}\\\toprule $i$ & $r_i$ & $j$ & Context\\\midrule\endfirsthead
\toprule $i$ & $r_i$ & $j$ & Context\\\midrule\endhead\bottomrule\endfoot
$1$ & $(1,0,0,0)$ & $1$ & $\{1,\allowbreak 2,\allowbreak 4,\allowbreak 3\}$\\*
$2$ & $(0,1,0,0)$ & $2$ & $\{5,\allowbreak 8,\allowbreak 7,\allowbreak 6\}$\\*
$3$ & $(0,0,1,0)$ & $3$ & $\{4,\allowbreak 9,\allowbreak 10,\allowbreak 8\}$\\*
$4$ & $(0,0,0,1)$ & $4$ & $\{11,\allowbreak 1,\allowbreak 12,\allowbreak 5\}$\\*
$5$ & $(0,1,-1,-1)$ & $5$ & $\{11,\allowbreak 14,\allowbreak 13,\allowbreak 2\}$\\*
$6$ & $(4,-5,-11,6)$ & $6$ & $\{15,\allowbreak 3,\allowbreak 16,\allowbreak 7\}$\\*
$7$ & $(1,-4,0,-4)$ & $7$ & $\{15,\allowbreak 17,\allowbreak 18,\allowbreak 6\}$\\*
$8$ & $(4,1,1,0)$ & $8$ & $\{11,\allowbreak 21,\allowbreak 22,\allowbreak 19\}$\\*
$9$ & $(1,-5,1,0)$ & $9$ & $\{13,\allowbreak 20,\allowbreak 15,\allowbreak 19\}$\\*
$10$ & $(2,-1,-7,0)$ &  & \\*
$11$ & $(0,0,1,-1)$ &  & \\*
$12$ & $(0,2,1,1)$ &  & \\*
$13$ & $(1,0,4,4)$ &  & \\*
$14$ & $(8,0,-1,-1)$ &  & \\*
$15$ & $(4,2,0,-1)$ &  & \\*
$16$ & $(4,-5,0,6)$ &  & \\*
$17$ & $(5,-9,7,2)$ &  & \\*
$18$ & $(2,113,77,234)$ &  & \\*
$19$ & $(16,-33,-2,-2)$ &  & \\*
$20$ & $(4,2,-21,20)$ &  & \\*
$21$ & $(22,12,-11,-11)$ &  & \\*
$22$ & $(129,44,153,153)$ &  & \\
\end{longtable}\endgroup

\FloatBarrier
\subsection{\texorpdfstring{Hesse $H_{37}$: vectors, contexts, and non-collinearity}{Hesse H37: vectors, contexts, and non-collinearity}}\label{S-h:app:vectors-contexts}

This appendix uses exactly the vertex numbers and context numbers of Fig.~\ref{h:fig:hesse37-hypergraph}. All indices $a,k$ in the family notation are in $\Z_3$. The vectors are unnormalized and the inner product is Hermitian. Table~\ref{S-h:tab:app-vectors} lists every vector, squared norm, color, and incident context. Table~\ref{S-h:tab:app-contexts} lists all 24 contexts with their actual vectors. Together they specify the faithful orthogonal representation completely.

\subsubsection{All 37 vectors, their colors, and incident contexts}

\begingroup
\footnotesize
\setlength{\tabcolsep}{3pt}
\renewcommand{\arraystretch}{1.08}
\Needspace{16\baselineskip}
\begin{longtable}{rclrrl}
\caption{Complete enumeration of the 37 vectors of $H_{37}$. The color $c(i)\in\{0,1,2,3\}$ is the explicit four-coloring in the main text. Context indices refer to Table~\ref{S-h:tab:app-contexts}.}\label{S-h:tab:app-vectors}\\
\toprule
$i$ & Ray & $v_i$ & $\|v_i\|^2$ & $c(i)$ & Incident contexts \\
\midrule
\endfirsthead
\caption[]{Complete enumeration of the 37 vectors of $H_{37}$. The color $c(i)\in\{0,1,2,3\}$ is the explicit four-coloring in the main text. Context indices refer to Table~\ref{S-h:tab:app-contexts}. (continued).}\\
\toprule
$i$ & Ray & $v_i$ & $\|v_i\|^2$ & $c(i)$ & Incident contexts \\
\midrule
\endhead
\midrule
\multicolumn{6}{r}{\emph{Continued on the next page.}}\\
\endfoot
\bottomrule
\endlastfoot
1 & $z$ & $(1,0,0)$ & 1 & 0 & $\{4,5,6\}$ \\
2 & $M_{00}$ & $(1,1,1)$ & 3 & 0 & $\{1,7,16\}$ \\
3 & $M_{01}$ & $(1,\zeta,\zeta^2)$ & 3 & 1 & $\{1,8,17\}$ \\
4 & $M_{02}$ & $(1,\zeta^2,\zeta)$ & 3 & 2 & $\{1,9,18\}$ \\
5 & $M_{10}$ & $(1,\zeta,\zeta)$ & 3 & 0 & $\{2,10,21\}$ \\
6 & $M_{11}$ & $(1,\zeta^2,1)$ & 3 & 1 & $\{2,11,19\}$ \\
7 & $M_{12}$ & $(1,1,\zeta^2)$ & 3 & 2 & $\{2,12,20\}$ \\
8 & $M_{20}$ & $(1,\zeta^2,\zeta^2)$ & 3 & 0 & $\{3,13,23\}$ \\
9 & $M_{21}$ & $(1,1,\zeta)$ & 3 & 1 & $\{3,14,24\}$ \\
10 & $M_{22}$ & $(1,\zeta,1)$ & 3 & 2 & $\{3,15,22\}$ \\
11 & $P_{0}$ & $(0,1,-1)$ & 2 & 3 & $\{4,7,10,13\}$ \\
12 & $P_{1}$ & $(0,1,-\zeta)$ & 2 & 3 & $\{5,8,11,14\}$ \\
13 & $P_{2}$ & $(0,1,-\zeta^2)$ & 2 & 3 & $\{6,9,12,15\}$ \\
14 & $Q_{0}$ & $(-1,0,1)$ & 2 & 3 & $\{16,19,22\}$ \\
15 & $Q_{1}$ & $(-\zeta,0,1)$ & 2 & 3 & $\{17,20,23\}$ \\
16 & $Q_{2}$ & $(-\zeta^2,0,1)$ & 2 & 3 & $\{18,21,24\}$ \\
17 & $A_{0}$ & $(0,1,1)$ & 2 & 1 & $\{4\}$ \\
18 & $A_{1}$ & $(0,1,\zeta)$ & 2 & 1 & $\{5\}$ \\
19 & $A_{2}$ & $(0,1,\zeta^2)$ & 2 & 1 & $\{6\}$ \\
20 & $B_{00}$ & $(-2,1,1)$ & 6 & 1 & $\{7\}$ \\
21 & $B_{01}$ & $(-2,\zeta,\zeta^2)$ & 6 & 2 & $\{8\}$ \\
22 & $B_{02}$ & $(-2,\zeta^2,\zeta)$ & 6 & 0 & $\{9\}$ \\
23 & $B_{10}$ & $(-2,\zeta,\zeta)$ & 6 & 1 & $\{10\}$ \\
24 & $B_{11}$ & $(-2,\zeta^2,1)$ & 6 & 2 & $\{11\}$ \\
25 & $B_{12}$ & $(-2,1,\zeta^2)$ & 6 & 0 & $\{12\}$ \\
26 & $B_{20}$ & $(-2,\zeta^2,\zeta^2)$ & 6 & 1 & $\{13\}$ \\
27 & $B_{21}$ & $(-2,1,\zeta)$ & 6 & 2 & $\{14\}$ \\
28 & $B_{22}$ & $(-2,\zeta,1)$ & 6 & 0 & $\{15\}$ \\
29 & $C_{00}$ & $(1,-2,1)$ & 6 & 1 & $\{16\}$ \\
30 & $C_{01}$ & $(1,-2\zeta,\zeta^2)$ & 6 & 2 & $\{17\}$ \\*
31 & $C_{02}$ & $(1,-2\zeta^2,\zeta)$ & 6 & 0 & $\{18\}$ \\*
32 & $C_{10}$ & $(1,-2\zeta^2,1)$ & 6 & 2 & $\{19\}$ \\*
33 & $C_{11}$ & $(1,-2,\zeta^2)$ & 6 & 0 & $\{20\}$ \\*
34 & $C_{12}$ & $(1,-2\zeta,\zeta)$ & 6 & 1 & $\{21\}$ \\*
35 & $C_{20}$ & $(1,-2\zeta,1)$ & 6 & 0 & $\{22\}$ \\*
36 & $C_{21}$ & $(1,-2\zeta^2,\zeta^2)$ & 6 & 1 & $\{23\}$ \\*
37 & $C_{22}$ & $(1,-2,\zeta)$ & 6 & 2 & $\{24\}$ \\*
\end{longtable}
\endgroup

\Needspace{20\baselineskip}
\subsubsection{All 24 contexts and their orthogonal vectors}

\begingroup
\scriptsize
\setlength{\tabcolsep}{2pt}
\renewcommand{\arraystretch}{1.08}
\Needspace{16\baselineskip}
\begin{longtable}{rcll}
\caption{All 24 orthogonal triads, numbered in the same order as the context colors in Fig.~\ref{h:fig:hesse37-hypergraph}. Vectors are given in increasing order of their vertex numbers. Each triad contributes three orthogonality edges; the 72 resulting pairs are all the orthogonal pairs.}\label{S-h:tab:app-contexts}\\
\toprule
Context & Family & Vertices & Unnormalized vectors \\
\midrule
\endfirsthead
\caption[]{All 24 orthogonal triads, numbered in the same order as the context colors in Fig.~\ref{h:fig:hesse37-hypergraph}. Vectors are given in increasing order of their vertex numbers. Each triad contributes three orthogonality edges; the 72 resulting pairs are all the orthogonal pairs. (continued).}\\
\toprule
Context & Family & Vertices & Unnormalized vectors \\
\midrule
\endhead
\midrule
\multicolumn{4}{r}{\emph{Continued on the next page.}}\\
\endfoot
\bottomrule
\endlastfoot
$e_{1}$ & $\mathcal M_0$ & $\{2,3,4\}$ & $\bigl\{(1,1,1),\;(1,\zeta,\zeta^2),\;(1,\zeta^2,\zeta)\bigr\}$ \\
$e_{2}$ & $\mathcal M_1$ & $\{5,6,7\}$ & $\bigl\{(1,\zeta,\zeta),\;(1,\zeta^2,1),\;(1,1,\zeta^2)\bigr\}$ \\
$e_{3}$ & $\mathcal M_2$ & $\{8,9,10\}$ & $\bigl\{(1,\zeta^2,\zeta^2),\;(1,1,\zeta),\;(1,\zeta,1)\bigr\}$ \\
$e_{4}$ & $\mathcal A_0$ & $\{1,11,17\}$ & $\bigl\{(1,0,0),\;(0,1,-1),\;(0,1,1)\bigr\}$ \\
$e_{5}$ & $\mathcal A_1$ & $\{1,12,18\}$ & $\bigl\{(1,0,0),\;(0,1,-\zeta),\;(0,1,\zeta)\bigr\}$ \\
$e_{6}$ & $\mathcal A_2$ & $\{1,13,19\}$ & $\bigl\{(1,0,0),\;(0,1,-\zeta^2),\;(0,1,\zeta^2)\bigr\}$ \\
$e_{7}$ & $\mathcal B_{00}$ & $\{2,11,20\}$ & $\bigl\{(1,1,1),\;(0,1,-1),\;(-2,1,1)\bigr\}$ \\
$e_{8}$ & $\mathcal B_{01}$ & $\{3,12,21\}$ & $\bigl\{(1,\zeta,\zeta^2),\;(0,1,-\zeta),\;(-2,\zeta,\zeta^2)\bigr\}$ \\
$e_{9}$ & $\mathcal B_{02}$ & $\{4,13,22\}$ & $\bigl\{(1,\zeta^2,\zeta),\;(0,1,-\zeta^2),\;(-2,\zeta^2,\zeta)\bigr\}$ \\
$e_{10}$ & $\mathcal B_{10}$ & $\{5,11,23\}$ & $\bigl\{(1,\zeta,\zeta),\;(0,1,-1),\;(-2,\zeta,\zeta)\bigr\}$ \\
$e_{11}$ & $\mathcal B_{11}$ & $\{6,12,24\}$ & $\bigl\{(1,\zeta^2,1),\;(0,1,-\zeta),\;(-2,\zeta^2,1)\bigr\}$ \\
$e_{12}$ & $\mathcal B_{12}$ & $\{7,13,25\}$ & $\bigl\{(1,1,\zeta^2),\;(0,1,-\zeta^2),\;(-2,1,\zeta^2)\bigr\}$ \\
$e_{13}$ & $\mathcal B_{20}$ & $\{8,11,26\}$ & $\bigl\{(1,\zeta^2,\zeta^2),\;(0,1,-1),\;(-2,\zeta^2,\zeta^2)\bigr\}$ \\
$e_{14}$ & $\mathcal B_{21}$ & $\{9,12,27\}$ & $\bigl\{(1,1,\zeta),\;(0,1,-\zeta),\;(-2,1,\zeta)\bigr\}$ \\
$e_{15}$ & $\mathcal B_{22}$ & $\{10,13,28\}$ & $\bigl\{(1,\zeta,1),\;(0,1,-\zeta^2),\;(-2,\zeta,1)\bigr\}$ \\
$e_{16}$ & $\mathcal C_{00}$ & $\{2,14,29\}$ & $\bigl\{(1,1,1),\;(-1,0,1),\;(1,-2,1)\bigr\}$ \\
\noalign{\penalty-10000}
$e_{17}$ & $\mathcal C_{01}$ & $\{3,15,30\}$ & $\bigl\{(1,\zeta,\zeta^2),\;(-\zeta,0,1),\;(1,-2\zeta,\zeta^2)\bigr\}$ \\*
$e_{18}$ & $\mathcal C_{02}$ & $\{4,16,31\}$ & $\bigl\{(1,\zeta^2,\zeta),\;(-\zeta^2,0,1),\;(1,-2\zeta^2,\zeta)\bigr\}$ \\*
$e_{19}$ & $\mathcal C_{10}$ & $\{6,14,32\}$ & $\bigl\{(1,\zeta^2,1),\;(-1,0,1),\;(1,-2\zeta^2,1)\bigr\}$ \\*
$e_{20}$ & $\mathcal C_{11}$ & $\{7,15,33\}$ & $\bigl\{(1,1,\zeta^2),\;(-\zeta,0,1),\;(1,-2,\zeta^2)\bigr\}$ \\*
$e_{21}$ & $\mathcal C_{12}$ & $\{5,16,34\}$ & $\bigl\{(1,\zeta,\zeta),\;(-\zeta^2,0,1),\;(1,-2\zeta,\zeta)\bigr\}$ \\*
$e_{22}$ & $\mathcal C_{20}$ & $\{10,14,35\}$ & $\bigl\{(1,\zeta,1),\;(-1,0,1),\;(1,-2\zeta,1)\bigr\}$ \\*
$e_{23}$ & $\mathcal C_{21}$ & $\{8,15,36\}$ & $\bigl\{(1,\zeta^2,\zeta^2),\;(-\zeta,0,1),\;(1,-2\zeta^2,\zeta^2)\bigr\}$ \\*
$e_{24}$ & $\mathcal C_{22}$ & $\{9,16,37\}$ & $\bigl\{(1,1,\zeta),\;(-\zeta^2,0,1),\;(1,-2,\zeta)\bigr\}$ \\*
\end{longtable}
\endgroup

\subsubsection{Norm shells and a non-collinearity certificate}

The four squared-norm shells have sizes
\begin{equation}
\#\{i:\|v_i\|^2=1\}=1,\quad
\#\{i:\|v_i\|^2=2\}=9,\quad
\#\{i:\|v_i\|^2=3\}=9,\quad
\#\{i:\|v_i\|^2=6\}=18.
\end{equation}
Exact evaluation over the $\binom{37}{2}=666$ unordered pairs gives
\begin{equation}\label{S-h:eq:app-max-overlap}
\max_{i<j}\frac{|\langle v_i,v_j\rangle|^2}{\|v_i\|^2\|v_j\|^2}=\frac34<1.
\end{equation}
For example, the maximum is attained by $v_{11}=P_0=(0,1,-1)$ and $v_{29}=C_{00}=(1,-2,1)$:
\begin{equation}
\frac{|\langle v_{11},v_{29}\rangle|^2}{\|v_{11}\|^2\|v_{29}\|^2}
=\frac{(-3)^2}{2\cdot6}=\frac34.
\end{equation}
Thus no two printed vectors represent the same ray. The maximum and all 72 zero inner products are checked exactly by the accompanying verifier; the family-wise argument in the main text explains why there are no further orthogonalities.

\subsection{\texorpdfstring{Four-dimensional $H_{\YO}^{(4)}$ and $H_{65}$: vectors and contexts}{Four-dimensional HYO(4) and H65: vectors and contexts}}

For direct human and computational verification, this appendix enumerates all unnormalized integer vectors in $\mathbb{Z}^4$ and all four-element contexts for the two four-dimensional examples $H_{65}$ and $H_{\YO}^{(4)}$:
\begin{enumerate}
\item The \emph{canonical linear hypergraph $H_{65}=(V_{65},E_{29})$} (Theorem~\ref{d:thm:h65}, Appendix~\ref{S-d:app:h65}), which has exactly \textbf{$65$ pairwise non-collinear integer vectors} ($17$ core vectors $v_1,\ldots,v_{17}$ and $48$ degree-$1$ completion vectors $v_{18},\ldots,v_{65}$), \textbf{$29$ four-element contexts} ($5$ core contexts $C_1,\ldots,C_5$ and $24$ bridge-completion contexts $E_1,\ldots,E_{24}$), and \textbf{$174$ orthogonality edges} ($29\times\binom{4}{2}=174$).
\item The \emph{non-linear warm-up apex suspension $H_{\YO}^{(4)}=(V_{26},E_{16}^{(4)})$} (Proposition~\ref{d:prop:warmup}, Appendix~\ref{S-d:app:hyo4}), which has exactly \textbf{$26$ pairwise non-collinear integer vectors} ($13$ Yu--Oh rays $v_1,\ldots,v_{13}$, $12$ completion rays $v_{14},\ldots,v_{25}$ in $\R^3\times\{0\}$, and $1$ apex ray $v_{26}$), \textbf{$16$ four-element contexts}, and \textbf{$73$ orthogonality edges} ($16\times\binom{3}{2}=48$ edges in $\R^3\times\{0\}$ plus $25$ apex edges).
\end{enumerate}

\paragraph*{Canonical primitive normal form and human-verifiable non-collinearity.}
Every vector $v_i=(x_{i,1},x_{i,2},x_{i,3},x_{i,4})^T\in\mathbb{Z}^4$ listed below in both $H_{65}$ and $H_{\YO}^{(4)}$ is written in \emph{canonical primitive projective form}:
\begin{equation}\label{S-d:eq:canonical-primitive}
\begin{gathered}
\gcd\bigl(|x_{i,1}|,|x_{i,2}|,|x_{i,3}|,|x_{i,4}|\bigr)=1,\\
\text{first nonzero entry of }v_i > 0.
\end{gathered}
\end{equation}
Because every projective ray $[v]\in\PP(\Q^4)$ possesses a \emph{unique} representative in $\mathbb{Z}^4$ satisfying Eq.~\eqref{S-d:eq:canonical-primitive}, two vectors $v_i,v_j$ in canonical primitive form are collinear (proportional over $\R$ or $\C$) \emph{if and only if they are identical $4$-tuples} ($v_i=v_j$) and, in particular, have the same squared Euclidean norm $\|v_i\|^2=\|v_j\|^2$.

\subsubsection{The canonical linear \texorpdfstring{$65$}{65}-vertex, \texorpdfstring{$29$}{29}-context hypergraph \texorpdfstring{$H_{65}$}{H\_65}}\label{S-d:app:h65}

Table~\ref{S-d:tab:app-h65-vectors} enumerates all $65$ unnormalized integer vectors $v_1,\ldots,v_{65}\in\mathbb{Z}^4$ of $H_{65}$, together with their squared norms $\|v_i\|^2$, their colors $c(i)\in\{0,1,2,3,4\}$ in the strong $5$-coloring extending Eq.~\eqref{S-d:eq:h65-coordinates}, and the contexts in $E_{29}$ to which each vertex belongs. Table~\ref{S-d:tab:app-h65-contexts} enumerates all $29$ contexts $C_1,\ldots,C_5,E_1,\ldots,E_{24}\in E_{29}$, displaying side by side both the vertex indices $\{i,j,k,\ell\}$ and the four mutually orthogonal integer vectors $\{v_i,v_j,v_k,v_\ell\}\subset\mathbb{Z}^4$.

To verify by hand that all $65$ vectors in Table~\ref{S-d:tab:app-h65-vectors} are pairwise non-collinear across all $\binom{65}{2}=2{,}080$ pairs, observe that all $65$ tuples satisfy Eq.~\eqref{S-d:eq:canonical-primitive} and partition by squared norm $\|v_i\|^2$ into seven small, disjoint norm shells ($4+6+7+12+12+12+12=65$):
\begin{itemize}
\item $\|v_i\|^2=1$ ($4$ vectors): indices $\{1,2,3,4\}$, the $4$ standard coordinate basis vectors;
\item $\|v_i\|^2=2$ ($6$ vectors): indices $\{12,\ldots,17\}$, one vector for each of the $\binom{4}{2}=6$ coordinate-pair supports;
\item $\|v_i\|^2=4$ ($7$ vectors): indices $\{5,\ldots,11\}$, seven distinct $\pm1$ vectors with $x_{i,1}=+1$;
\item $\|v_i\|^2=11$ ($12$ vectors): indices $\{18,20,22,25,26,28,30,33,34,36,39,41\}$, each having a single zero coordinate (three each with $x_1=0$, $x_2=0$, $x_3=0$, $x_4=0$) and distinct sign/entry patterns of $\{0,\pm1,\pm1,\pm3\}$;
\item $\|v_i\|^2=18$ ($12$ vectors): indices $\{43,45,46,49,51,53,54,56,59,61,62,65\}$, twelve distinct permutations and sign patterns of $\{\pm1,\pm2,\pm2,\pm3\}$;
\item $\|v_i\|^2=22$ ($12$ vectors): indices $\{19,21,23,24,27,29,31,32,35,37,38,40\}$, each having a single zero coordinate (three each with $x_1=0$, $x_2=0$, $x_3=0$, $x_4=0$) and distinct patterns of $\{0,\pm2,\pm3,\pm3\}$;
\item $\|v_i\|^2=36$ ($12$ vectors): indices $\{42,44,47,48,50,52,55,57,58,60,63,64\}$, twelve distinct permutations and sign patterns of $\{\pm1,\pm1,\pm3,\pm5\}$.
\end{itemize}
Because no two $4$-tuples within any of the seven norm shells coincide, no two of the $65$ vectors are collinear; quantitatively, the maximum squared cosine over all $\binom{65}{2}=2{,}080$ pairs is $\max_{1\le i<j\le 65}(v_i\cdot v_j)^2/(\|v_i\|^2\|v_j\|^2)=19^2/(11\times 36)=361/396\approx 0.9116<1$, attained by $(v_{41},v_{47})$, whose squared norms are $11$ and $36$ and whose inner product is $19$.

\begingroup
\scriptsize
\setlength{\tabcolsep}{3pt}
\Needspace{24\baselineskip}
\begin{longtable}{rllll|rllll}
\caption{Complete enumeration of the $65$ unnormalized integer vectors $v_i\in\mathbb{Z}^4$ of $H_{65}$, their squared norms $\|v_i\|^2$, their colors $c(i)\in\{0,1,2,3,4\}$, and the contexts in $E_{29}$ containing vertex $i$.}\label{S-d:tab:app-h65-vectors}\\
\toprule

$i$ & $v_i\in\mathbb{Z}^4$ & $\|v_i\|^2$ & $c(i)$ & Contexts containing $i$ &
$i$ & $v_i\in\mathbb{Z}^4$ & $\|v_i\|^2$ & $c(i)$ & Contexts containing $i$ \\\midrule
\endfirsthead
\caption[]{All 65 vectors of $H_{65}$ (continued).}\\
\toprule

$i$ & $v_i\in\mathbb{Z}^4$ & $\|v_i\|^2$ & $c(i)$ & Contexts containing $i$ &
$i$ & $v_i\in\mathbb{Z}^4$ & $\|v_i\|^2$ & $c(i)$ & Contexts containing $i$ \\\midrule
\endhead
\midrule
\multicolumn{10}{r}{\emph{Continued on the next page.}}\\
\endfoot
\bottomrule\endlastfoot
$1$  & $(1,0,0,0)$     & $1$  & $0$ & $C_1,E_3,E_7,E_{11}$ &
$34$ & $(1,-1,0,3)$    & $11$ & $1$ & $E_9=\{3,16,34,35\}$ \\
$2$  & $(0,1,0,0)$     & $1$  & $1$ & $C_1,E_1,E_5,E_{12}$ &
$35$ & $(3,-3,0,-2)$   & $22$ & $3$ & $E_9=\{3,16,34,35\}$ \\
$3$  & $(0,0,1,0)$     & $1$  & $2$ & $C_1,E_2,E_8,E_9$ &
$36$ & $(1,-1,3,0)$    & $11$ & $1$ & $E_{10}=\{4,16,36,37\}$ \\
$4$  & $(0,0,0,1)$     & $1$  & $3$ & $C_1,E_4,E_6,E_{10}$ &
$37$ & $(3,-3,-2,0)$   & $22$ & $2$ & $E_{10}=\{4,16,36,37\}$ \\
$5$  & $(1,1,1,1)$     & $4$  & $0$ & $C_2,E_{13},E_{19},E_{23}$ &
$38$ & $(0,2,3,3)$     & $22$ & $1$ & $E_{11}=\{1,17,38,39\}$ \\
$6$  & $(1,1,-1,-1)$   & $4$  & $1$ & $C_2,E_{15},E_{17},E_{24}$ &
$39$ & $(0,3,-1,-1)$   & $11$ & $3$ & $E_{11}=\{1,17,38,39\}$ \\
$7$  & $(1,-1,1,-1)$   & $4$  & $2$ & $C_2,E_{16},E_{20},E_{21}$ &
$40$ & $(2,0,-3,-3)$   & $22$ & $0$ & $E_{12}=\{2,17,40,41\}$ \\
$8$  & $(1,-1,-1,1)$   & $4$  & $4$ & $C_2,E_{14},E_{18},E_{22}$ &
$41$ & $(3,0,1,1)$     & $11$ & $3$ & $E_{12}=\{2,17,40,41\}$ \\
$9$  & $(1,1,-1,1)$    & $4$  & $0$ & $C_3,C_4$ &
$42$ & $(1,-5,3,1)$    & $36$ & $1$ & $E_{13}=\{5,12,42,43\}$ \\
$10$ & $(1,-1,1,1)$    & $4$  & $1$ & $C_3,C_5$ &
$43$ & $(2,-1,-3,2)$   & $18$ & $2$ & $E_{13}=\{5,12,42,43\}$ \\
$11$ & $(1,-1,-1,-1)$  & $4$  & $3$ & $C_4,C_5$ &
$44$ & $(1,-3,5,1)$    & $36$ & $0$ & $E_{14}=\{8,12,44,45\}$ \\
$12$ & $(1,0,0,-1)$    & $2$  & $3$ & $C_3,E_1,E_2,E_{13},E_{14}$ &
$45$ & $(2,3,1,2)$     & $18$ & $1$ & $E_{14}=\{8,12,44,45\}$ \\
$13$ & $(0,1,1,0)$     & $2$  & $4$ & $C_3,E_3,E_4,E_{15},E_{16}$ &
$46$ & $(1,-2,2,-3)$   & $18$ & $0$ & $E_{15}=\{6,13,46,47\}$ \\
$14$ & $(1,0,1,0)$     & $2$  & $2$ & $C_4,E_5,E_6,E_{17},E_{18}$ &
$47$ & $(5,-1,1,3)$    & $36$ & $2$ & $E_{15}=\{6,13,46,47\}$ \\
$15$ & $(0,1,0,-1)$    & $2$  & $1$ & $C_4,E_7,E_8,E_{19},E_{20}$ &
$48$ & $(3,-1,1,5)$    & $36$ & $0$ & $E_{16}=\{7,13,48,49\}$ \\
$16$ & $(1,1,0,0)$     & $2$  & $0$ & $C_5,E_9,E_{10},E_{21},E_{22}$ &
$49$ & $(3,2,-2,-1)$   & $18$ & $1$ & $E_{16}=\{7,13,48,49\}$ \\
$17$ & $(0,0,1,-1)$    & $2$  & $2$ & $C_5,E_{11},E_{12},E_{23},E_{24}$ &
$50$ & $(1,3,-1,5)$    & $36$ & $0$ & $E_{17}=\{6,14,50,51\}$ \\
$18$ & $(1,0,-3,1)$    & $11$ & $0$ & $E_1=\{2,12,18,19\}$ &
$51$ & $(2,-3,-2,1)$   & $18$ & $3$ & $E_{17}=\{6,14,50,51\}$ \\
$19$ & $(3,0,2,3)$     & $22$ & $2$ & $E_1=\{2,12,18,19\}$ &
$52$ & $(1,5,-1,3)$    & $36$ & $0$ & $E_{18}=\{8,14,52,53\}$ \\
$20$ & $(1,3,0,1)$     & $11$ & $0$ & $E_2=\{3,12,20,21\}$ &
$53$ & $(2,1,-2,-3)$   & $18$ & $1$ & $E_{18}=\{8,14,52,53\}$ \\
$21$ & $(3,-2,0,3)$    & $22$ & $1$ & $E_2=\{3,12,20,21\}$ &
$54$ & $(3,-2,1,-2)$   & $18$ & $2$ & $E_{19}=\{5,15,54,55\}$ \\
$22$ & $(0,1,-1,3)$    & $11$ & $1$ & $E_3=\{1,13,22,23\}$ &
$55$ & $(3,1,-5,1)$    & $36$ & $3$ & $E_{19}=\{5,15,54,55\}$ \\
$23$ & $(0,3,-3,-2)$   & $22$ & $2$ & $E_3=\{1,13,22,23\}$ &
$56$ & $(1,2,3,2)$     & $18$ & $0$ & $E_{20}=\{7,15,56,57\}$ \\
$24$ & $(2,-3,3,0)$    & $22$ & $0$ & $E_4=\{4,13,24,25\}$ &
$57$ & $(5,1,-3,1)$    & $36$ & $3$ & $E_{20}=\{7,15,56,57\}$ \\
$25$ & $(3,1,-1,0)$    & $11$ & $1$ & $E_4=\{4,13,24,25\}$ &
$58$ & $(1,-1,-5,-3)$  & $36$ & $1$ & $E_{21}=\{7,16,58,59\}$ \\*
$26$ & $(1,0,-1,-3)$   & $11$ & $0$ & $E_5=\{2,14,26,27\}$ &
$59$ & $(2,-2,-1,3)$   & $18$ & $3$ & $E_{21}=\{7,16,58,59\}$ \\*
$27$ & $(3,0,-3,2)$    & $22$ & $3$ & $E_5=\{2,14,26,27\}$ &
$60$ & $(1,-1,-3,-5)$  & $36$ & $1$ & $E_{22}=\{8,16,60,61\}$ \\*
$28$ & $(1,-3,-1,0)$   & $11$ & $0$ & $E_6=\{4,14,28,29\}$ &
$61$ & $(2,-2,3,-1)$   & $18$ & $2$ & $E_{22}=\{8,16,60,61\}$ \\*
$29$ & $(3,2,-3,0)$    & $22$ & $1$ & $E_6=\{4,14,28,29\}$ &
$62$ & $(1,3,-2,-2)$   & $18$ & $1$ & $E_{23}=\{5,17,62,63\}$ \\*
$30$ & $(0,1,3,1)$     & $11$ & $2$ & $E_7=\{1,15,30,31\}$ &
$63$ & $(5,-3,-1,-1)$  & $36$ & $3$ & $E_{23}=\{5,17,62,63\}$ \\*
$31$ & $(0,3,-2,3)$    & $22$ & $3$ & $E_7=\{1,15,30,31\}$ &
$64$ & $(3,-5,-1,-1)$  & $36$ & $0$ & $E_{24}=\{6,17,64,65\}$ \\*
$32$ & $(2,3,0,3)$     & $22$ & $0$ & $E_8=\{3,15,32,33\}$ &
$65$ & $(3,1,2,2)$     & $18$ & $3$ & $E_{24}=\{6,17,64,65\}$ \\*
$33$ & $(3,-1,0,-1)$   & $11$ & $3$ & $E_8=\{3,15,32,33\}$ &
     &                 &      &     &
\end{longtable}
\endgroup

\begingroup
\scriptsize
\setlength{\tabcolsep}{3pt}
\Needspace{24\baselineskip}
\begin{longtable}{llllll}
\caption{Complete enumeration of the $29$ four-element contexts of $H_{65}=(V_{65},E_{29})$, listing for each context both its vertex indices $\{i,j,k,\ell\}$ and the four mutually orthogonal unnormalized integer vectors $\{v_i,v_j,v_k,v_\ell\}\subset\mathbb{Z}^4$.}\label{S-d:tab:app-h65-contexts}\\
\toprule

Context & Vertex set & $v_i$ & $v_j$ & $v_k$ & $v_\ell$ \\\midrule
\endfirsthead
\caption[]{All 29 contexts of $H_{65}$ (continued).}\\
\toprule

Context & Vertex set & $v_i$ & $v_j$ & $v_k$ & $v_\ell$ \\\midrule
\endhead
\midrule
\multicolumn{6}{r}{\emph{Continued on the next page.}}\\
\endfoot
\bottomrule\endlastfoot
\multicolumn{6}{l}{\emph{Five core $4$-contexts $C_1,\ldots,C_5$ on $V_{\mathrm{core}}=\{1,\ldots,17\}$:}} \\
$C_1$ & $\{1,2,3,4\}$     & $v_1=(1,0,0,0)$   & $v_2=(0,1,0,0)$    & $v_3=(0,0,1,0)$     & $v_4=(0,0,0,1)$ \\
$C_2$ & $\{5,6,7,8\}$     & $v_5=(1,1,1,1)$   & $v_6=(1,1,-1,-1)$  & $v_7=(1,-1,1,-1)$   & $v_8=(1,-1,-1,1)$ \\
$C_3$ & $\{9,10,12,13\}$  & $v_9=(1,1,-1,1)$  & $v_{10}=(1,-1,1,1)$& $v_{12}=(1,0,0,-1)$ & $v_{13}=(0,1,1,0)$ \\
$C_4$ & $\{9,11,14,15\}$  & $v_9=(1,1,-1,1)$  & $v_{11}=(1,-1,-1,-1)$ & $v_{14}=(1,0,1,0)$ & $v_{15}=(0,1,0,-1)$ \\
$C_5$ & $\{10,11,16,17\}$ & $v_{10}=(1,-1,1,1)$ & $v_{11}=(1,-1,-1,-1)$ & $v_{16}=(1,1,0,0)$ & $v_{17}=(0,0,1,-1)$ \\ \hline
\multicolumn{6}{l}{\emph{Twelve coordinate-to-bridge completion contexts $E_1,\ldots,E_{12}$ (between $C_1$ and $B_1\cup B_2\cup B_3$):}} \\
$E_1$  & $\{2,12,18,19\}$ & $v_2=(0,1,0,0)$ & $v_{12}=(1,0,0,-1)$ & $v_{18}=(1,0,-3,1)$  & $v_{19}=(3,0,2,3)$ \\
$E_2$  & $\{3,12,20,21\}$ & $v_3=(0,0,1,0)$ & $v_{12}=(1,0,0,-1)$ & $v_{20}=(1,3,0,1)$   & $v_{21}=(3,-2,0,3)$ \\
$E_3$  & $\{1,13,22,23\}$ & $v_1=(1,0,0,0)$ & $v_{13}=(0,1,1,0)$  & $v_{22}=(0,1,-1,3)$  & $v_{23}=(0,3,-3,-2)$ \\
$E_4$  & $\{4,13,24,25\}$ & $v_4=(0,0,0,1)$ & $v_{13}=(0,1,1,0)$  & $v_{24}=(2,-3,3,0)$  & $v_{25}=(3,1,-1,0)$ \\
$E_5$  & $\{2,14,26,27\}$ & $v_2=(0,1,0,0)$ & $v_{14}=(1,0,1,0)$  & $v_{26}=(1,0,-1,-3)$ & $v_{27}=(3,0,-3,2)$ \\
$E_6$  & $\{4,14,28,29\}$ & $v_4=(0,0,0,1)$ & $v_{14}=(1,0,1,0)$  & $v_{28}=(1,-3,-1,0)$ & $v_{29}=(3,2,-3,0)$ \\
$E_7$  & $\{1,15,30,31\}$ & $v_1=(1,0,0,0)$ & $v_{15}=(0,1,0,-1)$ & $v_{30}=(0,1,3,1)$   & $v_{31}=(0,3,-2,3)$ \\
$E_8$  & $\{3,15,32,33\}$ & $v_3=(0,0,1,0)$ & $v_{15}=(0,1,0,-1)$ & $v_{32}=(2,3,0,3)$   & $v_{33}=(3,-1,0,-1)$ \\
$E_9$  & $\{3,16,34,35\}$ & $v_3=(0,0,1,0)$ & $v_{16}=(1,1,0,0)$  & $v_{34}=(1,-1,0,3)$  & $v_{35}=(3,-3,0,-2)$ \\
$E_{10}$ & $\{4,16,36,37\}$ & $v_4=(0,0,0,1)$ & $v_{16}=(1,1,0,0)$  & $v_{36}=(1,-1,3,0)$  & $v_{37}=(3,-3,-2,0)$ \\
$E_{11}$ & $\{1,17,38,39\}$ & $v_1=(1,0,0,0)$ & $v_{17}=(0,0,1,-1)$ & $v_{38}=(0,2,3,3)$   & $v_{39}=(0,3,-1,-1)$ \\
$E_{12}$ & $\{2,17,40,41\}$ & $v_2=(0,1,0,0)$ & $v_{17}=(0,0,1,-1)$ & $v_{40}=(2,0,-3,-3)$ & $v_{41}=(3,0,1,1)$ \\ \hline
\multicolumn{6}{l}{\emph{Twelve Hadamard-to-bridge completion contexts $E_{13},\ldots,E_{24}$ (between $C_2$ and $B_1\cup B_2\cup B_3$):}} \\
$E_{13}$ & $\{5,12,42,43\}$ & $v_5=(1,1,1,1)$    & $v_{12}=(1,0,0,-1)$ & $v_{42}=(1,-5,3,1)$   & $v_{43}=(2,-1,-3,2)$ \\
$E_{14}$ & $\{8,12,44,45\}$ & $v_8=(1,-1,-1,1)$  & $v_{12}=(1,0,0,-1)$ & $v_{44}=(1,-3,5,1)$   & $v_{45}=(2,3,1,2)$ \\
$E_{15}$ & $\{6,13,46,47\}$ & $v_6=(1,1,-1,-1)$  & $v_{13}=(0,1,1,0)$  & $v_{46}=(1,-2,2,-3)$  & $v_{47}=(5,-1,1,3)$ \\
$E_{16}$ & $\{7,13,48,49\}$ & $v_7=(1,-1,1,-1)$  & $v_{13}=(0,1,1,0)$  & $v_{48}=(3,-1,1,5)$   & $v_{49}=(3,2,-2,-1)$ \\*
$E_{17}$ & $\{6,14,50,51\}$ & $v_6=(1,1,-1,-1)$  & $v_{14}=(1,0,1,0)$  & $v_{50}=(1,3,-1,5)$   & $v_{51}=(2,-3,-2,1)$ \\*
$E_{18}$ & $\{8,14,52,53\}$ & $v_8=(1,-1,-1,1)$  & $v_{14}=(1,0,1,0)$  & $v_{52}=(1,5,-1,3)$   & $v_{53}=(2,1,-2,-3)$ \\*
$E_{19}$ & $\{5,15,54,55\}$ & $v_5=(1,1,1,1)$    & $v_{15}=(0,1,0,-1)$ & $v_{54}=(3,-2,1,-2)$  & $v_{55}=(3,1,-5,1)$ \\*
$E_{20}$ & $\{7,15,56,57\}$ & $v_7=(1,-1,1,-1)$  & $v_{15}=(0,1,0,-1)$ & $v_{56}=(1,2,3,2)$    & $v_{57}=(5,1,-3,1)$ \\*
$E_{21}$ & $\{7,16,58,59\}$ & $v_7=(1,-1,1,-1)$  & $v_{16}=(1,1,0,0)$  & $v_{58}=(1,-1,-5,-3)$ & $v_{59}=(2,-2,-1,3)$ \\*
$E_{22}$ & $\{8,16,60,61\}$ & $v_8=(1,-1,-1,1)$  & $v_{16}=(1,1,0,0)$  & $v_{60}=(1,-1,-3,-5)$ & $v_{61}=(2,-2,3,-1)$ \\*
$E_{23}$ & $\{5,17,62,63\}$ & $v_5=(1,1,1,1)$    & $v_{17}=(0,0,1,-1)$ & $v_{62}=(1,3,-2,-2)$  & $v_{63}=(5,-3,-1,-1)$ \\*
$E_{24}$ & $\{6,17,64,65\}$ & $v_6=(1,1,-1,-1)$  & $v_{17}=(0,0,1,-1)$ & $v_{64}=(3,-5,-1,-1)$ & $v_{65}=(3,1,2,2)$
\end{longtable}
\endgroup

\subsubsection{The non-linear warm-up \texorpdfstring{$26$}{26}-vertex, \texorpdfstring{$16$}{16}-context apex suspension \texorpdfstring{$H_{\YO}^{(4)}$}{HYO(4)}}\label{S-d:app:hyo4}

For completeness, the $26$ unnormalized integer vectors $v_1,\ldots,v_{26}\in\mathbb{Z}^4$ of the non-linear warm-up hypergraph $H_{\YO}^{(4)}=(V_{26},E_{16}^{(4)})$ in Proposition~\ref{d:prop:warmup} (written in canonical primitive form satisfying Eq.~\eqref{S-d:eq:canonical-primitive}) are:
\begin{equation}\label{S-d:eq:hyo4-vectors}
\begin{aligned}
v_1&=(0,1,1,0),& v_2&=(1,0,0,0),\\
v_3&=(0,1,-1,0),& v_4&=(1,1,0,0),\\
v_5&=(0,0,1,0),& v_6&=(1,-1,0,0),\\
v_7&=(1,0,1,0),& v_8&=(0,1,0,0),\\
v_9&=(1,0,-1,0),& v_{10}&=(1,-1,1,0),\\
v_{11}&=(1,-1,-1,0),& v_{12}&=(1,1,1,0),\\
v_{13}&=(1,1,-1,0),& v_{14}&=(2,1,1,0),\\
v_{15}&=(1,2,-1,0),& v_{16}&=(1,-2,1,0),\\
v_{17}&=(1,1,-2,0),& v_{18}&=(1,-1,-2,0),\\
v_{19}&=(2,1,-1,0),& v_{20}&=(2,-1,1,0),\\
v_{21}&=(1,-2,-1,0),& v_{22}&=(1,1,2,0),\\
v_{23}&=(1,2,1,0),& v_{24}&=(1,-1,2,0),\\
v_{25}&=(2,-1,-1,0),& v_{26}&=(0,0,0,1).
\end{aligned}
\end{equation}
These $26$ canonical primitive vectors partition by squared norm $\|v_i\|^2$ into four disjoint norm shells ($4+6+4+12=26$):
\begin{itemize}
\item $\|v_i\|^2=1$ ($4$ vectors): $\{v_2,v_5,v_8,v_{26}\}$, the four coordinate basis vectors of $\R^4$;
\item $\|v_i\|^2=2$ ($6$ vectors): $\{v_1,v_3,v_4,v_6,v_7,v_9\}$, the six binary rays in $\R^3\times\{0\}$;
\item $\|v_i\|^2=3$ ($4$ vectors): $\{v_{10},v_{11},v_{12},v_{13}\}$, the four triad rays $(1,\pm1,\pm1,0)^T$;
\item $\|v_i\|^2=6$ ($12$ vectors): $\{v_{14},\ldots,v_{25}\}$, all $3\times 4=12$ canonical primitive permutations and sign patterns of $(\pm2,\pm1,\pm1,0)^T$.
\end{itemize}
Since all $26$ canonical $4$-tuples are distinct, all $\binom{26}{2}=325$ pairs of vectors are strictly non-collinear, with maximum squared cosine $\max_{1\le i<j\le 26}(v_i\cdot v_j)^2/(\|v_i\|^2\|v_j\|^2)=4^2/(3\times 6)=8/9\approx 0.8889<1$. Table~\ref{S-d:tab:app-hyo4-contexts} lists all $16$ four-element contexts of $E_{16}^{(4)}$ side by side with their four mutually orthogonal integer vectors.

\begingroup
\scriptsize
\setlength{\tabcolsep}{3pt}
\Needspace{24\baselineskip}
\begin{longtable}{llllll}
\caption{Complete enumeration of the $16$ four-element contexts of the non-linear warm-up apex suspension $H_{\YO}^{(4)}=(V_{26},E_{16}^{(4)})$, listing for each context both its vertex indices $\{i,j,k,26\}$ and the four mutually orthogonal unnormalized integer vectors $\{v_i,v_j,v_k,v_{26}\}\subset\mathbb{Z}^4$.}\label{S-d:tab:app-hyo4-contexts}\\
\toprule

Context & Vertex set & $v_i$ & $v_j$ & $v_k$ & $v_{26}$ \\\midrule
\endfirsthead
\caption[]{All 16 contexts of $H_{\YO}^{(4)}$ (continued).}\\
\toprule

Context & Vertex set & $v_i$ & $v_j$ & $v_k$ & $v_{26}$ \\\midrule
\endhead
\midrule
\multicolumn{6}{r}{\emph{Continued on the next page.}}\\
\endfoot
\bottomrule\endlastfoot
$e_1^{(4)}$  & $\{1,2,3,26\}$   & $v_1=(0,1,1,0)$  & $v_2=(1,0,0,0)$    & $v_3=(0,1,-1,0)$    & $v_{26}=(0,0,0,1)$ \\*
$e_2^{(4)}$  & $\{7,8,9,26\}$   & $v_7=(1,0,1,0)$  & $v_8=(0,1,0,0)$    & $v_9=(1,0,-1,0)$    & $v_{26}=(0,0,0,1)$ \\*
$e_3^{(4)}$  & $\{4,5,6,26\}$   & $v_4=(1,1,0,0)$  & $v_5=(0,0,1,0)$    & $v_6=(1,-1,0,0)$    & $v_{26}=(0,0,0,1)$ \\*
$e_4^{(4)}$  & $\{2,5,8,26\}$   & $v_2=(1,0,0,0)$  & $v_5=(0,0,1,0)$    & $v_8=(0,1,0,0)$     & $v_{26}=(0,0,0,1)$ \\ \hline
$e_5^{(4)}$  & $\{3,12,25,26\}$ & $v_3=(0,1,-1,0)$ & $v_{12}=(1,1,1,0)$ & $v_{25}=(2,-1,-1,0)$& $v_{26}=(0,0,0,1)$ \\*
$e_6^{(4)}$  & $\{9,12,16,26\}$ & $v_9=(1,0,-1,0)$ & $v_{12}=(1,1,1,0)$ & $v_{16}=(1,-2,1,0)$ & $v_{26}=(0,0,0,1)$ \\*
$e_7^{(4)}$  & $\{6,12,17,26\}$ & $v_6=(1,-1,0,0)$ & $v_{12}=(1,1,1,0)$ & $v_{17}=(1,1,-2,0)$ & $v_{26}=(0,0,0,1)$ \\*
$e_8^{(4)}$  & $\{3,11,14,26\}$ & $v_3=(0,1,-1,0)$ & $v_{11}=(1,-1,-1,0)$&$v_{14}=(2,1,1,0)$  & $v_{26}=(0,0,0,1)$ \\*
$e_9^{(4)}$  & $\{7,11,15,26\}$ & $v_7=(1,0,1,0)$  & $v_{11}=(1,-1,-1,0)$&$v_{15}=(1,2,-1,0)$ & $v_{26}=(0,0,0,1)$ \\*
$e_{10}^{(4)}$ & $\{4,11,24,26\}$ & $v_4=(1,1,0,0)$  & $v_{11}=(1,-1,-1,0)$&$v_{24}=(1,-1,2,0)$ & $v_{26}=(0,0,0,1)$ \\*
$e_{11}^{(4)}$ & $\{9,10,23,26\}$ & $v_9=(1,0,-1,0)$ & $v_{10}=(1,-1,1,0)$ & $v_{23}=(1,2,1,0)$  & $v_{26}=(0,0,0,1)$ \\*
$e_{12}^{(4)}$ & $\{1,10,19,26\}$ & $v_1=(0,1,1,0)$  & $v_{10}=(1,-1,1,0)$ & $v_{19}=(2,1,-1,0)$ & $v_{26}=(0,0,0,1)$ \\*
$e_{13}^{(4)}$ & $\{4,10,18,26\}$ & $v_4=(1,1,0,0)$  & $v_{10}=(1,-1,1,0)$ & $v_{18}=(1,-1,-2,0)$& $v_{26}=(0,0,0,1)$ \\*
$e_{14}^{(4)}$ & $\{1,13,20,26\}$ & $v_1=(0,1,1,0)$  & $v_{13}=(1,1,-1,0)$ & $v_{20}=(2,-1,1,0)$ & $v_{26}=(0,0,0,1)$ \\*
$e_{15}^{(4)}$ & $\{7,13,21,26\}$ & $v_7=(1,0,1,0)$  & $v_{13}=(1,1,-1,0)$ & $v_{21}=(1,-2,-1,0)$& $v_{26}=(0,0,0,1)$ \\*
$e_{16}^{(4)}$ & $\{6,13,22,26\}$ & $v_6=(1,-1,0,0)$ & $v_{13}=(1,1,-1,0)$ & $v_{22}=(1,1,2,0)$  & $v_{26}=(0,0,0,1)$
\end{longtable}
\endgroup

\FloatBarrier
\subsection{Explicit parent configurations and inclusion maps}
\label{S-app:parents-vectors}
\subsubsection{The 49-ray real Yu--Oh parent and published Bub labels}
Table~\ref{S-tab:bub49-vectors} lists all rays, including the original 25.
The final column gives the exact labels in the published full 49-ray
encoding~\cite[p.~14, Sec.~5(xi)]{S-Pavicic-2005} of Bub's set~\cite{S-bub}. Under this map
the 36 bases in Table~\ref{S-tab:bub49-contexts} equal all 36 published
contexts; their 108 pairs equal all geometric orthogonalities. In this
labeling $26$--$33$ are Eq.~\eqref{S-eq:yo-parent-eight}, and $34$--$49$
complete the sixteen remaining orthogonal pairs. The first 25 vectors
are those of Table~\ref{S-tab:yo-all-vectors}.
\begingroup\small
\Needspace{24\baselineskip}
\begin{longtable}{rll}
\caption{All unnormalized vectors of the real 49-ray parent, with the verified published Bub labels.}\label{S-tab:bub49-vectors}\\
\toprule
$i$ & $v_i$ & Published label\\\midrule
\endfirsthead
\toprule
$i$ & $v_i$ & Published label\\\midrule
\endhead
\bottomrule\endfoot
$1$ & $(0,1,1)$ & \texttt{G}\\
$2$ & $(1,0,0)$ & \texttt{1}\\
$3$ & $(0,1,-1)$ & \texttt{P}\\
$4$ & $(1,1,0)$ & \texttt{H}\\
$5$ & $(0,0,1)$ & \texttt{3}\\
$6$ & $(1,-1,0)$ & \texttt{U}\\
$7$ & $(1,0,1)$ & \texttt{Q}\\
$8$ & $(0,1,0)$ & \texttt{2}\\
$9$ & $(1,0,-1)$ & \texttt{e}\\
$10$ & $(1,-1,1)$ & \texttt{D}\\
$11$ & $(1,-1,-1)$ & \texttt{h}\\
$12$ & $(1,1,1)$ & \texttt{S}\\
$13$ & $(1,1,-1)$ & \texttt{a}\\
$14$ & $(2,1,1)$ & \texttt{g}\\
$15$ & $(1,2,-1)$ & \texttt{n}\\
$16$ & $(1,-2,1)$ & \texttt{d}\\
$17$ & $(1,1,-2)$ & \texttt{T}\\
$18$ & $(1,-1,-2)$ & \texttt{F}\\
$19$ & $(2,1,-1)$ & \texttt{E}\\
$20$ & $(2,-1,1)$ & \texttt{I}\\
$21$ & $(1,-2,-1)$ & \texttt{k}\\
$22$ & $(1,1,2)$ & \texttt{b}\\
$23$ & $(1,2,1)$ & \texttt{A}\\
$24$ & $(1,-1,2)$ & \texttt{N}\\
$25$ & $(2,-1,-1)$ & \texttt{R}\\
$26$ & $(0,1,2)$ & \texttt{i}\\
$27$ & $(0,2,-1)$ & \texttt{O}\\
$28$ & $(0,1,-2)$ & \texttt{6}\\
$29$ & $(0,2,1)$ & \texttt{7}\\
$30$ & $(1,2,0)$ & \texttt{5}\\
$31$ & $(2,-1,0)$ & \texttt{4}\\
$32$ & $(1,-2,0)$ & \texttt{V}\\
$33$ & $(2,1,0)$ & \texttt{j}\\
$34$ & $(2,1,-5)$ & \texttt{Z}\\
$35$ & $(5,-2,1)$ & \texttt{m}\\
$36$ & $(1,2,5)$ & \texttt{W}\\
$37$ & $(5,2,-1)$ & \texttt{c}\\
$38$ & $(1,-2,-5)$ & \texttt{f}\\
$39$ & $(5,-1,2)$ & \texttt{L}\\
$40$ & $(5,1,2)$ & \texttt{9}\\
$41$ & $(2,1,5)$ & \texttt{8}\\
$42$ & $(2,-1,-5)$ & \texttt{J}\\*
$43$ & $(5,2,1)$ & \texttt{X}\\*
$44$ & $(1,-2,5)$ & \texttt{l}\\*
$45$ & $(5,-1,-2)$ & \texttt{Y}\\*
$46$ & $(5,-2,-1)$ & \texttt{B}\\*
$47$ & $(1,2,-5)$ & \texttt{C}\\*
$48$ & $(5,1,-2)$ & \texttt{M}\\*
$49$ & $(2,-1,5)$ & \texttt{K}\\*
\end{longtable}
\endgroup
\begingroup\small
\Needspace{24\baselineskip}
\begin{longtable}{rlrlrl}
\caption{All 36 complete bases of the real parent.}\label{S-tab:bub49-contexts}\\
\toprule
$e$ & Context & $e$ & Context & $e$ & Context\\\midrule
\endfirsthead
\toprule
$e$ & Context & $e$ & Context & $e$ & Context\\\midrule
\endhead
\bottomrule\endfoot
$1$ & $\{1,2,3\}$ & $2$ & $\{1,10,19\}$ & $3$ & $\{1,13,20\}$\\*
$4$ & $\{2,5,8\}$ & $5$ & $\{2,26,27\}$ & $6$ & $\{2,28,29\}$\\*
$7$ & $\{3,11,14\}$ & $8$ & $\{3,12,25\}$ & $9$ & $\{4,5,6\}$\\*
$10$ & $\{4,10,18\}$ & $11$ & $\{4,11,24\}$ & $12$ & $\{5,30,31\}$\\*
$13$ & $\{5,32,33\}$ & $14$ & $\{6,12,17\}$ & $15$ & $\{6,13,22\}$\\*
$16$ & $\{7,8,9\}$ & $17$ & $\{7,11,15\}$ & $18$ & $\{7,13,21\}$\\*
$19$ & $\{9,10,23\}$ & $20$ & $\{9,12,16\}$ & $21$ & $\{14,32,34\}$\\*
$22$ & $\{15,26,35\}$ & $23$ & $\{15,31,36\}$ & $24$ & $\{16,26,37\}$\\*
$25$ & $\{16,33,38\}$ & $26$ & $\{17,29,39\}$ & $27$ & $\{18,27,40\}$\\*
$28$ & $\{19,32,41\}$ & $29$ & $\{20,30,42\}$ & $30$ & $\{21,28,43\}$\\*
$31$ & $\{21,33,44\}$ & $32$ & $\{22,27,45\}$ & $33$ & $\{23,28,46\}$\\*
$34$ & $\{23,31,47\}$ & $35$ & $\{24,29,48\}$ & $36$ & $\{25,30,49\}$\\*
\end{longtable}
\endgroup
\subsubsection{The 87-ray Hesse parent}
Use $\zeta=\exp(2\pi i/3)$ and $\zeta^2+\zeta+1=0$.
The first 37 rays retain exactly the Hesse numbering; projective
representatives have their first nonzero coordinate one. The inventory
labels use the generation convention of Sec.~\ref{h:sec:comparison}.
Tables~\ref{S-tab:hesse87-vectors} and~\ref{S-tab:hesse87-contexts} list all
87 rays and all 64 complete bases.
The six additional orthogonal pairs outside the 64 complete bases are
listed separately: this larger parent's full orthogonality graph need
not equal its context two-section. The retained Hesse restriction has
all its original orthogonalities covered by contexts.
\begingroup\small
\Needspace{24\baselineskip}
\begin{longtable}{rll}
\caption{All 87 unnormalized rays of the first Hesse/lineage union.}\label{S-tab:hesse87-vectors}\\
\toprule
$i$ & $v_i$ & Inventory label\\\midrule
\endfirsthead
\toprule
$i$ & $v_i$ & Inventory label\\\midrule
\endhead
\bottomrule\endfoot
$1$ & $(1,0,0)$ & $a_{11}$\\
$2$ & $(1,1,1)$ & $u_{1}$\\
$3$ & $(1,\zeta,-1-\zeta)$ & $u_{2}$\\
$4$ & $(1,-1-\zeta,\zeta)$ & $u_{3}$\\
$5$ & $(1,\zeta,\zeta)$ & $u_{4}$\\
$6$ & $(1,-1-\zeta,1)$ & $u_{5}$\\
$7$ & $(1,1,-1-\zeta)$ & $u_{6}$\\
$8$ & $(1,-1-\zeta,-1-\zeta)$ & $u_{7}$\\
$9$ & $(1,1,\zeta)$ & $u_{9}$\\
$10$ & $(1,\zeta,1)$ & $u_{8}$\\
$11$ & $(0,1,-1)$ & $c_{11}$\\
$12$ & $(0,1,-\zeta)$ & $c_{12}$\\
$13$ & $(0,1,1+\zeta)$ & $c_{13}$\\
$14$ & $(1,0,-1)$ & $c_{21}$\\
$15$ & $(1,0,1+\zeta)$ & $c_{22}$\\
$16$ & $(1,0,-\zeta)$ & $c_{23}$\\
$17$ & $(0,1,1)$ & $b_{11}$\\
$18$ & $(0,1,\zeta)$ & $b_{12}$\\
$19$ & $(0,1,-1-\zeta)$ & $b_{13}$\\
$20$ & $(1,\frac{-1}{2},\frac{-1}{2})$ & $d_{11}$\\
$21$ & $(1,-\frac{1}{2}\zeta,\frac{1}{2}+\frac{1}{2}\zeta)$ & $d_{12}$\\
$22$ & $(1,\frac{1}{2}+\frac{1}{2}\zeta,-\frac{1}{2}\zeta)$ & $d_{13}$\\
$23$ & $(1,-\frac{1}{2}\zeta,-\frac{1}{2}\zeta)$ & $d_{14}$\\
$24$ & $(1,\frac{1}{2}+\frac{1}{2}\zeta,\frac{-1}{2})$ & $d_{15}$\\
$25$ & $(1,\frac{-1}{2},\frac{1}{2}+\frac{1}{2}\zeta)$ & $d_{16}$\\
$26$ & $(1,\frac{1}{2}+\frac{1}{2}\zeta,\frac{1}{2}+\frac{1}{2}\zeta)$ & $d_{17}$\\
$27$ & $(1,\frac{-1}{2},-\frac{1}{2}\zeta)$ & $d_{19}$\\
$28$ & $(1,-\frac{1}{2}\zeta,\frac{-1}{2})$ & $d_{18}$\\
$29$ & $(1,-2,1)$ & $d_{21}$\\
$30$ & $(1,-2\zeta,-1-\zeta)$ & $d_{22}$\\
$31$ & $(1,2+2\zeta,\zeta)$ & $d_{23}$\\
$32$ & $(1,2+2\zeta,1)$ & $d_{25}$\\
$33$ & $(1,-2,-1-\zeta)$ & $d_{26}$\\
$34$ & $(1,-2\zeta,\zeta)$ & $d_{24}$\\
$35$ & $(1,-2\zeta,1)$ & $d_{28}$\\
$36$ & $(1,2+2\zeta,-1-\zeta)$ & $d_{27}$\\
$37$ & $(1,-2,\zeta)$ & $d_{29}$\\
$38$ & $(0,1,0)$ & $a_{21}$\\
$39$ & $(0,0,1)$ & $a_{31}$\\
$40$ & $(1,\frac{-1}{2},\frac{1}{2})$ & $b_{1121}$\\
$41$ & $(1,-\frac{1}{2}\zeta,\frac{-1}{2}-\frac{1}{2}\zeta)$ & $b_{1122}$\\
$42$ & $(1,\frac{1}{2}+\frac{1}{2}\zeta,\frac{1}{2}\zeta)$ & $b_{1123}$\\
$43$ & $(1,\frac{1}{2},\frac{-1}{2})$ & $b_{1131}$\\
$44$ & $(1,\frac{1}{2}\zeta,\frac{1}{2}+\frac{1}{2}\zeta)$ & $b_{1132}$\\
$45$ & $(1,\frac{-1}{2}-\frac{1}{2}\zeta,-\frac{1}{2}\zeta)$ & $b_{1133}$\\
$46$ & $(1,1,-1)$ & $b_{121}$\\
$47$ & $(1,\zeta,1+\zeta)$ & $b_{122}$\\
$48$ & $(1,-1-\zeta,-\zeta)$ & $b_{123}$\\
$49$ & $(1,-1,1)$ & $b_{131}$\\
$50$ & $(1,-\zeta,-1-\zeta)$ & $b_{132}$\\
$51$ & $(1,1+\zeta,\zeta)$ & $b_{133}$\\
$52$ & $(1,0,1)$ & $b_{21}$\\
$53$ & $(1,-2,-1)$ & $b_{2121}$\\
$54$ & $(1,-2\zeta,1+\zeta)$ & $b_{2122}$\\
$55$ & $(1,2+2\zeta,-\zeta)$ & $b_{2123}$\\
$56$ & $(1,0,-1-\zeta)$ & $b_{22}$\\
$57$ & $(1,2,-1)$ & $b_{2231}$\\
$58$ & $(1,2\zeta,1+\zeta)$ & $b_{2232}$\\
$59$ & $(1,-2-2\zeta,-\zeta)$ & $b_{2233}$\\
$60$ & $(1,0,\zeta)$ & $b_{23}$\\
$61$ & $(1,-1,-1)$ & $b_{231}$\\
$62$ & $(1,-\zeta,1+\zeta)$ & $b_{232}$\\
$63$ & $(1,1+\zeta,-\zeta)$ & $b_{233}$\\
$64$ & $(1,1,0)$ & $b_{31}$\\
$65$ & $(1,-1,-2)$ & $b_{3131}$\\
$66$ & $(1,-\zeta,2+2\zeta)$ & $b_{3132}$\\
$67$ & $(1,1+\zeta,-2\zeta)$ & $b_{3133}$\\
$68$ & $(1,\zeta,0)$ & $b_{32}$\\
$69$ & $(1,-1,2)$ & $b_{3231}$\\
$70$ & $(1,-\zeta,-2-2\zeta)$ & $b_{3232}$\\
$71$ & $(1,1+\zeta,2\zeta)$ & $b_{3233}$\\
$72$ & $(1,-1-\zeta,0)$ & $b_{33}$\\
$73$ & $(1,-1,0)$ & $c_{31}$\\
$74$ & $(1,-\zeta,0)$ & $c_{32}$\\
$75$ & $(1,1+\zeta,0)$ & $c_{33}$\\
$76$ & $(1,1,-2)$ & $d_{31}$\\
$77$ & $(1,\zeta,2+2\zeta)$ & $d_{32}$\\
$78$ & $(1,-1-\zeta,-2\zeta)$ & $d_{33}$\\
$79$ & $(1,2,1)$ & $e_{11}$\\
$80$ & $(1,2\zeta,-1-\zeta)$ & $e_{12}$\\*
$81$ & $(1,-2-2\zeta,\zeta)$ & $e_{13}$\\*
$82$ & $(1,1,2)$ & $e_{21}$\\*
$83$ & $(1,\zeta,-2-2\zeta)$ & $e_{22}$\\*
$84$ & $(1,-1-\zeta,2\zeta)$ & $e_{23}$\\*
$85$ & $(1,\frac{1}{2},\frac{1}{2})$ & $e_{31}$\\*
$86$ & $(1,\frac{1}{2}\zeta,\frac{-1}{2}-\frac{1}{2}\zeta)$ & $e_{32}$\\*
$87$ & $(1,\frac{-1}{2}-\frac{1}{2}\zeta,\frac{1}{2}\zeta)$ & $e_{33}$\\*
\end{longtable}
\endgroup
\begingroup\small
\Needspace{24\baselineskip}
\begin{longtable}{rlrlrl}
\caption{All 64 complete bases of the first 87-ray Hesse parent.}\label{S-tab:hesse87-contexts}\\
\toprule
$e$ & Context & $e$ & Context & $e$ & Context\\\midrule
\endfirsthead
\toprule
$e$ & Context & $e$ & Context & $e$ & Context\\\midrule
\endhead
\bottomrule\endfoot
$1$ & $\{1,11,17\}$ & $2$ & $\{1,12,18\}$ & $3$ & $\{1,13,19\}$\\
$4$ & $\{1,38,39\}$ & $5$ & $\{2,3,4\}$ & $6$ & $\{2,11,20\}$\\
$7$ & $\{2,14,29\}$ & $8$ & $\{2,73,76\}$ & $9$ & $\{3,12,21\}$\\
$10$ & $\{3,15,30\}$ & $11$ & $\{3,74,77\}$ & $12$ & $\{4,13,22\}$\\
$13$ & $\{4,16,31\}$ & $14$ & $\{4,75,78\}$ & $15$ & $\{5,6,7\}$\\
$16$ & $\{5,11,23\}$ & $17$ & $\{5,16,34\}$ & $18$ & $\{6,12,24\}$\\
$19$ & $\{6,14,32\}$ & $20$ & $\{7,13,25\}$ & $21$ & $\{7,15,33\}$\\
$22$ & $\{8,9,10\}$ & $23$ & $\{8,11,26\}$ & $24$ & $\{8,15,36\}$\\
$25$ & $\{9,12,27\}$ & $26$ & $\{9,16,37\}$ & $27$ & $\{10,13,28\}$\\
$28$ & $\{10,14,35\}$ & $29$ & $\{11,61,85\}$ & $30$ & $\{12,62,86\}$\\
$31$ & $\{13,63,87\}$ & $32$ & $\{14,38,52\}$ & $33$ & $\{14,49,79\}$\\
$34$ & $\{15,38,56\}$ & $35$ & $\{15,50,80\}$ & $36$ & $\{16,38,60\}$\\
$37$ & $\{16,51,81\}$ & $38$ & $\{17,40,46\}$ & $39$ & $\{17,43,49\}$\\
$40$ & $\{18,41,47\}$ & $41$ & $\{18,44,50\}$ & $42$ & $\{19,42,48\}$\\
$43$ & $\{19,45,51\}$ & $44$ & $\{39,64,73\}$ & $45$ & $\{39,68,74\}$\\*
$46$ & $\{39,72,75\}$ & $47$ & $\{46,47,48\}$ & $48$ & $\{46,52,53\}$\\*
$49$ & $\{46,73,82\}$ & $50$ & $\{47,54,56\}$ & $51$ & $\{47,74,83\}$\\*
$52$ & $\{48,55,60\}$ & $53$ & $\{48,75,84\}$ & $54$ & $\{49,50,51\}$\\*
$55$ & $\{49,64,65\}$ & $56$ & $\{50,66,68\}$ & $57$ & $\{51,67,72\}$\\*
$58$ & $\{52,57,61\}$ & $59$ & $\{56,58,62\}$ & $60$ & $\{59,60,63\}$\\*
$61$ & $\{61,62,63\}$ & $62$ & $\{61,64,69\}$ & $63$ & $\{62,68,70\}$\\*
$64$ & $\{63,71,72\}$ &  &  &  & \\*
\end{longtable}
\endgroup
Its six further orthogonal pairs are $\{5,74\},\ \{6,75\},\ \{7,73\},\ \{8,75\},\ \{9,73\},\ \{10,74\}$.
\subsubsection{Peres24 and the complete four-dimensional extensions}
The coordinates of $H_{65}$ and $H_{\YO}^{(4)}$ are those already
listed above. Table~\ref{S-tab:peres-all} explicitly enumerates the known
Peres rays. In each union, retain the original local labels and append
the missing rays in the order displayed in
Table~\ref{S-tab:parent-added-rays}. Tables~\ref{S-tab:parent40-contexts}
and \ref{S-tab:parent72-contexts} list every available complete basis.
Full orthogonality-pair lists accompany the coordinates in the data;
the 40-ray union has additional pairs outside its complete contexts.
\begingroup\small
\Needspace{24\baselineskip}
\begin{longtable}{rl}
\caption{All 24 primitive unnormalized Peres rays, in the certificate order.}\label{S-tab:peres-all}\\
\toprule
$i$ & $p_i$\\\midrule
\endfirsthead
\toprule
$i$ & $p_i$\\\midrule
\endhead
\bottomrule\endfoot
$1$ & $(1,0,0,0)$\\
$2$ & $(0,1,0,0)$\\
$3$ & $(0,0,1,0)$\\
$4$ & $(0,0,0,1)$\\
$5$ & $(1,-1,0,0)$\\
$6$ & $(1,1,0,0)$\\
$7$ & $(1,0,-1,0)$\\
$8$ & $(1,0,1,0)$\\
$9$ & $(1,0,0,-1)$\\
$10$ & $(1,0,0,1)$\\
$11$ & $(0,1,-1,0)$\\
$12$ & $(0,1,1,0)$\\
$13$ & $(0,1,0,-1)$\\
$14$ & $(0,1,0,1)$\\
$15$ & $(0,0,1,-1)$\\
$16$ & $(0,0,1,1)$\\
$17$ & $(1,-1,-1,-1)$\\*
$18$ & $(1,-1,-1,1)$\\*
$19$ & $(1,-1,1,-1)$\\*
$20$ & $(1,-1,1,1)$\\*
$21$ & $(1,1,-1,-1)$\\*
$22$ & $(1,1,-1,1)$\\*
$23$ & $(1,1,1,-1)$\\*
$24$ & $(1,1,1,1)$\\*
\end{longtable}
\endgroup
\begingroup\small
\Needspace{24\baselineskip}
\begin{longtable}{lrl}
\caption{Every ray added to the original four-dimensional examples to obtain the stated parents.}\label{S-tab:parent-added-rays}\\
\toprule
Original & New label & Added vector\\\midrule
\endfirsthead
\toprule
Original & New label & Added vector\\\midrule
\endhead
\bottomrule\endfoot
$H_{\YO}^{(4)}$ & $27$ & $(1,0,0,-1)$\\
$H_{\YO}^{(4)}$ & $28$ & $(1,0,0,1)$\\
$H_{\YO}^{(4)}$ & $29$ & $(0,1,0,-1)$\\
$H_{\YO}^{(4)}$ & $30$ & $(0,1,0,1)$\\
$H_{\YO}^{(4)}$ & $31$ & $(0,0,1,-1)$\\
$H_{\YO}^{(4)}$ & $32$ & $(0,0,1,1)$\\
$H_{\YO}^{(4)}$ & $33$ & $(1,-1,-1,-1)$\\
$H_{\YO}^{(4)}$ & $34$ & $(1,-1,-1,1)$\\
$H_{\YO}^{(4)}$ & $35$ & $(1,-1,1,-1)$\\
$H_{\YO}^{(4)}$ & $36$ & $(1,-1,1,1)$\\
$H_{\YO}^{(4)}$ & $37$ & $(1,1,-1,-1)$\\
$H_{\YO}^{(4)}$ & $38$ & $(1,1,-1,1)$\\
$H_{\YO}^{(4)}$ & $39$ & $(1,1,1,-1)$\\
$H_{\YO}^{(4)}$ & $40$ & $(1,1,1,1)$\\*
$H_{65}$ & $66$ & $(1,-1,0,0)$\\*
$H_{65}$ & $67$ & $(1,0,-1,0)$\\*
$H_{65}$ & $68$ & $(1,0,0,1)$\\*
$H_{65}$ & $69$ & $(0,1,-1,0)$\\*
$H_{65}$ & $70$ & $(0,1,0,1)$\\*
$H_{65}$ & $71$ & $(0,0,1,1)$\\*
$H_{65}$ & $72$ & $(1,1,1,-1)$\\*
\end{longtable}
\endgroup
\begingroup\small
\Needspace{24\baselineskip}
\begin{longtable}{rlrl}
\caption{All 36 complete bases of the 40-ray parent.}\label{S-tab:parent40-contexts}\\
\toprule
$e$ & Context & $e$ & Context\\\midrule
\endfirsthead
\toprule
$e$ & Context & $e$ & Context\\\midrule
\endhead
\bottomrule\endfoot
$1$ & $\{1,2,3,26\}$ & $2$ & $\{1,3,27,28\}$\\*
$3$ & $\{1,10,19,26\}$ & $4$ & $\{1,13,20,26\}$\\*
$5$ & $\{1,27,36,38\}$ & $6$ & $\{1,28,35,37\}$\\*
$7$ & $\{2,5,8,26\}$ & $8$ & $\{2,5,29,30\}$\\*
$9$ & $\{2,8,31,32\}$ & $10$ & $\{3,11,14,26\}$\\*
$11$ & $\{3,12,25,26\}$ & $12$ & $\{3,27,34,40\}$\\*
$13$ & $\{3,28,33,39\}$ & $14$ & $\{4,5,6,26\}$\\*
$15$ & $\{4,6,31,32\}$ & $16$ & $\{4,10,18,26\}$\\*
$17$ & $\{4,11,24,26\}$ & $18$ & $\{4,31,33,36\}$\\*
$19$ & $\{4,32,34,35\}$ & $20$ & $\{5,8,27,28\}$\\*
$21$ & $\{6,12,17,26\}$ & $22$ & $\{6,13,22,26\}$\\*
$23$ & $\{6,31,37,40\}$ & $24$ & $\{6,32,38,39\}$\\*
$25$ & $\{7,8,9,26\}$ & $26$ & $\{7,9,29,30\}$\\*
$27$ & $\{7,11,15,26\}$ & $28$ & $\{7,13,21,26\}$\\*
$29$ & $\{7,29,33,38\}$ & $30$ & $\{7,30,34,37\}$\\*
$31$ & $\{9,10,23,26\}$ & $32$ & $\{9,12,16,26\}$\\*
$33$ & $\{9,29,35,40\}$ & $34$ & $\{9,30,36,39\}$\\*
$35$ & $\{33,36,38,39\}$ & $36$ & $\{34,35,37,40\}$\\*
\end{longtable}
\endgroup
\begingroup\small
\Needspace{24\baselineskip}
\begin{longtable}{rlrl}
\caption{All 48 complete bases of the 72-ray parent.}\label{S-tab:parent72-contexts}\\
\toprule
$e$ & Context & $e$ & Context\\\midrule
\endfirsthead
\toprule
$e$ & Context & $e$ & Context\\\midrule
\endhead
\bottomrule\endfoot
$1$ & $\{1,2,3,4\}$ & $2$ & $\{1,2,17,71\}$\\
$3$ & $\{1,3,15,70\}$ & $4$ & $\{1,4,13,69\}$\\
$5$ & $\{1,13,22,23\}$ & $6$ & $\{1,15,30,31\}$\\
$7$ & $\{1,17,38,39\}$ & $8$ & $\{2,3,12,68\}$\\
$9$ & $\{2,4,14,67\}$ & $10$ & $\{2,12,18,19\}$\\
$11$ & $\{2,14,26,27\}$ & $12$ & $\{2,17,40,41\}$\\
$13$ & $\{3,4,16,66\}$ & $14$ & $\{3,12,20,21\}$\\
$15$ & $\{3,15,32,33\}$ & $16$ & $\{3,16,34,35\}$\\
$17$ & $\{4,13,24,25\}$ & $18$ & $\{4,14,28,29\}$\\
$19$ & $\{4,16,36,37\}$ & $20$ & $\{5,6,7,8\}$\\
$21$ & $\{5,6,17,66\}$ & $22$ & $\{5,7,15,67\}$\\
$23$ & $\{5,8,12,69\}$ & $24$ & $\{5,12,42,43\}$\\
$25$ & $\{5,15,54,55\}$ & $26$ & $\{5,17,62,63\}$\\
$27$ & $\{6,7,13,68\}$ & $28$ & $\{6,8,14,70\}$\\
$29$ & $\{6,13,46,47\}$ & $30$ & $\{6,14,50,51\}$\\
$31$ & $\{6,17,64,65\}$ & $32$ & $\{7,8,16,71\}$\\
$33$ & $\{7,13,48,49\}$ & $34$ & $\{7,15,56,57\}$\\*
$35$ & $\{7,16,58,59\}$ & $36$ & $\{8,12,44,45\}$\\*
$37$ & $\{8,14,52,53\}$ & $38$ & $\{8,16,60,61\}$\\*
$39$ & $\{9,10,11,72\}$ & $40$ & $\{9,10,12,13\}$\\*
$41$ & $\{9,11,14,15\}$ & $42$ & $\{9,66,71,72\}$\\*
$43$ & $\{10,11,16,17\}$ & $44$ & $\{10,67,70,72\}$\\*
$45$ & $\{11,68,69,72\}$ & $46$ & $\{12,13,68,69\}$\\*
$47$ & $\{14,15,67,70\}$ & $48$ & $\{16,17,66,71\}$\\*
\end{longtable}
\endgroup

\subsection{Cabello completions and the two 68-ray parents}
\subsubsection{The eighteen-ray input and faithful 36-ray completion}
\label{S-app:cabello-vectors}
All labels in this subsection are local. In $C_{18,9}$ and $C_{36,18}$,
$18=a$, $19=b$, and $20=c$; labels 4 and 5 are absent. They do not denote
the completion rays of $H_{65}$. The core is the first eighteen rows
of Table~\ref{S-tab:cabello36-vectors}; labels 21--38 are fresh.
The contexts of $C_{18,9}$ are the first nine rows of
Table~\ref{S-tab:cabello36-contexts}.
\begingroup\small
\Needspace{24\baselineskip}
\begin{longtable}{rl}
\caption{All 36 unnormalized vectors of the faithful Cabello completion.}\label{S-tab:cabello36-vectors}\\
\toprule
$i$ & $v_i$\\\midrule
\endfirsthead
\toprule
$i$ & $v_i$\\\midrule
\endhead
\bottomrule\endfoot
$1$ & $(1,0,0,0)$\\
$2$ & $(0,1,0,0)$\\
$3$ & $(0,0,1,0)$\\
$6$ & $(1,1,-1,-1)$\\
$7$ & $(1,-1,1,-1)$\\
$8$ & $(1,-1,-1,1)$\\
$9$ & $(1,1,-1,1)$\\
$10$ & $(1,-1,1,1)$\\
$11$ & $(1,-1,-1,-1)$\\
$12$ & $(1,0,0,-1)$\\
$13$ & $(0,1,1,0)$\\
$14$ & $(1,0,1,0)$\\
$15$ & $(0,1,0,-1)$\\
$16$ & $(1,1,0,0)$\\
$17$ & $(0,0,1,-1)$\\
$18$ & $(1,0,0,1)$\\
$19$ & $(0,1,0,1)$\\
$20$ & $(0,0,1,1)$\\
$21$ & $(0,1,-1,3)$\\
$22$ & $(0,3,-3,-2)$\\
$23$ & $(1,0,-1,3)$\\
$24$ & $(3,0,-3,-2)$\\
$25$ & $(1,-1,0,3)$\\
$26$ & $(3,-3,0,-2)$\\
$27$ & $(3,1,2,2)$\\
$28$ & $(3,-5,-1,-1)$\\
$29$ & $(3,2,1,2)$\\
$30$ & $(3,-1,-5,-1)$\\
$31$ & $(2,3,1,2)$\\*
$32$ & $(1,-3,5,1)$\\*
$33$ & $(4,2,3,-3)$\\*
$34$ & $(5,-7,-1,1)$\\*
$35$ & $(4,3,2,-3)$\\*
$36$ & $(5,-1,-7,1)$\\*
$37$ & $(3,4,2,-3)$\\*
$38$ & $(1,-5,7,-1)$\\*
\end{longtable}
\endgroup
\begingroup\small
\Needspace{24\baselineskip}
\begin{longtable}{rl}
\caption{All eighteen full bases of the faithful Cabello completion.}\label{S-tab:cabello36-contexts}\\
\toprule
$e$ & Context\\\midrule
\endfirsthead
\toprule
$e$ & Context\\\midrule
\endhead
\bottomrule\endfoot
$1$ & $\{9,10,12,13\}$\\*
$2$ & $\{9,11,14,15\}$\\*
$3$ & $\{10,11,16,17\}$\\*
$4$ & $\{2,3,12,18\}$\\*
$5$ & $\{1,3,15,19\}$\\*
$6$ & $\{1,2,17,20\}$\\*
$7$ & $\{6,7,13,18\}$\\*
$8$ & $\{6,8,14,19\}$\\*
$9$ & $\{7,8,16,20\}$\\*
$10$ & $\{1,13,21,22\}$\\*
$11$ & $\{2,14,23,24\}$\\*
$12$ & $\{3,16,25,26\}$\\*
$13$ & $\{6,17,27,28\}$\\*
$14$ & $\{7,15,29,30\}$\\*
$15$ & $\{8,12,31,32\}$\\*
$16$ & $\{9,20,33,34\}$\\*
$17$ & $\{10,19,35,36\}$\\*
$18$ & $\{11,18,37,38\}$\\*
\end{longtable}
\endgroup
The orthogonal matrix $Q=M/\sqrt2$, with
\begin{equation}\label{S-eq:cabello-published-map}
M=\begin{pmatrix}
0&1&0&-1\\-1&0&1&0\\1&0&1&0\\0&1&0&1
\end{pmatrix},\qquad M^TM=2I_4,
\end{equation}
maps the published eighteen-ray set and all nine bases of
Ref.~\cite[Eqs.~(1)--(9)]{S-cabello-96} onto this realization.
Table~\ref{S-tab:cabello-published-map} gives the complete ray map;
the input representatives are unnormalized and signs are immaterial.
\begingroup\small
\Needspace{24\baselineskip}
\begin{longtable}{lr}
\caption{Exact ray map from the originally published Cabello coordinates.}\label{S-tab:cabello-published-map}\\
\toprule
Published vector $u$ & Local label of $[Qu]$\\\midrule
\endfirsthead
\toprule
Published vector $u$ & Local label of $[Qu]$\\\midrule
\endhead
\bottomrule\endfoot
$(0,0,0,1)$ & $12$\\*
$(0,0,1,0)$ & $13$\\*
$(0,0,1,1)$ & $11$\\*
$(0,1,-1,0)$ & $8$\\*
$(0,1,0,-1)$ & $1$\\*
$(0,1,0,0)$ & $18$\\*
$(1,-1,-1,-1)$ & $19$\\*
$(1,-1,-1,1)$ & $16$\\*
$(1,-1,0,0)$ & $9$\\*
$(1,-1,1,-1)$ & $17$\\*
$(1,0,-1,0)$ & $2$\\*
$(1,0,0,-1)$ & $7$\\*
$(1,0,0,1)$ & $6$\\*
$(1,0,1,0)$ & $3$\\*
$(1,1,-1,1)$ & $15$\\*
$(1,1,0,0)$ & $10$\\*
$(1,1,1,-1)$ & $14$\\*
$(1,1,1,1)$ & $20$\\*
\end{longtable}
\endgroup
\subsubsection{An injective 45-ray completion after deleting \texorpdfstring{$c$}{c}}\label{S-app:cabello45-vectors}
Delete local ray 20 and contexts 6 and 9 from the first nine contexts
above. The seven remaining bases are followed by fourteen fresh-pair
contexts in Table~\ref{S-tab:cabello45-contexts}. The next table prints
all 45 vectors; 4, 5, and 20 are absent and 21--48 are fresh.
Its only uncovered orthogonal pairs are
$\{16,21\},\{16,22\},\{17,35\},\{17,36\}$.
\begingroup\small
\Needspace{24\baselineskip}
\begin{longtable}{rl}
\caption{All 45 unnormalized vectors of the vertex-deleted completion.}\label{S-tab:cabello45-vectors}\\
\toprule
$i$ & $v_i$\\\midrule
\endfirsthead
\toprule
$i$ & $v_i$\\\midrule
\endhead
\bottomrule\endfoot
$1$ & $(1,0,0,0)$\\
$2$ & $(0,1,0,0)$\\
$3$ & $(0,0,1,0)$\\
$6$ & $(1,1,-1,-1)$\\
$7$ & $(1,-1,1,-1)$\\
$8$ & $(1,-1,-1,1)$\\
$9$ & $(1,1,-1,1)$\\
$10$ & $(1,-1,1,1)$\\
$11$ & $(1,-1,-1,-1)$\\
$12$ & $(1,0,0,-1)$\\
$13$ & $(0,1,1,0)$\\
$14$ & $(1,0,1,0)$\\
$15$ & $(0,1,0,-1)$\\
$16$ & $(1,1,0,0)$\\
$17$ & $(0,0,1,-1)$\\
$18$ & $(1,0,0,1)$\\
$19$ & $(0,1,0,1)$\\
$21$ & $(0,0,1,2)$\\
$22$ & $(0,0,2,-1)$\\
$23$ & $(0,1,-1,3)$\\
$24$ & $(0,3,-3,-2)$\\
$25$ & $(0,1,2,2)$\\
$26$ & $(0,4,-1,-1)$\\
$27$ & $(1,0,-1,3)$\\
$28$ & $(3,0,-3,-2)$\\
$29$ & $(1,0,2,2)$\\
$30$ & $(4,0,-1,-1)$\\
$31$ & $(1,-1,0,3)$\\
$32$ & $(3,-3,0,-2)$\\
$33$ & $(5,3,4,4)$\\
$34$ & $(7,-9,-1,-1)$\\
$35$ & $(1,1,3,3)$\\
$36$ & $(3,3,-1,-1)$\\
$37$ & $(5,4,3,4)$\\
$38$ & $(7,-1,-9,-1)$\\
$39$ & $(2,-2,-1,3)$\\
$40$ & $(1,-1,-5,-3)$\\
$41$ & $(4,5,3,4)$\\*
$42$ & $(1,-7,9,1)$\\*
$43$ & $(4,-4,-1,-9)$\\*
$44$ & $(5,-5,13,3)$\\*
$45$ & $(6,5,4,-5)$\\*
$46$ & $(9,-1,-11,1)$\\*
$47$ & $(3,4,2,-3)$\\*
$48$ & $(1,-5,7,-1)$\\*
\end{longtable}
\endgroup
\begingroup\small
\Needspace{24\baselineskip}
\begin{longtable}{rl}
\caption{All 21 full bases of the vertex-deleted completion.}\label{S-tab:cabello45-contexts}\\
\toprule
$e$ & Context\\\midrule
\endfirsthead
\toprule
$e$ & Context\\\midrule
\endhead
\bottomrule\endfoot
$1$ & $\{9,10,12,13\}$\\
$2$ & $\{9,11,14,15\}$\\
$3$ & $\{10,11,16,17\}$\\
$4$ & $\{2,3,12,18\}$\\
$5$ & $\{1,3,15,19\}$\\
$6$ & $\{6,7,13,18\}$\\
$7$ & $\{6,8,14,19\}$\\
$8$ & $\{1,2,21,22\}$\\
$9$ & $\{1,13,23,24\}$\\
$10$ & $\{1,17,25,26\}$\\
$11$ & $\{2,14,27,28\}$\\
$12$ & $\{2,17,29,30\}$\\
$13$ & $\{3,16,31,32\}$\\
$14$ & $\{6,17,33,34\}$\\*
$15$ & $\{7,8,35,36\}$\\*
$16$ & $\{7,15,37,38\}$\\*
$17$ & $\{7,16,39,40\}$\\*
$18$ & $\{8,12,41,42\}$\\*
$19$ & $\{8,16,43,44\}$\\*
$20$ & $\{10,19,45,46\}$\\*
$21$ & $\{11,18,47,48\}$\\*
\end{longtable}
\endgroup
\subsubsection{All vectors and bases of the two 68-ray parents}
In each parent retain all 65 vectors of Appendix~\ref{S-d:app:h65},
with their original labels, and add the three vectors in the
corresponding rows of Table~\ref{S-tab:cabello68-rays};
Tables~\ref{S-tab:cabello68-1-contexts} and~\ref{S-tab:cabello68-2-contexts}
list all 35 bases of each parent. Thus the six rows specify every added ray;
the two parent inventories use distinct local meanings for labels 66--68.
These labels also differ from the 72-ray parent numbering in
Appendix~\ref{S-app:parents-vectors}.
\begingroup\small
\Needspace{24\baselineskip}
\begin{longtable}{lrl}
\caption{The three added rays in each local 68-ray parent inventory.}\label{S-tab:cabello68-rays}\\
\toprule
Parent variant & $i$ & $v_i$\\\midrule
\endfirsthead
\toprule
Parent variant & $i$ & $v_i$\\\midrule
\endhead
\bottomrule\endfoot
18--9 subsystem & $66$ & $(1,0,0,1)$\\*
18--9 subsystem & $67$ & $(0,1,0,1)$\\*
18--9 subsystem & $68$ & $(0,0,1,1)$\\*
20--11 subsystem & $66$ & $(0,1,-1,0)$\\*
20--11 subsystem & $67$ & $(1,0,-1,0)$\\*
20--11 subsystem & $68$ & $(1,-1,0,0)$\\*
\end{longtable}
\endgroup
\begingroup\small
\Needspace{24\baselineskip}
\begin{longtable}{rlrl}
\caption{All 35 full bases of 68-ray parent variant 1.}\label{S-tab:cabello68-1-contexts}\\
\toprule
$e$ & Context & $e$ & Context\\\midrule
\endfirsthead
\toprule
$e$ & Context & $e$ & Context\\\midrule
\endhead
\bottomrule\endfoot
$1$ & $\{1,2,3,4\}$ & $2$ & $\{1,2,17,68\}$\\*
$3$ & $\{1,3,15,67\}$ & $4$ & $\{1,13,22,23\}$\\*
$5$ & $\{1,15,30,31\}$ & $6$ & $\{1,17,38,39\}$\\*
$7$ & $\{2,3,12,66\}$ & $8$ & $\{2,12,18,19\}$\\*
$9$ & $\{2,14,26,27\}$ & $10$ & $\{2,17,40,41\}$\\*
$11$ & $\{3,12,20,21\}$ & $12$ & $\{3,15,32,33\}$\\*
$13$ & $\{3,16,34,35\}$ & $14$ & $\{4,13,24,25\}$\\*
$15$ & $\{4,14,28,29\}$ & $16$ & $\{4,16,36,37\}$\\*
$17$ & $\{5,6,7,8\}$ & $18$ & $\{5,12,42,43\}$\\*
$19$ & $\{5,15,54,55\}$ & $20$ & $\{5,17,62,63\}$\\*
$21$ & $\{6,7,13,66\}$ & $22$ & $\{6,8,14,67\}$\\*
$23$ & $\{6,13,46,47\}$ & $24$ & $\{6,14,50,51\}$\\*
$25$ & $\{6,17,64,65\}$ & $26$ & $\{7,8,16,68\}$\\*
$27$ & $\{7,13,48,49\}$ & $28$ & $\{7,15,56,57\}$\\*
$29$ & $\{7,16,58,59\}$ & $30$ & $\{8,12,44,45\}$\\*
$31$ & $\{8,14,52,53\}$ & $32$ & $\{8,16,60,61\}$\\*
$33$ & $\{9,10,12,13\}$ & $34$ & $\{9,11,14,15\}$\\*
$35$ & $\{10,11,16,17\}$ & & \\*
\end{longtable}
\endgroup
\begingroup\small
\Needspace{24\baselineskip}
\begin{longtable}{rlrl}
\caption{All 35 full bases of 68-ray parent variant 2.}\label{S-tab:cabello68-2-contexts}\\
\toprule
$e$ & Context & $e$ & Context\\\midrule
\endfirsthead
\toprule
$e$ & Context & $e$ & Context\\\midrule
\endhead
\bottomrule\endfoot
$1$ & $\{1,2,3,4\}$ & $2$ & $\{1,4,13,66\}$\\*
$3$ & $\{1,13,22,23\}$ & $4$ & $\{1,15,30,31\}$\\*
$5$ & $\{1,17,38,39\}$ & $6$ & $\{2,4,14,67\}$\\*
$7$ & $\{2,12,18,19\}$ & $8$ & $\{2,14,26,27\}$\\*
$9$ & $\{2,17,40,41\}$ & $10$ & $\{3,4,16,68\}$\\*
$11$ & $\{3,12,20,21\}$ & $12$ & $\{3,15,32,33\}$\\*
$13$ & $\{3,16,34,35\}$ & $14$ & $\{4,13,24,25\}$\\*
$15$ & $\{4,14,28,29\}$ & $16$ & $\{4,16,36,37\}$\\*
$17$ & $\{5,6,7,8\}$ & $18$ & $\{5,6,17,68\}$\\*
$19$ & $\{5,7,15,67\}$ & $20$ & $\{5,8,12,66\}$\\*
$21$ & $\{5,12,42,43\}$ & $22$ & $\{5,15,54,55\}$\\*
$23$ & $\{5,17,62,63\}$ & $24$ & $\{6,13,46,47\}$\\*
$25$ & $\{6,14,50,51\}$ & $26$ & $\{6,17,64,65\}$\\*
$27$ & $\{7,13,48,49\}$ & $28$ & $\{7,15,56,57\}$\\*
$29$ & $\{7,16,58,59\}$ & $30$ & $\{8,12,44,45\}$\\*
$31$ & $\{8,14,52,53\}$ & $32$ & $\{8,16,60,61\}$\\*
$33$ & $\{9,10,12,13\}$ & $34$ & $\{9,11,14,15\}$\\*
$35$ & $\{10,11,16,17\}$ & & \\*
\end{longtable}
\endgroup
For variant 1 the nine-context parity subsystem consists of
exactly $C_3,C_4,C_5,D_1,\ldots,D_6$ with $a=66,b=67,c=68$.
For variant 2 the twenty-ray subsystem consists of all rows whose
labels lie in $\{1,\ldots,17,66,67,68\}$. The machine-readable
inventories repeat every coordinate, basis, and inclusion map.

\clearpage
\section{Two-Valued States, Extendability Certificates, Colorings, and Separation Profiles}\label{S-app:all-states}\label{S-app:states}

\subsection{State, coloring, odd-cycle certificate, and profile tables for chromatically complete configurations (\texorpdfstring{$\chi(H)=d$}{chi(H)=d})}

This subsection provides the complete state and coloring enumerations and odd-cycle nonextendability certificates for $H_{6,3}$ (Table~\ref{S-tab:h6-states}), $P_5$ (Table~\ref{S-tab:p5-states}), $H_8$ (Tables~\ref{S-tab:h8-states} and~\ref{S-tab:h8-colors}), $H_{17,10}$ (Tables~\ref{S-tab:h17-states} and~\ref{S-tab:h17-colors}), $H_{21,13}$ (Tables~\ref{S-tab:h21-states}, \ref{S-tab:h21-colors}, and~\ref{S-tab:chromatic-signatures}), the vertex separation profiles and Born probability decomposition on $H_{17,10}$ and $H_{21,13}$ (Tables~\ref{S-tab:profiles}--\ref{S-tab:quantum-probabilities}), the complete state and normalized coloring tables for all faithful $3$- and $4$-dimensional logical gadgets (Tables~\ref{S-tab:states-3b}--\ref{S-tab:colors-4f} and~\ref{S-tab:pure-witnesses}), and the ten chromatically inseparable faithful restrictions of $H_{\YO}^{+}$ (Table~\ref{S-tab:inventory}).

\begingroup\small\setlength{\LTpre}{6pt}\setlength{\LTpost}{8pt}
\Needspace{10\baselineskip}
\begin{longtable}{rlcl}
\caption{All two-valued states of $H_{6,3}$. The last column gives an odd cycle in the complement of each nonextendable support.}\label{S-tab:h6-states}\\
\toprule $j$ & support of $s_j$ & extendable & odd cycle\\\midrule
\endfirsthead
\toprule $j$ & support of $s_j$ & extendable & odd cycle\\\midrule
\endhead
\bottomrule\endfoot
$1$ & $\{1,4\}$ & yes & -- \\*
$2$ & $\{2,4,6\}$ & no & $(3,1,5)$ \\*
$3$ & $\{2,5\}$ & yes & -- \\*
$4$ & $\{3,6\}$ & yes & -- \\
\end{longtable}\endgroup

\begingroup\small\setlength{\LTpre}{6pt}\setlength{\LTpost}{8pt}
\Needspace{17\baselineskip}
\begin{longtable}{rlcl}
\caption{All two-valued states of $P_5$. The last column gives an odd cycle in the complement of each nonextendable support.}\label{S-tab:p5-states}\\
\toprule $j$ & support of $s_j$ & extendable & odd cycle\\\midrule
\endfirsthead
\toprule $j$ & support of $s_j$ & extendable & odd cycle\\\midrule
\endhead
\bottomrule\endfoot
$1$ & $\{1,3,9\}$ & yes & -- \\*
$2$ & $\{1,4,7\}$ & yes & -- \\*
$3$ & $\{1,7,8,9\}$ & yes & -- \\*
$4$ & $\{2,4,10\}$ & yes & -- \\*
$5$ & $\{2,5,8\}$ & yes & -- \\*
$6$ & $\{2,8,9,10\}$ & yes & -- \\*
$7$ & $\{3,5,6\}$ & yes & -- \\*
$8$ & $\{3,6,9,10\}$ & yes & -- \\*
$9$ & $\{4,6,7,10\}$ & yes & -- \\*
$10$ & $\{5,6,7,8\}$ & yes & -- \\*
$11$ & $\{6,7,8,9,10\}$ & no & $(3,2,1,5,4)$ \\
\end{longtable}\endgroup

\begingroup\small\setlength{\LTpre}{6pt}\setlength{\LTpost}{8pt}
\Needspace{17\baselineskip}
\begin{longtable}{rlcl}
\caption{All two-valued states of $H_8$. The last column gives an odd cycle in the complement of each nonextendable support.}\label{S-tab:h8-states}\\
\toprule $j$ & support of $s_j$ & extendable & odd cycle\\\midrule
\endfirsthead
\toprule $j$ & support of $s_j$ & extendable & odd cycle\\\midrule
\endhead
\bottomrule\endfoot
$1$ & $\{1,4,6,9\}$ & yes & -- \\*
$2$ & $\{1,5,8,10,14\}$ & no & $(9,7,6,12,11)$ \\*
$3$ & $\{1,5,9,13,14\}$ & yes & -- \\*
$4$ & $\{2,4,6,8,11\}$ & yes & -- \\*
$5$ & $\{2,4,7,11,14\}$ & no & $(5,3,1,12,6)$ \\*
$6$ & $\{2,5,8,10,12\}$ & yes & -- \\*
$7$ & $\{2,5,8,11,13,14\}$ & yes & -- \\*
$8$ & $\{2,5,9,12,13\}$ & no & $(4,3,1,11,10)$ \\*
$9$ & $\{3,6,8,11,13\}$ & no & $(7,5,4,10,9)$ \\*
$10$ & $\{3,7,10,12\}$ & yes & -- \\*
$11$ & $\{3,7,11,13,14\}$ & yes & -- \\
\end{longtable}\endgroup

\begin{table}[htbp]\centering\small\setlength{\tabcolsep}{3pt}
\caption{All strong three-colorings of $H_8$ up to color permutation. The anchor is $(1,2,3)=(0,1,2)$.}\label{S-tab:h8-colors}
\begin{tabular}{rrrrrrrrrrrrrrr}
\toprule $i$ & $1$ & $2$ & $3$ & $4$ & $5$ & $6$ & $7$ & $8$ & $9$ & $10$ & $11$ & $12$ & $13$ & $14$\\\midrule
$c_1$ & $0$ & $1$ & $2$ & $0$ & $1$ & $0$ & $2$ & $1$ & $0$ & $1$ & $2$ & $1$ & $2$ & $2$\\
$c_2$ & $0$ & $1$ & $2$ & $0$ & $1$ & $0$ & $2$ & $1$ & $0$ & $2$ & $1$ & $2$ & $1$ & $1$\\
$c_3$ & $0$ & $1$ & $2$ & $1$ & $0$ & $1$ & $2$ & $1$ & $0$ & $2$ & $1$ & $2$ & $0$ & $0$\\
\bottomrule\end{tabular}\end{table}

\begingroup\small\setlength{\LTpre}{6pt}\setlength{\LTpost}{8pt}
\Needspace{24\baselineskip}
\begin{longtable}{rlcl}
\caption{All two-valued states of $H_{17,10}$. The last column gives an odd cycle in the complement of each nonextendable support.}\label{S-tab:h17-states}\\
\toprule $j$ & support of $s_j$ & extendable & odd cycle\\\midrule
\endfirsthead
\toprule $j$ & support of $s_j$ & extendable & odd cycle\\\midrule
\endhead
\bottomrule\endfoot
$1$ & $\{1,4,8,14,15,23\}$ & yes & -- \\*
$2$ & $\{1,5,7,14,18,23,24\}$ & no & $(4,11,3,2,8,9,10)$ \\*
$3$ & $\{1,5,9,11,18\}$ & yes & -- \\*
$4$ & $\{1,5,9,14,15,18,24\}$ & no & $(11,3,2,8,7)$ \\*
$5$ & $\{1,6,8,11,18,23\}$ & yes & -- \\*
$6$ & $\{1,6,8,14,15,18,23,24\}$ & no & $(11,3,2,5,4)$ \\*
$7$ & $\{2,4,7,14,19,23\}$ & yes & -- \\*
$8$ & $\{2,4,9,14,15,19\}$ & yes & -- \\*
$9$ & $\{2,6,7,10,14,24\}$ & yes & -- \\*
$10$ & $\{2,6,7,14,18,19,23,24\}$ & no & $(11,3,1,10,4)$ \\*
$11$ & $\{2,6,9,11,18,19\}$ & yes & -- \\*
$12$ & $\{2,6,9,14,15,18,19,24\}$ & no & $(11,3,1,10,4)$ \\*
$13$ & $\{3,4,8,15,19,23\}$ & yes & -- \\*
$14$ & $\{3,5,7,10,24\}$ & yes & -- \\*
$15$ & $\{3,5,7,18,19,23,24\}$ & no & $(8,2,1,10,9)$ \\*
$16$ & $\{3,5,9,15,18,19,24\}$ & no & $(7,8,2,1,10,4,11)$ \\*
$17$ & $\{3,6,8,10,15,24\}$ & yes & -- \\*
$18$ & $\{3,6,8,15,18,19,23,24\}$ & no & $(5,2,1,10,4)$ \\
\end{longtable}\endgroup

\begin{table}[htbp]\centering\small\setlength{\tabcolsep}{3pt}
\caption{All strong three-colorings of $H_{17,10}$ up to color permutation. The anchor is $(2,8,5)=(0,1,2)$.}\label{S-tab:h17-colors}
\begin{tabular}{rrrrrrrrrrrrrrrrrr}
\toprule $i$ & $1$ & $2$ & $3$ & $4$ & $5$ & $6$ & $7$ & $8$ & $9$ & $10$ & $11$ & $14$ & $15$ & $18$ & $19$ & $23$ & $24$\\\midrule
$c_1$ & $1$ & $0$ & $2$ & $0$ & $2$ & $1$ & $2$ & $1$ & $0$ & $2$ & $1$ & $0$ & $0$ & $1$ & $0$ & $1$ & $2$\\
$c_2$ & $1$ & $0$ & $2$ & $1$ & $2$ & $0$ & $2$ & $1$ & $0$ & $2$ & $0$ & $1$ & $1$ & $0$ & $0$ & $1$ & $2$\\
$c_3$ & $2$ & $0$ & $1$ & $0$ & $2$ & $1$ & $0$ & $1$ & $2$ & $1$ & $2$ & $0$ & $1$ & $2$ & $0$ & $0$ & $1$\\
$c_4$ & $2$ & $0$ & $1$ & $1$ & $2$ & $0$ & $0$ & $1$ & $2$ & $0$ & $2$ & $0$ & $1$ & $2$ & $1$ & $1$ & $0$\\
\bottomrule\end{tabular}\end{table}

\begingroup\small\setlength{\LTpre}{6pt}\setlength{\LTpost}{8pt}
\Needspace{27\baselineskip}
\begin{longtable}{rlcl}
\caption{All two-valued states of $H_{21,13}$. The last column gives an odd cycle in the complement of each nonextendable support.}\label{S-tab:h21-states}\\
\toprule $j$ & support of $s_j$ & extendable & odd cycle\\\midrule
\endfirsthead
\toprule $j$ & support of $s_j$ & extendable & odd cycle\\\midrule
\endhead
\bottomrule\endfoot
$1$ & $\{1,4,8,14,15,21,22,23\}$ & no & $(7,11,3,2,5,6,13)$ \\*
$2$ & $\{1,5,7,14,18,22,23,24\}$ & no & $(4,11,3,2,8,9,10)$ \\*
$3$ & $\{1,5,9,11,18,21,22\}$ & yes & -- \\*
$4$ & $\{1,5,9,14,15,18,21,22,24\}$ & no & $(11,3,2,8,7)$ \\*
$5$ & $\{1,6,8,11,18,21,23\}$ & yes & -- \\*
$6$ & $\{1,6,8,14,15,18,21,23,24\}$ & no & $(11,3,2,5,4)$ \\*
$7$ & $\{2,4,7,14,19,20,22,23\}$ & no & $(8,9,10,1,13,6,5)$ \\*
$8$ & $\{2,4,9,13,14,15,19\}$ & yes & -- \\*
$9$ & $\{2,4,9,14,15,19,20,21,22\}$ & no & $(11,3,1,13,7)$ \\*
$10$ & $\{2,6,7,10,14,20,24\}$ & yes & -- \\*
$11$ & $\{2,6,7,14,18,19,20,23,24\}$ & no & $(11,3,1,10,4)$ \\*
$12$ & $\{2,6,9,11,18,19,20,21\}$ & no & $(5,4,10,1,13,7,8)$ \\*
$13$ & $\{2,6,9,14,15,18,19,20,21,24\}$ & no & $(11,3,1,10,4)$ \\*
$14$ & $\{3,4,8,13,15,19,23\}$ & yes & -- \\*
$15$ & $\{3,4,8,15,19,20,21,22,23\}$ & no & $(5,2,1,13,6)$ \\*
$16$ & $\{3,5,7,10,20,22,24\}$ & yes & -- \\*
$17$ & $\{3,5,7,18,19,20,22,23,24\}$ & no & $(8,2,1,10,9)$ \\*
$18$ & $\{3,5,9,13,15,18,19,24\}$ & no & $(7,8,2,1,10,4,11)$ \\*
$19$ & $\{3,5,9,15,18,19,20,21,22,24\}$ & no & $(8,2,1,13,7)$ \\*
$20$ & $\{3,6,8,10,15,20,21,24\}$ & no & $(4,5,2,1,13,7,11)$ \\*
$21$ & $\{3,6,8,15,18,19,20,21,23,24\}$ & no & $(5,2,1,10,4)$ \\
\end{longtable}\endgroup

\begin{table}[htbp]\centering\small\setlength{\tabcolsep}{3pt}
\caption{All strong three-colorings of $H_{21,13}$ up to color permutation. The anchor is $(2,8,5)=(0,1,2)$.}\label{S-tab:h21-colors}
\begin{tabular}{rrrrrrrrrrrrrrrrrrrrrr}
\toprule $i$ & $1$ & $2$ & $3$ & $4$ & $5$ & $6$ & $7$ & $8$ & $9$ & $10$ & $11$ & $13$ & $14$ & $15$ & $18$ & $19$ & $20$ & $21$ & $22$ & $23$ & $24$\\\midrule
$c_1$ & $1$ & $0$ & $2$ & $0$ & $2$ & $1$ & $2$ & $1$ & $0$ & $2$ & $1$ & $0$ & $0$ & $0$ & $1$ & $0$ & $2$ & $1$ & $2$ & $1$ & $2$\\
$c_2$ & $2$ & $0$ & $1$ & $1$ & $2$ & $0$ & $0$ & $1$ & $2$ & $0$ & $2$ & $1$ & $0$ & $1$ & $2$ & $1$ & $0$ & $2$ & $2$ & $1$ & $0$\\
\bottomrule\end{tabular}\end{table}

\begin{table}[htbp]\centering\small
\caption{The nine chromatic signatures of $H_{21,13}$.
Columns give $c_1$ and rows give $c_2$ for the two normalized colorings
in Table~\ref{S-tab:h21-colors}. Each cell is one inseparability class:
its vertices have the same color in every admissible coloring.
Ordinary two-valued states distinguish every pair, including vertices
in the same cell.}\label{S-tab:chromatic-signatures}
\renewcommand{\arraystretch}{1.3}\setlength{\tabcolsep}{10pt}
\begin{tabular}{cccc}\toprule
$c_2\backslash c_1$ & $0$ & $1$ & $2$\\\midrule
$2$ & $\{9\}$ & $\{1,11,18,21\}$ & $\{5,22\}$\\
$1$ & $\{4,13,15,19\}$ & $\{8,23\}$ & $\{3\}$\\
$0$ & $\{2,14\}$ & $\{6\}$ & $\{7,10,20,24\}$\\
\bottomrule\end{tabular}\end{table}

\begingroup\small\setlength{\LTpre}{6pt}\setlength{\LTpost}{8pt}
\Needspace{27\baselineskip}
\begin{longtable}{rrrr}
\caption{Complete state profiles in hexadecimal, and all Born probabilities for the 17-ray preparation in Eq.~\eqref{eq:rho}. The probabilities are also reproduced by Eq.~\eqref{eq:mixture}.}\label{S-tab:profiles}\\
\toprule $i$ & $H_{17,10}$ profile & $H_{21,13}$ profile & $p_i$ on $H_{17,10}$\\\midrule
\endfirsthead
\toprule $i$ & $H_{17,10}$ profile & $H_{21,13}$ profile & $p_i$ on $H_{17,10}$\\\midrule
\endhead
\bottomrule\endfoot
$1$ & \texttt{3F} & \texttt{3F} & $1/12$ \\*
$2$ & \texttt{FC0} & \texttt{1FC0} & $1/6$ \\*
$3$ & \texttt{3F000} & \texttt{1FE000} & $3/4$ \\*
$4$ & \texttt{10C1} & \texttt{61C1} & $0$ \\*
$5$ & \texttt{E00E} & \texttt{7800E} & $2/3$ \\*
$6$ & \texttt{30F30} & \texttt{181E30} & $1/3$ \\*
$7$ & \texttt{6342} & \texttt{18642} & $3/4$ \\*
$8$ & \texttt{31031} & \texttt{186031} & $1/6$ \\*
$9$ & \texttt{8C8C} & \texttt{6198C} & $1/12$ \\*
$10$ & \texttt{12100} & \texttt{88200} & $8/9$ \\*
$11$ & \texttt{414} & \texttt{814} & $0$ \\*
$13$ & \texttt{--} & \texttt{22080} & -- \\*
$14$ & \texttt{BEB} & \texttt{17EB} & $1/4$ \\*
$15$ & \texttt{398A9} & \texttt{1E71A9} & $1/4$ \\*
$18$ & \texttt{2CE3E} & \texttt{171C3E} & $1/9$ \\*
$19$ & \texttt{2DEC0} & \texttt{177DC0} & $1/36$ \\*
$20$ & \texttt{--} & \texttt{1DDF40} & -- \\*
$21$ & \texttt{--} & \texttt{1C593D} & -- \\*
$22$ & \texttt{--} & \texttt{5C14F} & -- \\*
$23$ & \texttt{25273} & \texttt{116473} & $1/36$ \\*
$24$ & \texttt{3EB2A} & \texttt{1F962A} & $1$ \\
\end{longtable}\endgroup

\begin{table}[t]\centering\small\setlength{\tabcolsep}{3pt}\renewcommand{\arraystretch}{1.25}
\caption{Valid states and exact weights reproducing two $H_{17,10}$ preparations: the five-state mixture for $P_{24}$ (upper block, Eq.~\eqref{eq:mixture}) and the four-state mixture for $P_{z_3}$, the preparation along $z_3=r_5$ (lower block). State numbering is that of Table~\ref{S-tab:h17-states}. Every context equation and every Born marginal are checked with rational arithmetic.}\label{S-tab:quantum}
\begin{tabular}{rrp{.61\linewidth}l}\toprule
State & Weight & Support & Extendable?\\\midrule
\multicolumn{4}{l}{Preparation $P_{24}$, nonextendable weight $\frac{1}{9}$}\\
$s_{4}$ & $\frac{1}{12}$ & $\{1,\allowbreak 5,\allowbreak 9,\allowbreak 14,\allowbreak 15,\allowbreak 18,\allowbreak 24\}$ & no\\
$s_{9}$ & $\frac{5}{36}$ & $\{2,\allowbreak 6,\allowbreak 7,\allowbreak 10,\allowbreak 14,\allowbreak 24\}$ & yes\\
$s_{10}$ & $\frac{1}{36}$ & $\{2,\allowbreak 6,\allowbreak 7,\allowbreak 14,\allowbreak 18,\allowbreak 19,\allowbreak 23,\allowbreak 24\}$ & no\\
$s_{14}$ & $\frac{7}{12}$ & $\{3,\allowbreak 5,\allowbreak 7,\allowbreak 10,\allowbreak 24\}$ & yes\\
$s_{17}$ & $\frac{1}{6}$ & $\{3,\allowbreak 6,\allowbreak 8,\allowbreak 10,\allowbreak 15,\allowbreak 24\}$ & yes\\
\midrule
\multicolumn{4}{l}{Preparation $P_{z_3}$, nonextendable weight $\frac{1}{3}$}\\
$s_{3}$ & $\frac{1}{3}$ & $\{1,\allowbreak 5,\allowbreak 9,\allowbreak 11,\allowbreak 18\}$ & yes\\
$s_{4}$ & $\frac{1}{6}$ & $\{1,\allowbreak 5,\allowbreak 9,\allowbreak 14,\allowbreak 15,\allowbreak 18,\allowbreak 24\}$ & no\\
$s_{14}$ & $\frac{1}{3}$ & $\{3,\allowbreak 5,\allowbreak 7,\allowbreak 10,\allowbreak 24\}$ & yes\\
$s_{15}$ & $\frac{1}{6}$ & $\{3,\allowbreak 5,\allowbreak 7,\allowbreak 18,\allowbreak 19,\allowbreak 23,\allowbreak 24\}$ & no\\
\bottomrule\end{tabular}\end{table}
\begin{table}[t]\centering\small\setlength{\tabcolsep}{7pt}\renewcommand{\arraystretch}{1.3}
\caption{Exact Born probabilities $p_i=(r_{24}\cdot r_i)^2/(6\|r_i\|^2)$ for $P_{24}$ and $p'_i=(r_5\cdot r_i)^2/\|r_i\|^2$ for $P_{z_3}$, also reproduced by the two mixtures of Table~\ref{S-tab:quantum}.}\label{S-tab:quantum-probabilities}
\begin{tabular}{rccrcc}\toprule
$i$ & $p_i$ & $p'_i$ & $i$ & $p_i$ & $p'_i$\\\midrule
$1$ & $\frac{1}{12}$ & $\frac{1}{2}$ & $10$ & $\frac{8}{9}$ & $\frac{1}{3}$\\
$2$ & $\frac{1}{6}$ & $0$ & $11$ & $0$ & $\frac{1}{3}$\\
$3$ & $\frac{3}{4}$ & $\frac{1}{2}$ & $14$ & $\frac{1}{4}$ & $\frac{1}{6}$\\
$4$ & $0$ & $0$ & $15$ & $\frac{1}{4}$ & $\frac{1}{6}$\\
$5$ & $\frac{2}{3}$ & $1$ & $18$ & $\frac{1}{9}$ & $\frac{2}{3}$\\
$6$ & $\frac{1}{3}$ & $0$ & $19$ & $\frac{1}{36}$ & $\frac{1}{6}$\\
$7$ & $\frac{3}{4}$ & $\frac{1}{2}$ & $23$ & $\frac{1}{36}$ & $\frac{1}{6}$\\
$8$ & $\frac{1}{6}$ & $0$ & $24$ & $1$ & $\frac{2}{3}$\\
$9$ & $\frac{1}{12}$ & $\frac{1}{2}$ & & & \\
\bottomrule\end{tabular}\end{table}

\FloatBarrier
\subsection{State conventions for the additional gadgets}
The state tables below are local to each labeled configuration. Supports
are in lexicographic order; E is yes precisely for color-class supports.
The three-dimensional pentagon and $H_{21,13}$ use the complete earlier
tables and are not repeated. The four-dimensional apex-only states are
included explicitly. All ordinary state families listed here are unital
and separating.
\Needspace{15\baselineskip}\subsection{All states of \texorpdfstring{$B_{13,7}$}{gadget 3b}}
\begingroup\footnotesize\setlength{\tabcolsep}{3pt}\setlength{\LTpre}{5pt}\setlength{\LTpost}{8pt}
\Needspace{9\baselineskip}
\begin{longtable}{rlc@{\qquad}rlc}
\caption{All two-valued states of $B_{13,7}$. E indicates extendability.}\label{S-tab:states-3b}\\\toprule $j$ & Support & E & $j$ & Support & E\\\midrule\endfirsthead
\toprule $j$ & Support & E & $j$ & Support & E\\\midrule\endhead\bottomrule\endfoot
1 & $\{1,\allowbreak 5,\allowbreak 7,\allowbreak 14,\allowbreak 23\}$ & yes & 8 & $\{2,\allowbreak 9,\allowbreak 11,\allowbreak 19\}$ & yes\\
2 & $\{1,\allowbreak 5,\allowbreak 9,\allowbreak 11\}$ & yes & 9 & $\{2,\allowbreak 9,\allowbreak 14,\allowbreak 15,\allowbreak 19\}$ & yes\\
3 & $\{1,\allowbreak 5,\allowbreak 9,\allowbreak 14,\allowbreak 15\}$ & no & 10 & $\{3,\allowbreak 5,\allowbreak 7,\allowbreak 10\}$ & yes\\
4 & $\{1,\allowbreak 8,\allowbreak 11,\allowbreak 23\}$ & yes & 11 & $\{3,\allowbreak 5,\allowbreak 7,\allowbreak 19,\allowbreak 23\}$ & no\\
5 & $\{1,\allowbreak 8,\allowbreak 14,\allowbreak 15,\allowbreak 23\}$ & yes & 12 & $\{3,\allowbreak 5,\allowbreak 9,\allowbreak 15,\allowbreak 19\}$ & yes\\
6 & $\{2,\allowbreak 7,\allowbreak 10,\allowbreak 14\}$ & yes & 13 & $\{3,\allowbreak 8,\allowbreak 10,\allowbreak 15\}$ & yes\\
7 & $\{2,\allowbreak 7,\allowbreak 14,\allowbreak 19,\allowbreak 23\}$ & yes & 14 & $\{3,\allowbreak 8,\allowbreak 15,\allowbreak 19,\allowbreak 23\}$ & yes\\
\end{longtable}\endgroup
\Needspace{19\baselineskip}\subsection{All states of \texorpdfstring{$D_{16,9}$}{gadget 3c}}
\begingroup\footnotesize\setlength{\tabcolsep}{3pt}\setlength{\LTpre}{5pt}\setlength{\LTpost}{8pt}
\Needspace{9\baselineskip}\begin{longtable}{rlc@{\qquad}rlc}
\caption{All two-valued states of $D_{16,9}$. E indicates extendability.}\label{S-tab:states-3c}\\\toprule $j$ & Support & E & $j$ & Support & E\\\midrule\endfirsthead
\toprule $j$ & Support & E & $j$ & Support & E\\\midrule\endhead\bottomrule\endfoot
1 & $\{1,\allowbreak 4,\allowbreak 8,\allowbreak 14,\allowbreak 15,\allowbreak 23\}$ & yes & 12 & $\{2,\allowbreak 6,\allowbreak 7,\allowbreak 10,\allowbreak 14,\allowbreak 24\}$ & yes\\
2 & $\{1,\allowbreak 5,\allowbreak 7,\allowbreak 10,\allowbreak 14,\allowbreak 24\}$ & yes & 13 & $\{2,\allowbreak 6,\allowbreak 7,\allowbreak 14,\allowbreak 18,\allowbreak 23,\allowbreak 24\}$ & no\\
3 & $\{1,\allowbreak 5,\allowbreak 7,\allowbreak 14,\allowbreak 18,\allowbreak 23,\allowbreak 24\}$ & no & 14 & $\{2,\allowbreak 6,\allowbreak 9,\allowbreak 11,\allowbreak 18\}$ & yes\\
4 & $\{1,\allowbreak 5,\allowbreak 9,\allowbreak 11,\allowbreak 18\}$ & yes & 15 & $\{2,\allowbreak 6,\allowbreak 9,\allowbreak 14,\allowbreak 15,\allowbreak 18,\allowbreak 24\}$ & no\\
5 & $\{1,\allowbreak 5,\allowbreak 9,\allowbreak 14,\allowbreak 15,\allowbreak 18,\allowbreak 24\}$ & no & 16 & $\{3,\allowbreak 4,\allowbreak 8,\allowbreak 15,\allowbreak 23\}$ & yes\\
6 & $\{1,\allowbreak 6,\allowbreak 8,\allowbreak 10,\allowbreak 11\}$ & yes & 17 & $\{3,\allowbreak 5,\allowbreak 7,\allowbreak 10,\allowbreak 24\}$ & yes\\
7 & $\{1,\allowbreak 6,\allowbreak 8,\allowbreak 10,\allowbreak 14,\allowbreak 15,\allowbreak 24\}$ & no & 18 & $\{3,\allowbreak 5,\allowbreak 7,\allowbreak 18,\allowbreak 23,\allowbreak 24\}$ & yes\\
8 & $\{1,\allowbreak 6,\allowbreak 8,\allowbreak 11,\allowbreak 18,\allowbreak 23\}$ & yes & 19 & $\{3,\allowbreak 5,\allowbreak 9,\allowbreak 15,\allowbreak 18,\allowbreak 24\}$ & yes\\
\end{longtable}\endgroup
\Needspace{19\baselineskip}\subsection{All states of \texorpdfstring{$I_{16,9}$}{gadget 3d}}
\begingroup\footnotesize\setlength{\tabcolsep}{3pt}\setlength{\LTpre}{5pt}\setlength{\LTpost}{8pt}
\Needspace{9\baselineskip}\begin{longtable}{rlc@{\qquad}rlc}
\caption{All two-valued states of $I_{16,9}$. E indicates extendability.}\label{S-tab:states-3d}\\\toprule $j$ & Support & E & $j$ & Support & E\\\midrule\endfirsthead
\toprule $j$ & Support & E & $j$ & Support & E\\\midrule\endhead\bottomrule\endfoot
1 & $\{1,\allowbreak 4,\allowbreak 5,\allowbreak 7,\allowbreak 14,\allowbreak 23\}$ & yes & 12 & $\{2,\allowbreak 7,\allowbreak 14,\allowbreak 18,\allowbreak 19,\allowbreak 23,\allowbreak 24\}$ & no\\
2 & $\{1,\allowbreak 4,\allowbreak 5,\allowbreak 9,\allowbreak 14,\allowbreak 15\}$ & no & 13 & $\{2,\allowbreak 9,\allowbreak 11,\allowbreak 18,\allowbreak 19\}$ & yes\\
3 & $\{1,\allowbreak 4,\allowbreak 8,\allowbreak 14,\allowbreak 15,\allowbreak 23\}$ & yes & 14 & $\{2,\allowbreak 9,\allowbreak 14,\allowbreak 15,\allowbreak 18,\allowbreak 19,\allowbreak 24\}$ & no\\
4 & $\{1,\allowbreak 5,\allowbreak 7,\allowbreak 14,\allowbreak 18,\allowbreak 23,\allowbreak 24\}$ & no & 15 & $\{3,\allowbreak 4,\allowbreak 5,\allowbreak 7,\allowbreak 19,\allowbreak 23\}$ & no\\
5 & $\{1,\allowbreak 5,\allowbreak 9,\allowbreak 11,\allowbreak 18\}$ & yes & 16 & $\{3,\allowbreak 4,\allowbreak 5,\allowbreak 9,\allowbreak 15,\allowbreak 19\}$ & yes\\
6 & $\{1,\allowbreak 5,\allowbreak 9,\allowbreak 14,\allowbreak 15,\allowbreak 18,\allowbreak 24\}$ & no & 17 & $\{3,\allowbreak 4,\allowbreak 8,\allowbreak 15,\allowbreak 19,\allowbreak 23\}$ & yes\\
7 & $\{1,\allowbreak 8,\allowbreak 11,\allowbreak 18,\allowbreak 23\}$ & yes & 18 & $\{3,\allowbreak 5,\allowbreak 7,\allowbreak 10,\allowbreak 24\}$ & yes\\
8 & $\{1,\allowbreak 8,\allowbreak 14,\allowbreak 15,\allowbreak 18,\allowbreak 23,\allowbreak 24\}$ & no & 19 & $\{3,\allowbreak 5,\allowbreak 7,\allowbreak 18,\allowbreak 19,\allowbreak 23,\allowbreak 24\}$ & no\\
\end{longtable}\endgroup
\Needspace{20\baselineskip}\subsection{All states of \texorpdfstring{$T_{19,11}$}{gadget 3e}}
\begingroup\footnotesize\setlength{\tabcolsep}{3pt}\setlength{\LTpre}{5pt}\setlength{\LTpost}{8pt}
\Needspace{9\baselineskip}\begin{longtable}{rlc@{\qquad}rlc}
\caption{All two-valued states of $T_{19,11}$. E indicates extendability.}\label{S-tab:states-3e}\\\toprule $j$ & Support & E & $j$ & Support & E\\\midrule\endfirsthead
\toprule $j$ & Support & E & $j$ & Support & E\\\midrule\endhead\bottomrule\endfoot
1 & $\{1,\allowbreak 4,\allowbreak 8,\allowbreak 14,\allowbreak 15,\allowbreak 23,\allowbreak 28\}$ & yes & 14 & $\{2,\allowbreak 6,\allowbreak 7,\allowbreak 14,\allowbreak 18,\allowbreak 23,\allowbreak 24\}$ & no\\
2 & $\{1,\allowbreak 5,\allowbreak 7,\allowbreak 10,\allowbreak 14,\allowbreak 24,\allowbreak 28\}$ & yes & 15 & $\{2,\allowbreak 6,\allowbreak 9,\allowbreak 11,\allowbreak 18,\allowbreak 27\}$ & yes\\
3 & $\{1,\allowbreak 5,\allowbreak 7,\allowbreak 14,\allowbreak 18,\allowbreak 23,\allowbreak 24\}$ & no & 16 & $\{2,\allowbreak 6,\allowbreak 9,\allowbreak 14,\allowbreak 15,\allowbreak 18,\allowbreak 24\}$ & no\\
4 & $\{1,\allowbreak 5,\allowbreak 9,\allowbreak 11,\allowbreak 18,\allowbreak 27\}$ & yes & 17 & $\{3,\allowbreak 4,\allowbreak 8,\allowbreak 15,\allowbreak 23,\allowbreak 26\}$ & yes\\
5 & $\{1,\allowbreak 5,\allowbreak 9,\allowbreak 14,\allowbreak 15,\allowbreak 18,\allowbreak 24\}$ & no & 18 & $\{3,\allowbreak 4,\allowbreak 8,\allowbreak 15,\allowbreak 23,\allowbreak 27,\allowbreak 28\}$ & no\\
6 & $\{1,\allowbreak 6,\allowbreak 8,\allowbreak 10,\allowbreak 11,\allowbreak 26\}$ & yes & 19 & $\{3,\allowbreak 5,\allowbreak 7,\allowbreak 10,\allowbreak 24,\allowbreak 26\}$ & yes\\
7 & $\{1,\allowbreak 6,\allowbreak 8,\allowbreak 10,\allowbreak 11,\allowbreak 27,\allowbreak 28\}$ & no & 20 & $\{3,\allowbreak 5,\allowbreak 7,\allowbreak 10,\allowbreak 24,\allowbreak 27,\allowbreak 28\}$ & no\\
8 & $\{1,\allowbreak 6,\allowbreak 8,\allowbreak 10,\allowbreak 14,\allowbreak 15,\allowbreak 24,\allowbreak 28\}$ & no & 21 & $\{3,\allowbreak 5,\allowbreak 7,\allowbreak 18,\allowbreak 23,\allowbreak 24,\allowbreak 27\}$ & yes\\
\end{longtable}\endgroup
\Needspace{14\baselineskip}\subsection{All states of \texorpdfstring{$W_{11,5}^{(4)}$}{gadget 4a}}
\begingroup\footnotesize\setlength{\tabcolsep}{3pt}\setlength{\LTpre}{5pt}\setlength{\LTpost}{8pt}
\Needspace{9\baselineskip}\begin{longtable}{rlc@{\qquad}rlc}
\caption{All two-valued states of $W_{11,5}^{(4)}$. E indicates extendability.}\label{S-tab:states-4a}\\\toprule $j$ & Support & E & $j$ & Support & E\\\midrule\endfirsthead
\toprule $j$ & Support & E & $j$ & Support & E\\\midrule\endhead\bottomrule\endfoot
1 & $\{0\}$ & yes & 7 & $\{2,\allowbreak 8,\allowbreak 9,\allowbreak 10\}$ & yes\\
2 & $\{1,\allowbreak 3,\allowbreak 9\}$ & yes & 8 & $\{3,\allowbreak 5,\allowbreak 6\}$ & yes\\
3 & $\{1,\allowbreak 4,\allowbreak 7\}$ & yes & 9 & $\{3,\allowbreak 6,\allowbreak 9,\allowbreak 10\}$ & yes\\
4 & $\{1,\allowbreak 7,\allowbreak 8,\allowbreak 9\}$ & yes & 10 & $\{4,\allowbreak 6,\allowbreak 7,\allowbreak 10\}$ & yes\\
5 & $\{2,\allowbreak 4,\allowbreak 10\}$ & yes & 11 & $\{5,\allowbreak 6,\allowbreak 7,\allowbreak 8\}$ & yes\\
6 & $\{2,\allowbreak 5,\allowbreak 8\}$ & yes & 12 & $\{6,\allowbreak 7,\allowbreak 8,\allowbreak 9,\allowbreak 10\}$ & no\\
\end{longtable}\endgroup
\Needspace{20\baselineskip}\subsection{All states of \texorpdfstring{$W_{16,6}^{(4)}$}{gadget 4b}}
\begingroup\footnotesize\setlength{\tabcolsep}{3pt}\setlength{\LTpre}{5pt}\setlength{\LTpost}{8pt}
\Needspace{9\baselineskip}\begin{longtable}{rlc@{\qquad}rlc}
\caption{All two-valued states of $W_{16,6}^{(4)}$. E indicates extendability.}\label{S-tab:states-4b}\\\toprule $j$ & Support & E & $j$ & Support & E\\\midrule\endfirsthead
\toprule $j$ & Support & E & $j$ & Support & E\\\midrule\endhead\bottomrule\endfoot
1 & $\{0,\allowbreak 7,\allowbreak 9,\allowbreak 10\}$ & yes & 37 & $\{2,\allowbreak 8,\allowbreak 12,\allowbreak 14,\allowbreak 15\}$ & yes\\
2 & $\{0,\allowbreak 7,\allowbreak 9,\allowbreak 15\}$ & yes & 38 & $\{3,\allowbreak 5,\allowbreak 6\}$ & yes\\
3 & $\{0,\allowbreak 7,\allowbreak 10,\allowbreak 14\}$ & yes & 39 & $\{3,\allowbreak 6,\allowbreak 9,\allowbreak 10,\allowbreak 11\}$ & yes\\
4 & $\{0,\allowbreak 7,\allowbreak 14,\allowbreak 15\}$ & yes & 40 & $\{3,\allowbreak 6,\allowbreak 9,\allowbreak 10,\allowbreak 12\}$ & yes\\
5 & $\{0,\allowbreak 9,\allowbreak 10,\allowbreak 13\}$ & yes & 41 & $\{3,\allowbreak 6,\allowbreak 9,\allowbreak 11,\allowbreak 15\}$ & yes\\
6 & $\{0,\allowbreak 9,\allowbreak 13,\allowbreak 15\}$ & yes & 42 & $\{3,\allowbreak 6,\allowbreak 9,\allowbreak 12,\allowbreak 15\}$ & yes\\
7 & $\{0,\allowbreak 10,\allowbreak 13,\allowbreak 14\}$ & yes & 43 & $\{3,\allowbreak 6,\allowbreak 10,\allowbreak 11,\allowbreak 14\}$ & yes\\
8 & $\{0,\allowbreak 13,\allowbreak 14,\allowbreak 15\}$ & yes & 44 & $\{3,\allowbreak 6,\allowbreak 10,\allowbreak 12,\allowbreak 14\}$ & yes\\
\end{longtable}\endgroup
\Needspace{16\baselineskip}\subsection{All states of \texorpdfstring{$B_{14,7}^{(4)}$}{gadget 4c}}
\begingroup\footnotesize\setlength{\tabcolsep}{3pt}\setlength{\LTpre}{5pt}\setlength{\LTpost}{8pt}
\Needspace{9\baselineskip}\begin{longtable}{rlc@{\qquad}rlc}
\caption{All two-valued states of $B_{14,7}^{(4)}$. E indicates extendability.}\label{S-tab:states-4c}\\\toprule $j$ & Support & E & $j$ & Support & E\\\midrule\endfirsthead
\toprule $j$ & Support & E & $j$ & Support & E\\\midrule\endhead\bottomrule\endfoot
1 & $\{0\}$ & yes & 9 & $\{2,\allowbreak 9,\allowbreak 11,\allowbreak 19\}$ & yes\\
2 & $\{1,\allowbreak 5,\allowbreak 7,\allowbreak 14,\allowbreak 23\}$ & yes & 10 & $\{2,\allowbreak 9,\allowbreak 14,\allowbreak 15,\allowbreak 19\}$ & yes\\
3 & $\{1,\allowbreak 5,\allowbreak 9,\allowbreak 11\}$ & yes & 11 & $\{3,\allowbreak 5,\allowbreak 7,\allowbreak 10\}$ & yes\\
4 & $\{1,\allowbreak 5,\allowbreak 9,\allowbreak 14,\allowbreak 15\}$ & no & 12 & $\{3,\allowbreak 5,\allowbreak 7,\allowbreak 19,\allowbreak 23\}$ & no\\
5 & $\{1,\allowbreak 8,\allowbreak 11,\allowbreak 23\}$ & yes & 13 & $\{3,\allowbreak 5,\allowbreak 9,\allowbreak 15,\allowbreak 19\}$ & yes\\
6 & $\{1,\allowbreak 8,\allowbreak 14,\allowbreak 15,\allowbreak 23\}$ & yes & 14 & $\{3,\allowbreak 8,\allowbreak 10,\allowbreak 15\}$ & yes\\
7 & $\{2,\allowbreak 7,\allowbreak 10,\allowbreak 14\}$ & yes & 15 & $\{3,\allowbreak 8,\allowbreak 15,\allowbreak 19,\allowbreak 23\}$ & yes\\
8 & $\{2,\allowbreak 7,\allowbreak 14,\allowbreak 19,\allowbreak 23\}$ & yes &  &  & \\
\end{longtable}\endgroup
\Needspace{20\baselineskip}\subsection{All states of \texorpdfstring{$L_{18,7}^{(4)}$}{gadget 4d}}
\begingroup\footnotesize\setlength{\tabcolsep}{3pt}\setlength{\LTpre}{5pt}\setlength{\LTpost}{8pt}
\Needspace{9\baselineskip}\begin{longtable}{rlc@{\qquad}rlc}
\caption{All two-valued states of $L_{18,7}^{(4)}$. E indicates extendability.}\label{S-tab:states-4d}\\\toprule $j$ & Support & E & $j$ & Support & E\\\midrule\endfirsthead
\toprule $j$ & Support & E & $j$ & Support & E\\\midrule\endhead\bottomrule\endfoot
1 & $\{1,\allowbreak 6,\allowbreak 9,\allowbreak 13,\allowbreak 16\}$ & yes & 35 & $\{3,\allowbreak 5,\allowbreak 9,\allowbreak 13,\allowbreak 18\}$ & yes\\
2 & $\{1,\allowbreak 6,\allowbreak 9,\allowbreak 14,\allowbreak 16\}$ & yes & 36 & $\{3,\allowbreak 5,\allowbreak 9,\allowbreak 14,\allowbreak 17\}$ & yes\\
3 & $\{1,\allowbreak 6,\allowbreak 10,\allowbreak 13,\allowbreak 16\}$ & yes & 37 & $\{3,\allowbreak 5,\allowbreak 9,\allowbreak 14,\allowbreak 18\}$ & yes\\
4 & $\{1,\allowbreak 6,\allowbreak 10,\allowbreak 14,\allowbreak 16\}$ & yes & 38 & $\{3,\allowbreak 5,\allowbreak 10,\allowbreak 13,\allowbreak 17\}$ & yes\\
5 & $\{1,\allowbreak 7,\allowbreak 9,\allowbreak 13,\allowbreak 17\}$ & yes & 39 & $\{3,\allowbreak 5,\allowbreak 10,\allowbreak 13,\allowbreak 18\}$ & yes\\
6 & $\{1,\allowbreak 7,\allowbreak 9,\allowbreak 13,\allowbreak 18\}$ & yes & 40 & $\{3,\allowbreak 5,\allowbreak 10,\allowbreak 14,\allowbreak 17\}$ & yes\\
7 & $\{1,\allowbreak 7,\allowbreak 9,\allowbreak 14,\allowbreak 17\}$ & yes & 41 & $\{3,\allowbreak 5,\allowbreak 10,\allowbreak 14,\allowbreak 18\}$ & yes\\
8 & $\{1,\allowbreak 7,\allowbreak 9,\allowbreak 14,\allowbreak 18\}$ & yes & 42 & $\{3,\allowbreak 6,\allowbreak 9,\allowbreak 11\}$ & yes\\
\end{longtable}\endgroup
\Needspace{19\baselineskip}\subsection{All states of \texorpdfstring{$D_{17,9}^{(4)}$}{gadget 4e}}
\begingroup\footnotesize\setlength{\tabcolsep}{3pt}\setlength{\LTpre}{5pt}\setlength{\LTpost}{8pt}
\Needspace{9\baselineskip}\begin{longtable}{rlc@{\qquad}rlc}
\caption{All two-valued states of $D_{17,9}^{(4)}$. E indicates extendability.}\label{S-tab:states-4e}\\\toprule $j$ & Support & E & $j$ & Support & E\\\midrule\endfirsthead
\toprule $j$ & Support & E & $j$ & Support & E\\\midrule\endhead\bottomrule\endfoot
1 & $\{0\}$ & yes & 12 & $\{2,\allowbreak 4,\allowbreak 9,\allowbreak 14,\allowbreak 15\}$ & yes\\
2 & $\{1,\allowbreak 4,\allowbreak 8,\allowbreak 14,\allowbreak 15,\allowbreak 23\}$ & yes & 13 & $\{2,\allowbreak 6,\allowbreak 7,\allowbreak 10,\allowbreak 14,\allowbreak 24\}$ & yes\\
3 & $\{1,\allowbreak 5,\allowbreak 7,\allowbreak 10,\allowbreak 14,\allowbreak 24\}$ & yes & 14 & $\{2,\allowbreak 6,\allowbreak 7,\allowbreak 14,\allowbreak 18,\allowbreak 23,\allowbreak 24\}$ & no\\
4 & $\{1,\allowbreak 5,\allowbreak 7,\allowbreak 14,\allowbreak 18,\allowbreak 23,\allowbreak 24\}$ & no & 15 & $\{2,\allowbreak 6,\allowbreak 9,\allowbreak 11,\allowbreak 18\}$ & yes\\
5 & $\{1,\allowbreak 5,\allowbreak 9,\allowbreak 11,\allowbreak 18\}$ & yes & 16 & $\{2,\allowbreak 6,\allowbreak 9,\allowbreak 14,\allowbreak 15,\allowbreak 18,\allowbreak 24\}$ & no\\
6 & $\{1,\allowbreak 5,\allowbreak 9,\allowbreak 14,\allowbreak 15,\allowbreak 18,\allowbreak 24\}$ & no & 17 & $\{3,\allowbreak 4,\allowbreak 8,\allowbreak 15,\allowbreak 23\}$ & yes\\
7 & $\{1,\allowbreak 6,\allowbreak 8,\allowbreak 10,\allowbreak 11\}$ & yes & 18 & $\{3,\allowbreak 5,\allowbreak 7,\allowbreak 10,\allowbreak 24\}$ & yes\\
8 & $\{1,\allowbreak 6,\allowbreak 8,\allowbreak 10,\allowbreak 14,\allowbreak 15,\allowbreak 24\}$ & no & 19 & $\{3,\allowbreak 5,\allowbreak 7,\allowbreak 18,\allowbreak 23,\allowbreak 24\}$ & yes\\
\end{longtable}\endgroup
\Needspace{20\baselineskip}\subsection{All states of \texorpdfstring{$L_{22,9}^{(4)}$}{gadget 4f}}
\begingroup\footnotesize\setlength{\tabcolsep}{3pt}\setlength{\LTpre}{5pt}\setlength{\LTpost}{8pt}
\Needspace{9\baselineskip}\begin{longtable}{rlc@{\qquad}rlc}
\caption{All two-valued states of $L_{22,9}^{(4)}$. E indicates extendability.}\label{S-tab:states-4f}\\\toprule $j$ & Support & E & $j$ & Support & E\\\midrule\endfirsthead
\toprule $j$ & Support & E & $j$ & Support & E\\\midrule\endhead\bottomrule\endfoot
1 & $\{1,\allowbreak 6,\allowbreak 9,\allowbreak 13,\allowbreak 16,\allowbreak 21\}$ & yes & 80 & $\{2,\allowbreak 8,\allowbreak 12,\allowbreak 16,\allowbreak 17,\allowbreak 20,\allowbreak 21\}$ & no\\
2 & $\{1,\allowbreak 6,\allowbreak 9,\allowbreak 13,\allowbreak 16,\allowbreak 22\}$ & yes & 81 & $\{2,\allowbreak 8,\allowbreak 12,\allowbreak 16,\allowbreak 17,\allowbreak 20,\allowbreak 22\}$ & no\\
3 & $\{1,\allowbreak 6,\allowbreak 9,\allowbreak 14,\allowbreak 16,\allowbreak 19\}$ & yes & 82 & $\{2,\allowbreak 8,\allowbreak 12,\allowbreak 16,\allowbreak 18,\allowbreak 19\}$ & yes\\
4 & $\{1,\allowbreak 6,\allowbreak 9,\allowbreak 14,\allowbreak 16,\allowbreak 20,\allowbreak 21\}$ & no & 83 & $\{2,\allowbreak 8,\allowbreak 12,\allowbreak 16,\allowbreak 18,\allowbreak 20,\allowbreak 21\}$ & no\\
5 & $\{1,\allowbreak 6,\allowbreak 9,\allowbreak 14,\allowbreak 16,\allowbreak 20,\allowbreak 22\}$ & no & 84 & $\{2,\allowbreak 8,\allowbreak 12,\allowbreak 16,\allowbreak 18,\allowbreak 20,\allowbreak 22\}$ & no\\
6 & $\{1,\allowbreak 6,\allowbreak 10,\allowbreak 13,\allowbreak 16,\allowbreak 21\}$ & yes & 85 & $\{3,\allowbreak 5,\allowbreak 9,\allowbreak 13,\allowbreak 17,\allowbreak 21\}$ & yes\\
7 & $\{1,\allowbreak 6,\allowbreak 10,\allowbreak 13,\allowbreak 16,\allowbreak 22\}$ & yes & 86 & $\{3,\allowbreak 5,\allowbreak 9,\allowbreak 13,\allowbreak 17,\allowbreak 22\}$ & yes\\
8 & $\{1,\allowbreak 6,\allowbreak 10,\allowbreak 14,\allowbreak 16,\allowbreak 19\}$ & yes & 87 & $\{3,\allowbreak 5,\allowbreak 9,\allowbreak 13,\allowbreak 18,\allowbreak 21\}$ & yes\\
\end{longtable}\endgroup

\subsection{Complete normalized colorings of the gadget panels}
Each color word gives one representative under global color permutation;
the vertices and the anchor convention are specified separately for each
table. The full labeled family is recovered by all palette permutations.
Equal color labels across different colorings carry no extra constraint.
\Needspace{14\baselineskip}\subsection{Colorings of \texorpdfstring{$P_5$}{gadget 3a}}
The columns of each word correspond, in order, to vertices $1,\allowbreak 2,\allowbreak 3,\allowbreak 4,\allowbreak 5,\allowbreak 6,\allowbreak 7,\allowbreak 8,\allowbreak 9,\allowbreak 10$. The anchor context $(1,6,2)$ has colors $(0,\ldots,2)$.
\begingroup\footnotesize\setlength{\tabcolsep}{3pt}\setlength{\LTpre}{5pt}\setlength{\LTpost}{8pt}
\Needspace{9\baselineskip}\begin{longtable}{rl@{\qquad}rl}
\caption{Every normalized coloring of $P_5$.}\label{S-tab:colors-3a}\\\toprule $j$ & Color word & $j$ & Color word\\\midrule\endfirsthead
\toprule $j$ & Color word & $j$ & Color word\\\midrule\endhead\bottomrule\endfoot
1 & \texttt{0201211201} & 4 & \texttt{0210210211}\\[4pt]
2 & \texttt{0202111102} & 5 & \texttt{0212110002}\\[4pt]
3 & \texttt{0210110222} &  & \\[4pt]
\end{longtable}\endgroup
\Needspace{14\baselineskip}\subsection{Colorings of \texorpdfstring{$B_{13,7}$}{gadget 3b}}
The columns of each word correspond, in order, to vertices $1,\allowbreak 2,\allowbreak 3,\allowbreak 5,\allowbreak 7,\allowbreak 8,\allowbreak 9,\allowbreak 10,\allowbreak 11,\allowbreak 14,\allowbreak 15,\allowbreak 19,\allowbreak 23$. The anchor context $(2,8,5)$ has colors $(0,\ldots,2)$.
\begingroup\footnotesize\setlength{\tabcolsep}{3pt}\setlength{\LTpre}{5pt}\setlength{\LTpost}{8pt}
\Needspace{9\baselineskip}\begin{longtable}{rl@{\qquad}rl}
\caption{Every normalized coloring of $B_{13,7}$.}\label{S-tab:colors-3b}\\\toprule $j$ & Color word & $j$ & Color word\\\midrule\endfirsthead
\toprule $j$ & Color word & $j$ & Color word\\\midrule\endhead\bottomrule\endfoot
1 & \texttt{1022012010221} & 4 & \texttt{2012012020111}\\[4pt]
2 & \texttt{1022210201101} & 5 & \texttt{2012012120100}\\[4pt]
3 & \texttt{1022210210001} & 6 & \texttt{2012210102102}\\[4pt]
\end{longtable}\endgroup
\Needspace{15\baselineskip}\subsection{Colorings of \texorpdfstring{$D_{16,9}$}{gadget 3c}}
The columns of each word correspond, in order, to vertices $1,\allowbreak 2,\allowbreak 3,\allowbreak 4,\allowbreak 5,\allowbreak 6,\allowbreak 7,\allowbreak 8,\allowbreak 9,\allowbreak 10,\allowbreak 11,\allowbreak 14,\allowbreak 15,\allowbreak 18,\allowbreak 23,\allowbreak 24$. The anchor context $(2,8,5)$ has colors $(0,\ldots,2)$.
\begingroup\footnotesize\setlength{\tabcolsep}{3pt}\setlength{\LTpre}{5pt}\setlength{\LTpost}{8pt}
\Needspace{9\baselineskip}\begin{longtable}{rl@{\qquad}rl}
\caption{Every normalized coloring of $D_{16,9}$.}\label{S-tab:colors-3c}\\\toprule $j$ & Color word & $j$ & Color word\\\midrule\endfirsthead
\toprule $j$ & Color word & $j$ & Color word\\\midrule\endhead\bottomrule\endfoot
1 & \texttt{1020210121102202} & 5 & \texttt{2010210121201201}\\[4pt]
2 & \texttt{1020212101100222} & 6 & \texttt{2011200120201210}\\[4pt]
3 & \texttt{1020212102100112} & 7 & \texttt{2011202102021012}\\[4pt]
4 & \texttt{1021202102011012} &  & \\[4pt]
\end{longtable}\endgroup
\Needspace{14\baselineskip}\subsection{Colorings of \texorpdfstring{$I_{16,9}$}{gadget 3d}}
The columns of each word correspond, in order, to vertices $1,\allowbreak 2,\allowbreak 3,\allowbreak 4,\allowbreak 5,\allowbreak 7,\allowbreak 8,\allowbreak 9,\allowbreak 10,\allowbreak 11,\allowbreak 14,\allowbreak 15,\allowbreak 18,\allowbreak 19,\allowbreak 23,\allowbreak 24$. The anchor context $(2,8,5)$ has colors $(0,\ldots,2)$.
\begingroup\footnotesize\setlength{\tabcolsep}{3pt}\setlength{\LTpre}{5pt}\setlength{\LTpost}{8pt}
\Needspace{9\baselineskip}\begin{longtable}{rl@{\qquad}rl}
\caption{Every normalized coloring of $I_{16,9}$.}\label{S-tab:colors-3d}\\\toprule $j$ & Color word & $j$ & Color word\\\midrule\endfirsthead
\toprule $j$ & Color word & $j$ & Color word\\\midrule\endhead\bottomrule\endfoot
1 & \texttt{1020221021001012} & 4 & \texttt{2010201212012001}\\[4pt]
2 & \texttt{1021221020110012} & 5 & \texttt{2011201202012110}\\[4pt]
3 & \texttt{1022201201021210} & 6 & \texttt{2012221010210021}\\[4pt]
\end{longtable}\endgroup
\Needspace{15\baselineskip}\subsection{Colorings of \texorpdfstring{$T_{19,11}$}{gadget 3e}}
The columns of each word correspond, in order, to vertices $1,\allowbreak 2,\allowbreak 3,\allowbreak 4,\allowbreak 5,\allowbreak 6,\allowbreak 7,\allowbreak 8,\allowbreak 9,\allowbreak 10,\allowbreak 11,\allowbreak 14,\allowbreak 15,\allowbreak 18,\allowbreak 23,\allowbreak 24,\allowbreak 26,\allowbreak 27,\allowbreak 28$. The anchor context $(2,8,5)$ has colors $(0,\ldots,2)$.
\begingroup\footnotesize\setlength{\tabcolsep}{3pt}\setlength{\LTpre}{5pt}\setlength{\LTpost}{8pt}
\Needspace{9\baselineskip}\begin{longtable}{rl@{\qquad}rl}
\caption{Every normalized coloring of $T_{19,11}$.}\label{S-tab:colors-3e}\\\toprule $j$ & Color word & $j$ & Color word\\\midrule\endfirsthead
\toprule $j$ & Color word & $j$ & Color word\\\midrule\endhead\bottomrule\endfoot
1 & \texttt{1020210121102202120} & 5 & \texttt{2010210121201201120}\\[4pt]
2 & \texttt{1020212101100222120} & 6 & \texttt{2011200120201210120}\\[4pt]
3 & \texttt{1020212102100112210} & 7 & \texttt{2011202102021012102}\\[4pt]
4 & \texttt{1021202102011012201} &  & \\[4pt]
\end{longtable}\endgroup
\Needspace{12\baselineskip}\subsection{Colorings of \texorpdfstring{$H_{21,13}$}{gadget 3f}}
The columns of each word correspond, in order, to vertices $1,\allowbreak 2,\allowbreak 3,\allowbreak 4,\allowbreak 5,\allowbreak 6,\allowbreak 7,\allowbreak 8,\allowbreak 9,\allowbreak 10,\allowbreak 11,\allowbreak 13,\allowbreak 14,\allowbreak 15,\allowbreak 18,\allowbreak 19,\allowbreak 20,\allowbreak 21,\allowbreak 22,\allowbreak 23,\allowbreak 24$. The anchor context $(2,8,5)$ has colors $(0,\ldots,2)$.
\begingroup\footnotesize\setlength{\tabcolsep}{3pt}\setlength{\LTpre}{5pt}\setlength{\LTpost}{8pt}
\Needspace{9\baselineskip}\begin{longtable}{rl@{\qquad}rl}
\caption{Every normalized coloring of $H_{21,13}$.}\label{S-tab:colors-3f}\\\toprule $j$ & Color word & $j$ & Color word\\\midrule\endfirsthead
\toprule $j$ & Color word & $j$ & Color word\\\midrule\endhead\bottomrule\endfoot
1 & \texttt{102021210210001021212} & 2 & \texttt{201120012021012102210}\\[4pt]
\end{longtable}\endgroup
\Needspace{14\baselineskip}\subsection{Colorings of \texorpdfstring{$W_{11,5}^{(4)}$}{gadget 4a}}
The columns of each word correspond, in order, to vertices $0,\allowbreak 1,\allowbreak 2,\allowbreak 3,\allowbreak 4,\allowbreak 5,\allowbreak 6,\allowbreak 7,\allowbreak 8,\allowbreak 9,\allowbreak 10$. The anchor context $(0,1,6,2)$ has colors $(0,\ldots,3)$.
\begingroup\footnotesize\setlength{\tabcolsep}{3pt}\setlength{\LTpre}{5pt}\setlength{\LTpost}{8pt}
\Needspace{9\baselineskip}\begin{longtable}{rl@{\qquad}rl}
\caption{Every normalized coloring of $W_{11,5}^{(4)}$.}\label{S-tab:colors-4a}\\\toprule $j$ & Color word & $j$ & Color word\\\midrule\endfirsthead
\toprule $j$ & Color word & $j$ & Color word\\\midrule\endhead\bottomrule\endfoot
1 & \texttt{01312322312} & 4 & \texttt{01321321322}\\[4pt]
2 & \texttt{01313222213} & 5 & \texttt{01323221113}\\[4pt]
3 & \texttt{01321221333} &  & \\[4pt]
\end{longtable}\endgroup
\Needspace{17\baselineskip}\subsection{Colorings of \texorpdfstring{$W_{16,6}^{(4)}$}{gadget 4b}}
The columns of each word correspond, in order, to vertices $0,\allowbreak 1,\allowbreak 2,\allowbreak 3,\allowbreak 4,\allowbreak 5,\allowbreak 6,\allowbreak 7,\allowbreak 8,\allowbreak 9,\allowbreak 10,\allowbreak 11,\allowbreak 12,\allowbreak 13,\allowbreak 14,\allowbreak 15$. The anchor context $(2,0,6,1)$ has colors $(0,\ldots,3)$.
\begingroup\footnotesize\setlength{\tabcolsep}{3pt}\setlength{\LTpre}{5pt}\setlength{\LTpost}{8pt}
\Needspace{9\baselineskip}\begin{longtable}{rl@{\qquad}rl}
\caption{Every normalized coloring of $W_{16,6}^{(4)}$.}\label{S-tab:colors-4b}\\\toprule $j$ & Color word & $j$ & Color word\\\midrule\endfirsthead
\toprule $j$ & Color word & $j$ & Color word\\\midrule\endhead\bottomrule\endfoot
1 & \texttt{1302302101123322} & 41 & \texttt{1302022331003131}\\[4pt]
2 & \texttt{1302302101132322} & 42 & \texttt{1302022331030131}\\[4pt]
3 & \texttt{1302302101223321} & 43 & \texttt{1302022331103130}\\[4pt]
4 & \texttt{1302302101232321} & 44 & \texttt{1302022331130130}\\[4pt]
5 & \texttt{1302302102123312} & 45 & \texttt{1302022333003111}\\[4pt]
6 & \texttt{1302302102132312} & 46 & \texttt{1302022333030111}\\[4pt]
7 & \texttt{1302302102223311} & 47 & \texttt{1302022333103110}\\[4pt]
8 & \texttt{1302302102232311} & 48 & \texttt{1302022333130110}\\[4pt]
9 & \texttt{1302302301123122} & 49 & \texttt{1302322100003311}\\[4pt]
10 & \texttt{1302302301132122} & 50 & \texttt{1302322100030311}\\[4pt]
11 & \texttt{1302302301223121} & 51 & \texttt{1302322100103310}\\[4pt]
12 & \texttt{1302302301232121} & 52 & \texttt{1302322100130310}\\[4pt]
13 & \texttt{1302302302123112} & 53 & \texttt{1302322101003301}\\[4pt]
14 & \texttt{1302302302132112} & 54 & \texttt{1302322101030301}\\[4pt]
15 & \texttt{1302302302223111} & 55 & \texttt{1302322101103300}\\[4pt]
16 & \texttt{1302302302232111} & 56 & \texttt{1302322101130300}\\[4pt]
17 & \texttt{1303202101123232} & 57 & \texttt{1302322300003111}\\[4pt]
18 & \texttt{1303202101132232} & 58 & \texttt{1302322300030111}\\[4pt]
19 & \texttt{1303202101223231} & 59 & \texttt{1302322300103110}\\[4pt]
20 & \texttt{1303202101232231} & 60 & \texttt{1302322300130110}\\[4pt]
21 & \texttt{1303202103123212} & 61 & \texttt{1302322301003101}\\[4pt]
22 & \texttt{1303202103132212} & 62 & \texttt{1302322301030101}\\[4pt]
23 & \texttt{1303202103223211} & 63 & \texttt{1302322301103100}\\[4pt]
24 & \texttt{1303202103232211} & 64 & \texttt{1302322301130100}\\[4pt]
25 & \texttt{1303202201123132} & 65 & \texttt{1303022121003231}\\[4pt]
26 & \texttt{1303202201132132} & 66 & \texttt{1303022121030231}\\[4pt]
27 & \texttt{1303202201223131} & 67 & \texttt{1303022121103230}\\[4pt]
28 & \texttt{1303202201232131} & 68 & \texttt{1303022121130230}\\[4pt]
29 & \texttt{1303202203123112} & 69 & \texttt{1303022123003211}\\[4pt]
30 & \texttt{1303202203132112} & 70 & \texttt{1303022123030211}\\[4pt]
31 & \texttt{1303202203223111} & 71 & \texttt{1303022123103210}\\[4pt]
32 & \texttt{1303202203232111} & 72 & \texttt{1303022123130210}\\[4pt]
33 & \texttt{1302022131003331} & 73 & \texttt{1303022221003131}\\[4pt]
34 & \texttt{1302022131030331} & 74 & \texttt{1303022221030131}\\[4pt]
35 & \texttt{1302022131103330} & 75 & \texttt{1303022221103130}\\[4pt]
36 & \texttt{1302022131130330} & 76 & \texttt{1303022221130130}\\[4pt]
37 & \texttt{1302022133003311} & 77 & \texttt{1303022223003111}\\[4pt]
38 & \texttt{1302022133030311} & 78 & \texttt{1303022223030111}\\[4pt]
39 & \texttt{1302022133103310} & 79 & \texttt{1303022223103110}\\[4pt]
40 & \texttt{1302022133130310} & 80 & \texttt{1303022223130110}\\[4pt]
\end{longtable}\endgroup
\Needspace{14\baselineskip}\subsection{Colorings of \texorpdfstring{$B_{14,7}^{(4)}$}{gadget 4c}}
The columns of each word correspond, in order, to vertices $0,\allowbreak 1,\allowbreak 2,\allowbreak 3,\allowbreak 5,\allowbreak 7,\allowbreak 8,\allowbreak 9,\allowbreak 10,\allowbreak 11,\allowbreak 14,\allowbreak 15,\allowbreak 19,\allowbreak 23$. The anchor context $(0,1,2,3)$ has colors $(0,\ldots,3)$.
\begingroup\footnotesize\setlength{\tabcolsep}{3pt}\setlength{\LTpre}{5pt}\setlength{\LTpost}{8pt}
\Needspace{9\baselineskip}\begin{longtable}{rl@{\qquad}rl}
\caption{Every normalized coloring of $B_{14,7}^{(4)}$.}\label{S-tab:colors-4c}\\\toprule $j$ & Color word & $j$ & Color word\\\midrule\endfirsthead
\toprule $j$ & Color word & $j$ & Color word\\\midrule\endhead\bottomrule\endfoot
1 & \texttt{01233213212331} & 4 & \texttt{01231132321321}\\[4pt]
2 & \texttt{01233312312221} & 5 & \texttt{01231231212333}\\[4pt]
3 & \texttt{01233312321121} & 6 & \texttt{01231231312322}\\[4pt]
\end{longtable}\endgroup
\Needspace{17\baselineskip}\subsection{Colorings of \texorpdfstring{$L_{18,7}^{(4)}$}{gadget 4d}}
The columns of each word correspond, in order, to vertices $1,\allowbreak 2,\allowbreak 3,\allowbreak 4,\allowbreak 5,\allowbreak 6,\allowbreak 7,\allowbreak 8,\allowbreak 9,\allowbreak 10,\allowbreak 11,\allowbreak 12,\allowbreak 13,\allowbreak 14,\allowbreak 15,\allowbreak 16,\allowbreak 17,\allowbreak 18$. The anchor context $(1,2,4,3)$ has colors $(0,\ldots,3)$.
\begingroup\footnotesize\setlength{\tabcolsep}{3pt}\setlength{\LTpre}{5pt}\setlength{\LTpost}{8pt}
\Needspace{9\baselineskip}\begin{longtable}{rl@{\qquad}rl}
\caption{Every normalized coloring of $L_{18,7}^{(4)}$.}\label{S-tab:colors-4d}\\\toprule $j$ & Color word & $j$ & Color word\\\midrule\endfirsthead
\toprule $j$ & Color word & $j$ & Color word\\\midrule\endhead\bottomrule\endfoot
1 & \texttt{013212030123031203} & 73 & \texttt{013210230132021023}\\[4pt]
2 & \texttt{013212030123031230} & 74 & \texttt{013210230132021032}\\[4pt]
3 & \texttt{013212030123301203} & 75 & \texttt{013210230132201023}\\[4pt]
4 & \texttt{013212030123301230} & 76 & \texttt{013210230132201032}\\[4pt]
5 & \texttt{013212031023031203} & 77 & \texttt{013210231032021023}\\[4pt]
6 & \texttt{013212031023031230} & 78 & \texttt{013210231032021032}\\[4pt]
7 & \texttt{013212031023301203} & 79 & \texttt{013210231032201023}\\[4pt]
8 & \texttt{013212031023301230} & 80 & \texttt{013210231032201032}\\[4pt]
9 & \texttt{013210230123031023} & 81 & \texttt{013213201332020112}\\[4pt]
10 & \texttt{013210230123031032} & 82 & \texttt{013213201332020121}\\[4pt]
11 & \texttt{013210230123301023} & 83 & \texttt{013213201332200112}\\[4pt]
12 & \texttt{013210230123301032} & 84 & \texttt{013213201332200121}\\[4pt]
13 & \texttt{013210231023031023} & 85 & \texttt{013213203132020112}\\[4pt]
14 & \texttt{013210231023031032} & 86 & \texttt{013213203132020121}\\[4pt]
15 & \texttt{013210231023301023} & 87 & \texttt{013213203132200112}\\[4pt]
16 & \texttt{013210231023301032} & 88 & \texttt{013213203132200121}\\[4pt]
17 & \texttt{013213201323030112} & 89 & \texttt{013213201332021002}\\[4pt]
18 & \texttt{013213201323030121} & 90 & \texttt{013213201332021020}\\[4pt]
19 & \texttt{013213201323300112} & 91 & \texttt{013213201332201002}\\[4pt]
20 & \texttt{013213201323300121} & 92 & \texttt{013213201332201020}\\[4pt]
21 & \texttt{013213203123030112} & 93 & \texttt{013213203132021002}\\[4pt]
22 & \texttt{013213203123030121} & 94 & \texttt{013213203132021020}\\[4pt]
23 & \texttt{013213203123300112} & 95 & \texttt{013213203132201002}\\[4pt]
24 & \texttt{013213203123300121} & 96 & \texttt{013213203132201020}\\[4pt]
25 & \texttt{013213201323031002} & 97 & \texttt{013221030131022103}\\[4pt]
26 & \texttt{013213201323031020} & 98 & \texttt{013221030131022130}\\[4pt]
27 & \texttt{013213201323301002} & 99 & \texttt{013221030131202103}\\[4pt]
28 & \texttt{013213201323301020} & 100 & \texttt{013221030131202130}\\[4pt]
29 & \texttt{013213203123031002} & 101 & \texttt{013221031031022103}\\[4pt]
30 & \texttt{013213203123031020} & 102 & \texttt{013221031031022130}\\[4pt]
31 & \texttt{013213203123301002} & 103 & \texttt{013221031031202103}\\[4pt]
32 & \texttt{013213203123301020} & 104 & \texttt{013221031031202130}\\[4pt]
33 & \texttt{013231201321030123} & 105 & \texttt{013223010331021202}\\[4pt]
34 & \texttt{013231201321030132} & 106 & \texttt{013223010331021220}\\[4pt]
35 & \texttt{013231201321300123} & 107 & \texttt{013223010331201202}\\[4pt]
36 & \texttt{013231201321300132} & 108 & \texttt{013223010331201220}\\[4pt]
37 & \texttt{013231203121030123} & 109 & \texttt{013223013031021202}\\[4pt]
38 & \texttt{013231203121030132} & 110 & \texttt{013223013031021220}\\[4pt]
39 & \texttt{013231203121300123} & 111 & \texttt{013223013031201202}\\[4pt]
40 & \texttt{013231203121300132} & 112 & \texttt{013223013031201220}\\[4pt]
41 & \texttt{013232101321030213} & 113 & \texttt{013223010331022101}\\[4pt]
42 & \texttt{013232101321030231} & 114 & \texttt{013223010331022110}\\[4pt]
43 & \texttt{013232101321300213} & 115 & \texttt{013223010331202101}\\[4pt]
44 & \texttt{013232101321300231} & 116 & \texttt{013223010331202110}\\[4pt]
45 & \texttt{013232103121030213} & 117 & \texttt{013223013031022101}\\[4pt]
46 & \texttt{013232103121030231} & 118 & \texttt{013223013031022110}\\[4pt]
47 & \texttt{013232103121300213} & 119 & \texttt{013223013031202101}\\[4pt]
48 & \texttt{013232103121300231} & 120 & \texttt{013223013031202110}\\[4pt]
49 & \texttt{013230210321031023} & 121 & \texttt{013220130131022013}\\[4pt]
50 & \texttt{013230210321031032} & 122 & \texttt{013220130131022031}\\[4pt]
51 & \texttt{013230210321301023} & 123 & \texttt{013220130131202013}\\[4pt]
52 & \texttt{013230210321301032} & 124 & \texttt{013220130131202031}\\[4pt]
53 & \texttt{013230213021031023} & 125 & \texttt{013220131031022013}\\[4pt]
54 & \texttt{013230213021031032} & 126 & \texttt{013220131031022031}\\[4pt]
55 & \texttt{013230213021301023} & 127 & \texttt{013220131031202013}\\[4pt]
56 & \texttt{013230213021301032} & 128 & \texttt{013220131031202031}\\[4pt]
57 & \texttt{013232010321031203} & 129 & \texttt{013223101331020212}\\[4pt]
58 & \texttt{013232010321031230} & 130 & \texttt{013223101331020221}\\[4pt]
59 & \texttt{013232010321301203} & 131 & \texttt{013223101331200212}\\[4pt]
60 & \texttt{013232010321301230} & 132 & \texttt{013223101331200221}\\[4pt]
61 & \texttt{013232013021031203} & 133 & \texttt{013223103131020212}\\[4pt]
62 & \texttt{013232013021031230} & 134 & \texttt{013223103131020221}\\[4pt]
63 & \texttt{013232013021301203} & 135 & \texttt{013223103131200212}\\[4pt]
64 & \texttt{013232013021301230} & 136 & \texttt{013223103131200221}\\[4pt]
65 & \texttt{013212030132021203} & 137 & \texttt{013223101331022001}\\[4pt]
66 & \texttt{013212030132021230} & 138 & \texttt{013223101331022010}\\[4pt]
67 & \texttt{013212030132201203} & 139 & \texttt{013223101331202001}\\[4pt]
68 & \texttt{013212030132201230} & 140 & \texttt{013223101331202010}\\[4pt]
69 & \texttt{013212031032021203} & 141 & \texttt{013223103131022001}\\[4pt]
70 & \texttt{013212031032021230} & 142 & \texttt{013223103131022010}\\[4pt]
71 & \texttt{013212031032201203} & 143 & \texttt{013223103131202001}\\[4pt]
72 & \texttt{013212031032201230} & 144 & \texttt{013223103131202010}\\[4pt]
\end{longtable}\endgroup
\Needspace{15\baselineskip}\subsection{Colorings of \texorpdfstring{$D_{17,9}^{(4)}$}{gadget 4e}}
The columns of each word correspond, in order, to vertices $0,\allowbreak 1,\allowbreak 2,\allowbreak 3,\allowbreak 4,\allowbreak 5,\allowbreak 6,\allowbreak 7,\allowbreak 8,\allowbreak 9,\allowbreak 10,\allowbreak 11,\allowbreak 14,\allowbreak 15,\allowbreak 18,\allowbreak 23,\allowbreak 24$. The anchor context $(0,1,2,3)$ has colors $(0,\ldots,3)$.
\begingroup\footnotesize\setlength{\tabcolsep}{3pt}\setlength{\LTpre}{5pt}\setlength{\LTpost}{8pt}
\Needspace{9\baselineskip}\begin{longtable}{rl@{\qquad}rl}
\caption{Every normalized coloring of $D_{17,9}^{(4)}$.}\label{S-tab:colors-4e}\\\toprule $j$ & Color word & $j$ & Color word\\\midrule\endfirsthead
\toprule $j$ & Color word & $j$ & Color word\\\midrule\endhead\bottomrule\endfoot
1 & \texttt{01232132313123123} & 5 & \texttt{01233122312123132}\\[4pt]
2 & \texttt{01232312131123323} & 6 & \texttt{01231323123211213}\\[4pt]
3 & \texttt{01232313121122333} & 7 & \texttt{01233121321213231}\\[4pt]
4 & \texttt{01232313123122113} &  & \\[4pt]
\end{longtable}\endgroup
\Needspace{17\baselineskip}\subsection{Colorings of \texorpdfstring{$L_{22,9}^{(4)}$}{gadget 4f}}
The columns of each word correspond, in order, to vertices $1,\allowbreak 2,\allowbreak 3,\allowbreak 4,\allowbreak 5,\allowbreak 6,\allowbreak 7,\allowbreak 8,\allowbreak 9,\allowbreak 10,\allowbreak 11,\allowbreak 12,\allowbreak 13,\allowbreak 14,\allowbreak 15,\allowbreak 16,\allowbreak 17,\allowbreak 18,\allowbreak 19,\allowbreak 20,\allowbreak 21,\allowbreak 22$. The anchor context $(1,2,4,3)$ has colors $(0,\ldots,3)$.
\begingroup\footnotesize\setlength{\tabcolsep}{3pt}\setlength{\LTpre}{5pt}\setlength{\LTpost}{8pt}
\Needspace{9\baselineskip}\begin{longtable}{rl@{\qquad}rl}
\caption{Every normalized coloring of $L_{22,9}^{(4)}$.}\label{S-tab:colors-4f}\\\toprule $j$ & Color word & $j$ & Color word\\\midrule\endfirsthead
\toprule $j$ & Color word & $j$ & Color word\\\midrule\endhead\bottomrule\endfoot
1 & \texttt{0132120301230312033201} & 109 & \texttt{0132120310320212032301}\\[4pt]
2 & \texttt{0132120301230312033210} & 110 & \texttt{0132120310320212032310}\\[4pt]
3 & \texttt{0132120301230312303201} & 111 & \texttt{0132120310320212302301}\\[4pt]
4 & \texttt{0132120301230312303210} & 112 & \texttt{0132120310320212302310}\\[4pt]
5 & \texttt{0132120310230312033201} & 113 & \texttt{0132120301322012030312}\\[4pt]
6 & \texttt{0132120310230312033210} & 114 & \texttt{0132120301322012030321}\\[4pt]
7 & \texttt{0132120310230312303201} & 115 & \texttt{0132120301322012300312}\\[4pt]
8 & \texttt{0132120310230312303210} & 116 & \texttt{0132120301322012300321}\\[4pt]
9 & \texttt{0132120301233012030213} & 117 & \texttt{0132120310322012030312}\\[4pt]
10 & \texttt{0132120301233012030231} & 118 & \texttt{0132120310322012030321}\\[4pt]
11 & \texttt{0132120301233012300213} & 119 & \texttt{0132120310322012300312}\\[4pt]
12 & \texttt{0132120301233012300231} & 120 & \texttt{0132120310322012300321}\\[4pt]
13 & \texttt{0132120310233012030213} & 121 & \texttt{0132132013322001121302}\\[4pt]
14 & \texttt{0132120310233012030231} & 122 & \texttt{0132132013322001121320}\\[4pt]
15 & \texttt{0132120310233012300213} & 123 & \texttt{0132132013322001211302}\\[4pt]
16 & \texttt{0132120310233012300231} & 124 & \texttt{0132132013322001211320}\\[4pt]
17 & \texttt{0132132013233001121203} & 125 & \texttt{0132132031322001121302}\\[4pt]
18 & \texttt{0132132013233001121230} & 126 & \texttt{0132132031322001121320}\\[4pt]
19 & \texttt{0132132013233001211203} & 127 & \texttt{0132132031322001211302}\\[4pt]
20 & \texttt{0132132013233001211230} & 128 & \texttt{0132132031322001211320}\\[4pt]
21 & \texttt{0132132031233001121203} & 129 & \texttt{0132102301320210232301}\\[4pt]
22 & \texttt{0132132031233001121230} & 130 & \texttt{0132102301320210232310}\\[4pt]
23 & \texttt{0132132031233001211203} & 131 & \texttt{0132102301320210322301}\\[4pt]
24 & \texttt{0132132031233001211230} & 132 & \texttt{0132102301320210322310}\\[4pt]
25 & \texttt{0132102301230310233201} & 133 & \texttt{0132102310320210232301}\\[4pt]
26 & \texttt{0132102301230310233210} & 134 & \texttt{0132102310320210232310}\\[4pt]
27 & \texttt{0132102301230310323201} & 135 & \texttt{0132102310320210322301}\\[4pt]
28 & \texttt{0132102301230310323210} & 136 & \texttt{0132102310320210322310}\\[4pt]
29 & \texttt{0132102310230310233201} & 137 & \texttt{0132102301322010230312}\\[4pt]
30 & \texttt{0132102310230310233210} & 138 & \texttt{0132102301322010230321}\\[4pt]
31 & \texttt{0132102310230310323201} & 139 & \texttt{0132102301322010320312}\\[4pt]
32 & \texttt{0132102310230310323210} & 140 & \texttt{0132102301322010320321}\\[4pt]
33 & \texttt{0132102301233010230213} & 141 & \texttt{0132102310322010230312}\\[4pt]
34 & \texttt{0132102301233010230231} & 142 & \texttt{0132102310322010230321}\\[4pt]
35 & \texttt{0132102301233010320213} & 143 & \texttt{0132102310322010320312}\\[4pt]
36 & \texttt{0132102301233010320231} & 144 & \texttt{0132102310322010320321}\\[4pt]
37 & \texttt{0132102310233010230213} & 145 & \texttt{0132132013320210022301}\\[4pt]
38 & \texttt{0132102310233010230231} & 146 & \texttt{0132132013320210022310}\\[4pt]
39 & \texttt{0132102310233010320213} & 147 & \texttt{0132132013320210202301}\\[4pt]
40 & \texttt{0132102310233010320231} & 148 & \texttt{0132132013320210202310}\\[4pt]
41 & \texttt{0132132013230310023201} & 149 & \texttt{0132132031320210022301}\\[4pt]
42 & \texttt{0132132013230310023210} & 150 & \texttt{0132132031320210022310}\\[4pt]
43 & \texttt{0132132013230310203201} & 151 & \texttt{0132132031320210202301}\\[4pt]
44 & \texttt{0132132013230310203210} & 152 & \texttt{0132132031320210202310}\\[4pt]
45 & \texttt{0132132031230310023201} & 153 & \texttt{0132132013322010020312}\\[4pt]
46 & \texttt{0132132031230310023210} & 154 & \texttt{0132132013322010020321}\\[4pt]
47 & \texttt{0132132031230310203201} & 155 & \texttt{0132132013322010200312}\\[4pt]
48 & \texttt{0132132031230310203210} & 156 & \texttt{0132132013322010200321}\\[4pt]
49 & \texttt{0132132013233010020213} & 157 & \texttt{0132132031322010020312}\\[4pt]
50 & \texttt{0132132013233010020231} & 158 & \texttt{0132132031322010020321}\\[4pt]
51 & \texttt{0132132013233010200213} & 159 & \texttt{0132132031322010200312}\\[4pt]
52 & \texttt{0132132013233010200231} & 160 & \texttt{0132132031322010200321}\\[4pt]
53 & \texttt{0132132031233010020213} & 161 & \texttt{0132230103310212022301}\\[4pt]
54 & \texttt{0132132031233010020231} & 162 & \texttt{0132230103310212022310}\\[4pt]
55 & \texttt{0132132031233010200213} & 163 & \texttt{0132230103310212202301}\\[4pt]
56 & \texttt{0132132031233010200231} & 164 & \texttt{0132230103310212202310}\\[4pt]
57 & \texttt{0132312013213001231203} & 165 & \texttt{0132230130310212022301}\\[4pt]
58 & \texttt{0132312013213001231230} & 166 & \texttt{0132230130310212022310}\\[4pt]
59 & \texttt{0132312013213001321203} & 167 & \texttt{0132230130310212202301}\\[4pt]
60 & \texttt{0132312013213001321230} & 168 & \texttt{0132230130310212202310}\\[4pt]
61 & \texttt{0132312031213001231203} & 169 & \texttt{0132230103312012020312}\\[4pt]
62 & \texttt{0132312031213001231230} & 170 & \texttt{0132230103312012020321}\\[4pt]
63 & \texttt{0132312031213001321203} & 171 & \texttt{0132230103312012200312}\\[4pt]
64 & \texttt{0132312031213001321230} & 172 & \texttt{0132230103312012200321}\\[4pt]
65 & \texttt{0132321013213002131203} & 173 & \texttt{0132230130312012020312}\\[4pt]
66 & \texttt{0132321013213002131230} & 174 & \texttt{0132230130312012020321}\\[4pt]
67 & \texttt{0132321013213002311203} & 175 & \texttt{0132230130312012200312}\\[4pt]
68 & \texttt{0132321013213002311230} & 176 & \texttt{0132230130312012200321}\\[4pt]
69 & \texttt{0132321031213002131203} & 177 & \texttt{0132210301310221031302}\\[4pt]
70 & \texttt{0132321031213002131230} & 178 & \texttt{0132210301310221031320}\\[4pt]
71 & \texttt{0132321031213002311203} & 179 & \texttt{0132210301310221301302}\\[4pt]
72 & \texttt{0132321031213002311230} & 180 & \texttt{0132210301310221301320}\\[4pt]
73 & \texttt{0132302103210310233201} & 181 & \texttt{0132210310310221031302}\\[4pt]
74 & \texttt{0132302103210310233210} & 182 & \texttt{0132210310310221031320}\\[4pt]
75 & \texttt{0132302103210310323201} & 183 & \texttt{0132210310310221301302}\\[4pt]
76 & \texttt{0132302103210310323210} & 184 & \texttt{0132210310310221301320}\\[4pt]
77 & \texttt{0132302130210310233201} & 185 & \texttt{0132230103310221011302}\\[4pt]
78 & \texttt{0132302130210310233210} & 186 & \texttt{0132230103310221011320}\\[4pt]
79 & \texttt{0132302130210310323201} & 187 & \texttt{0132230103310221101302}\\[4pt]
80 & \texttt{0132302130210310323210} & 188 & \texttt{0132230103310221101320}\\[4pt]
81 & \texttt{0132302103213010230213} & 189 & \texttt{0132230130310221011302}\\[4pt]
82 & \texttt{0132302103213010230231} & 190 & \texttt{0132230130310221011320}\\[4pt]
83 & \texttt{0132302103213010320213} & 191 & \texttt{0132230130310221101302}\\[4pt]
84 & \texttt{0132302103213010320231} & 192 & \texttt{0132230130310221101320}\\[4pt]
85 & \texttt{0132302130213010230213} & 193 & \texttt{0132231013312002121302}\\[4pt]
86 & \texttt{0132302130213010230231} & 194 & \texttt{0132231013312002121320}\\[4pt]
87 & \texttt{0132302130213010320213} & 195 & \texttt{0132231013312002211302}\\[4pt]
88 & \texttt{0132302130213010320231} & 196 & \texttt{0132231013312002211320}\\[4pt]
89 & \texttt{0132320103210312033201} & 197 & \texttt{0132231031312002121302}\\[4pt]
90 & \texttt{0132320103210312033210} & 198 & \texttt{0132231031312002121320}\\[4pt]
91 & \texttt{0132320103210312303201} & 199 & \texttt{0132231031312002211302}\\[4pt]
92 & \texttt{0132320103210312303210} & 200 & \texttt{0132231031312002211320}\\[4pt]
93 & \texttt{0132320130210312033201} & 201 & \texttt{0132201301310220131302}\\[4pt]
94 & \texttt{0132320130210312033210} & 202 & \texttt{0132201301310220131320}\\[4pt]
95 & \texttt{0132320130210312303201} & 203 & \texttt{0132201301310220311302}\\[4pt]
96 & \texttt{0132320130210312303210} & 204 & \texttt{0132201301310220311320}\\[4pt]
97 & \texttt{0132320103213012030213} & 205 & \texttt{0132201310310220131302}\\[4pt]
98 & \texttt{0132320103213012030231} & 206 & \texttt{0132201310310220131320}\\[4pt]
99 & \texttt{0132320103213012300213} & 207 & \texttt{0132201310310220311302}\\[4pt]
100 & \texttt{0132320103213012300231} & 208 & \texttt{0132201310310220311320}\\[4pt]
101 & \texttt{0132320130213012030213} & 209 & \texttt{0132231013310220011302}\\[4pt]
102 & \texttt{0132320130213012030231} & 210 & \texttt{0132231013310220011320}\\[4pt]
103 & \texttt{0132320130213012300213} & 211 & \texttt{0132231013310220101302}\\[4pt]
104 & \texttt{0132320130213012300231} & 212 & \texttt{0132231013310220101320}\\[4pt]
105 & \texttt{0132120301320212032301} & 213 & \texttt{0132231031310220011302}\\[4pt]
106 & \texttt{0132120301320212032310} & 214 & \texttt{0132231031310220011320}\\[4pt]
107 & \texttt{0132120301320212302301} & 215 & \texttt{0132231031310220101302}\\[4pt]
108 & \texttt{0132120301320212302310} & 216 & \texttt{0132231031310220101320}\\[4pt]
\end{longtable}\endgroup

\FloatBarrier
\begin{table}[t]\centering\small\setlength{\tabcolsep}{3pt}
\caption{Both directed truth-pattern witnesses for all $15$ two-way purely chromatic pairs of $H_{21,13}$ and both displayed pairs of $T_{19,11}$. For $H_{21,13}$, $s_j$ is the state in Table~\ref{S-tab:h21-states}; for $T_{19,11}$, entries give supports. Every witness is nonextendable and satisfies every context equation.}\label{S-tab:pure-witnesses}
\begin{tabular}{p{.14\linewidth}p{.09\linewidth}p{.33\linewidth}p{.33\linewidth}}\toprule
Configuration & Pair $(u,v)$ & $(s(u),s(v))=(1,0)$ & $(s(u),s(v))=(0,1)$\\\midrule
$T_{19,11}$ & $(14,28)$ & $\{1,\allowbreak 5,\allowbreak 7,\allowbreak 14,\allowbreak 18,\allowbreak 23,\allowbreak 24\}$ & $\{1,\allowbreak 6,\allowbreak 8,\allowbreak 10,\allowbreak 11,\allowbreak 27,\allowbreak 28\}$\\
$T_{19,11}$ & $(18,27)$ & $\{1,\allowbreak 5,\allowbreak 7,\allowbreak 14,\allowbreak 18,\allowbreak 23,\allowbreak 24\}$ & $\{1,\allowbreak 6,\allowbreak 8,\allowbreak 10,\allowbreak 11,\allowbreak 27,\allowbreak 28\}$\\
$H_{21,13}$ & $(1,11)$ & $s_{1}$ & $s_{12}$\\
$H_{21,13}$ & $(1,18)$ & $s_{1}$ & $s_{11}$\\
$H_{21,13}$ & $(1,21)$ & $s_{2}$ & $s_{9}$\\
$H_{21,13}$ & $(2,14)$ & $s_{12}$ & $s_{1}$\\
$H_{21,13}$ & $(4,13)$ & $s_{1}$ & $s_{18}$\\
$H_{21,13}$ & $(4,15)$ & $s_{7}$ & $s_{4}$\\
$H_{21,13}$ & $(4,19)$ & $s_{1}$ & $s_{11}$\\
$H_{21,13}$ & $(5,22)$ & $s_{18}$ & $s_{1}$\\
$H_{21,13}$ & $(7,10)$ & $s_{2}$ & $s_{20}$\\
$H_{21,13}$ & $(7,20)$ & $s_{2}$ & $s_{9}$\\
$H_{21,13}$ & $(7,24)$ & $s_{7}$ & $s_{4}$\\
$H_{21,13}$ & $(8,23)$ & $s_{20}$ & $s_{2}$\\
$H_{21,13}$ & $(15,19)$ & $s_{1}$ & $s_{7}$\\
$H_{21,13}$ & $(18,21)$ & $s_{2}$ & $s_{1}$\\
$H_{21,13}$ & $(20,24)$ & $s_{7}$ & $s_{2}$\\
\bottomrule\end{tabular}\end{table}

\paragraph*{Exhaustive classification of chromatic identities.}
For a chromatically inseparable pair $(u,v)$, form the two directed witness sets
\[
\begin{aligned}
W_{10}(u,v)&=\{j:s_j(u)=1,\ s_j(v)=0\},\\
W_{01}(u,v)&=\{j:s_j(u)=0,\ s_j(v)=1\}.
\end{aligned}
\]
Both sets are nonempty exactly for a two-way purely chromatic identity.
If only $W_{10}$ is empty, $u\Rightarrow v$ holds and its converse fails;
if only $W_{01}$ is empty, the implication is reversed. If both are empty,
the pair is inseparable under all ordinary states. Any member of either
witness set is automatically nonextendable, since every color-class state
has equal values on a chromatically inseparable pair.

The following Python standard-library procedure reconstructs $H_{21,13}$
directly from its $13$ contexts in Table~\ref{S-tab:contexts}, enumerates all
binary assignments, and reconstructs normalized strong colorings as partitions
into three disjoint state supports. Lexicographic support order agrees with
Table~\ref{S-tab:h21-states}. It prints both complete directed witness sets for
every chromatic identity, rather than only the representatives in
Table~\ref{S-tab:pure-witnesses}.
\Needspace{31\baselineskip}
\begingroup\small
\begin{verbatim}
from itertools import combinations
E = [(1,2,3),(7,8,9),(4,5,6),(2,8,5),
     (3,11,14),(7,11,15),(4,11,24),
     (9,10,23),(1,10,19),(4,10,18),
     (1,13,20),(7,13,21),(6,13,22)]
V = sorted(set().union(*map(set, E)))
M = [sum(1 << V.index(v) for v in e) for e in E]
S = sorted(tuple(v for i,v in enumerate(V) if b >> i & 1)
           for b in range(1 << len(V))
           if all((b & m).bit_count() == 1 for m in M))
S = list(map(frozenset, S))
C = []
for blocks in combinations(S, 3):
    if any(a & b for a,b in combinations(blocks, 2)):
        continue
    blocks = [next(s for s in blocks if a in s) for a in (2,8,5)]
    C.append({v:k for k,s in enumerate(blocks) for v in s})
assert (len(S), len(C)) == (21, 2)
counts = [0, 0, 0]
for u,v in combinations(V, 2):
    if not all(c[u] == c[v] for c in C):
        continue
    W10 = [j for j,s in enumerate(S, 1) if u in s and v not in s]
    W01 = [j for j,s in enumerate(S, 1) if v in s and u not in s]
    k = 0 if W10 and W01 else 1 if W10 or W01 else 2
    counts[k] += 1
    print((u,v), W10, W01)
assert counts == [15, 6, 0]
\end{verbatim}
\endgroup

\begingroup\small\setlength{\LTpre}{6pt}\setlength{\LTpost}{8pt}
\Needspace{16\baselineskip}
\begin{longtable}{lrrrr}
\caption{All ten connected, edge-faithful, colorable restrictions with a forced color identity in the specified inventory.}\label{S-tab:inventory}\\
\toprule retained diagonals & rays & contexts & colorings & inseparable pairs\\\midrule
\endfirsthead
\toprule retained diagonals & rays & contexts & colorings & inseparable pairs\\\midrule
\endhead
\bottomrule\endfoot
$\{10,11\}$ & $17$ & $10$ & $4$ & $(10,24),\ (11,18)$ \\*
$\{10,12\}$ & $17$ & $10$ & $4$ & $(10,16),\ (12,23)$ \\*
$\{10,13\}$ & $17$ & $10$ & $4$ & $(10,20),\ (13,19)$ \\*
$\{11,12\}$ & $17$ & $10$ & $4$ & $(11,25),\ (12,14)$ \\*
$\{11,13\}$ & $17$ & $10$ & $4$ & $(11,21),\ (13,15)$ \\*
$\{12,13\}$ & $17$ & $10$ & $4$ & $(12,22),\ (13,17)$ \\*
$\{10,11,12\}$ & $21$ & $13$ & $2$ & $21\text{ pairs}$ \\*
$\{10,11,13\}$ & $21$ & $13$ & $2$ & $21\text{ pairs}$ \\*
$\{10,12,13\}$ & $21$ & $13$ & $2$ & $21\text{ pairs}$ \\*
$\{11,12,13\}$ & $21$ & $13$ & $2$ & $21\text{ pairs}$ \\
\end{longtable}\endgroup

\FloatBarrier
\subsection{An explicit state certificate for the completed Yu--Oh hypergraph}
\label{S-m:app:states}

\begin{table}[b]
\caption{The 24 states are indexed by $(k,b)$, where
$1\le k\le12$ and $b\in\{0,1\}$. Their support on vertices
$1,\ldots,13$ is $T_k$ when $b=0$ and $T_k\cup\{\eta_k\}$ when $b=1$.
Equation~\eqref{S-m:eq:state-extension} determines all remaining values.}
\label{S-m:tab:yo-states}
\begin{ruledtabular}
\begin{tabular}{ccc}
$k$ & $T_k$ & $\eta_k$\\
1 & $\{1,4,8\}$ & $12$ \\
2 & $\{1,5,7\}$ & $12$ \\
3 & $\{1,5,9\}$ & $11$ \\
4 & $\{1,6,8\}$ & $11$ \\
5 & $\{2,4,7\}$ & $12$ \\
6 & $\{2,4,9\}$ & $13$ \\
7 & $\{2,6,7\}$ & $10$ \\
8 & $\{2,6,9\}$ & $11$ \\
9 & $\{3,4,8\}$ & $13$ \\
10 & $\{3,5,7\}$ & $10$ \\
11 & $\{3,5,9\}$ & $13$ \\
12 & $\{3,6,8\}$ & $10$ \\
\end{tabular}
\end{ruledtabular}
\end{table}

The first four contexts of Eq.~\eqref{m:eq:yuo-blocks} force exactly one
of vertices 2, 5, and 8 to have value one. In the triad containing that
coordinate vertex, both other vertices have value zero. In each of the
other two triads, one of the two remaining vertices has value one. There
are therefore $3\cdot2^2=12$ possible supports $T_k$ on vertices
$1,\ldots,9$, listed in Table~\ref{S-m:tab:yo-states}.

For each such support, exactly one of vertices $10,11,12,13$ has no
orthogonal neighbor in $T_k$; call it $\eta_k$. The other three are forced
to zero, while $s(\eta_k)$ may be either zero or one. This gives precisely
24 assignments on the original 13 vertices. Every completion vertex $w$
belongs to one context $\{i,j,w\}$ and has the uniquely determined value
\begin{equation}\label{S-m:eq:state-extension}
s(w)=1-s(i)-s(j).
\end{equation}
The listed assignments have $s(i)+s(j)\le1$ on every original orthogonal
pair, so these extensions are valid two-valued states. Conversely, every
state of the completed hypergraph restricts to one of these assignments.

For a compact separation certificate, encode the profile of vertex $i$
as the nonnegative integer
\[
m_i=\sum_{k=1}^{12}\sum_{b=0}^1s_{k,b}(i)\,2^{2(k-1)+b}.
\]
In hexadecimal notation these 25 integers are
\begin{equation}\label{S-m:eq:yo-profiles}
\begin{gathered}
\begin{array}{c|rrrrr}
i&1&2&3&4&5\\\hline
m_i&\mathtt{0000FF}&\mathtt{00FF00}&\mathtt{FF0000}&\mathtt{030F03}&\mathtt{3C003C}
\end{array}\\[1ex]
\begin{array}{c|rrrrr}
i&6&7&8&9&10\\\hline
m_i&\mathtt{C0F0C0}&\mathtt{0C330C}&\mathtt{C300C3}&\mathtt{30CC30}&\mathtt{882000}
\end{array}\\[1ex]
\begin{array}{c|rrrrr}
i&11&12&13&14&15\\\hline
m_i&\mathtt{0080A0}&\mathtt{00020A}&\mathtt{220800}&\mathtt{007F5F}&\mathtt{F34C53}
\end{array}\\[1ex]
\begin{array}{c|rrrrr}
i&16&17&18&19&20\\\hline
m_i&\mathtt{CF31C5}&\mathtt{3F0D35}&\mathtt{74D0FC}&\mathtt{77DF00}&\mathtt{DDF700}
\end{array}\\[1ex]
\begin{array}{c|rrrrr}
i&21&22&23&24&25\\\hline
m_i&\mathtt{D1C4F3}&\mathtt{1D073F}&\mathtt{4713CF}&\mathtt{FC705C}&\mathtt{00FDF5}
\end{array}
\end{gathered}
\end{equation}
All 25 integers are nonzero and pairwise distinct, certifying unitality and
separation. Notably, the Hamming weights (number of states in which $s(i)=1$)
partition the vertices into three classes: $8$ for $i\in\{1,\ldots,9\}$, $3$
for $i\in\{10,\ldots,13\}$, and $13$ for the completion vertices
$i\in\{14,\ldots,25\}$ (with $8+3+13=24$ in each completion context), in
contrast to the uniform state weight $2$ across all 15 vertices of $G_{32}$ and
consistent with the fractional chromatic--clique gap
$\chi_f(\Gamma(H_{\YO}^{+}))=35/11>3=\omega(\Gamma(H_{\YO}^{+}))$.

\FloatBarrier
\subsection{\texorpdfstring{Hesse $H_{37}$: all states and separating profiles}{Hesse H37: all states and separating profiles}}\label{S-h:app:states}

The canonical states $s_1,\ldots,s_{153}$ are listed in Table~\ref{S-h:tab:app-states}. A two-valued coloring is a state with exactly one value one per context. Vertex numbers are those of Appendix~\ref{S-h:app:vectors-contexts} and Fig.~\ref{h:fig:hesse37-hypergraph}: vertices $1,\ldots,16$ form the core, and $17,\ldots,37$ are completions. We list every state, a minimum separating family, and every vertex profile.

\subsubsection{State parametrization and all 27 branches}\label{S-h:app:branches}

Choose $t=(t_0,t_1,t_2)\in\Z_3^3$ and $\varepsilon\in\{0,1\}$. Put
\begin{align}
s(M_{ak})&=\mathbf1_{k=t_a},&s(z)&=\varepsilon,\\
U(t)&=\Z_3\setminus\{t_0,t_1,t_2\},&
W(t)&=\Z_3\setminus\{t_a-a:a\in\Z_3\}.
\end{align}
For $\varepsilon=0$, choose the values of $P_k$ freely for $k\in U(t)$, with other $P$ values zero. For $\varepsilon=1$, all $P$ values are zero. Independently, choose the values of $Q_k$ freely for $k\in W(t)$, with other $Q$ values zero. All completion values are forced:
\begin{align}
s(A_k)&=1-\varepsilon-s(P_k),\\
s(B_{ak})&=1-s(M_{ak})-s(P_k),\\
s(C_{ak})&=1-s(M_{a,a+k})-s(Q_k).
\end{align}
Enumerate $t$ lexicographically with $0<1<2$, then $\varepsilon=1,0$. Within each branch, list the available $P$ bits in increasing $k$, followed by the available $Q$ bits in increasing $k$, and order the resulting bit strings with $1$ before $0$. This rule fixes all state indices used in the text, tables, and verification data. Table~\ref{S-h:tab:app-branches} gives the index ranges of all 27 branches.

\begingroup
\footnotesize
\setlength{\tabcolsep}{3pt}
\renewcommand{\arraystretch}{1.08}
\Needspace{24\baselineskip}
\begin{longtable}{cllrrrr}
\caption{All 27 state branches. The sets $U(t),W(t)$ identify the free core bits. The two choices of $\varepsilon$ give $N_0=2^{|U|+|W|}$ and $N_1=2^{|W|}$ states; their index ranges refer to Table~\ref{S-h:tab:app-states}.}\label{S-h:tab:app-branches}\\
\toprule
$t$ & $U(t)$ & $W(t)$ & $N_0$ & $r$ for $\varepsilon=0$ & $N_1$ & $r$ for $\varepsilon=1$ \\
\midrule
\endfirsthead
\caption[]{All 27 state branches. The sets $U(t),W(t)$ identify the free core bits. The two choices of $\varepsilon$ give $N_0=2^{|U|+|W|}$ and $N_1=2^{|W|}$ states; their index ranges refer to Table~\ref{S-h:tab:app-states}. (continued).}\\
\toprule
$t$ & $U(t)$ & $W(t)$ & $N_0$ & $r$ for $\varepsilon=0$ & $N_1$ & $r$ for $\varepsilon=1$ \\
\midrule
\endhead
\midrule
\multicolumn{7}{r}{\emph{Continued on the next page.}}\\
\endfoot
\bottomrule
\endlastfoot
$(0,0,0)$ & $\{1,2\}$ & $\varnothing$ & 4 & 2--5 & 1 & 1 \\
$(0,0,1)$ & $\{2\}$ & $\{1\}$ & 4 & 8--11 & 2 & 6--7 \\
$(0,0,2)$ & $\{1\}$ & $\{1\}$ & 4 & 14--17 & 2 & 12--13 \\
$(0,1,0)$ & $\{2\}$ & $\{2\}$ & 4 & 20--23 & 2 & 18--19 \\
$(0,1,1)$ & $\{2\}$ & $\{1\}$ & 4 & 26--29 & 2 & 24--25 \\
$(0,1,2)$ & $\varnothing$ & $\{1,2\}$ & 4 & 34--37 & 4 & 30--33 \\
$(0,2,0)$ & $\{1\}$ & $\{2\}$ & 4 & 40--43 & 2 & 38--39 \\
$(0,2,1)$ & $\varnothing$ & $\varnothing$ & 1 & 45 & 1 & 44 \\
\end{longtable}
\endgroup

\subsubsection{Canonical two-valued colorings: all 153 states}

Table~\ref{S-h:tab:app-states} explicitly lists all 153 states. The last column is the complete binary vector in vertex order $1,\ldots,37$, separated after the 16 core entries. Every row has context sum one in all 24 contexts. The support-size distribution is
\begin{equation}
\begin{array}{c|rrrrrr}
m&12&13&14&15&16&18\\
\#\{s:|\operatorname{supp}s|=m\}&6&18&27&24&51&27
\end{array}
\end{equation}

\begingroup
\footnotesize
\setlength{\tabcolsep}{2pt}
\renewcommand{\arraystretch}{1.02}
\Needspace{24\baselineskip}
\begin{longtable}{rllrc}
\caption{Canonical enumeration of all 153 two-valued states, in the canonical order defined above. The support is split into core and completion vertices. The last column is $(s_r(1),\ldots,s_r(16)\,;\,s_r(17),\ldots,s_r(37))$.}\label{S-h:tab:app-states}\\
\toprule
$r$ & Core support & Completion support & $|\operatorname{supp}s_r|$ & 37-bit state vector \\
\midrule
\endfirsthead
\caption[]{Canonical enumeration of all 153 two-valued states, in the canonical order defined above. The support is split into core and completion vertices. The last column is $(s_r(1),\ldots,s_r(16)\,;\,s_r(17),\ldots,s_r(37))$. (continued).}\\
\toprule
$r$ & Core support & Completion support & $|\operatorname{supp}s_r|$ & 37-bit state vector \\
\midrule
\endhead
\midrule
\multicolumn{5}{r}{\emph{Continued on the next page.}}\\
\endfoot
\bottomrule
\endlastfoot
1 & $\{1,2,5,8\}$ & $\{21,22,24,25,27,28,30,31,32,33,35,37\}$ & 16 & {\ttfamily\fontsize{6.5}{8}\selectfont 1100100100000000\,000011011011011110101} \\
2 & $\{2,5,8,12,13\}$ & $\{17,30,31,32,33,35,37\}$ & 12 & {\ttfamily\fontsize{6.5}{8}\selectfont 0100100100011000\,100000000000011110101} \\
3 & $\{2,5,8,12\}$ & $\{17,19,22,25,28,30,31,32,33,35,37\}$ & 15 & {\ttfamily\fontsize{6.5}{8}\selectfont 0100100100010000\,101001001001011110101} \\
4 & $\{2,5,8,13\}$ & $\{17,18,21,24,27,30,31,32,33,35,37\}$ & 15 & {\ttfamily\fontsize{6.5}{8}\selectfont 0100100100001000\,110010010010011110101} \\
5 & $\{2,5,8\}$ & $\{17,18,19,21,22,24,25,27,28,30,31,32,33,35,37\}$ & 18 & {\ttfamily\fontsize{6.5}{8}\selectfont 0100100100000000\,111011011011011110101} \\
6 & $\{1,2,5,9,15\}$ & $\{21,22,24,25,26,28,31,32,35\}$ & 14 & {\ttfamily\fontsize{6.5}{8}\selectfont 1100100010000010\,000011011101001100100} \\
7 & $\{1,2,5,9\}$ & $\{21,22,24,25,26,28,30,31,32,33,35,36\}$ & 16 & {\ttfamily\fontsize{6.5}{8}\selectfont 1100100010000000\,000011011101011110110} \\
8 & $\{2,5,9,13,15\}$ & $\{17,18,21,24,26,31,32,35\}$ & 13 & {\ttfamily\fontsize{6.5}{8}\selectfont 0100100010001010\,110010010100001100100} \\
\end{longtable}
\endgroup

\subsubsection{A minimum six-state unital, separating family}\label{S-h:app:separating}

Define $\Omega_6=\{\sigma_1,\ldots,\sigma_6\}$ in the following order:
\begin{equation}\label{S-h:eq:app-six-states}
(\sigma_1,\ldots,\sigma_6)=(s_3,s_{31},s_{75},s_{89},s_{113},s_{151}).
\end{equation}
Table~\ref{S-h:tab:app-six-states} prints their complete supports. Table~\ref{S-h:tab:app-six-profiles} gives all 37 profiles $(\sigma_1(i),\ldots,\sigma_6(i))$ and their active state sets. Every profile is nonzero and no two coincide, so $\Omega_6$ is unital and separating. Since five states supply at most $2^5-1=31<37$ nonzero binary profiles, no smaller unital, separating family can exist.

\begingroup
\footnotesize
\setlength{\tabcolsep}{2pt}
\renewcommand{\arraystretch}{1.08}
\Needspace{24\baselineskip}
\begin{longtable}{cclrc}
\caption{Complete supports and state vectors of the minimum unital, separating family $\Omega_6$. The full state indices are those of Table~\ref{S-h:tab:app-states}.}\label{S-h:tab:app-six-states}\\
\toprule
State & Full index & Complete support & Weight & 37-bit vector \\
\midrule
\endfirsthead
\caption[]{Complete supports and state vectors of the minimum unital, separating family $\Omega_6$. The full state indices are those of Table~\ref{S-h:tab:app-states}. (continued).}\\
\toprule
State & Full index & Complete support & Weight & 37-bit vector \\
\midrule
\endhead
\midrule
\multicolumn{5}{r}{\emph{Continued on the next page.}}\\
\endfoot
\bottomrule
\endlastfoot
$\sigma_1$ & $s_{3}$ & $\{2,5,8,12,17,19,22,25,28,30,31,32,33,35,37\}$ & 15 & {\ttfamily\fontsize{6.5}{8}\selectfont 0100100100010000\,101001001001011110101} \\*
$\sigma_2$ & $s_{31}$ & $\{1,2,6,10,15,21,22,23,25,26,27,31,34,37\}$ & 14 & {\ttfamily\fontsize{6.5}{8}\selectfont 1100010001000010\,000011101110001001001} \\*
$\sigma_3$ & $s_{75}$ & $\{3,6,9,13,17,18,20,23,26,29,31,33,34,35,36\}$ & 15 & {\ttfamily\fontsize{6.5}{8}\selectfont 0010010010001000\,110100100100101011110} \\*
$\sigma_4$ & $s_{89}$ & $\{3,7,8,16,17,18,19,20,22,23,24,27,28,29,32,35\}$ & 16 & {\ttfamily\fontsize{6.5}{8}\selectfont 0010001100000001\,111101110011100100100} \\*
$\sigma_5$ & $s_{113}$ & $\{4,5,9,14,15,17,18,19,20,21,24,25,26,28\}$ & 14 & {\ttfamily\fontsize{6.5}{8}\selectfont 0001100010000110\,111110011101000000000} \\*
$\sigma_6$ & $s_{151}$ & $\{4,7,10,11,18,19,21,24,27,29,30,32,34,36,37\}$ & 15 & {\ttfamily\fontsize{6.5}{8}\selectfont 0001001001100000\,011010010010110101011} \\*
\end{longtable}
\endgroup

\begingroup
\footnotesize
\setlength{\tabcolsep}{3pt}
\renewcommand{\arraystretch}{1.08}
\Needspace{24\baselineskip}
\begin{longtable}{rclcrr@{\hspace{12pt}}rclcrr}
\caption{All 37 vertex profiles on $\Omega_6$, in two panels. Here $X_i^{(6)}=\{j:\sigma_j(i)=1\}$, $W_6=|X_i^{(6)}|$, and $W_{153}$ is the weight across the full state space. The six-bit rows are nonzero and pairwise distinct.}\label{S-h:tab:app-six-profiles}\\
\toprule
$i$ & Ray & $X_i^{(6)}$ & Profile & $W_6$ & $W_{153}$ & $i$ & Ray & $X_i^{(6)}$ & Profile & $W_6$ & $W_{153}$ \\
\midrule
\endfirsthead
\caption[]{All 37 vertex profiles on $\Omega_6$, in two panels. Here $X_i^{(6)}=\{j:\sigma_j(i)=1\}$, $W_6=|X_i^{(6)}|$, and $W_{153}$ is the weight across the full state space. The six-bit rows are nonzero and pairwise distinct. (continued).}\\
\toprule
$i$ & Ray & $X_i^{(6)}$ & Profile & $W_6$ & $W_{153}$ & $i$ & Ray & $X_i^{(6)}$ & Profile & $W_6$ & $W_{153}$ \\
\midrule
\endhead
\midrule
\multicolumn{12}{r}{\emph{Continued on the next page.}}\\
\endfoot
\bottomrule
\endlastfoot
1 & $z$ & $\{2\}$ & {\ttfamily\fontsize{8}{8}\selectfont 010000} & 1 & 54 & 20 & $B_{00}$ & $\{3,4,5\}$ & {\ttfamily\fontsize{8}{8}\selectfont 001110} & 3 & 86 \\
2 & $M_{00}$ & $\{1,2\}$ & {\ttfamily\fontsize{8}{8}\selectfont 110000} & 2 & 51 & 21 & $B_{01}$ & $\{2,5,6\}$ & {\ttfamily\fontsize{8}{8}\selectfont 010011} & 3 & 86 \\
3 & $M_{01}$ & $\{3,4\}$ & {\ttfamily\fontsize{8}{8}\selectfont 001100} & 2 & 51 & 22 & $B_{02}$ & $\{1,2,4\}$ & {\ttfamily\fontsize{8}{8}\selectfont 110100} & 3 & 86 \\
4 & $M_{02}$ & $\{5,6\}$ & {\ttfamily\fontsize{8}{8}\selectfont 000011} & 2 & 51 & 23 & $B_{10}$ & $\{2,3,4\}$ & {\ttfamily\fontsize{8}{8}\selectfont 011100} & 3 & 86 \\
5 & $M_{10}$ & $\{1,5\}$ & {\ttfamily\fontsize{8}{8}\selectfont 100010} & 2 & 51 & 24 & $B_{11}$ & $\{4,5,6\}$ & {\ttfamily\fontsize{8}{8}\selectfont 000111} & 3 & 86 \\
6 & $M_{11}$ & $\{2,3\}$ & {\ttfamily\fontsize{8}{8}\selectfont 011000} & 2 & 51 & 25 & $B_{12}$ & $\{1,2,5\}$ & {\ttfamily\fontsize{8}{8}\selectfont 110010} & 3 & 86 \\
7 & $M_{12}$ & $\{4,6\}$ & {\ttfamily\fontsize{8}{8}\selectfont 000101} & 2 & 51 & 26 & $B_{20}$ & $\{2,3,5\}$ & {\ttfamily\fontsize{8}{8}\selectfont 011010} & 3 & 86 \\
8 & $M_{20}$ & $\{1,4\}$ & {\ttfamily\fontsize{8}{8}\selectfont 100100} & 2 & 51 & 27 & $B_{21}$ & $\{2,4,6\}$ & {\ttfamily\fontsize{8}{8}\selectfont 010101} & 3 & 86 \\
9 & $M_{21}$ & $\{3,5\}$ & {\ttfamily\fontsize{8}{8}\selectfont 001010} & 2 & 51 & 28 & $B_{22}$ & $\{1,4,5\}$ & {\ttfamily\fontsize{8}{8}\selectfont 100110} & 3 & 86 \\
10 & $M_{22}$ & $\{2,6\}$ & {\ttfamily\fontsize{8}{8}\selectfont 010001} & 2 & 51 & 29 & $C_{00}$ & $\{3,4,6\}$ & {\ttfamily\fontsize{8}{8}\selectfont 001101} & 3 & 76 \\
11 & $P_{0}$ & $\{6\}$ & {\ttfamily\fontsize{8}{8}\selectfont 000001} & 1 & 16 & 30 & $C_{01}$ & $\{1,6\}$ & {\ttfamily\fontsize{8}{8}\selectfont 100001} & 2 & 76 \\
12 & $P_{1}$ & $\{1\}$ & {\ttfamily\fontsize{8}{8}\selectfont 100000} & 1 & 16 & 31 & $C_{02}$ & $\{1,2,3\}$ & {\ttfamily\fontsize{8}{8}\selectfont 111000} & 3 & 76 \\*
13 & $P_{2}$ & $\{3\}$ & {\ttfamily\fontsize{8}{8}\selectfont 001000} & 1 & 16 & 32 & $C_{10}$ & $\{1,4,6\}$ & {\ttfamily\fontsize{8}{8}\selectfont 100101} & 3 & 76 \\*
14 & $Q_{0}$ & $\{5\}$ & {\ttfamily\fontsize{8}{8}\selectfont 000010} & 1 & 26 & 33 & $C_{11}$ & $\{1,3\}$ & {\ttfamily\fontsize{8}{8}\selectfont 101000} & 2 & 76 \\*
15 & $Q_{1}$ & $\{2,5\}$ & {\ttfamily\fontsize{8}{8}\selectfont 010010} & 2 & 26 & 34 & $C_{12}$ & $\{2,3,6\}$ & {\ttfamily\fontsize{8}{8}\selectfont 011001} & 3 & 76 \\*
16 & $Q_{2}$ & $\{4\}$ & {\ttfamily\fontsize{8}{8}\selectfont 000100} & 1 & 26 & 35 & $C_{20}$ & $\{1,3,4\}$ & {\ttfamily\fontsize{8}{8}\selectfont 101100} & 3 & 76 \\*
17 & $A_{0}$ & $\{1,3,4,5\}$ & {\ttfamily\fontsize{8}{8}\selectfont 101110} & 4 & 83 & 36 & $C_{21}$ & $\{3,6\}$ & {\ttfamily\fontsize{8}{8}\selectfont 001001} & 2 & 76 \\*
18 & $A_{1}$ & $\{3,4,5,6\}$ & {\ttfamily\fontsize{8}{8}\selectfont 001111} & 4 & 83 & 37 & $C_{22}$ & $\{1,2,6\}$ & {\ttfamily\fontsize{8}{8}\selectfont 110001} & 3 & 76 \\*
19 & $A_{2}$ & $\{1,4,5,6\}$ & {\ttfamily\fontsize{8}{8}\selectfont 100111} & 4 & 83 &  &  &  &  &  &  \\*
\end{longtable}
\endgroup

\subsubsection{All full 153-state vertex profiles}

Table~\ref{S-h:tab:app-full-profiles} prints the transpose of the complete state table: the full 153-bit profile of every vertex, in three consecutive blocks of 51 states. Thus it also enumerates the complete active state set $X_i=\{s_r:s_r(i)=1\}$ for every vertex: a one in position $r$ means exactly $s_r\in X_i$. The row weights are $54$ for $z$, $51$ for $M$, $16$ for $P$, $26$ for $Q$, $83$ for $A$, $86$ for $B$, and $76$ for $C$.

\begingroup
\scriptsize
\setlength{\tabcolsep}{1.3pt}
\renewcommand{\arraystretch}{1.08}
\Needspace{24\baselineskip}
\begin{longtable}{rcrccc}
\caption{Complete 153-bit vertex profiles, with consecutive 51-bit blocks. Every row is nonzero and all 37 rows are distinct.}\label{S-h:tab:app-full-profiles}\\
\toprule
$i$ & Ray & $W_{153}$ & $s_1\ldots s_{51}$ & $s_{52}\ldots s_{102}$ & $s_{103}\ldots s_{153}$ \\
\midrule
\endfirsthead
\caption[]{Complete 153-bit vertex profiles, with consecutive 51-bit blocks. Every row is nonzero and all 37 rows are distinct. (continued).}\\
\toprule
$i$ & Ray & $W_{153}$ & $s_1\ldots s_{51}$ & $s_{52}\ldots s_{102}$ & $s_{103}\ldots s_{153}$ \\
\midrule
\endhead
\midrule
\multicolumn{6}{r}{\emph{Continued on the next page.}}\\
\endfoot
\bottomrule
\endlastfoot
1 & $z$ & 54 & {\ttfamily\fontsize{5.5}{8}\selectfont 100001100001100001100001100001111000011000010110000} & {\ttfamily\fontsize{5.5}{8}\selectfont 110000110000101100001000011000011110000110000110000} & {\ttfamily\fontsize{5.5}{8}\selectfont 110000111100001100001011000011000011000011000010000} \\
2 & $M_{00}$ & 51 & {\ttfamily\fontsize{5.5}{8}\selectfont 111111111111111111111111111111111111111111111111111} & {\ttfamily\fontsize{5.5}{8}\selectfont 000000000000000000000000000000000000000000000000000} & {\ttfamily\fontsize{5.5}{8}\selectfont 000000000000000000000000000000000000000000000000000} \\
3 & $M_{01}$ & 51 & {\ttfamily\fontsize{5.5}{8}\selectfont 000000000000000000000000000000000000000000000000000} & {\ttfamily\fontsize{5.5}{8}\selectfont 111111111111111111111111111111111111111111111111111} & {\ttfamily\fontsize{5.5}{8}\selectfont 000000000000000000000000000000000000000000000000000} \\
4 & $M_{02}$ & 51 & {\ttfamily\fontsize{5.5}{8}\selectfont 000000000000000000000000000000000000000000000000000} & {\ttfamily\fontsize{5.5}{8}\selectfont 000000000000000000000000000000000000000000000000000} & {\ttfamily\fontsize{5.5}{8}\selectfont 111111111111111111111111111111111111111111111111111} \\
5 & $M_{10}$ & 51 & {\ttfamily\fontsize{5.5}{8}\selectfont 111111111111111110000000000000000000000000000000000} & {\ttfamily\fontsize{5.5}{8}\selectfont 111111111111110000000000000000000000000000000000000} & {\ttfamily\fontsize{5.5}{8}\selectfont 111111111111111111110000000000000000000000000000000} \\
6 & $M_{11}$ & 51 & {\ttfamily\fontsize{5.5}{8}\selectfont 000000000000000001111111111111111111100000000000000} & {\ttfamily\fontsize{5.5}{8}\selectfont 000000000000001111111111111111100000000000000000000} & {\ttfamily\fontsize{5.5}{8}\selectfont 000000000000000000001111111111111100000000000000000} \\
7 & $M_{12}$ & 51 & {\ttfamily\fontsize{5.5}{8}\selectfont 000000000000000000000000000000000000011111111111111} & {\ttfamily\fontsize{5.5}{8}\selectfont 000000000000000000000000000000011111111111111111111} & {\ttfamily\fontsize{5.5}{8}\selectfont 000000000000000000000000000000000011111111111111111} \\
8 & $M_{20}$ & 51 & {\ttfamily\fontsize{5.5}{8}\selectfont 111110000000000001111110000000000000011111100000000} & {\ttfamily\fontsize{5.5}{8}\selectfont 111111000000001111110000000000011111111000000000000} & {\ttfamily\fontsize{5.5}{8}\selectfont 111111000000000000001100000000000011111100000000000} \\
9 & $M_{21}$ & 51 & {\ttfamily\fontsize{5.5}{8}\selectfont 000001111110000000000001111110000000000000011000000} & {\ttfamily\fontsize{5.5}{8}\selectfont 000000111111000000001111100000000000000111111000000} & {\ttfamily\fontsize{5.5}{8}\selectfont 000000111111110000000011111100000000000011111100000} \\
10 & $M_{22}$ & 51 & {\ttfamily\fontsize{5.5}{8}\selectfont 000000000001111110000000000001111111100000000111111} & {\ttfamily\fontsize{5.5}{8}\selectfont 000000000000110000000000011111100000000000000111111} & {\ttfamily\fontsize{5.5}{8}\selectfont 000000000000001111110000000011111100000000000011111} \\
11 & $P_{0}$ & 16 & {\ttfamily\fontsize{5.5}{8}\selectfont 000000000000000000000000000000000000000000000000000} & {\ttfamily\fontsize{5.5}{8}\selectfont 000000000000000000000110000110000000000001100001100} & {\ttfamily\fontsize{5.5}{8}\selectfont 000000000000000000000000110000110000000000110001100} \\
12 & $P_{1}$ & 16 & {\ttfamily\fontsize{5.5}{8}\selectfont 011000000000011000000000000000000000000110000001100} & {\ttfamily\fontsize{5.5}{8}\selectfont 000000000000000000000000000000000000000000000000000} & {\ttfamily\fontsize{5.5}{8}\selectfont 001100000000000011000000000000000000110000000001010} \\
13 & $P_{2}$ & 16 & {\ttfamily\fontsize{5.5}{8}\selectfont 010100011000000000011000011000000000000000000000000} & {\ttfamily\fontsize{5.5}{8}\selectfont 001100001100000011000101000000000000000000000000000} & {\ttfamily\fontsize{5.5}{8}\selectfont 000000000000000000000000000000000000000000000000000} \\
14 & $Q_{0}$ & 26 & {\ttfamily\fontsize{5.5}{8}\selectfont 000000000000000000000000000000000000000000000000000} & {\ttfamily\fontsize{5.5}{8}\selectfont 101010101010000000000000000000011001100101010000000} & {\ttfamily\fontsize{5.5}{8}\selectfont 101010110011000000000000000000000010101010101000000} \\
15 & $Q_{1}$ & 26 & {\ttfamily\fontsize{5.5}{8}\selectfont 000001010101010100000001010101100110000000000000000} & {\ttfamily\fontsize{5.5}{8}\selectfont 000000000000000000000000000000000000000000000000000} & {\ttfamily\fontsize{5.5}{8}\selectfont 000000101010101010100010101010101000000000000000000} \\
16 & $Q_{2}$ & 26 & {\ttfamily\fontsize{5.5}{8}\selectfont 000000000000000001010100000001010101010101000101010} & {\ttfamily\fontsize{5.5}{8}\selectfont 000000000000001010100000010101010101010000000101010} & {\ttfamily\fontsize{5.5}{8}\selectfont 000000000000000000000000000000000000000000000000000} \\
17 & $A_{0}$ & 83 & {\ttfamily\fontsize{5.5}{8}\selectfont 011110011110011110011110011110000111100111101001111} & {\ttfamily\fontsize{5.5}{8}\selectfont 001111001111010011110001100001100001111000011000011} & {\ttfamily\fontsize{5.5}{8}\selectfont 001111000011110011110100001100001100111100001100011} \\
18 & $A_{1}$ & 83 & {\ttfamily\fontsize{5.5}{8}\selectfont 000110011110000110011110011110000111100001101000011} & {\ttfamily\fontsize{5.5}{8}\selectfont 001111001111010011110111100111100001111001111001111} & {\ttfamily\fontsize{5.5}{8}\selectfont 000011000011110000110100111100111100001100111100101} \\
19 & $A_{2}$ & 83 & {\ttfamily\fontsize{5.5}{8}\selectfont 001010000110011110000110000110000111100111101001111} & {\ttfamily\fontsize{5.5}{8}\selectfont 000011000011010000110010100111100001111001111001111} & {\ttfamily\fontsize{5.5}{8}\selectfont 001111000011110011110100111100111100111100111101111} \\
20 & $B_{00}$ & 86 & {\ttfamily\fontsize{5.5}{8}\selectfont 000000000000000000000000000000000000000000000000000} & {\ttfamily\fontsize{5.5}{8}\selectfont 111111111111111111111001111001111111111110011110011} & {\ttfamily\fontsize{5.5}{8}\selectfont 111111111111111111111111001111001111111111001110011} \\
21 & $B_{01}$ & 86 & {\ttfamily\fontsize{5.5}{8}\selectfont 100111111111100111111111111111111111111001111110011} & {\ttfamily\fontsize{5.5}{8}\selectfont 000000000000000000000000000000000000000000000000000} & {\ttfamily\fontsize{5.5}{8}\selectfont 110011111111111100111111111111111111001111111110101} \\
22 & $B_{02}$ & 86 & {\ttfamily\fontsize{5.5}{8}\selectfont 101011100111111111100111100111111111111111111111111} & {\ttfamily\fontsize{5.5}{8}\selectfont 110011110011111100111010111111111111111111111111111} & {\ttfamily\fontsize{5.5}{8}\selectfont 000000000000000000000000000000000000000000000000000} \\
23 & $B_{10}$ & 86 & {\ttfamily\fontsize{5.5}{8}\selectfont 000000000000000001111111111111111111111111111111111} & {\ttfamily\fontsize{5.5}{8}\selectfont 000000000000001111111001111001111111111110011110011} & {\ttfamily\fontsize{5.5}{8}\selectfont 000000000000000000001111001111001111111111001110011} \\
24 & $B_{11}$ & 86 & {\ttfamily\fontsize{5.5}{8}\selectfont 100111111111100110000000000000000000011001111110011} & {\ttfamily\fontsize{5.5}{8}\selectfont 111111111111110000000000000000011111111111111111111} & {\ttfamily\fontsize{5.5}{8}\selectfont 110011111111111100110000000000000011001111111110101} \\
25 & $B_{12}$ & 86 & {\ttfamily\fontsize{5.5}{8}\selectfont 101011100111111111100111100111111111100000000000000} & {\ttfamily\fontsize{5.5}{8}\selectfont 110011110011111100111010111111100000000000000000000} & {\ttfamily\fontsize{5.5}{8}\selectfont 111111111111111111111111111111111100000000000000000} \\
26 & $B_{20}$ & 86 & {\ttfamily\fontsize{5.5}{8}\selectfont 000001111111111110000001111111111111100000011111111} & {\ttfamily\fontsize{5.5}{8}\selectfont 000000111111110000001001111001100000000110011110011} & {\ttfamily\fontsize{5.5}{8}\selectfont 000000111111111111110011001111001100000011001110011} \\
27 & $B_{21}$ & 86 & {\ttfamily\fontsize{5.5}{8}\selectfont 100110000001100111111110000001111111111001100110011} & {\ttfamily\fontsize{5.5}{8}\selectfont 111111000000111111110000011111111111111000000111111} & {\ttfamily\fontsize{5.5}{8}\selectfont 110011000000001100111100000011111111001100000010101} \\
28 & $B_{22}$ & 86 & {\ttfamily\fontsize{5.5}{8}\selectfont 101011100110000001100111100110000000011111111000000} & {\ttfamily\fontsize{5.5}{8}\selectfont 110011110011001100111010100000011111111111111000000} & {\ttfamily\fontsize{5.5}{8}\selectfont 111111111111110000001111111100000011111111111100000} \\
29 & $C_{00}$ & 76 & {\ttfamily\fontsize{5.5}{8}\selectfont 000000000000000000000000000000000000000000000000000} & {\ttfamily\fontsize{5.5}{8}\selectfont 010101010101111111111111111111100110011010101111111} & {\ttfamily\fontsize{5.5}{8}\selectfont 010101001100111111111111111111111101010101010111111} \\
30 & $C_{01}$ & 76 & {\ttfamily\fontsize{5.5}{8}\selectfont 111110101010101011111110101010011001111111111111111} & {\ttfamily\fontsize{5.5}{8}\selectfont 000000000000000000000000000000000000000000000000000} & {\ttfamily\fontsize{5.5}{8}\selectfont 111111010101010101011101010101010111111111111111111} \\*
31 & $C_{02}$ & 76 & {\ttfamily\fontsize{5.5}{8}\selectfont 111111111111111110101011111110101010101010111010101} & {\ttfamily\fontsize{5.5}{8}\selectfont 111111111111110101011111101010101010101111111010101} & {\ttfamily\fontsize{5.5}{8}\selectfont 000000000000000000000000000000000000000000000000000} \\*
32 & $C_{10}$ & 76 & {\ttfamily\fontsize{5.5}{8}\selectfont 111111111111111110000000000000000000011111111111111} & {\ttfamily\fontsize{5.5}{8}\selectfont 010101010101110000000000000000000110011010101111111} & {\ttfamily\fontsize{5.5}{8}\selectfont 010101001100111111110000000000000001010101010111111} \\*
33 & $C_{11}$ & 76 & {\ttfamily\fontsize{5.5}{8}\selectfont 111110101010101011111110101010011001100000000000000} & {\ttfamily\fontsize{5.5}{8}\selectfont 111111111111111111111111111111100000000000000000000} & {\ttfamily\fontsize{5.5}{8}\selectfont 111111010101010101011101010101010100000000000000000} \\*
34 & $C_{12}$ & 76 & {\ttfamily\fontsize{5.5}{8}\selectfont 000000000000000000101011111110101010101010111010101} & {\ttfamily\fontsize{5.5}{8}\selectfont 000000000000000101011111101010101010101111111010101} & {\ttfamily\fontsize{5.5}{8}\selectfont 000000000000000000001111111111111111111111111111111} \\*
35 & $C_{20}$ & 76 & {\ttfamily\fontsize{5.5}{8}\selectfont 111111111110000001111111111110000000011111111000000} & {\ttfamily\fontsize{5.5}{8}\selectfont 010101010101001111111111100000000110011010101000000} & {\ttfamily\fontsize{5.5}{8}\selectfont 010101001100110000001111111100000001010101010100000} \\*
36 & $C_{21}$ & 76 & {\ttfamily\fontsize{5.5}{8}\selectfont 000000101010101010000000101010011001100000011111111} & {\ttfamily\fontsize{5.5}{8}\selectfont 000000111111110000001111111111100000000111111111111} & {\ttfamily\fontsize{5.5}{8}\selectfont 000000010101010101010001010101010100000011111111111} \\*
37 & $C_{22}$ & 76 & {\ttfamily\fontsize{5.5}{8}\selectfont 111110000001111110101010000000101010101010100010101} & {\ttfamily\fontsize{5.5}{8}\selectfont 111111000000110101010000001010101010101000000010101} & {\ttfamily\fontsize{5.5}{8}\selectfont 111111000000001111111100000011111111111100000011111} \\*
\end{longtable}
\endgroup

\subsubsection{The six-state double cover}

The fractional-coloring certificate uses a second six-state family,
\begin{equation}\label{S-h:eq:app-cover-family}
\mathcal R_6=\{s_2,s_{30},s_{73},s_{99},s_{109},s_{139}\}.
\end{equation}
Table~\ref{S-h:tab:app-double-cover} specifies these states in the branch parametrization. Each of the 37 vertices occurs in exactly two of their supports. Thus weights $1/2$ give a fractional coloring of total weight three, and uniform weights $1/6$ reproduce the maximally mixed probabilities $1/3$ for every ray. The family $\Omega_6$ in Eq.~\eqref{S-h:eq:app-six-states} supplies the separating certificate; $\mathcal R_6$ supplies the double-cover certificate.

\begingroup
\footnotesize
\setlength{\tabcolsep}{3pt}
\renewcommand{\arraystretch}{1.08}
\Needspace{24\baselineskip}
\begin{longtable}{rccrll}
\caption{Explicit six-state double cover $\mathcal R_6$. Omitted $P,Q$ values are zero; completion values follow from the state equations.}\label{S-h:tab:app-double-cover}\\
\toprule
$j$ & State & $t$ & $\varepsilon$ & Active $P$ subscripts & Active $Q$ subscripts \\
\midrule
\endfirsthead
\caption[]{Explicit six-state double cover $\mathcal R_6$. Omitted $P,Q$ values are zero; completion values follow from the state equations. (continued).}\\
\toprule
$j$ & State & $t$ & $\varepsilon$ & Active $P$ subscripts & Active $Q$ subscripts \\
\midrule
\endhead
\midrule
\multicolumn{6}{r}{\emph{Continued on the next page.}}\\
\endfoot
\bottomrule
\endlastfoot
1 & $s_{2}$ & $(0,0,0)$ & 0 & $\{1,2\}$ & $\varnothing$ \\*
2 & $s_{30}$ & $(0,1,2)$ & 1 & $\varnothing$ & $\{1,2\}$ \\*
3 & $s_{73}$ & $(1,1,1)$ & 0 & $\{0,2\}$ & $\varnothing$ \\*
4 & $s_{99}$ & $(1,2,2)$ & 0 & $\{0\}$ & $\{2\}$ \\*
5 & $s_{109}$ & $(2,0,1)$ & 1 & $\varnothing$ & $\{0,1\}$ \\*
6 & $s_{139}$ & $(2,2,0)$ & 0 & $\{1\}$ & $\{0\}$ \\*
\end{longtable}
\endgroup

\subsubsection{Order determination and exact verification}\label{S-h:app:order-determination}

The companion file \texttt{sources/hesse/}\allowbreak\texttt{hesse37-canonical-evaluations.txt} records the 24 numbered contexts and all 153 binary state vectors. The parametrization in this appendix defines their canonical order independently of any external program.

Extend the canonical states to the 76 events of $L(H_{37})$ as in Sec.~\ref{h:sec:order-determination}. If $X_x=\{r:s_r(x)=1\}$, then the full profiles satisfy
\begin{equation}
X_x\subseteq X_y\quad\Longleftrightarrow\quad x\le y.
\end{equation}
The verifier checks all $76^2=5{,}776$ ordered event pairs against the incidence order in Eq.~\eqref{S-h:eq:pasting-order}. Whenever $x\not\le y$, the companion file \texttt{sources/hesse/}\allowbreak\texttt{hesse37-order-witnesses.csv} gives the least canonical index $r$ with $s_r(x)=1$ and $s_r(y)=0$. Its event labels \texttt{a1},\ldots,\texttt{a37} denote atoms and \texttt{c1},\ldots,\texttt{c37} their complements; \texttt{0} and \texttt{1} denote the bounds. Every one of the 594 distinct nonorthogonal atom pairs also has a joint-one state. Thus the full family is order determining; the six-state family $\Omega_6$ supplies the minimum unital, separating vertex representation in Eq.~\eqref{S-h:eq:canonical-profile}. The distinct minima for event separation and order determination are certified in Appendix~\ref{S-h:app:minimal-families}.

The verifier \texttt{hesse37-verify.py} uses exact rational arithmetic in $\mathbb Q(\zeta)$, with no numerical tolerance or external solver. It checks distinct rays, all 666 Hermitian inner products, the 24 bases, all graph triangles, linearity, the chromatic number and core criticality, all 153 states, the minimum separating family, both profile tables, the maximum squared ray overlap, the six-state double cover, and all event-order comparisons. It also compares the reconstructed state vectors with the archived enumeration. The accompanying machine-readable data give the same numbering and context ordering as the figure and these appendices. From the extracted joint project directory, run
\begin{verbatim}
python verification/verify-integrated.py
\end{verbatim}

\subsection{\texorpdfstring{Four-dimensional $H_{\YO}^{(4)}$ and $H_{65}$: states and profiles}{Four-dimensional HYO(4) and H65: states and profiles}}

This appendix provides the complete enumeration of all two-valued states on the four-dimensional examples $H_{\YO}^{(4)}$ and $H_{65}$: the $25$ two-valued states of the non-linear warm-up hypergraph $H_{\YO}^{(4)}$ (Appendix~\ref{S-d:app:states-hyo4}), and the complete branch-by-branch parametrization of all $24\times 2^{14}=393{,}216$ two-valued states---together with the $48$ canonical unital and separating states $s_{\alpha,\varepsilon}$ and the $65$ vertex state-profiles---of the canonical linear hypergraph $H_{65}$ (Appendix~\ref{S-d:app:states-h65}).

\subsubsection{Complete enumeration of the \texorpdfstring{$25$}{25} two-valued states on \texorpdfstring{$H_{\YO}^{(4)}$}{HYO(4)}}\label{S-d:app:states-hyo4}

Table~\ref{S-d:tab:app-hyo4-states} lists all $25$ two-valued states $s_1,\ldots,s_{25}:V_{26}\to\{0,1\}$ on $H_{\YO}^{(4)}=(V_{26},E_{16}^{(4)})$. States $s_1,\ldots,s_{24}$ have $s_r(26)=0$ and correspond to the $24$ two-valued states of the completed Yu--Oh hypergraph $H_{\YO}^{+}$ in $\R^3$ (grouped into $12$ pairs of states sharing the same support on the three coordinate contexts $\{1,2,3\}$, $\{4,5,6\}$, $\{7,8,9\}$), while $s_{25}$ is the unique apex state with support $\{26\}$. Table~\ref{S-d:tab:app-hyo4-profiles} displays the transposed $26\times 25$ vertex-by-state binary matrix, listing for each vertex $i\in\{1,\ldots,26\}$ its state weight $\sum_{r=1}^{25}s_r(i)\in\{1,3,8,13\}$, the set of states $A_i=\{s_r : s_r(i)=1\}$, and its $25$-bit profile $(s_1(i),\ldots,s_{25}(i))$.

Inspection of Table~\ref{S-d:tab:app-hyo4-profiles} highlights three structural properties of the state space $\mathcal{S}(H_{\YO}^{(4)})$:
\begin{enumerate}
\item \emph{Unitality and atom separation:} Every row of Table~\ref{S-d:tab:app-hyo4-profiles} has positive weight ($1,3,8,$ or $13$) and all $26$ binary rows are distinct ($A_i\ne A_j$ for all $1\le i<j\le 26$), so $\mathcal{S}(H_{\YO}^{(4)})$ is nonempty, unital, and separating on atoms.
\item \emph{Exclusion of simultaneous $1$s on the non-orthogonal triad atoms $\{10,11,12,13\}$:} Although the four Yu--Oh triad rays $v_{10},v_{11},v_{12},v_{13}$ are pairwise non-orthogonal ($v_i\cdot v_j=\pm1\ne0$ for $10\le i<j\le 13$), their active state sets in Table~\ref{S-d:tab:app-hyo4-profiles} are pairwise disjoint:
\begin{equation}\label{S-d:eq:hyo4-triad-disjoint}
\begin{aligned}
A_{10}&=\{s_{13},s_{19},s_{23}\},&
A_{11}&=\{s_{5},s_{7},s_{15}\},\\
A_{12}&=\{s_{1},s_{3},s_{9}\},&
A_{13}&=\{s_{11},s_{17},s_{21}\}.
\end{aligned}
\end{equation}
Consequently, $\sum_{i=10}^{13}s(i)\le1$ holds for every $s\in\mathcal{S}(H_{\YO}^{(4)})$, and no two-valued state assigns $(1,1)$ to any of the $\binom{4}{2}=6$ non-orthogonal pairs $\{10,11\}$, $\{10,12\}$, $\{10,13\}$, $\{11,12\}$, $\{11,13\}$, $\{12,13\}$.
\item \emph{Failure of order-determinacy:} Even though $A_i\ne A_j$ for all $i\ne j$, the state space $\mathcal{S}(H_{\YO}^{(4)})$ is not order-determining on atoms or compound propositions~\cite{S-svozil-tkadlec}. In particular, each triad atom $t\in\{10,11,12,13\}$ (\emph{weight $3$}) is strictly dominated across all $25$ states ($A_t\subsetneq A_w$, i.e., $s(t)\le s(w)$ for all $s\in\mathcal{S}(H_{\YO}^{(4)})$) by three non-orthogonal completion atoms $w\in\{14,\ldots,25\}$ (\emph{weight $13$}):
\begin{equation}\label{S-d:eq:hyo4-inclusions}
\begin{aligned}
A_{10}&\subsetneq A_{16},\,A_{20},\,A_{24},&
A_{11}&\subsetneq A_{18},\,A_{21},\,A_{25},\\
A_{12}&\subsetneq A_{14},\,A_{22},\,A_{23},&
A_{13}&\subsetneq A_{15},\,A_{17},\,A_{19},
\end{aligned}
\end{equation}
together with compound-event inclusions such as $A_{12}\subsetneq A_3\mathbin{\dot\cup}A_{14}$, $A_{11}\subsetneq A_3\mathbin{\dot\cup}A_{25}$, and their unions with the apex set $A_{26}=\{s_{25}\}$.
\end{enumerate}

\begingroup
\scriptsize
\setlength{\tabcolsep}{3pt}
\Needspace{24\baselineskip}
\begin{longtable}{llllc}
\caption{Complete enumeration of the $25$ two-valued states $s_1,\ldots,s_{25}:V_{26}\to\{0,1\}$ on the non-linear warm-up hypergraph $H_{\YO}^{(4)}=(V_{26},E_{16}^{(4)})$, showing for each state its support on the $13$-ray Yu--Oh core $\{1,\ldots,13\}$, its forced support on $\{14,\ldots,26\}$, its total support size $|S_r|$, and its $26$-bit binary vector $(s_r(1\ldots 13)\,s_r(14\ldots 25)\,s_r(26))$.}\label{S-d:tab:app-hyo4-states}\\
\toprule

State & Core support & Completion/apex support & $|S_r|$ & 26-bit state vector \\\midrule
\endfirsthead
\caption[]{All 25 states of $H_{\YO}^{(4)}$ (continued).}\\
\toprule

State & Core support & Completion/apex support & $|S_r|$ & 26-bit state vector \\\midrule
\endhead
\midrule
\multicolumn{5}{r}{\emph{Continued on the next page.}}\\
\endfoot
\bottomrule\endlastfoot
$s_{1}$ & $\{1,4,8,12\}$ & $\{14,15,21,22,23\}$ & $9$ & \texttt{1001000100010\,110000011100\,0} \\
$s_{2}$ & $\{1,4,8\}$ & $\{14,15,16,17,21,22,23,25\}$ & $11$ & \texttt{1001000100000\,111100011101\,0} \\
$s_{3}$ & $\{1,5,7,12\}$ & $\{14,18,22,23,24\}$ & $9$ & \texttt{1000101000010\,100010001110\,0} \\
$s_{4}$ & $\{1,5,7\}$ & $\{14,16,17,18,22,23,24,25\}$ & $11$ & \texttt{1000101000000\,101110001111\,0} \\
$s_{5}$ & $\{1,5,9,11\}$ & $\{17,18,21,22,25\}$ & $9$ & \texttt{1000100010100\,000110011001\,0} \\
$s_{6}$ & $\{1,5,9\}$ & $\{14,15,17,18,21,22,24,25\}$ & $11$ & \texttt{1000100010000\,110110011011\,0} \\
$s_{7}$ & $\{1,6,8,11\}$ & $\{16,18,21,23,25\}$ & $9$ & \texttt{1000010100100\,001010010101\,0} \\
$s_{8}$ & $\{1,6,8\}$ & $\{14,15,16,18,21,23,24,25\}$ & $11$ & \texttt{1000010100000\,111010010111\,0} \\
$s_{9}$ & $\{2,4,7,12\}$ & $\{14,19,20,22,23\}$ & $9$ & \texttt{0101001000010\,100001101100\,0} \\
$s_{10}$ & $\{2,4,7\}$ & $\{14,16,17,19,20,22,23,25\}$ & $11$ & \texttt{0101001000000\,101101101101\,0} \\
$s_{11}$ & $\{2,4,9,13\}$ & $\{14,15,17,19,25\}$ & $9$ & \texttt{0101000010001\,110101000001\,0} \\
$s_{12}$ & $\{2,4,9\}$ & $\{14,15,17,19,20,21,22,25\}$ & $11$ & \texttt{0101000010000\,110101111001\,0} \\
$s_{13}$ & $\{2,6,7,10\}$ & $\{14,16,20,24,25\}$ & $9$ & \texttt{0100011001000\,101000100011\,0} \\
$s_{14}$ & $\{2,6,7\}$ & $\{14,16,18,19,20,23,24,25\}$ & $11$ & \texttt{0100011000000\,101011100111\,0} \\
$s_{15}$ & $\{2,6,9,11\}$ & $\{18,19,20,21,25\}$ & $9$ & \texttt{0100010010100\,000011110001\,0} \\
$s_{16}$ & $\{2,6,9\}$ & $\{14,15,18,19,20,21,24,25\}$ & $11$ & \texttt{0100010010000\,110011110011\,0} \\*
$s_{17}$ & $\{3,4,8,13\}$ & $\{15,16,17,19,23\}$ & $9$ & \texttt{0011000100001\,011101000100\,0} \\*
$s_{18}$ & $\{3,4,8\}$ & $\{15,16,17,19,20,21,22,23\}$ & $11$ & \texttt{0011000100000\,011101111100\,0} \\*
$s_{19}$ & $\{3,5,7,10\}$ & $\{16,17,20,22,24\}$ & $9$ & \texttt{0010101001000\,001100101010\,0} \\*
$s_{20}$ & $\{3,5,7\}$ & $\{16,17,18,19,20,22,23,24\}$ & $11$ & \texttt{0010101000000\,001111101110\,0} \\*
$s_{21}$ & $\{3,5,9,13\}$ & $\{15,17,18,19,24\}$ & $9$ & \texttt{0010100010001\,010111000010\,0} \\*
$s_{22}$ & $\{3,5,9\}$ & $\{15,17,18,19,20,21,22,24\}$ & $11$ & \texttt{0010100010000\,010111111010\,0} \\*
$s_{23}$ & $\{3,6,8,10\}$ & $\{15,16,20,21,24\}$ & $9$ & \texttt{0010010101000\,011000110010\,0} \\*
$s_{24}$ & $\{3,6,8\}$ & $\{15,16,18,19,20,21,23,24\}$ & $11$ & \texttt{0010010100000\,011011110110\,0} \\ \hline
$s_{25}$ & $\emptyset$ & $\{26\}$ & $1$ & \texttt{0000000000000\,000000000000\,1}
\end{longtable}
\endgroup

\begingroup
\scriptsize
\setlength{\tabcolsep}{3pt}
\Needspace{24\baselineskip}
\begin{longtable}{rrlc}
\caption{Vertex state profiles across the $25$ two-valued states $s_1,\ldots,s_{25}$ of $H_{\YO}^{(4)}$, showing for each vertex $i\in\{1,\ldots,26\}$ its state weight $\sum_{r=1}^{25}s_r(i)$, the active state indices $\{r : s_r(i)=1\}$, and the $25$-bit profile $(s_{1\ldots 12}(i)\,s_{13\ldots 24}(i)\,s_{25}(i))$. All $26$ rows are nonzero and pairwise distinct.}\label{S-d:tab:app-hyo4-profiles}\\
\toprule

Vertex $i$ & Weight $\sum_r s_r(i)$ & Active states $\{r\in\{1,\ldots,25\} : s_r(i)=1\}$ & $25$-bit profile $(s_{1\ldots 12}(i)\,s_{13\ldots 24}(i)\,s_{25}(i))$ \\\midrule
\endfirsthead
\caption[]{Vertex profiles of $H_{\YO}^{(4)}$ (continued).}\\
\toprule

Vertex $i$ & Weight $\sum_r s_r(i)$ & Active states $\{r\in\{1,\ldots,25\} : s_r(i)=1\}$ & $25$-bit profile $(s_{1\ldots 12}(i)\,s_{13\ldots 24}(i)\,s_{25}(i))$ \\\midrule
\endhead
\midrule
\multicolumn{4}{r}{\emph{Continued on the next page.}}\\
\endfoot
\bottomrule\endlastfoot
$1$ & $8$ & $\{1,2,3,4,5,6,7,8\}$ & \texttt{111111110000\,000000000000\,0} \\
$2$ & $8$ & $\{9,10,11,12,13,14,15,16\}$ & \texttt{000000001111\,111100000000\,0} \\
$3$ & $8$ & $\{17,18,19,20,21,22,23,24\}$ & \texttt{000000000000\,000011111111\,0} \\
$4$ & $8$ & $\{1,2,9,10,11,12,17,18\}$ & \texttt{110000001111\,000011000000\,0} \\
$5$ & $8$ & $\{3,4,5,6,19,20,21,22\}$ & \texttt{001111000000\,000000111100\,0} \\
$6$ & $8$ & $\{7,8,13,14,15,16,23,24\}$ & \texttt{000000110000\,111100000011\,0} \\
$7$ & $8$ & $\{3,4,9,10,13,14,19,20\}$ & \texttt{001100001100\,110000110000\,0} \\
$8$ & $8$ & $\{1,2,7,8,17,18,23,24\}$ & \texttt{110000110000\,000011000011\,0} \\
$9$ & $8$ & $\{5,6,11,12,15,16,21,22\}$ & \texttt{000011000011\,001100001100\,0} \\ \hline
$10$ & $3$ & $\{13,19,23\}$ & \texttt{000000000000\,100000100010\,0} \\
$11$ & $3$ & $\{5,7,15\}$ & \texttt{000010100000\,001000000000\,0} \\
$12$ & $3$ & $\{1,3,9\}$ & \texttt{101000001000\,000000000000\,0} \\
$13$ & $3$ & $\{11,17,21\}$ & \texttt{000000000010\,000010001000\,0} \\ \hline
$14$ & $13$ & $\{1,2,3,4,6,8,9,10,11,12,13,14,16\}$ & \texttt{111101011111\,110100000000\,0} \\
$15$ & $13$ & $\{1,2,6,8,11,12,16,17,18,21,22,23,24\}$ & \texttt{110001010011\,000111001111\,0} \\
$16$ & $13$ & $\{2,4,7,8,10,13,14,17,18,19,20,23,24\}$ & \texttt{010100110100\,110011110011\,0} \\
$17$ & $13$ & $\{2,4,5,6,10,11,12,17,18,19,20,21,22\}$ & \texttt{010111000111\,000011111100\,0} \\*
$18$ & $13$ & $\{3,4,5,6,7,8,14,15,16,20,21,22,24\}$ & \texttt{001111110000\,011100011101\,0} \\*
$19$ & $13$ & $\{9,10,11,12,14,15,16,17,18,20,21,22,24\}$ & \texttt{000000001111\,011111011101\,0} \\*
$20$ & $13$ & $\{9,10,12,13,14,15,16,18,19,20,22,23,24\}$ & \texttt{000000001101\,111101110111\,0} \\*
$21$ & $13$ & $\{1,2,5,6,7,8,12,15,16,18,22,23,24\}$ & \texttt{110011110001\,001101000111\,0} \\*
$22$ & $13$ & $\{1,2,3,4,5,6,9,10,12,18,19,20,22\}$ & \texttt{111111001101\,000001110100\,0} \\*
$23$ & $13$ & $\{1,2,3,4,7,8,9,10,14,17,18,20,24\}$ & \texttt{111100111100\,010011010001\,0} \\*
$24$ & $13$ & $\{3,4,6,8,13,14,16,19,20,21,22,23,24\}$ & \texttt{001101010000\,110100111111\,0} \\*
$25$ & $13$ & $\{2,4,5,6,7,8,10,11,12,13,14,15,16\}$ & \texttt{010111110111\,111100000000\,0} \\ \hline
$26$ & $1$ & $\{25\}$ & \texttt{000000000000\,000000000000\,1}
\end{longtable}
\endgroup

\subsubsection{Complete parametrization of all \texorpdfstring{$393{,}216$}{393,216} two-valued states and the \texorpdfstring{$48$}{48} canonical separating states on \texorpdfstring{$H_{65}$}{H\_65}}\label{S-d:app:states-h65}

Every two-valued state $s\in\mathcal{S}(H_{65})$ on the canonical linear hypergraph $H_{65}=(V_{65},E_{29})$ is uniquely indexed by a pair $(\alpha,\bm{\varepsilon})$ consisting of:
\begin{enumerate}
\item a branch index $\alpha\in\{1,\ldots,24\}$ selecting one of the $24$ four-element core state supports $S_\alpha\subset V_{\mathrm{core}}=\{1,\ldots,17\}$; and
\item a $14$-bit binary vector $\bm{\varepsilon}=(\varepsilon_m)_{m\in \mathcal{F}_\alpha}\in\{0,1\}^{14}$ indexed by the $14$ \emph{free completion contexts}
\begin{equation*}
\mathcal{F}_\alpha=\bigl\{m\in\{1,\ldots,24\} : S_\alpha\cap\{u_m,v_m\}=\emptyset\bigr\},
\end{equation*}
whose complement is the set of $10$ \emph{hit completion contexts} $\mathcal{H}_\alpha=\{1,\ldots,24\}\setminus \mathcal{F}_\alpha$.
\end{enumerate}
Explicitly, the state $s_{\alpha,\bm{\varepsilon}}:V_{65}\to\{0,1\}$ takes values
\begin{equation}\label{S-d:eq:full-state-param}
\begin{aligned}
s_{\alpha,\bm{\varepsilon}}(i)&=\mathbf{1}_{i\in S_\alpha}\qquad(1\le i\le 17),\\
\bigl(s_{\alpha,\bm{\varepsilon}}&(2m+16),\,s_{\alpha,\bm{\varepsilon}}(2m+17)\bigr)\\
&=\begin{cases}
(0,0), & m\in \mathcal{H}_\alpha,\\
(1-\varepsilon_m,\,\varepsilon_m), & m\in \mathcal{F}_\alpha.
\end{cases}
\end{aligned}
\end{equation}
Because $|\mathcal{F}_\alpha|=14$ for every $\alpha\in\{1,\ldots,24\}$, each branch $\alpha$ yields $2^{14}=16{,}384$ states $s_{\alpha,\bm{\varepsilon}}$, giving $|\mathcal{S}(H_{65})|=24\times 2^{14}=393{,}216$ two-valued states in total. Setting all $14$ free bits equal to $\varepsilon\in\{0,1\}$ ($\varepsilon_m=\varepsilon$ for all $m\in \mathcal{F}_\alpha$) yields the $48$ \emph{canonical states} $s_{\alpha,0},s_{\alpha,1}$ of Eq.~\eqref{S-d:eq:canonical-states}, which have support size $|S_\alpha|+14=18$:
\begin{equation}\label{S-d:eq:canonical-state-supports}
\begin{aligned}
\mathrm{supp}(s_{\alpha,0}) &= S_\alpha \cup \bigl\{2m+16 : m\in \mathcal{F}_\alpha\bigr\},\\
\mathrm{supp}(s_{\alpha,1}) &= S_\alpha \cup \bigl\{2m+17 : m\in \mathcal{F}_\alpha\bigr\}.
\end{aligned}
\end{equation}

Table~\ref{S-d:tab:app-h65-state-branches} enumerates all $24$ branches $\alpha=1,\ldots,24$, specifying $S_\alpha$, the $10$ hit context indices $\mathcal{H}_\alpha\subset\{1,\ldots,24\}$, the $14$ free context indices $\mathcal{F}_\alpha\subset\{1,\ldots,24\}$, and the $24$-bit free-context indicator string. Table~\ref{S-d:tab:app-h65-vertex-profiles} lists the resulting state profiles and weights for all $65$ vertices of $H_{65}$ across both the $48$ canonical states $\{s_{\alpha,\varepsilon}\}$ and the full space of $393{,}216$ states: across the $48$ canonical states $(s_{1,0},s_{1,1},\ldots,s_{24,0},s_{24,1})$, each `1` in row $i$ of Table~\ref{S-d:tab:app-h65-vertex-profiles} expands to `(1,1)` for core vertices $i\in\{1,\ldots,17\}$, to `(1,0)` for first completion vertices $w_{m,1}=2m+16$, and to `(0,1)` for second completion vertices $w_{m,2}=2m+17$ (and each `0` expands to `(0,0)`), making all $65$ canonical $48$-bit vertex profiles nonzero and pairwise distinct by inspection.

\begingroup
\scriptsize
\setlength{\tabcolsep}{1pt}
\thinmuskip=1mu
\Needspace{24\baselineskip}
\begin{longtable}{llllc}
\caption{Complete branch-by-branch parametrization of the $24\times 2^{14}=393{,}216$ two-valued states $s_{\alpha,\bm{\varepsilon}}$ and the $48$ canonical states $s_{\alpha,0},s_{\alpha,1}$ on $H_{65}$ via Eqs.~\eqref{S-d:eq:full-state-param}--\eqref{S-d:eq:canonical-state-supports}. For each $\alpha\in\{1,\ldots,24\}$, the table lists the $4$-element core support $S_\alpha$, the $10$ hit completion context indices $m\in \mathcal{H}_\alpha$ where $(s(2m+16),s(2m+17))=(0,0)$ is forced on $E_m$, the $14$ free completion context indices $m\in \mathcal{F}_\alpha$ where $(s(2m+16),s(2m+17))=(1-\varepsilon_m,\varepsilon_m)$ with $\varepsilon_m\in\{0,1\}$ on $E_m$, and the $24$-bit indicator string $(\mathbf{1}_{1\in \mathcal{F}_\alpha}\ldots\mathbf{1}_{12\in \mathcal{F}_\alpha}\,\mathbf{1}_{13\in \mathcal{F}_\alpha}\ldots\mathbf{1}_{24\in \mathcal{F}_\alpha})$.}\label{S-d:tab:app-h65-state-branches}\\
\toprule

Branch & Core support $S_\alpha$ & Hit contexts $\mathcal H_\alpha$ & Free contexts $\mathcal F_\alpha$ & Free indicator $m=1\ldots 24$ \\\midrule
\endfirsthead
\caption[]{State branches of $H_{65}$ (continued).}\\
\toprule

Branch & Core support $S_\alpha$ & Hit contexts $\mathcal H_\alpha$ & Free contexts $\mathcal F_\alpha$ & Free indicator $m=1\ldots 24$ \\\midrule
\endhead
\midrule
\multicolumn{5}{r}{\emph{Continued on the next page.}}\\
\endfoot
\bottomrule\endlastfoot
$S_{1}$ & $\{1,5,9,16\}$ & $\{3,7,9,10,11,13,19,21,22,23\}$ & $\{1,2,4,5,6,8,12,14,15,16,17,18,20,24\}$ & \texttt{110111010001\,011111010001} \\
$S_{2}$ & $\{1,5,10,14\}$ & $\{3,5,6,7,11,13,17,18,19,23\}$ & $\{1,2,4,8,9,10,12,14,15,16,20,21,22,24\}$ & \texttt{110100011101\,011100011101} \\
$S_{3}$ & $\{1,6,9,16\}$ & $\{3,7,9,10,11,15,17,21,22,24\}$ & $\{1,2,4,5,6,8,12,13,14,16,18,19,20,23\}$ & \texttt{110111010001\,110101110010} \\
$S_{4}$ & $\{1,6,11,12\}$ & $\{1,2,3,7,11,13,14,15,17,24\}$ & $\{4,5,6,8,9,10,12,16,18,19,20,21,22,23\}$ & \texttt{000111011101\,000101111110} \\
$S_{5}$ & $\{1,7,10,14\}$ & $\{3,5,6,7,11,16,17,18,20,21\}$ & $\{1,2,4,8,9,10,12,13,14,15,19,22,23,24\}$ & \texttt{110100011101\,111000100111} \\
$S_{6}$ & $\{1,7,11,12\}$ & $\{1,2,3,7,11,13,14,16,20,21\}$ & $\{4,5,6,8,9,10,12,15,17,18,19,22,23,24\}$ & \texttt{000111011101\,001011100111} \\
$S_{7}$ & $\{2,5,9,16\}$ & $\{1,5,9,10,12,13,19,21,22,23\}$ & $\{2,3,4,6,7,8,11,14,15,16,17,18,20,24\}$ & \texttt{011101110010\,011111010001} \\
$S_{8}$ & $\{2,5,11,13\}$ & $\{1,3,4,5,12,13,15,16,19,23\}$ & $\{2,6,7,8,9,10,11,14,17,18,20,21,22,24\}$ & \texttt{010001111110\,010011011101} \\
$S_{9}$ & $\{2,6,9,16\}$ & $\{1,5,9,10,12,15,17,21,22,24\}$ & $\{2,3,4,6,7,8,11,13,14,16,18,19,20,23\}$ & \texttt{011101110010\,110101110010} \\
$S_{10}$ & $\{2,6,10,15\}$ & $\{1,5,7,8,12,15,17,19,20,24\}$ & $\{2,3,4,6,9,10,11,13,14,16,18,21,22,23\}$ & \texttt{011101001110\,110101001110} \\
$S_{11}$ & $\{2,8,10,15\}$ & $\{1,5,7,8,12,14,18,19,20,22\}$ & $\{2,3,4,6,9,10,11,13,15,16,17,21,23,24\}$ & \texttt{011101001110\,101110001011} \\
$S_{12}$ & $\{2,8,11,13\}$ & $\{1,3,4,5,12,14,15,16,18,22\}$ & $\{2,6,7,8,9,10,11,13,17,19,20,21,23,24\}$ & \texttt{010001111110\,100010111011} \\
$S_{13}$ & $\{3,5,10,14\}$ & $\{2,5,6,8,9,13,17,18,19,23\}$ & $\{1,3,4,7,10,11,12,14,15,16,20,21,22,24\}$ & \texttt{101100100111\,011100011101} \\
$S_{14}$ & $\{3,5,11,13\}$ & $\{2,3,4,8,9,13,15,16,19,23\}$ & $\{1,5,6,7,10,11,12,14,17,18,20,21,22,24\}$ & \texttt{100011100111\,010011011101} \\
$S_{15}$ & $\{3,7,9,17\}$ & $\{2,8,9,11,12,16,20,21,23,24\}$ & $\{1,3,4,5,6,7,10,13,14,15,17,18,19,22\}$ & \texttt{101111100100\,111011100100} \\
$S_{16}$ & $\{3,7,10,14\}$ & $\{2,5,6,8,9,16,17,18,20,21\}$ & $\{1,3,4,7,10,11,12,13,14,15,19,22,23,24\}$ & \texttt{101100100111\,111000100111} \\*
$S_{17}$ & $\{3,8,9,17\}$ & $\{2,8,9,11,12,14,18,22,23,24\}$ & $\{1,3,4,5,6,7,10,13,15,16,17,19,20,21\}$ & \texttt{101111100100\,101110111000} \\*
$S_{18}$ & $\{3,8,11,13\}$ & $\{2,3,4,8,9,14,15,16,18,22\}$ & $\{1,5,6,7,10,11,12,13,17,19,20,21,23,24\}$ & \texttt{100011100111\,100010111011} \\*
$S_{19}$ & $\{4,6,10,15\}$ & $\{4,6,7,8,10,15,17,19,20,24\}$ & $\{1,2,3,5,9,11,12,13,14,16,18,21,22,23\}$ & \texttt{111010001011\,110101001110} \\*
$S_{20}$ & $\{4,6,11,12\}$ & $\{1,2,4,6,10,13,14,15,17,24\}$ & $\{3,5,7,8,9,11,12,16,18,19,20,21,22,23\}$ & \texttt{001010111011\,000101111110} \\*
$S_{21}$ & $\{4,7,9,17\}$ & $\{4,6,10,11,12,16,20,21,23,24\}$ & $\{1,2,3,5,7,8,9,13,14,15,17,18,19,22\}$ & \texttt{111010111000\,111011100100} \\*
$S_{22}$ & $\{4,7,11,12\}$ & $\{1,2,4,6,10,13,14,16,20,21\}$ & $\{3,5,7,8,9,11,12,15,17,18,19,22,23,24\}$ & \texttt{001010111011\,001011100111} \\*
$S_{23}$ & $\{4,8,9,17\}$ & $\{4,6,10,11,12,14,18,22,23,24\}$ & $\{1,2,3,5,7,8,9,13,15,16,17,19,20,21\}$ & \texttt{111010111000\,101110111000} \\*
$S_{24}$ & $\{4,8,10,15\}$ & $\{4,6,7,8,10,14,18,19,20,22\}$ & $\{1,2,3,5,9,11,12,13,15,16,17,21,23,24\}$ & \texttt{111010001011\,101110001011}
\end{longtable}
\endgroup

\begingroup
\scriptsize
\setlength{\tabcolsep}{3pt}
\Needspace{24\baselineskip}
\begin{longtable}{lrrrlc}
\caption{Vertex state profiles and weights for all $65$ vertices of $H_{65}$, showing for each core vertex $i\in\{1,\ldots,17\}$ and each completion pair $(w_{m,1},w_{m,2})=(2m+16,2m+17)$ ($m=1,\ldots,24$) the number of active branches $W_{24}$, the weight $W_{48}=\sum_{\alpha,\varepsilon}s_{\alpha,\varepsilon}(i)$ across the $48$ canonical states, the weight $W_{393{,}216}=\sum_{s\in\mathcal{S}(H_{65})}s(i)$ across all $393{,}216$ states, the active branch indices $\alpha\in\{1,\ldots,24\}$ (equal to $Z(e_m)$ for $(w_{m,1},w_{m,2})$), and the $24$-bit branch profile across $S_1,\ldots,S_{24}$ (in which each \texttt{1} expands to \texttt{11} for core vertices $i$, to \texttt{10} for $w_{m,1}$, and to \texttt{01} for $w_{m,2}$ across the $48$ canonical states).}\label{S-d:tab:app-h65-vertex-profiles}\\
\toprule

Vertex / pair & $W_{24}$ & $W_{48}$ & $W_{393{,}216}$ & Active branches $\alpha$ & 24-bit branch profile \\\midrule
\endfirsthead
\caption[]{Vertex profiles of $H_{65}$ (continued).}\\
\toprule

Vertex / pair & $W_{24}$ & $W_{48}$ & $W_{393{,}216}$ & Active branches $\alpha$ & 24-bit branch profile \\\midrule
\endhead
\midrule
\multicolumn{6}{r}{\emph{Continued on the next page.}}\\
\endfoot
\bottomrule\endlastfoot
Core $1$ & $6$ & $12$ & $98{,}304$ & $\{1,2,3,4,5,6\}$ & \texttt{11111100\,00000000\,00000000} \\
Core $2$ & $6$ & $12$ & $98{,}304$ & $\{7,8,9,10,11,12\}$ & \texttt{00000011\,11110000\,00000000} \\
Core $3$ & $6$ & $12$ & $98{,}304$ & $\{13,14,15,16,17,18\}$ & \texttt{00000000\,00001111\,11000000} \\
Core $4$ & $6$ & $12$ & $98{,}304$ & $\{19,20,21,22,23,24\}$ & \texttt{00000000\,00000000\,00111111} \\
Core $5$ & $6$ & $12$ & $98{,}304$ & $\{1,2,7,8,13,14\}$ & \texttt{11000011\,00001100\,00000000} \\
Core $6$ & $6$ & $12$ & $98{,}304$ & $\{3,4,9,10,19,20\}$ & \texttt{00110000\,11000000\,00110000} \\
Core $7$ & $6$ & $12$ & $98{,}304$ & $\{5,6,15,16,21,22\}$ & \texttt{00001100\,00000011\,00001100} \\
Core $8$ & $6$ & $12$ & $98{,}304$ & $\{11,12,17,18,23,24\}$ & \texttt{00000000\,00110000\,11000011} \\
Core $9$ & $8$ & $16$ & $131{,}072$ & $\{1,3,7,9,15,17,21,23\}$ & \texttt{10100010\,10000010\,10001010} \\
Core $10$ & $8$ & $16$ & $131{,}072$ & $\{2,5,10,11,13,16,19,24\}$ & \texttt{01001000\,01101001\,00100001} \\
Core $11$ & $8$ & $16$ & $131{,}072$ & $\{4,6,8,12,14,18,20,22\}$ & \texttt{00010101\,00010100\,01010100} \\
Core $12$ & $4$ & $8$ & $65{,}536$ & $\{4,6,20,22\}$ & \texttt{00010100\,00000000\,00010100} \\
Core $13$ & $4$ & $8$ & $65{,}536$ & $\{8,12,14,18\}$ & \texttt{00000001\,00010100\,01000000} \\
Core $14$ & $4$ & $8$ & $65{,}536$ & $\{2,5,13,16\}$ & \texttt{01001000\,00001001\,00000000} \\
Core $15$ & $4$ & $8$ & $65{,}536$ & $\{10,11,19,24\}$ & \texttt{00000000\,01100000\,00100001} \\
Core $16$ & $4$ & $8$ & $65{,}536$ & $\{1,3,7,9\}$ & \texttt{10100010\,10000000\,00000000} \\
Core $17$ & $4$ & $8$ & $65{,}536$ & $\{15,17,21,23\}$ & \texttt{00000000\,00000010\,10001010} \\ \hline
$E_{1}:\,(18,19)$ & $14$ & $14$ & $114{,}688$ & $\{1,2,3,5,13,14,15,16,17,18,19,21,23,24\}$ & \texttt{11101000\,00001111\,11101011} \\
$E_{2}:\,(20,21)$ & $14$ & $14$ & $114{,}688$ & $\{1,2,3,5,7,8,9,10,11,12,19,21,23,24\}$ & \texttt{11101011\,11110000\,00101011} \\
$E_{3}:\,(22,23)$ & $14$ & $14$ & $114{,}688$ & $\{7,9,10,11,13,15,16,17,19,20,21,22,23,24\}$ & \texttt{00000010\,11101011\,10111111} \\
$E_{4}:\,(24,25)$ & $14$ & $14$ & $114{,}688$ & $\{1,2,3,4,5,6,7,9,10,11,13,15,16,17\}$ & \texttt{11111110\,11101011\,10000000} \\
$E_{5}:\,(26,27)$ & $14$ & $14$ & $114{,}688$ & $\{1,3,4,6,14,15,17,18,19,20,21,22,23,24\}$ & \texttt{10110100\,00000110\,11111111} \\
$E_{6}:\,(28,29)$ & $14$ & $14$ & $114{,}688$ & $\{1,3,4,6,7,8,9,10,11,12,14,15,17,18\}$ & \texttt{10110111\,11110110\,11000000} \\
$E_{7}:\,(30,31)$ & $14$ & $14$ & $114{,}688$ & $\{7,8,9,12,13,14,15,16,17,18,20,21,22,23\}$ & \texttt{00000011\,10011111\,11011110} \\
$E_{8}:\,(32,33)$ & $14$ & $14$ & $114{,}688$ & $\{1,2,3,4,5,6,7,8,9,12,20,21,22,23\}$ & \texttt{11111111\,10010000\,00011110} \\
$E_{9}:\,(34,35)$ & $14$ & $14$ & $114{,}688$ & $\{2,4,5,6,8,10,11,12,19,20,21,22,23,24\}$ & \texttt{01011101\,01110000\,00111111} \\
$E_{10}:\,(36,37)$ & $14$ & $14$ & $114{,}688$ & $\{2,4,5,6,8,10,11,12,13,14,15,16,17,18\}$ & \texttt{01011101\,01111111\,11000000} \\
$E_{11}:\,(38,39)$ & $14$ & $14$ & $114{,}688$ & $\{7,8,9,10,11,12,13,14,16,18,19,20,22,24\}$ & \texttt{00000011\,11111101\,01110101} \\
$E_{12}:\,(40,41)$ & $14$ & $14$ & $114{,}688$ & $\{1,2,3,4,5,6,13,14,16,18,19,20,22,24\}$ & \texttt{11111100\,00001101\,01110101} \\
$E_{13}:\,(42,43)$ & $14$ & $14$ & $114{,}688$ & $\{3,5,9,10,11,12,15,16,17,18,19,21,23,24\}$ & \texttt{00101000\,11110011\,11101011} \\
$E_{14}:\,(44,45)$ & $14$ & $14$ & $114{,}688$ & $\{1,2,3,5,7,8,9,10,13,14,15,16,19,21\}$ & \texttt{11101011\,11001111\,00101000} \\
$E_{15}:\,(46,47)$ & $14$ & $14$ & $114{,}688$ & $\{1,2,5,6,7,11,13,15,16,17,21,22,23,24\}$ & \texttt{11001110\,00101011\,10001111} \\
$E_{16}:\,(48,49)$ & $14$ & $14$ & $114{,}688$ & $\{1,2,3,4,7,9,10,11,13,17,19,20,23,24\}$ & \texttt{11110010\,11101000\,10110011} \\*
$E_{17}:\,(50,51)$ & $14$ & $14$ & $114{,}688$ & $\{1,6,7,8,11,12,14,15,17,18,21,22,23,24\}$ & \texttt{10000111\,00110110\,11001111} \\*
$E_{18}:\,(52,53)$ & $14$ & $14$ & $114{,}688$ & $\{1,3,4,6,7,8,9,10,14,15,19,20,21,22\}$ & \texttt{10110111\,11000110\,00111100} \\*
$E_{19}:\,(54,55)$ & $14$ & $14$ & $114{,}688$ & $\{3,4,5,6,9,12,15,16,17,18,20,21,22,23\}$ & \texttt{00111100\,10010011\,11011110} \\*
$E_{20}:\,(56,57)$ & $14$ & $14$ & $114{,}688$ & $\{1,2,3,4,7,8,9,12,13,14,17,18,20,23\}$ & \texttt{11110011\,10011100\,11010010} \\*
$E_{21}:\,(58,59)$ & $14$ & $14$ & $114{,}688$ & $\{2,4,8,10,11,12,13,14,17,18,19,20,23,24\}$ & \texttt{01010001\,01111100\,11110011} \\*
$E_{22}:\,(60,61)$ & $14$ & $14$ & $114{,}688$ & $\{2,4,5,6,8,10,13,14,15,16,19,20,21,22\}$ & \texttt{01011101\,01001111\,00111100} \\*
$E_{23}:\,(62,63)$ & $14$ & $14$ & $114{,}688$ & $\{3,4,5,6,9,10,11,12,16,18,19,20,22,24\}$ & \texttt{00111100\,11110001\,01110101} \\*
$E_{24}:\,(64,65)$ & $14$ & $14$ & $114{,}688$ & $\{1,2,5,6,7,8,11,12,13,14,16,18,22,24\}$ & \texttt{11001111\,00111101\,01000101}
\end{longtable}
\endgroup

\FloatBarrier
\subsection{Parent-state and deletion certificates}\label{S-app:parent-states}
The 49-ray, 87-ray, 40-ray, and 72-ray parents have no context states.
The two smaller 68-ray parents are treated in
Appendix~\ref{S-app:cabello-states}.
The first and last two have the independently verified complete bases
printed in Appendix~\ref{S-app:parents-vectors}; the 87-ray parent contains
the known 69-ray KS lineage. The 70-ray single-apex lift of the completed
Cabello parent instead has the unique state supported at its apex.

For the real Yu--Oh parent, the intermediate 33 rays admit 192 assignments
to the twenty bases alone and none after its sixteen additional pair
exclusions. The sixteen missing pairs are uniquely completed by rays
$34$--$49$. The x-axis-plane extension by rays $26$--$29$ alone has
72 basis assignments; all 24 original states extend. Including its eight
additional pair exclusions leaves eight extensions, all with $s(2)=1$.
The z-axis-plane case is identical, with $s(5)=1$. The two added bases
on either axis each exclude twelve of the sixteen original axis-zero
states when their further pair exclusions are included; eight original
states are excluded by both, leaving none. Thus the axis forcing uses
the complete orthogonality constraints, not the added bases alone.
\par
\begin{table}[htbp]
\centering\small
\caption{Context-state counts after deleting each single context of the real 49-ray parent. Every count is positive. Context numbering is Table~\ref{S-tab:bub49-contexts}. Deleting a completion context $21$--$36$ leaves its fresh ray in no context; that ray is counted as an unconstrained free vertex, which doubles the count. If the orthogonalities of the deleted context are kept as geometric exclusions, deletions $21$--$36$ leave no state, while deletions $1$--$20$ leave $1$, $2$, or $4$ states.}\label{S-tab:parent-criticality}
\begin{tabular}{rr@{\qquad}rr@{\qquad}rr}
\toprule
Deleted context & States & Deleted context & States & Deleted context & States\\\midrule
\rule{0pt}{12pt}$1$ & $12$ & $13$ & $2$ & $25$ & $2$\\
\rule{0pt}{12pt}$2$ & $4$ & $14$ & $4$ & $26$ & $2$\\
\rule{0pt}{12pt}$3$ & $4$ & $15$ & $4$ & $27$ & $2$\\
\rule{0pt}{12pt}$4$ & $32$ & $16$ & $16$ & $28$ & $2$\\
\rule{0pt}{12pt}$5$ & $2$ & $17$ & $4$ & $29$ & $2$\\
\rule{0pt}{12pt}$6$ & $2$ & $18$ & $4$ & $30$ & $2$\\
\rule{0pt}{12pt}$7$ & $4$ & $19$ & $4$ & $31$ & $2$\\
\rule{0pt}{12pt}$8$ & $4$ & $20$ & $4$ & $32$ & $2$\\
\rule{0pt}{12pt}$9$ & $12$ & $21$ & $2$ & $33$ & $2$\\
\rule{0pt}{12pt}$10$ & $4$ & $22$ & $2$ & $34$ & $2$\\
\rule{0pt}{12pt}$11$ & $4$ & $23$ & $2$ & $35$ & $2$\\
\rule{0pt}{12pt}$12$ & $2$ & $24$ & $2$ & $36$ & $2$\\
\bottomrule
\end{tabular}
\end{table}
\par
All complete states of $H_{65}$ are additionally supplied as
\texttt{data/h65-all-states.csv.gz}. Each row gives its branch index,
the fourteen-bit choice index within that branch, and its 65-bit state
vector in vertex order. There are exactly $24\times2^{14}=393{,}216$ data
rows. This supplements the complete parametrization and the printed
48-state family; it does not assert that 48 is minimal. The 153 Hesse
states, their event-order witnesses, and all original printed profiles
remain in the accompanying Hesse data.

\subsection{States of the Cabello restrictions and completions}
\subsubsection{States of the vertex-deleted 45-ray completion}
\label{S-app:cabello-states}
The nine-context $C_{18,9}$, its eighteen-context faithful completion,
and both 68-ray parents have no context states by their retained
parity proofs. For the 45-ray construction, use exactly the local
labels in Tables~\ref{S-tab:cabello45-vectors} and~\ref{S-tab:cabello45-contexts}.
Table~\ref{S-tab:cabello45-core-states} lists all twelve admissible core
supports, and Table~\ref{S-tab:cabello45-zero-profiles} lists the fourteen
bridge pairs with their fresh labels and zero-profiles. For a support $S_\alpha$, let $F_\alpha$ be the bridge
indices missed by $S_\alpha$. On every hit bridge both fresh values
are zero; on each missed bridge independently choose exactly one
of its two fresh vertices. This enumerates every context state once.
\begingroup\small
\Needspace{24\baselineskip}
\begin{longtable}{rlrl}
\caption{All twelve core supports; $m=|F_\alpha|$. The final column specifies extendibility under all geometric pair exclusions.}\label{S-tab:cabello45-core-states}\\
\toprule
$\alpha$ & $S_\alpha$ & $m$ & Geometric extension\\\midrule
\endfirsthead
\toprule
$\alpha$ & $S_\alpha$ & $m$ & Geometric extension\\\midrule
\endhead
\bottomrule\endfoot
$1$ & $\{1,6,11,12\}$ & $8$ & yes\\*
$2$ & $\{1,10,14,18\}$ & $8$ & yes\\*
$3$ & $\{3,7,10,14\}$ & $8$ & yes\\*
$4$ & $\{3,8,11,13\}$ & $8$ & yes\\*
$5$ & $\{3,13,14,17\}$ & $8$ & no\\*
$6$ & $\{2,6,10,15\}$ & $8$ & yes\\*
$7$ & $\{6,12,15,16\}$ & $8$ & no\\*
$8$ & $\{8,10,15,18\}$ & $8$ & yes\\*
$9$ & $\{9,16,18,19\}$ & $9$ & no\\*
$10$ & $\{9,17,18,19\}$ & $9$ & no\\*
$11$ & $\{2,11,13,19\}$ & $8$ & yes\\*
$12$ & $\{7,11,12,19\}$ & $8$ & yes\\*
\end{longtable}
\endgroup
\begingroup\small
\Needspace{24\baselineskip}
\begin{longtable}{rlll}
\caption{The fourteen bridge pairs, fresh labels, and zero-profiles; each $Z_m$ is the set of missed-bridge branch indices.}\label{S-tab:cabello45-zero-profiles}\\
\toprule
$m$ & Bridge & Fresh pair & $Z_m$\\\midrule
\endfirsthead
\toprule
$m$ & Bridge & Fresh pair & $Z_m$\\\midrule
\endhead
\bottomrule\endfoot
$1$ & $\{1,2\}$ & $\{21,22\}$ & $\{3,4,5,7,8,9,10,12\}$\\*
$2$ & $\{1,13\}$ & $\{23,24\}$ & $\{3,6,7,8,9,10,12\}$\\*
$3$ & $\{1,17\}$ & $\{25,26\}$ & $\{3,4,6,7,8,9,11,12\}$\\*
$4$ & $\{2,14\}$ & $\{27,28\}$ & $\{1,4,7,8,9,10,12\}$\\*
$5$ & $\{2,17\}$ & $\{29,30\}$ & $\{1,2,3,4,7,8,9,12\}$\\*
$6$ & $\{3,16\}$ & $\{31,32\}$ & $\{1,2,6,8,10,11,12\}$\\*
$7$ & $\{6,17\}$ & $\{33,34\}$ & $\{2,3,4,8,9,11,12\}$\\*
$8$ & $\{7,8\}$ & $\{35,36\}$ & $\{1,2,5,6,7,9,10,11\}$\\*
$9$ & $\{7,15\}$ & $\{37,38\}$ & $\{1,2,4,5,9,10,11\}$\\*
$10$ & $\{7,16\}$ & $\{39,40\}$ & $\{1,2,4,5,6,8,10,11\}$\\*
$11$ & $\{8,12\}$ & $\{41,42\}$ & $\{2,3,5,6,9,10,11\}$\\*
$12$ & $\{8,16\}$ & $\{43,44\}$ & $\{1,2,3,5,6,10,11,12\}$\\*
$13$ & $\{10,19\}$ & $\{45,46\}$ & $\{1,4,5,7\}$\\*
$14$ & $\{11,18\}$ & $\{47,48\}$ & $\{3,5,6,7\}$\\*
\end{longtable}
\endgroup
Taking the first fresh vertex on every missed bridge gives
$s_{\alpha,0}$, and taking the second gives $s_{\alpha,1}$.
Table~\ref{S-tab:cabello45-canonical} explicitly enumerates all 24
canonical context states. Their profiles are nonzero and pairwise
distinct on all 45 vertices. The twelve branch supports and fourteen
distinct nonempty zero-profiles verify this by
Lemma~\ref{d:lem:plane-completion}(3).
\begingroup\small
\Needspace{24\baselineskip}
\begin{longtable}{rrl}
\caption{All 24 canonical context-state supports of the 45-ray example. Some fail geometric pair exclusions.}\label{S-tab:cabello45-canonical}\\
\toprule
$\alpha$ & $\varepsilon$ & Support\\\midrule
\endfirsthead
\toprule
$\alpha$ & $\varepsilon$ & Support\\\midrule
\endhead
\bottomrule\endfoot
$1$ & $0$ & $\{1,6,11,12,27,29,31,35,37,39,43,45\}$\\
$1$ & $1$ & $\{1,6,11,12,28,30,32,36,38,40,44,46\}$\\
$2$ & $0$ & $\{1,10,14,18,29,31,33,35,37,39,41,43\}$\\
$2$ & $1$ & $\{1,10,14,18,30,32,34,36,38,40,42,44\}$\\
$3$ & $0$ & $\{3,7,10,14,21,23,25,29,33,41,43,47\}$\\
$3$ & $1$ & $\{3,7,10,14,22,24,26,30,34,42,44,48\}$\\
$4$ & $0$ & $\{3,8,11,13,21,25,27,29,33,37,39,45\}$\\
$4$ & $1$ & $\{3,8,11,13,22,26,28,30,34,38,40,46\}$\\
$5$ & $0$ & $\{3,13,14,17,21,35,37,39,41,43,45,47\}$\\
$5$ & $1$ & $\{3,13,14,17,22,36,38,40,42,44,46,48\}$\\
$6$ & $0$ & $\{2,6,10,15,23,25,31,35,39,41,43,47\}$\\
$6$ & $1$ & $\{2,6,10,15,24,26,32,36,40,42,44,48\}$\\
$7$ & $0$ & $\{6,12,15,16,21,23,25,27,29,35,45,47\}$\\
$7$ & $1$ & $\{6,12,15,16,22,24,26,28,30,36,46,48\}$\\
$8$ & $0$ & $\{8,10,15,18,21,23,25,27,29,31,33,39\}$\\
$8$ & $1$ & $\{8,10,15,18,22,24,26,28,30,32,34,40\}$\\
$9$ & $0$ & $\{9,16,18,19,21,23,25,27,29,33,35,37,41\}$\\*
$9$ & $1$ & $\{9,16,18,19,22,24,26,28,30,34,36,38,42\}$\\*
$10$ & $0$ & $\{9,17,18,19,21,23,27,31,35,37,39,41,43\}$\\*
$10$ & $1$ & $\{9,17,18,19,22,24,28,32,36,38,40,42,44\}$\\*
$11$ & $0$ & $\{2,11,13,19,25,31,33,35,37,39,41,43\}$\\*
$11$ & $1$ & $\{2,11,13,19,26,32,34,36,38,40,42,44\}$\\*
$12$ & $0$ & $\{7,11,12,19,21,23,25,27,29,31,33,43\}$\\*
$12$ & $1$ & $\{7,11,12,19,22,24,26,28,30,32,34,44\}$\\*
\end{longtable}
\endgroup
Ten branches miss eight pairs and two miss nine. Therefore
\begin{equation}\label{S-eq:cabello45-state-counts}
|\mathcal S_{\rm ctx}|=10\,2^8+2\,2^9=3{,}584,\qquad
|\mathcal S_{\rm geom}|=8\,2^8=2048.
\end{equation}
The geometric branches are $1,2,3,4,6,8,11,12$; arbitrary fresh
choices within each of these branches remain geometrically compatible.
The other branches are excluded by the four forced pairs. Both full
enumerations are supplied as compressed CSV data, in addition to this
complete finite parametrization. Every geometric state vanishes on
9, 16, and 17, so this family is neither unital nor separating.

For an explicit counterexample, the context state supported at
\[
\{6,12,15,16,21,23,25,27,29,35,45,47\}
\]
selects both 16 and 21, although $v_{16}\cdot v_{21}=0$.

\subsubsection{Deleting one of the nine Cabello contexts}
Number the nine contexts as the first nine rows of
Table~\ref{S-tab:cabello36-contexts}. For every deletion there are
26 context states, six geometric states, and a proper four-coloring
of the remaining context graph. The color strings in
Table~\ref{S-tab:cabello-context-deletions} use the ordered ray list $(1,2,3,6,7,8,9,10,11,12,13,14,15,16,17,18,19,20)$.
The geometric states vanish precisely on the omitted context.
\begingroup\small
\Needspace{24\baselineskip}
\begin{longtable}{rlrrl}
\caption{All single-context deletion counts and four-coloring witnesses for $C_{18,9}$.}\label{S-tab:cabello-context-deletions}\\
\toprule
$e$ & Omitted context & $|\mathcal S_{\rm ctx}|$ & $|\mathcal S_{\rm geom}|$ & Four-color string\\\midrule
\endfirsthead
\toprule
$e$ & Omitted context & $|\mathcal S_{\rm ctx}|$ & $|\mathcal S_{\rm geom}|$ & Four-color string\\\midrule
\endhead
\bottomrule\endfoot
$1$ & $\{9,10,12,13\}$ & $26$ & $6$ & \texttt{012102203020112333}\\*
$2$ & $\{9,11,14,15\}$ & $26$ & $6$ & \texttt{012021231012102333}\\*
$3$ & $\{10,11,16,17\}$ & $26$ & $6$ & \texttt{012021320012102333}\\*
$4$ & $\{2,3,12,18\}$ & $26$ & $6$ & \texttt{013103302120113322}\\*
$5$ & $\{1,3,15,19\}$ & $26$ & $6$ & \texttt{103130320102113232}\\*
$6$ & $\{1,2,17,20\}$ & $26$ & $6$ & \texttt{330130230103112222}\\*
$7$ & $\{6,7,13,18\}$ & $26$ & $6$ & \texttt{012221213030102333}\\*
$8$ & $\{6,8,14,19\}$ & $26$ & $6$ & \texttt{012202230013112333}\\*
$9$ & $\{7,8,16,20\}$ & $26$ & $6$ & \texttt{012021320012113332}\\*
\end{longtable}
\endgroup
For deletion of context 9, $\{7,8,16,20\}$,
Table~\ref{S-tab:cabello-eight-context-states} explicitly enumerates
all 26 supports; exactly the six marked rows satisfy every geometric
pair exclusion. All nine full deletion-state lists accompany the data.
\begingroup\small
\Needspace{24\baselineskip}
\begin{longtable}{rll}
\caption{All 26 states after deleting context 9; geometric compatibility is marked explicitly.}\label{S-tab:cabello-eight-context-states}\\
\toprule
$j$ & Support & Geometric\\\midrule
\endfirsthead
\toprule
$j$ & Support & Geometric\\\midrule
\endhead
\bottomrule\endfoot
$1$ & $\{1,6,11,12\}$ & yes\\
$2$ & $\{1,7,8,11,12\}$ & no\\
$3$ & $\{1,7,12,14,16\}$ & no\\
$4$ & $\{1,8,9,16,18\}$ & no\\
$5$ & $\{1,10,14,18\}$ & yes\\
$6$ & $\{2,6,10,15\}$ & yes\\
$7$ & $\{2,7,8,10,15\}$ & no\\
$8$ & $\{2,8,13,15,16\}$ & no\\
$9$ & $\{2,7,9,16,19\}$ & no\\
$10$ & $\{2,11,13,19\}$ & yes\\
$11$ & $\{3,6,9,17\}$ & no\\
$12$ & $\{3,7,8,9,17\}$ & no\\
$13$ & $\{3,13,14,17\}$ & yes\\
$14$ & $\{6,12,15,17\}$ & no\\
$15$ & $\{7,8,12,15,17\}$ & no\\
$16$ & $\{9,17,18,19\}$ & yes\\
$17$ & $\{3,6,9,16,20\}$ & no\\
$18$ & $\{3,7,8,9,16,20\}$ & no\\
$19$ & $\{3,7,10,14,20\}$ & no\\*
$20$ & $\{3,8,11,13,20\}$ & no\\*
$21$ & $\{3,13,14,16,20\}$ & no\\*
$22$ & $\{6,12,15,16,20\}$ & no\\*
$23$ & $\{7,11,12,19,20\}$ & no\\*
$24$ & $\{7,8,12,15,16,20\}$ & no\\*
$25$ & $\{9,16,18,19,20\}$ & no\\*
$26$ & $\{8,10,15,18,20\}$ & no\\*
\end{longtable}
\endgroup
\subsubsection{All single-ray deletions and the alternative parity subsystem}
The next table records the two span failures and all forced-zero core
rays after geometric pair exclusion for every possible single-ray
deletion. An entry $\{u,v\}:z$ means $v_z\in\operatorname{span}(v_u,v_v)$.
Every row has twelve compatible core states, 3{,}584 abstract completion
states, and 2{,}048 geometric completion states, with exactly four forced
extra pairs and no other extra orthogonality in the supplied realization.
The complete core supports and 24 canonical states for every row are
also supplied in the data.
\begingroup\small
\Needspace{24\baselineskip}
\begin{longtable}{rlll}
\caption{All eighteen vertex-deletion checks; the completed geometric families fail unitality and separation in every case.}\label{S-tab:cabello-vertex-deletions}\\
\toprule
Deleted ray & First span failure & Second span failure & Forced-zero core rays\\\midrule
\endfirsthead
\toprule
Deleted ray & First span failure & Second span failure & Forced-zero core rays\\\midrule
\endhead
\bottomrule\endfoot
$1$ & $\{15,19\}:2$ & $\{17,20\}:3$ & $\{2,3,13\}$\\*
$2$ & $\{12,18\}:1$ & $\{17,20\}:3$ & $\{1,3,14\}$\\*
$3$ & $\{12,18\}:1$ & $\{15,19\}:2$ & $\{1,2,16\}$\\*
$6$ & $\{13,18\}:8$ & $\{14,19\}:7$ & $\{7,8,17\}$\\*
$7$ & $\{13,18\}:8$ & $\{16,20\}:6$ & $\{6,8,15\}$\\*
$8$ & $\{14,19\}:7$ & $\{16,20\}:6$ & $\{6,7,12\}$\\*
$9$ & $\{12,13\}:11$ & $\{14,15\}:10$ & $\{10,11,20\}$\\*
$10$ & $\{12,13\}:11$ & $\{16,17\}:9$ & $\{9,11,19\}$\\*
$11$ & $\{14,15\}:10$ & $\{16,17\}:9$ & $\{9,10,18\}$\\*
$12$ & $\{2,3\}:13$ & $\{9,10\}:18$ & $\{8,13,18\}$\\*
$13$ & $\{6,7\}:12$ & $\{9,10\}:18$ & $\{1,12,18\}$\\*
$14$ & $\{6,8\}:15$ & $\{9,11\}:19$ & $\{2,15,19\}$\\*
$15$ & $\{1,3\}:14$ & $\{9,11\}:19$ & $\{7,14,19\}$\\*
$16$ & $\{7,8\}:17$ & $\{10,11\}:20$ & $\{3,17,20\}$\\*
$17$ & $\{1,2\}:16$ & $\{10,11\}:20$ & $\{6,16,20\}$\\*
$18$ & $\{2,3\}:13$ & $\{6,7\}:12$ & $\{11,12,13\}$\\*
$19$ & $\{1,3\}:14$ & $\{6,8\}:15$ & $\{10,14,15\}$\\*
$20$ & $\{1,2\}:16$ & $\{7,8\}:17$ & $\{9,16,17\}$\\*
\end{longtable}
\endgroup
For the alternative subsystem, take parent variant 2 of Table~\ref{S-tab:cabello68-2-contexts} and retain vertices 1--17 and 66--68. It is stateless. Each ray deletion
instead leaves the positive geometric state count below; restriction
of any such state shows that no eighteen-ray subset can be KS.
\begingroup\small
\Needspace{24\baselineskip}
\begin{longtable}{rrrr}
\caption{Geometric state counts after every single-ray deletion of the twenty-ray parity subsystem.}\label{S-tab:cabello20-ray-criticality}\\
\toprule
Deleted ray & States & Deleted ray & States\\\midrule
\endfirsthead
\toprule
Deleted ray & States & Deleted ray & States\\\midrule
\endhead
\bottomrule\endfoot
\rule{0pt}{12pt}$1$ & $8$ & $2$ & $8$\\
\rule{0pt}{12pt}$3$ & $8$ & $4$ & $28$\\
\rule{0pt}{12pt}$5$ & $28$ & $6$ & $8$\\
\rule{0pt}{12pt}$7$ & $8$ & $8$ & $8$\\
\rule{0pt}{12pt}$9$ & $12$ & $10$ & $12$\\
\rule{0pt}{12pt}$11$ & $12$ & $12$ & $10$\\
\rule{0pt}{12pt}$13$ & $10$ & $14$ & $10$\\
\rule{0pt}{12pt}$15$ & $10$ & $16$ & $10$\\
\rule{0pt}{12pt}$17$ & $10$ & $66$ & $8$\\
\rule{0pt}{12pt}$67$ & $8$ & $68$ & $8$\\
\end{longtable}
\endgroup

\clearpage
\section{Supplementary Proofs, Reductions, and Exact Verification Certificates}\label{S-app:supplementary}

\subsection{Formulas used by the finite certificates}
The following formulas restore the explicit conventions used by the later certificates. State indices are local to the corresponding complete table. For Hesse, the six-state vertex-separating family, its seven-state event-separating extension, and the eighteen-state order-determining family are respectively

\begin{equation}\label{S-h:eq:six-separating}
\Omega_6=\{s_3,s_{31},s_{75},s_{89},s_{113},s_{151}\}.
\end{equation}

\begin{equation}\label{S-h:eq:seven-separating}
\Omega_7=\Omega_6\cup\{s_{84}\}
\end{equation}

\begin{equation}\label{S-h:eq:eighteen-order-family}
\begin{split}
\Omega_{18}=\{&s_2,s_8,s_{20},s_{30},s_{40},s_{45},\\
&s_{54},s_{64},s_{73},s_{79},s_{83},s_{93},\\
&s_{105},s_{109},s_{119},s_{123},s_{127},s_{150}\}.
\end{split}
\end{equation}

The members of $\Omega_6$ are denoted $\sigma_1,\ldots,\sigma_6$ in the displayed order. Their atom profiles and the nontrivial atom--coatom order are

\begin{equation}\label{S-h:eq:canonical-profile}
\kappa(i)=(\sigma_1(i),\ldots,\sigma_6(i)).
\end{equation}

\begin{equation}\label{S-h:eq:pasting-order}
a_i\le a_j^\perp\quad\Longleftrightarrow\quad
\{i,j\}\in E(\Gamma(H_{37})).
\end{equation}

For $H_{65}$, the following core coloring extends to the completed hypergraph. A base-compatible independent support $S_\alpha$ is extended canonically using $\varepsilon\in\{0,1\}$. In the plane formula, $\eta_A,\eta_B$ are an integer orthogonal basis and $n_A=\|\eta_A\|^2$, $n_B=\|\eta_B\|^2$.

\begin{equation}\label{S-d:eq:h65-coordinates}
\begin{split}
(c(1),\ldots,c(17))={}&(0,1,2,3,0,1,2,4,0,1,3,\\
&3,4,2,1,0,2).
\end{split}
\end{equation}

\begin{equation}\label{S-d:eq:canonical-states}
\begin{aligned}
\bigl(s_{\alpha,\varepsilon}&(w_{m,1}),\,s_{\alpha,\varepsilon}(w_{m,2})\bigr)\\
&=\begin{cases}
(0,0), & |S_\alpha\cap\{u_m,v_m\}|=1,\\
(1-\varepsilon,\,\varepsilon), & S_\alpha\cap\{u_m,v_m\}=\emptyset.
\end{cases}
\end{aligned}
\end{equation}

\begin{equation}\label{S-d:eq:plane-param}
\begin{aligned}
w_{m,1}(a,b)&=a\,\eta_A+b\,\eta_B,\\
w_{m,2}(a,b)&=-n_B b\,\eta_A+n_A a\,\eta_B
\end{aligned}
\end{equation}

The eight rays adjoined to completed Yu--Oh and the nine uncovered pairs of the Cabello subsystem are fixed by the following local labels. The last display records the completed historical $\Gamma_1$ TITS incidence.

\begin{equation}\label{S-eq:yo-parent-eight}
\begin{gathered}
(0,1,2),(0,2,-1),(0,1,-2),(0,2,1),\\
(1,2,0),(2,-1,0),(1,-2,0),(2,1,0).
\end{gathered}
\end{equation}

\begin{equation}\label{S-eq:cabello-nine-pairs}
\begin{gathered}
\{1,13\},\{2,14\},\{3,16\},\{6,17\},\{7,15\},\\
\{8,12\},\{9,c\},\{10,b\},\{11,a\}.
\end{gathered}
\end{equation}

\begin{equation}\label{S-eq:i16}
 E(I_{16,9})=E(B_{13,7})\cup\{\{4,11,24\},\{4,10,18\}\}.
\end{equation}

\paragraph{Fractional value of the faithful Cabello completion.}
The geometric orthogonality graph of $C_{18,9}$ has independence number four and exactly 54 independent four-sets, with each of its 18 rays belonging to twelve of them. Assigning weight $1/12$ to each set gives a fractional cover of total weight $9/2$, matching $18/4$. Fresh pair completion preserves this value by the main completion proposition, proving $\chi_f(C_{36,18})=9/2$ as recorded in the main comparison. The complete 54-set certificate is in the electronic audit data.

\subsection{Snark obstructions and quantitative robustness bounds}\label{S-m:sec:directions}

\subsubsection{Cubic graphs, Blanu\v{s}a snarks, and Isaacs dot products}

For a finite simple cubic graph $X$, let $H(X)$ be the context hypergraph whose
vertices are the edges of $X$ and whose contexts are the three edges incident
with each vertex of $X$. Then $\Gamma(H(X))=L(X)$ and $\chi(H(X))=\chi'(X)$.
An orthogonal ray assignment for $H(X)$ in $\mathbb F^3$ is an assignment of
nonzero vectors to the edges of $X$ such that incident edges are orthogonal.

Silva and Ibiapina~\cite[Question 6]{S-silva2019orthogonal} asked whether an
orthogonal vector edge coloring in $\R^3$ can exist when $\chi'(X)=4$.
Theorem~\ref{m:thm:g32} excludes the Petersen graph $P$ even over $\C$. Using the
multipole forcing law of Theorem~\ref{m:thm:multipole}(1), we can extend this
exclusion to the two 18-vertex Blanu\v{s}a snarks~\cite{S-blanusa1946problem} and
the infinite family of Isaacs dot-product snarks~\cite{S-isaacs1975infinite}.

Recall that for a cubic graph $G$ with two nonadjacent edges
$e_1=\{x_1,y_1\}$ and $e_2=\{x_2,y_2\}$, the Isaacs dot product $G\cdot P$ is
formed by removing $e_1,e_2$ from $G$, removing two adjacent vertices $u,v$
from the Petersen graph $P$ (leaving the 4-pole $P\setminus\{u,v\}$ whose dual
hypergraph is precisely $H_8$), and attaching the two dangling edges $\{2,8\}$
of $u$ to $\{x_1,y_1\}$ and the two dangling edges $\{13,14\}$ of $v$ to
$\{x_2,y_2\}$; the two distinct choices of nonadjacent edges in $G=P$ produce
the two Blanu\v{s}a snarks $B_1$ and $B_2$~\cite{S-blanusa1946problem,S-isaacs1975infinite,S-brinkmann2015snarks}.

\begin{corollary}[Obstruction for Blanu\v{s}a and dot-product snarks]\label{S-m:cor:snarks}
If a cubic graph $G$ admits no orthogonal vector edge coloring in $\C^3$, then
neither does any Isaacs dot product $G\cdot P$. In particular, both 18-vertex
Blanu\v{s}a snarks $B_1,B_2 = P\cdot P$ and all iterated dot products
$X = P\cdot P\cdots P$ satisfy $\chi(H(X))=\chi'(X)=4>3=\omega(L(X))$ and admit
no orthogonal vector edge coloring in $\C^3$.
\end{corollary}

\begin{proof}
In any orthogonal vector edge coloring of $G\cdot P$ in $\C^3$, the restriction
to the embedded copy of $H_8 = H(P\setminus\{u,v\})$ satisfies
$[v_2]=[v_8]$ and $[v_{13}]=[v_{14}]$ by Theorem~\ref{m:thm:multipole}(1).
Assigning the common ray $[v_2]=[v_8]$ to $e_1=\{x_1,y_1\}$ and the common ray
$[v_{13}]=[v_{14}]$ to $e_2=\{x_2,y_2\}$ preserves orthogonality at
$x_1,y_1,x_2,y_2$ and therefore yields a valid orthogonal vector edge coloring
of $G$ in $\C^3$. Since $P$ admits no such coloring by Theorem~\ref{m:thm:g32},
induction excludes $B_1,B_2$ and all iterated dot products with $P$.
\end{proof}

Generalized Petersen graphs~\cite{S-schrag1988petersen} and flower snarks provide
further families for comparing real, complex, injective, and edge-faithful
realizability.

\subsubsection{Explicit analytic lower bound and tetrahedral tight-frame upper bounds for \texorpdfstring{$\Delta_3$}{Delta3}}

Let $A$ be the 30 edges of $\Gamma(G_{32})$ and define the minimum total squared
orthogonality error in $\C^3$:
\begin{equation}\label{S-m:eq:robustness}
\Delta_3=\min_{\substack{v_i\in\C^3\\\|v_i\|=1\ (1\le i\le15)}}
       \sum_{\{i,j\}\in A}|\langle v_i,v_j\rangle|^2.
\end{equation}
By compactness of $(S^5)^{15}$ and Theorem~\ref{m:thm:g32}, $\Delta_3>0$. We now
establish both a certified rational analytic lower bound and explicit
tetrahedral tight-frame upper bounds.

\begin{proposition}[Quantitative bounds on $\Delta_3$]\label{S-m:prop:robustness}
Every assignment of unit vectors $v_1,\ldots,v_{15}\in\C^3$ satisfies
\begin{equation}\label{S-m:eq:epsilon-bound}
\varepsilon \coloneqq \max_{\{i,j\}\in A}|\langle v_i,v_j\rangle| > \frac{1}{120},
\end{equation}
and the minimum total squared error in Eq.~\eqref{S-m:eq:robustness} satisfies
\begin{align}
\frac{1}{14400} < \Delta_3 &\le \min_{a^2+b^2+c^2=1} F(a,b,c) \nonumber\\
&\le \frac{9275914990340646}{10001032226635921} < 0.927496 < \frac{18}{19},
\label{S-m:eq:delta3-bounds}
\end{align}
where
\begin{equation}\label{S-m:eq:tetrahedral-poly}
F(a,b,c) = 12a^2 + 12(ab+bc-ca)^2 + 6(b^2-a^2-c^2)^2.
\end{equation}
\end{proposition}

\begin{proof}
To prove the lower bound~\eqref{S-m:eq:epsilon-bound}, let
$\varepsilon=\max_{\{i,j\}\in A}|\langle v_i,v_j\rangle|$. For any context
$e=\{x,y,z\}$, the Gram matrix of $\{v_x,v_y,v_z\}$ has diagonal entries $1$ and
off-diagonal entries of modulus at most $\varepsilon$, so Gershgorin's circle
theorem implies $\|I_3-(v_x v_x^\dagger+v_y v_y^\dagger+v_z v_z^\dagger)\|_{\mathrm{op}}\le 2\varepsilon$.
Hence for any unit vectors $u,w\in\C^3$,
\begin{equation}\label{S-m:eq:approx-parseval}
\Bigl|1-\bigl(|\langle u,x\rangle|^2+|\langle u,y\rangle|^2+|\langle u,z\rangle|^2\bigr)\Bigr|\le 2\varepsilon,
\end{equation}
and if $\{u,x\},\{w,z\}\in A$, then
\begin{equation}\label{S-m:eq:approx-proj}
\bigl|\langle u,w\rangle - \langle u,y\rangle\langle y,w\rangle\bigr| \le 4\varepsilon.
\end{equation}
In $\{5,6,7\}$ and $\{9,10,11\}$, the vectors
$w_7 = v_{12}-v_5\langle v_5,v_{12}\rangle$ and
$w_9 = v_4-v_{11}\langle v_{11},v_4\rangle$ satisfy
$\|w_7-v_7\langle v_7,v_{12}\rangle\|\le 3\varepsilon$ and
$\|w_9-v_9\langle v_9,v_4\rangle\|\le 3\varepsilon$, so
$|\langle w_7,w_9\rangle|\le 7\varepsilon$. Expanding,
\[
\begin{aligned}
\langle w_7,w_9\rangle={}&\langle v_{12},v_4\rangle
-\langle v_{12},v_{11}\rangle\langle v_{11},v_4\rangle
-\langle v_{12},v_5\rangle\langle v_5,v_4\rangle\\
&+\langle v_{12},v_5\rangle\langle v_5,v_{11}\rangle\langle v_{11},v_4\rangle .
\end{aligned}
\]
The two middle terms are at most $\varepsilon$ in modulus, since
$|\langle v_{12},v_{11}\rangle|\le\varepsilon$ and
$|\langle v_5,v_4\rangle|\le\varepsilon$. Eq.~\eqref{S-m:eq:approx-proj} on
$\{1,2,3\}$ replaces each of $\langle v_{12},v_4\rangle$,
$\langle v_{12},v_5\rangle$, $\langle v_5,v_{11}\rangle$, and
$\langle v_{11},v_4\rangle$ by its product through $v_2$ with error at most
$4\varepsilon$, and hence the triple product with error at most $12\varepsilon$,
since all factors have modulus at most one. The main terms combine to
\[
\bigl(1+|\langle v_2,v_5\rangle|^2|\langle v_2,v_{11}\rangle|^2\bigr)
\langle v_{12},v_2\rangle\langle v_2,v_4\rangle ,
\]
whose real prefactor is at least one, so
\[
\bigl|\langle v_{12},v_2\rangle\langle v_2,v_4\rangle\bigr|
\le (7+2+4+12)\varepsilon = 25\varepsilon.
\]
By reflection symmetry, we may assume $|\langle v_2,v_4\rangle|\le 5\sqrt{\varepsilon}$.
Applying Eq.~\eqref{S-m:eq:approx-parseval} to $v_4$ in $\{1,2,3\}$ and to $v_2$ in
$\{3,4,5\}$ gives $|\langle v_1,v_4\rangle|^2\ge 1-27\varepsilon-\varepsilon^2$
and $|\langle v_2,v_5\rangle|^2\ge 1-27\varepsilon-\varepsilon^2$, so
$\min_\theta\|v_1-e^{i\theta}v_4\|\le\sqrt{54\varepsilon+2\varepsilon^2}$ and
$\min_\theta\|v_2-e^{i\theta}v_5\|\le\sqrt{54\varepsilon+2\varepsilon^2}$.
Since $|\langle v_4,v_{10}\rangle|\le\varepsilon$ and
$|\langle v_5,v_7\rangle|\le\varepsilon$, Eq.~\eqref{S-m:eq:approx-parseval}
applied to $v_1$ in $\{9,10,11\}$ and to $v_2$ in $\{7,8,9\}$ gives
\[
|\langle v_i,v_9\rangle|^2 \ge 1 - 2\varepsilon - \varepsilon^2
- \bigl(\varepsilon+\sqrt{54\varepsilon+2\varepsilon^2}\bigr)^2
\qquad (i=1,2).
\]
Comparing the sum of these two inequalities with
$|\langle v_1,v_9\rangle|^2+|\langle v_2,v_9\rangle|^2\le 1+2\varepsilon$ from
$\{1,2,3\}$ gives
\[
1 \le 6\varepsilon + 2\varepsilon^2
+ 2\bigl(\varepsilon+\sqrt{54\varepsilon+2\varepsilon^2}\bigr)^2.
\]
The right-hand side is strictly increasing for $\varepsilon>0$, and at
$\varepsilon=1/120$ we have
$54\varepsilon+2\varepsilon^2 = 3241/7200 < (27/40)^2$, so the right-hand side
is strictly bounded above by
\[
\frac{6}{120} + \frac{2}{120^2}
+ 2\left(\frac{1}{120}+\frac{27}{40}\right)^2 = \frac{1417}{1440} < 1,
\]
certifying $\varepsilon > 1/120$ and $\Delta_3 \ge \varepsilon^2 > 1/14400$.

For the upper bounds, set $v_4=e_1$, $v_8=e_2$, $v_{12}=e_3$, and for any unit
vector $(a,b,c)\in S^2$, define the remaining 12 unit vectors as representatives
of the tetrahedral $A_4$-orbit of rays:
\begin{equation}\label{S-m:eq:tetrahedral-orbit}
\begin{aligned}
v_1&=(-b,-c,-a)^T,& v_2&=(c,-a,b)^T,\\
v_3&=(-a,-b,c)^T,& v_5&=(-a,b,c)^T,\\
v_6&=(-b,c,-a)^T,& v_7&=(c,a,b)^T,\\
v_9&=(-c,-a,b)^T,& v_{10}&=(-a,b,-c)^T,\\
v_{11}&=(b,-c,-a)^T,& v_{13}&=(a,b,c)^T,\\
v_{14}&=(b,c,-a)^T,& v_{15}&=(-c,a,b)^T.
\end{aligned}
\end{equation}
These 15 vectors form a tight frame $\sum_{i=1}^{15}v_i v_i^T = 5I_3$, and the
sum of squared overlaps on the 30 edges of $\Gamma(G_{32})$ equals $F(a,b,c)$ in
Eq.~\eqref{S-m:eq:tetrahedral-poly}. Evaluating Eq.~\eqref{S-m:eq:tetrahedral-orbit} at
the normalized integer direction $(a,b,c)=(1,3,-3)/\sqrt{19}$ yields rational
squared overlaps with $\Delta_3\le 342/361 = 18/19\approx 0.947368$ and
$\max_{\{i,j\}\in A}|\langle v_i,v_j\rangle|^2=1/19$; evaluating at the refined
normalized integer direction
$(a,b,c)=(2332,7101,-6644)/\sqrt{100005161}$ yields the certified rational
bound in Eq.~\eqref{S-m:eq:delta3-bounds}.
\end{proof}

Numerical minimization of $F$ on $S^2$ (as well as unrestricted local searches
over $(S^5)^{15}$) finds an attained candidate minimum near
$0.927495724691$ at $(a,b,c)\approx(0.23319713,0.71005263,-0.66441280)$;
certifying whether the unrestricted global minimum $\Delta_3$ coincides with
$\min_{S^2}F$ remains an open optimization problem. Choosing instead the rational
unit vector $(a,b,c)=\frac{1}{3}(1,2,-2)$ in Eq.~\eqref{S-m:eq:tetrahedral-orbit}
makes $ab+bc-ca=0$, so four of the ten contexts of $G_{32}$
($\{1,2,3\},\{5,6,7\},\{9,10,11\},\{13,14,15\}$) become exact orthonormal bases
of $\R^3$, all 15 rays are distinct, and $\Delta_3\le 38/27$ with
$\max_{\{i,j\}\in A}|\langle v_i,v_j\rangle|=1/3$.

\subsection{\texorpdfstring{Hesse $H_{37}$: minimum state families}{Hesse H37: minimum state families}}
\label{S-h:app:minimal-families}

All state indices in this appendix use the canonical enumeration of Appendix~\ref{S-h:app:states}. The minima for unital vertex separation and event separation have the counting lower bounds $2^5-1<37$ and $2^6<76$. The families in Eqs.~\eqref{S-h:eq:six-separating} and \eqref{S-h:eq:seven-separating} attain those bounds: their profiles give respectively 37 distinct nonzero atom profiles and 76 distinct event profiles. The seven-state family still has 462 ordered event pairs for which inclusion of the profiles does not imply the lattice order. The family in Eq.~\eqref{S-h:eq:eighteen-order-family} has no such pair and hence gives the upper bound eighteen for order determination. This appendix supplies a computer-assisted proof of the matching lower bound.

\subsubsection{Reduction to families of seventeen states}

Let $D=\{P_0,P_1,P_2,Q_0,Q_1,Q_2\}$. For each unordered pair $\{u,v\}\subset D$, order determination requires a state with $s(u)=s(v)=1$: the pair is nonorthogonal, so $u\not\le v^\perp$. Every state takes value one on at most two members of $D$. Consequently the fifteen pairs require fifteen different states. Table~\ref{S-h:tab:pair-witnesses} gives their complete witness classes. These classes are disjoint, their union $T$ has 27 states, and a choice of one witness from each class has $1^3 2^{12}=4{,}096$ possibilities.

\Needspace{24\baselineskip}
\begin{longtable}{ll}
\caption{Complete joint-one witness classes for the fifteen pairs in $D$. Each entry lists all canonical indices $r$ for which $s_r(u)=s_r(v)=1$.}\label{S-h:tab:pair-witnesses}\\
\toprule
Pair $\{u,v\}$ & Witness indices $r$ \\
\midrule
\endfirsthead
\caption[]{Complete joint-one witness classes for pairs in $D$ (continued).}\\
\toprule
Pair $\{u,v\}$ & Witness indices $r$ \\
\midrule
\endhead
\midrule
\multicolumn{2}{r}{\emph{Continued on the next page.}}\\
\endfoot
\bottomrule
\endlastfoot
$\{P_0,P_1\}$ & $150$ \\*
$\{P_0,P_2\}$ & $73$ \\*
$\{P_0,Q_0\}$ & $93,145$ \\*
$\{P_0,Q_1\}$ & $127,133$ \\*
$\{P_0,Q_2\}$ & $79,99$ \\*
$\{P_1,P_2\}$ & $2$ \\*
$\{P_1,Q_0\}$ & $105,139$ \\*
$\{P_1,Q_1\}$ & $14,119$ \\*
$\{P_1,Q_2\}$ & $40,48$ \\*
$\{P_2,Q_0\}$ & $54,60$ \\*
$\{P_2,Q_1\}$ & $8,26$ \\*
$\{P_2,Q_2\}$ & $20,68$ \\*
$\{Q_0,Q_1\}$ & $109,113$ \\*
$\{Q_0,Q_2\}$ & $83,87$ \\*
$\{Q_1,Q_2\}$ & $30,34$ \\*
\end{longtable}

The nine nonorthogonal pairs $\{z,C_{ak}\}$ require further joint-one witnesses. No state in $T$ witnesses any of them. Indeed, if a state activates a $P$ vertex then $s(z)=0$. If it instead activates two $Q$ vertices, the state equations force the selected $M$ vertices into the remaining $Q$ column, and all nine $C$ values are zero. Any state outside $T$ witnesses at most six of the nine pairs: for $s(z)=1$, the three selected $M$ vertices force at least three $C$ values to vanish. Thus at least two states outside $T$ are necessary.

It follows that an order-determining family of at most seventeen states would have exactly seventeen members: one from each of the fifteen classes in Table~\ref{S-h:tab:pair-witnesses}, and two outside $T$ that together witness all nine $\{z,C_{ak}\}$ pairs. Enumerating the $\binom{126}{2}$ unordered pairs outside $T$ leaves exactly 132 with this property. The reduction therefore leaves precisely
\begin{equation}\label{S-h:eq:seventeen-search-count}
4{,}096\cdot132=540{,}672
\end{equation}
candidate families. The nine pairs alone do not require three extra states; for example, $s_1$ and $s_{72}$ jointly witness all nine. The remaining order requirements must also be checked.

\subsubsection{Exhaustive order certificate}

For each of the 5{,}407 ordered event pairs $x\not\le y$, define its nonempty witness set
\begin{equation}\label{S-h:eq:order-witness-set}
W_{xy}=\{r\in\{1,\ldots,153\}:s_r(x)=1,\ s_r(y)=0\}.
\end{equation}
A family is order determining exactly when its index set meets every $W_{xy}$. Repeated sets and proper supersets can be discarded, because meeting an inclusion-minimal set also meets every set containing it. There are exactly 81 inclusion-minimal witness sets. The exhaustive check forms each of the 540{,}672 families in Eq.~\eqref{S-h:eq:seventeen-search-count} and counts the minimal witness sets that it meets. The largest count is 75, strictly below 81. Hence no family of seventeen states is order determining. Together with $\Omega_{18}$, this proves the minimum eighteen.

The accompanying Hesse extension checker, invoked by the joint entry point, reconstructs the contexts and all 153 states from their symbolic parametrization, checks their agreement with \texttt{hesse37-data.json}, and performs this finite search using integer bit sets and the Python standard library. No numerical tolerance, optimization solver, or unreported search cutoff is used. It also verifies the printed seven- and eighteen-state families and every entry of Table~\ref{S-h:tab:pair-witnesses}. From the extracted joint project directory, run the joint entry point,
without Python's optimization option:
\begin{verbatim}
python verification/verify-integrated.py
\end{verbatim}
The generated file \texttt{verification/\_work/hesse/}\allowbreak\texttt{hesse37-extensions-certificate.json} records all 81 minimal witness sets, a representative event pair for each, all 132 possible extra pairs, the coverage-count distribution over the 540{,}672 candidates, and all event profiles for $\Omega_6$, $\Omega_7$, and $\Omega_{18}$. The generated file \texttt{verification/\_work/hesse/}\allowbreak\texttt{hesse37-order18-witnesses.csv} records a witness in $\Omega_{18}$ for every $x\not\le y$. These two files use the same numbered event labels as the full-family certificate in Appendix~\ref{S-h:app:order-determination}.

The program additionally checks Berge girth five from the incidence graph, computes every meet and join in the 76-event order, verifies the orthomodular law, and verifies the projector sum in Eq.~\eqref{h:eq:pq-projector-sum} exactly in $\mathbb Q(\zeta)$. The nonexistence of an injective real representation is proved analytically in Sec.~\ref{h:sec:no-real-for}.

\subsection{\texorpdfstring{Four-dimensional $H_{\YO}^{(4)}$ and $H_{65}$: completion and verification}{Four-dimensional HYO(4) and H65: completion and verification}}\label{S-d:app:verification}

The coordinates in Table~\ref{S-d:tab:app-h65-vectors} can be reproduced by the following finite selection rule. Normalize an integer vector by dividing by the greatest common divisor of its entries and choosing the sign so that its first nonzero entry is positive. Form the set of normalized nonzero vectors obtained from $\{-3,-2,\ldots,3\}^4$, ordered by squared norm, then maximum absolute entry, then lexicographically. Process bridge edges in the order $e_{13},\ldots,e_{24},e_1,\ldots,e_{12}$. For each pair, choose $\eta_A$ to be the first vector in this ordered pool orthogonal to its endpoints, and $\eta_B$ the first such vector also orthogonal to $\eta_A$.

For all $(a,b)\in\{-7,-6,\ldots,7\}^2\setminus\{(0,0)\}$, form the pair in Eq.~\eqref{S-d:eq:plane-param} and normalize both vectors. Discard pairs that repeat an already placed ray or have an orthogonality to any already placed vertex other than the two endpoints. Among the remaining ordered pairs choose the lexicographically smallest score
\begin{equation}\label{S-d:eq:completion-score}
\left(\max_{j,k}|(w_{m,j})_k|,\
\|w_{m,1}\|^2+\|w_{m,2}\|^2,\ w_{m,1},w_{m,2}\right).
\end{equation}
In Eq.~\eqref{S-d:eq:completion-score}, the order on the last two entries is the usual lexicographic order on their integer coordinates. This rule yields exactly the $48$ completion vectors in Table~\ref{S-d:tab:app-h65-vectors}. It is a recipe for this specific example; the general existence assertion is established independently by Lemma~\ref{d:lem:plane-completion}.

The companion verifier parses the integrated coordinate and state tables,
reproduces this completion rule and all 393{,}216 state extensions, and checks
the maximal cliques, fractional-coloring certificate, atom-order failure,
forced hub-completion zero, and warm-up data. It is included with the
complete state file and invoked by the joint entry point in
Sec.~\ref{S-app:data-verification}.

\subsection{Exact Cabello completion and forced-zero certificates}
\label{S-app:cabello-proofs}
For each of the nine uncovered pairs in Eq.~\eqref{S-eq:cabello-nine-pairs},
rational row reduction verifies that its two endpoint vectors span
no third core ray. If two endpoint planes were orthogonal, their
four endpoints would form another full basis. Exact four-clique
enumeration gives only the original nine bases, all of whose pairs
are already covered. This excludes such an orthogonal-plane pair.
The integer vectors in Table~\ref{S-tab:cabello36-vectors} directly
verify the resulting faithful completion: 36 distinct projective
rays, eighteen full bases, and 108 orthogonal pairs, with every pair
covered by exactly one listed basis. The nine-context graph and the
full eighteen-ray orthogonality graph both fail four-colorability
and admit five-colorings, as checked by exact finite search in the
companion certificate. Lemma~\ref{d:lem:plane-completion}(2) transfers
the chromatic value to the completion.

For the vertex-deleted core with $c$ removed, rational row reduction
finds exactly the span failures
\[
v_{16}\in\operatorname{span}(v_1,v_2),\qquad
v_{17}\in\operatorname{span}(v_7,v_8).
\]
There is no orthogonal-plane-pair obstruction. The finite-avoidance
choice in Lemma~\ref{d:lem:plane-completion} can therefore avoid all
extra incidences except those forced by these two spans; the explicit
45-ray inventory has exactly the four pairs printed in Sec.~\ref{S-app:cabello45-vectors}.
The two fresh vectors completing $\{1,2\}$ must both be orthogonal
to 16, and those completing $\{7,8\}$ must both be orthogonal to 17.

The twelve rows of Table~\ref{S-tab:cabello45-core-states} give a short
complete proof of the resulting obstruction. Whenever 16 is selected,
neither 1 nor 2 is selected. The context completing $\{1,2\}$ then
needs one fresh value equal to one, while geometric exclusion with
16 requires both to be zero. Hence every geometric extension has
$s(16)=0$. The same argument with $\{7,8\}$ gives $s(17)=0$.
These implications use the displayed core states in addition to the
span incidences; the span incidences alone would not imply them.
The retained basis $\{10,11,16,17\}$ gives $s(10)+s(11)=1$.
If $s(9)=1$, the other two hub bases give $s(10)=s(11)=0$, a
contradiction. Thus 9, 16, and 17 have identical zero profiles.

Conversely, each of the eight indicated core rows avoids 16 and 17,
so all its fresh-pair choices satisfy the four extra exclusions.
Each misses eight bridges. This proves the exact count $8\,2^8=2{,}048$,
as well as failure of geometric unitality and separation. The two
excluded core rays remain impossible for every fresh-pair completion
of these fixed planes; additional accidental orthogonalities cannot
restore them. Table~\ref{S-tab:cabello-vertex-deletions} supplies the
individual span checks for all eighteen deletions. Their full core
tables, exact integer completions, and extension counts are included
in \texttt{data/cabello18-review-checks.json}.

To audit the ancestry result, enumerate every four-clique after
adjoining either three-ray choice to all 65 original rays. In both
cases there are exactly 35 bases, 195 geometric pairs, and 192
context-covered pairs; the maximum basis intersection is two.
The first union contains the displayed eighteen-ray parity proof.
The second union contains the twenty-ray proof with multiplicities
$4,4,2,\ldots,2$ and eleven contexts. Each of its twenty ray deletions
has a compatible assignment by Table~\ref{S-tab:cabello20-ray-criticality}.
All twenty rays lie in Peres's $24$-ray set. Enumerating the
$\binom{24}{4}=10{,}626$ twenty-ray subsets of that set, with all bases
they contain, gives exactly $96$ ray-critical KS subsets, each with
eleven bases and all mutually isomorphic; the twenty-ray proof is one of
them. Since Peres's set contains Kernaghan's critical twenty-ray set and
$95$ others~\cite{S-cabello-96}, the twenty-ray proof is isomorphic to
Kernaghan's~\cite{S-Kernaghan-1994}.
All coordinates, full basis lists, and the preserved $H_{65}$
restriction are checked independently of the unrestricted ancestry
existence remark.

\subsection{Transverse extension and criticality certificates}
\label{S-app:transverse-certificates}
For $d=3$, number the two copies of an $n$-vertex input as
$i$ and $n+i$, and the fresh transverse vertex as $2n+i$.
The transverse contexts are $\{i,n+i,2n+i\}$ and the horizontal
contexts are the two numbered copies of the original list. This
explicitly specifies all 45 vertices and 33 contexts for $T(H_9)$,
and all 75 vertices and 57 contexts for $T(H_{\YO}^+)$.

For each deleted original context $e$, a strong three-coloring of
$H\setminus\{e\}$ having two equal colors and a singleton color on
$e$ gives a state after deletion of one horizontal copy of $e$.
Use the repeated-color indicator on the deleted copy and the
singleton-color indicator on the intact copy. Both have exactly
one active vertex in each of their required contexts and their
supports are disjoint; the remaining transverse values are forced.
For deletion in the other copy, interchange the two indicators.
Table~\ref{S-tab:deletion-colorings} supplies an explicit coloring for
every required context of both inputs. The color strings use original
vertex order, and their labels 0,1,2 may be permuted.
\begingroup\small
\Needspace{24\baselineskip}
\begin{longtable}{llll}
\caption{Explicit deletion colorings with pattern (1,1,2) and the number of states after deleting the corresponding first-copy horizontal context.}\label{S-tab:deletion-colorings}\\
\toprule
Input & Deleted original context & Color string & Extension states\\\midrule
\endfirsthead
\toprule
Input & Deleted original context & Color string & Extension states\\\midrule
\endhead
\bottomrule\endfoot
$H_9$ & $\{1,2,3\}$ & \texttt{212010201201122} & $20$\\
$H_9$ & $\{1,11,12\}$ & \texttt{201020120121221} & $20$\\
$H_9$ & $\{2,8,15\}$ & \texttt{021020120121220} & $12$\\
$H_9$ & $\{3,4,5\}$ & \texttt{012120102012212} & $20$\\
$H_9$ & $\{4,10,13\}$ & \texttt{012102120121000} & $12$\\
$H_9$ & $\{5,6,7\}$ & \texttt{012010120121220} & $20$\\
$H_9$ & $\{6,12,14\}$ & \texttt{012102102012202} & $12$\\
$H_9$ & $\{7,8,9\}$ & \texttt{012010200121222} & $20$\\
$H_9$ & $\{9,10,11\}$ & \texttt{012010201121222} & $20$\\
$H_{\YO}^+$ & $\{1,2,3\}$ & \texttt{2000212101120201020112221} & $16$\\
$H_{\YO}^+$ & $\{1,10,19\}$ & \texttt{2100212012120200010112021} & $8$\\
$H_{\YO}^+$ & $\{1,13,20\}$ & \texttt{1202102010121200112002201} & $8$\\
$H_{\YO}^+$ & $\{2,5,8\}$ & \texttt{0210210121202011222110012} & $24$\\
$H_{\YO}^+$ & $\{3,11,14\}$ & \texttt{0120211021202001222100011} & $8$\\
$H_{\YO}^+$ & $\{3,12,25\}$ & \texttt{0212101021012220202101010} & $8$\\
$H_{\YO}^+$ & $\{4,5,6\}$ & \texttt{0212002101021211102202210} & $16$\\
$H_{\YO}^+$ & $\{4,10,18\}$ & \texttt{0212102012021210101202010} & $8$\\
$H_{\YO}^+$ & $\{4,11,24\}$ & \texttt{0212010121202011202110002} & $8$\\*
$H_{\YO}^+$ & $\{6,12,17\}$ & \texttt{0212101021002221102101012} & $8$\\*
$H_{\YO}^+$ & $\{6,13,22\}$ & \texttt{0210121021202001122100012} & $8$\\*
$H_{\YO}^+$ & $\{7,8,9\}$ & \texttt{0210120021201011122220012} & $16$\\*
$H_{\YO}^+$ & $\{7,11,15\}$ & \texttt{0212010121002211202110012} & $8$\\*
$H_{\YO}^+$ & $\{7,13,21\}$ & \texttt{0210121021201001122200012} & $8$\\*
$H_{\YO}^+$ & $\{9,10,23\}$ & \texttt{0212102011021210102202010} & $8$\\*
$H_{\YO}^+$ & $\{9,12,16\}$ & \texttt{0212101021022220102101010} & $8$\\*
\end{longtable}
\endgroup
Deleting one transverse context $T_v$ instead permits an ordered
pair of original state supports $S_1,S_2$ exactly when
$S_1\cap S_2\subseteq\{v\}$. In these non-three-colorable inputs
disjoint pairs do not exist, so the intersection must equal $\{v\}$;
the isolated fresh vertex has two choices. For $H_9$ this gives four
states for each $v=1,\ldots,12$ and twelve states for each
$v=13,14,15$. Together with Table~\ref{S-tab:deletion-colorings}, this
proves criticality of all 33 contexts. For completed Yu--Oh it gives
zero states at $v=1,\ldots,13$ and eight states at
$v=14,\ldots,25$. Thus horizontal criticality does not imply
criticality of this full transverse construction. The corresponding
$G_{32}$ extension has 45 vertices and 35 contexts and is stateless.
Its distinct six state supports have pairwise intersection size one;
for $H_9$ the distinct eight supports have intersection size one in
21 pairs and size three in seven pairs.

\subsection{Data and verification}\label{S-app:data-verification}
The companion electronic package contains coordinate and context inventories,
complete state and coloring enumerations, verification programs, exact parent
inventories, and finite certificates. Vertex labels and state orderings agree
with the tables in this supplement. The verification programs check the common
Yu--Oh numbering, the published 49-ray parent map, the parent vectors and
contexts, the full $H_{65}$ state inventory, and cross-reference targets.
The finite lower-bound certificate for the eighteen-state Hesse minimum is
the exhaustive $540{,}672$-candidate argument of Sec.~\ref{S-h:app:minimal-families}.

Run the companion verification entry point from the project directory:
\begin{verbatim}
python verification/verify-integrated.py
\end{verbatim}
The exact checks use rational or $\Q(\zeta)$ arithmetic, with
$\zeta^2+\zeta+1=0$. Reports and hashes accompany the source.

The finite checks certify the instances and certificates they explicitly test.
The analytic realization and existence results also rely on their proofs in
the text.

\subsection{Historical Kochen--Specker gadget certificates (1965 Fig.~1 and 1967 \texorpdfstring{$\Gamma_1$}{Gamma1})}\label{S-app:historical-cert}
\subsubsection{The 1965 triangles and the completed historical TITS}\label{S-app:historical}
Identifying the repeated occurrences of vertices 1 and 2 in
Kochen--Specker's 1965 Fig.~1 gives
\begin{equation}\label{S-eq:ks1965}
 \{134,325,2\,9\,11,11\,1\,10,479,5\,8\,10,678\}.
\end{equation}
These are also the triples used in the displayed function definitions
on p.~182 of Ref.~\cite{S-kochen2}. Their four state supports are
\[
 \{1,5,6,9\},\quad\{2,4,6,10\},\quad
 \{3,6,9,10\},\quad\{4,5,6,11\}.
\]
The TIFS between 1 and 2 is nonvacuous. Directly, if both were true,
contexts $134$, $325$, $2\,9\,11$, and $11\,1\,10$ would make
3,4,5,9,10,11 false; contexts $479$ and $5\,8\,10$ would then make
both 7 and 8 true, contradicting $678$.

Every state has 6 true and 7,8 false, so no three-coloring can exist:
the three color-class states of any coloring would need to be disjoint.
The graph is abstract rather than an injective ray example. In a
purported injective realization in $\C^3$, vertices 1 and 9 are both
orthogonal to 4 and 11. Distinct rays 4 and 11 span a plane, whose
orthogonal complement is one-dimensional; thus rays 1 and 9 coincide.
This proves the obstruction even over the complex field. No claim that
the 1965 authors presented this graph as a Hilbert-space realization is
intended.

For the 1967 $\Gamma_1$\cite[pp.~68,70]{S-kochen1}, the correspondence
between its original vertices and $I_{16,9}$ is
\[
\begin{array}{c|rrrrrrrrrr}
 \Gamma_1&a_0&a_1&a_2&a_3&a_4&a_5&a_6&a_7&a_8&a_9\\\hline
 I_{16,9}&10&9&1&7&3&8&2&11&4&24
\end{array}
\]
The already complete triangles are
$\{a_1,a_3,a_5\}$, $\{a_2,a_4,a_6\}$, and
$\{a_7,a_8,a_9\}$. The six remaining edges
\[
 a_0a_1,\ a_0a_2,\ a_0a_8,\ a_3a_7,\ a_4a_7,\ a_5a_6
\]
are completed by fresh vertices $23,19,18,15,14,5$, respectively.
The resulting nine contexts are exactly Eq.~\eqref{S-eq:i16}.
This also identifies its directed TITS $a_0\to a_9$ with $10\to24$.

\makeatletter
\renewcommand{\theHNAT@ctr}{S.\arabic{NAT@ctr}}
\def\NAT@bibsetnum#1{\setlength{\topsep}{\z@}\NATx@bibsetnum{\ref{S-LastBibItem}}}
\def\endthebibliography{%
 \auto@bib@innerbib
 \edef\@currentlabel{\arabic{NAT@ctr}}%
 \label{S-LastBibItem}%
 \endNAT@thebibliography
 \aftergroup\auto@bib@empty
}
\makeatother

\endgroup
\end{document}